\documentclass[acmsmall,screen, nonacm]{acmart}
 
\usepackage{soul,color}
\usepackage{xcolor,colortbl}
\usepackage{makecell}
\usepackage{cuted}
\usepackage{flushend}
\usepackage{subfigure}
\usepackage{enumerate}
\usepackage[linesnumbered,ruled,vlined]{algorithm2e}

\SetCommentSty{mycommfont}
\SetKwInput{KwInput}{Input}               
\SetKwInput{KwOutput}{Output}  
\usepackage{booktabs}

\usepackage{newtxmath}
\usepackage{bm}
\usepackage{graphicx}
\usepackage{hyperref}
\usepackage{url}
\usepackage{tikz}
\usetikzlibrary{positioning,chains,fit,shapes,calc}
\usetikzlibrary{snakes}
\usepackage{subfigure} 
\usepackage{lipsum}
\usepackage{bm}
\usepackage[utf8]{inputenc}
\usepackage[linesnumbered,ruled,vlined]{algorithm2e}

\usepackage{array}
\usepackage{booktabs}

\usepackage[edges]{forest}

\usepackage{enumitem}
\usetikzlibrary{shapes.geometric, arrows, calc, fit, positioning, chains, arrows.meta}

\usepackage{hyperref}

\AtBeginDocument{%
  \providecommand\BibTeX{{%
    \normalfont B\kern-0.5em{\scshape i\kern-0.25em b}\kern-0.8em\TeX}}}

\setcopyright{acmcopyright}
\copyrightyear{2021}
\acmYear{2021}

\acmJournal{TKDD}
\acmVolume{15}
\acmNumber{4}
\acmArticle{}
\acmMonth{3}

\begin{document}
\title{Scale-Invariant Strength Assortativity of Streaming Butterflies
}

\author{Aida Sheshbolouki}
\affiliation{%
  \institution{University of Waterloo}
  \country{Canada}
  }
\email{aida.sheshbolouki@uwaterloo.ca}
\orcid{0000-0001-5725-1781}

\author{M. Tamer {\"O}zsu}
\affiliation{%
  \institution{University of Waterloo}
  \country{Canada}
  }
\email{tamer.ozsu@uwaterloo.ca}
\orcid{0000-0002-8126-1717}

\renewcommand{\shortauthors}{Aida Sheshbolouki and M. Tamer {\"O}zsu}

\begin{abstract}
Bipartite graphs are rich data structures with prevalent applications and identifier structural features. However, less is known about their growth patterns, particularly in streaming settings. Current works study the patterns of static or aggregated temporal graphs optimized for certain down-stream analytics or ignoring multipartite /non-stationary data distributions, emergence patterns of subgraphs, and streaming paradigms. To address these, we perform statistical network analysis over web log streams and identify the governing patterns underlying the bursty emergence of mesoscopic building blocks, 2,2-bicliques known as butterflies, leading to a phenomenon that we call ``scale-invariant strength assortativity of streaming butterflie''. We provide the graph-theoretic explanation of this phenomenon. We further introduce a set of micro-mechanics in the body of a streaming growth algorithm, \emph{sGrow}, to pinpoint the generative origins. \emph{sGrow} supports streaming paradigms, emergence of 4-vertex graphlets, and provides user-specified configurations for the scale, burstiness, level of strength assortativity, probability of out-of-order records, generation time, and time-sensitive connections. Comprehensive Evaluations on pattern reproducing and stress testing validate the effectiveness, efficiency, and robustness of \emph{sGrow} in realization of the observed patterns independent of initial conditions, scale, temporal characteristics, and model configurations. Theoretical and experimental analysis verify the \emph{robust} ability of \emph{sGrow} in generating streaming graphs based on user-specified configurations that affect the scale and burstiness of the stream, level of strength assortativity, probability of-of-order streaming records, generation time, and time-sensitive connections.
\end{abstract}

\maketitle
\newpage
\tableofcontents
\clearpage
\newpage
\section{Introduction}
We study the bursty emergence of meso-scale building blocks in bipartite streaming graphs representing web logs to uncover the crucial mixing patterns and identify/explain their microscopic generative origins by introducing a streaming growth model. All complex networks have an underlying bipartite structure~\cite{vasques2020transitivity, guillaume2004bipartite}. Even those networks that are naturally unipartite, e.g. social networks, have an inherent bipartite structure driving the topological structure of the unipartite version \cite{newman2003social, vasques2020transitivity, guillaume2004bipartite, vasques2018degree}. These bipartite structures are captured by bipartite graphs, 
which are rich data models that provide full representation without information loss for interactions that naturally occur in one mode (compressed datasets as unipartite graphs~\cite{zhou2007bipartite}), or multiple modes (high order interconnections as hyper graphs~\cite{ito2020fast, aksoy2020hypernetwork, wang2009hypersum}). A prevalent use case is where bipartite graphs capture the interactions of users with entities spanning different domains such as social networks (users-hashtags~\cite{zhang2019language}), web-based services (users-websites, multimedia services, and products~\cite{ting2013personalized, jamali2011modeling, sun2020neighbor, wei2019mmgcn, wang2021bipartite}), financial systems (users-donation campaigns~\cite{akoglu2010oddball}), transportation systems (users-registered vehicles~\cite{ji2018detection}), and  communication systems (users-phone calls~\cite{ying2018fraudetector+}).

Despite the widespread applications, richness and critical role in determining the topological properties, less is known about the growth and generative patterns of bipartite interconnected data particularly in streaming settings~\cite{bonifati2020graph} where the continous rapid temporal evolutions lead to unbounded/unknown stream length and non-stationary distributions of the underlying data snapshots. In this context, the temporal evolutions usually occur wrt the most recent graph topology (i.e. update events are not global); the evolving streaming rates leads to non-uniform inter-event statistics; and multiple generative sources (as well as factors such as transmission delays) cause  out-of-order arrival of data records to a processing unit which has no control over the arrival order or data rate~\cite{golab2010data, ozsu2019stream}. 
Streaming graphs are different from aggregated temporal graphs that are a sequence of graph snapshots (representing a dynamic graph with an entirely available structure that undergoes temporal changes). Moreover, weight addition patterns~\cite{mcglohon2008weighted}, streaming context~\cite{ozsu2019big}, and data-driven semantics~\cite{barabasi2005origin} lead to burstiness in streaming record arrivals. An example is the case of user-product interactions in Alibaba e-commerce services that incurred a processing rate of $470$ million event logs per second during a peak interval~\cite{ozsu2019big}. Current works study and model the generative patterns of static or aggregated temporal graphs  commonly optimized for down stream analytics or ignore (1) multipartite/non-stationary data distributions, (2) \emph{emergence} patterns (not just existence) of building blocks, and (3) streaming paradigms  such as unbounded/time-sensitive updates, evolving streaming rates, and out-of-order/bursty records (e.g.,~\cite{akoglu2008rtm, gorke2012efficient, zeno2021dymond, aksoy2017measuring, leskovec2005graphs, ammar2013wgb, wang2021bipartite, zhou2020data}). In this paper, we perform statistical analysis over web log streams to infer the key features governing the emergence of the mesoscale building blocks of the bipartite streaming graphs. 

\textbf{Mesoscopic network inference.} Frequent subgraphs (motifs \cite{motifs} or graphlets \cite{rossi2020heterogeneous}) as the building blocks of graphs~\cite{motifs} play an important role in understanding the structure of graphs~\cite{ahmed2015algorithms, rossi2020heterogeneous, benson2016higher, liu2019sampling, syed2010motif, wu2011characterizing, paranjape2017motifs, li2021analyzing, ma2019linc, kovanen2011temporal, zeno2021dymond, purohit2018temporal, hu2019discovering, wang2019vertex, sheshbolouki2015feedback, sheshbolouki2021sgrapp}.
Particularly, the presence of butterflies ($2,2$-bicliques) as the simplest and most local form of cycles in the bipartite graphs has been identified as the main driver of transitivity and degree assortativity in the corresponding unipartite graphs (projections)~\cite{vasques2020transitivity}. Moreover, butterflies are of great importance in measuring properties such as cohesion, network stability and error tolerance~\cite{zhu2020hurricane}.  
Recently, various butterfly-based data models and analytic algorithms have been proposed for measuring rectangle-based connectivity, estimation of maximal bicliques and bitruss decomposition in heterogeneous information networks~\cite{fang2021cohesive, hao2020cohesive,aksoy2017measuring, li2021approximately, zheng2020butterfly, dong2021butterfly}. In a previous study~\cite{sheshbolouki2021sgrapp}, 
we have shown that butterflies are temporal motifs with bursty emergence patterns giving rise to significantly higher occurrence numbers over the timeline of edge arrivals compared to random (null) graphs. The quantitative emergence pattern is formulated as the butterfly densification power law (BPL) which states that the number of butterflies at time $t$ follows a power law function of the number of edges at time $t$. In real world graphs displaying BPL, there is a strong positive correlation between the butterfly support and degree of vertices and the probability of butterflies incorporating hubs (i.e. vertices with degree higher than average unique degree of seen vertices) is high. Also, there is a correlation between the frequency and average degree of hubs with the frequency of butterflies.  Hence, the hubs are identified as the main contributors to bursty formation of butterflies. However, BPL does not happen in random bipartite graphs constructed by the preferential attachment model. 
In preferential attachment model~\cite{albert2002statistical}
, the degree of hubs increases over time with the arrival of new vertices; however, the number of butterflies does not grow as rapidly as that of real-world graphs.  Therefore, there exist other factors besides the vertex degree that are impactful in butterfly densification. In this study, our goal is to discover these factors. We aim to answer two questions: \textit{how do butterflies as the building blocks of bipartite streaming graphs emerge over time?} and \textit{what is the generative process underlying their emergence?} The answers to these questions allow us to model the realistic growth patterns in bipartite streaming graphs. 
To further focus our research, we investigate the organizing principles in web-based user-item streams. The sequences of user-item interactions in web-services are typically associated with a weight that can be an explicit value such as rating, or an implicit value denoting the multiplicity of interactions between a pair of vertices. Moreover, the time-labeled interactions are continuously generated with a non-stable rate giving rise to emergence of an unbounded dynamic structure. The edge weights and fine grained temporal information enable exploring the temporal and connectivity patterns. We use publicly available data in which the timestamp and weight are explicitly given in the data records (common in rating graphs). We do not consider implicit weights computed by aggregating multiple edges between two vertices since such aggregations require aggregating the timestamps as well, which in turn manipulates the temporal properties and makes the temporal analysis unreliable. We study the bipartite structures directly as opposed to investigating the unipartite projections. Although projection enables using standard tools for analyzing unipartite graphs, three main issues exist: (1) information loss as the projection is not bijective and timestamps and weights are not captured in projection, (2) edge inflation as high degree vertices are transformed to dense cliques, and (3) unreliable patterns such as community structure and degree mixing patterns rooted from artificial edge densification and subgraph formation~\cite{grujic2009mixing, newman2001random, guillaume2004bipartite, guimera2007module, barber2007modularity, latapy2008basic}. 

\textbf{Contributions.} We explore the preference of vertices forming butterflies to connect to each other wrt strength similarity (i.e. strength assortativity) by integrating the vertex degrees and edge weights as vertex strength and considering burst of edges. To this end, we show the limitation of conventional approaches to study such mixing patterns and introduce a new quantification approach, based on tracking the localization of a vector embedding the data distributions in sequential burst-based graph snapshots, to enable effective temporal analysis of strength assortativity in graphs having abundant cliques, different scales, and multiple/skewed strength distributions; Utilizing this approach, we unveil the “scale-invariant strength assortativity of streaming butterflies” which represents a co-occurrence of three patterns: 
\begin{enumerate}
    \item \textit{Butterfly densification}: The butterfly count grows super-linearly wrt the edge count.
    \item \textit{Strength diversification}: Butterflies display a wide range of vertex strengths with a right skewed distribution and the skew increases over time.
    \item \textit{Steady strength assortativity}: Butterflies display a stable strength assortativity since the right skewed distribution of strength difference of connected butterfly vertices is fixed-shape, although the skew increases over time, and strength differences are localized below the mean.
\end{enumerate}
The co-occurrence of these patterns is counter-intuitive and interesting: As the stream and the number of butterflies grow rapidly, we observe that diversity of strengths for butterfly vertices increases and strong (i.e. high strength) vertices get stronger and obtain weak neighbors with the increasing of variance of strength differences. Therefore, we expect an increasing trend of disassortativity. However, the majority of butterfly edges are formed by vertices with similar strength and this assortativity remains at a fixed level regardless of stream size or butterfly count.
Therefore, the co-occurrence of these patterns implies non-trivial mixing patterns. Moreover, our analyses of existing growth mechanisms that yield skewed distributions, degree correlation, and cohesive structures highlight the essence of new growth models. We explain the confounding data-driven semantics in the domain of user-item interactions as these patterns relate to three graph theory concepts: burstiness, rich-get-richer, and core-periphery. 
Furthermore, we introduce \emph{sGrow}, a streaming growth model that explains the observed patterns and preserves related concepts. 
\emph{sGrow} is based on addition of edge bursts which satisfies streaming data paradigms, preserves realistic patterns quantitatively and qualitatively, and also makes the stream generation scalable. Moreover, \emph{sGrow} enables generating a sequence of bipartite edges attributed with timestamps and weights, isolated/out-of-order edges, and four-vertex graphlets. 
Our evaluations validate that \emph{sGrow} efficiently and effectively reproduces the bursty emergence patterns of streaming butterflies, independent of initial conditions, scale, temporal characteristics, and model configurations. Our experiments also verify the robustness of \emph{sGrow} in generating realistic streaming graphs configured with user-specified properties that affect generation time, scale and burstiness of the stream, level of strength assortativity, probability of out-of-order streaming records, and time-sensitive connections. Our contributions in this paper are the following:
\begin{itemize}
    \item A new measure for assortativity analysis. We introduce an effective quantification approach for meso\-scale mixing patterns in weighted bipartite streaming graphs which is also applicable to static and/or unipartite graphs.
    \item Butterfly emergence patterns. We uncover the \textit{scale-invariant strength assortativity of streaming butterflies} rooted from three real-world streaming graph patterns (butterfly densification, strength diversification, and steady strength assortativity) and three graph theory concepts (burstiness, rich-get-richer, and core-periphery).
    \item A streaming growth model. We introduce \emph{sGrow}, a streaming growth algorithm that explains 
    the observed patterns and preserves the confounding concepts while supporting streaming paradigms, emergence of 4-vertex graphlets, and user-specified configurations. Accordingly, we introduce: 
    \begin{itemize}
        \item A reference guide supported by extensive stress-testing experiments for configuring the parameters in benchmarking applications.
        \item A set of microscopic mechanisms to benefit development of streaming algorithms and optimized-models.
    \end{itemize}
\end{itemize}

\textbf{Impact.} While our quantification approach can benefit network inference over temporal graphs, our analysis of bipartite graph streams and our growth model are impactful in the following cases:
\begin{itemize}
    \item \textit{Streaming graph benchmarks}. Performance evaluation of algorithms including but not limited to butterfly-based algorithms relies on considering the characteristics of input graphs~\cite{alucc2014diversified, ammar2013wgb, bonifati2018survey}. Given the lack of streaming graph generators, our work help understand the important graph characteristics as well as providing realistic \emph{configurable} and scalable graphs for stress testing purposes.
    \item \textit{Machine learning benchmarks}. Collecting/annotating the training/testing datasets for graph-based models in domains such as outlier detection and computer vision is challenging due to the nature of data (e.g., rare outliers and diverse image instances), and expenses of manually labeling the  instances~\cite{zhao2020using, wang2021bipartite, park2021deep}.
    Solutions include building benchmark datasets via artificial instance injection
    ~\cite{campos2016evaluation, zhao2020using, emmott2013systematic} 
    and applying weakly-supervised techniques 
    ~\cite{liu2020weakly,wang2021bipartite, gu2014superpixel}. 
    Our analytical approach can be extended to such domains to inform the design of graph-based models with the temporal connectivity patterns 
    and also generating realistic yet synthetic datasets to which the artificial instances are injected. 
    \item \textit{Concept drift modeling}. To improve the performance of online adaptive learning algorithms in stream-based recommender systems for web activities, 
    it is important to consider the temporal evolution of modeled concepts due to a change in the distribution of log data or a change in the relation between data and target variable  (i.e. concept drift) ~\cite{al2021survey, rappaz2021recommendation, gama2014survey}. Considering a butterfly as two users with mutual preferences and two items with mutual perceptions, our work impacts modeling the parallel drift of concepts such as user preferences and item perception.
    \item \textit{Algorithm developments}. Graph analytics and generative models utilize microscopic mechanisms and graph pattern for algorithm design~\cite{wang2021bipartite, sheshbolouki2021sgrapp}. 
    Our 
    microscopic mechanisms and growth patterns benefit these cases as well. 
\end{itemize}

The rest of paper is organized as following: Section~\ref{sec:preliminaries} reviews preliminaries and datasets; Section~\ref{sec:metric} introduces a new perspective for strength assortativity quantification; Section~\ref{sec:observations} is devoted to our observations in real-world and synthetic streams as well as a brief survey of related works on graph patterns and growth rules; Section~\ref{sec:model} introduces sGrow; Section~\ref{sec:evaluations} reports the evaluations; Section~\ref{sec:conclusion} concludes the paper.
\section{Preliminaries and datasets}\label{sec:preliminaries}
We denote a bipartite graph as $G=(V , E)$, where $V=V_i\cup V_j$, $V_i \cap V_j=\emptyset$ and $E \subseteq V_i \times V_j$. We refer to $V_i=\{ v_i\}$ and $V_j=\{v_j\}$ as i-vertices and j-vertices. We denote the set of immediate/nearest neighbors of an i-vertex $v_i$, called j-neighbors, as $N_j(v_i)$; similarly for i-neighbors $N_i(v_j)$.

\begin{definition}[\hypertarget{sgr}{Streaming Graph Record}]\label{def:sgr}
A streaming graph record (\textit{sgr}) is a quadruple $r^m=\langle v_i^m,v_j^m,\omega_{ij}^m,\tau^m \rangle$ where $m$ is the sgr index, $\omega_{ij}^m$ is the weight of the edge between $v_i$ and $v_j$, and $\tau^m$ is the timestamp of the sgr assigned by the generative source.
\end{definition}

\begin{definition}[\hypertarget{wbsg}{Weighted Bipartite Streaming Graph}]\label{def:wbsgraph} A weighted bipartite streaming graph  is an unbounded 
sequence of sgrs denoted as $\Re= \langle r^1, r^2, \cdots \rangle$. 
\end{definition}

The sequence of sgrs can be ordered by either timestamps or arrival times. We consider the latter to support out-of-order sgrs (late arrivals).

\begin{definition}[\hypertarget{Burst}{Burst}]\label{def:burst}
A burst is the batch of subsequent sgrs with same timestamp. The number of bursts in the stream is denoted as $N_b$ and equals to the number of unique timestamps.
\end{definition}

We highlight that a burst is the batch of \emph{subsequent}, not \emph{all}, sgrs with same timestamp. This definition enables burst-based analysis of out-of-order sgrs. 

\begin{definition}[\hypertarget{bbgs}{Burst-based Graph Snapshot}]\label{def:burstbasedgraphsnapshot}
A burst-based graph snapshot, $G_{N_b}$, is the bipartite graph formed by the prefix of sgrs seen since the first timestamp $\tau^1$ until $N_b$-th unique timestamp. i.e. $G_{N_b}=(V,E)$, s.t. $V=V_i\cup V_j$, $V_i \cap V_j=\emptyset$, $E \subseteq V_i \times V_j$, and $E=\{r^m|\tau^m\in [\tau^0, \tau^{N_b}] \}$.
\end{definition}



Weighted bipartite streaming graphs display bursty patterns since (1) weight additions in temporal graphs follow bursty patterns~\cite{mcglohon2008weighted}, (2) data streams are commonly characterized as bursty~\cite{zhong2021burstsketch, ozsu2019stream}, and (3) data-driven semantics imply burstiness (for instance, human-initiated events are driven by the queuing processes of human decision making leading to non-Poisson inter-event statistics ~\cite{barabasi2005origin}). That is why real-world user-item interactions are characterized as bursty and the distribution of timestamps change over time. In other words, groups of interactions occur in short periods of intense activity (bursts) separated by relatively long gaps of inactivity. We study user-item data sets and due to the bursty characterization of these data streams, we conduct our analysis over graph snapshots created based on the number of bursts in the stream.
This enables comparing the graph snapshots at different scales (number of vertices/edges) and also evaluating the temporal properties in parallel to structural analyses.

\begin{definition}[Vertex Strength]\label{def:strength}
Vertex strength (shortly strength) is defined as the total weight of edges connected to the vertex, $S_i=\Sigma_{j\in N_j(i)}\omega_{ij}$, $S_j=\Sigma_{i\in N_i(j)}\omega_{ij}$~\cite{yook2001weighted, barrat2004architecture, barrat2004weighted}.
\end{definition}
Vertex strength is a natural generalization of the connectivity of graphs and is considered as a significant measure of graph properties in terms of weights~\cite{barrat2004weighted}. We use this notion for parallel study of edge weights and vertex degrees and their impact on the emergence patterns of butterflies. 

\textbf{Datasets}-- We use seven real-world graph datasets: Ciao~\cite{guo2014etaf}, Epinions~\cite{tangetal12b}, WikiLens~\cite{frankowski2007recommenders}, MovieLens100k  (ML100k)~\cite{herlocker2017algorithmic}, MovieLens1m (ML1m)~\cite{harper2015movielens}, Amazon~\cite{jindal2008opinion, lim2010detecting, mukherjee2012spotting}, and Yahoo songs (Yahoo)~\cite{dror2012yahoo}. All datasets are available at public repositories KONECT~\cite{kunegis2013konect} and Netzschleuder~\cite{netzscleuder}. These datasets include naturally occurring bipartite interactions as a set of records including the user ID, item ID, rating, and timestamp. The rating values are in the set $\{1,2,3,4,5\}$ in all datasets except for WikiLens with ratings in $\{0,0.5,1,..,4.5,5\}$. In WikiLens, we rounded the ratings and replaced ratings equal to $0$ with $1$ to convert the rating scale to $1-5$. 
Table~\ref{tab:graphs} provides the statistics of all the graph streams. These streams cover different structural properties (edge density, average vertex degree, and wedge (i.e. two-path) count) and temporal characteristics (number and average size of \hyperlink{Burst}{bursts}). Ciao and Amazon have low average degree of both i- and j-vertices, while they are bursty streams. Epinions has higher average degree of i-vertices compared to that of j-vertices with a very high number of wedges (the building blocks of butterflies),  and it is a bursty stream with large bursts. WikiLens has high average degree of i-vertices but it is not bursty.  ML100k has high average degree of i- and j-vertices and high number of wedges and it is roughly as bursty as ML1m and Yahoo which have higher average degree of i- and j-vertices and higher number of wedges.
\begin{table*}[ht]\caption{Real-world user-item graph datasets. $d_i$ and $d_j$ denote the average degree of i-vertices and j-vertices, respectively. $N_b$ and $b=|E|/N_b$ denote the number and the average size of bursts, respectively. $\bigwedge$ denotes the number of wedges (two-paths).} 
\small \centering
    \begin{tabular}{p{1.1cm} p{1.1cm} p{1.1cm} p{1.4cm} p{0.7cm} p{0.7cm} p{1.5cm} p{0.8cm} p{2cm}} \hline
    & $|V_i|$ & $|V_j|$ & $|E|$ & $d_i$ & $d_j$& $N_b$ & $b$ &  $\bigwedge$ 
    \\  \hline 
    Ciao & $17,615$ & $16,121$ & $72,665$ & $4.1$ & $4.5$ & $4,919$ & $14.8$ & $4,896,641$
    \\
    Epinions & $120,492$ & $755,760$ & $13,668,320$ & $113.4$ & $18$ & $501$ & $27,282$
    & $69,245,866,714$
    \\ 
    WikiLens & $326$ & $5,111$ & $26,937$ & $82.6$ & $5.2$ & $26,239$ & $1$ & $6,316,744$ 
    \\
    ML100k & $943$ & $1,682$ & $100,000$ & $106$& $59.4$ & $49,282$ & $2$ & $18,367,254$ 
    \\ 
    ML1m & $6,040$ & $3,706$ & $1,000,210$ & $165.6$& $269.9$ & $458,455$ & $2.2$ & $602,009,923$ 
    \\ 
    Amazon & $2,146,057$ & $1,230,915$ & $5,838,041$ & $2.7$ & $4.7$ & $3,329$ & $1,753.7$ & $627,186,651$
    \\
    Yahoo & $1,000,990$ & $624,961$ & $256,804,235$ & $256.5$ & $410.9$ & $105,331,405$ & $2.4$ & $4,627,224,528,654$
    \end{tabular}\label{tab:graphs}
    
\end{table*}

Our analyses rely on exact butterfly listing over landmark windows which is computationally expensive in bursty streams. We use the exact algorithm in sGrapp suit~\cite{sheshbolouki2021sgrapp} to list the butterflies over sequential \hyperlink{bbgs}{burst-based graph snapshots}. We study the emergence of a certain number of butterflies in different streams with different structural/temporal properties. That is, we consider the prefix of streams until the arrival of up to $\approx$$6.5\times10^6$ butterflies which covers the entire stream in WikiLens with $26220$ bursts and a prefix of $10000$, $9600$, $460$, $2000$, and $15000$ bursts in ML1m, Ml100k, Epinions, Amazon, and Yahoo, respectively. In Ciao, we checked the entire stream with $4900$ bursts and $\approx$$6.4\times10^5$ butterflies (Table~\ref{tab:snapshots}). We divide the corresponding timeline of burst arrival into  $20$ equally distanced points and at each point we study the butterflies in the burst-based graph snapshot (Definition~\ref{def:burstbasedgraphsnapshot}). In our analyses we care about the value and the trend of data points and the number of graph snapshots (here $20$) just change the smoothness of the plots and does not affect the results since we check streams with different distribution of timestamps and the scale of graph snapshots differs in various streams. The number of edges/butterflies in each burst varies in different graphs depending on the burstiness of the graph stream. 

\begin{table}[ht]\caption{The number of edges $|E^{20}|$, the number of bursts $N_b^{20}$, and the butterfly count $\bowtie^{20}$ of the $20$th graph snapshot in real-world graph streams. } 
\small \centering
    \begin{tabular}{p{1cm} p{1cm} p{1cm} p{1.2cm}}\hline
    & $|E^{20}|$ & $N_b^{20}$ & $\bowtie^{20}$\\  \hline 
    Ciao & $72,574$ & $4,900$ & $636,440$\\
    Epinions & $296,665$ & $460$ & $6,418,862$\\ 
    WikiLens & $26,918$ & $26,220$ & $6,556,913$\\
    ML100k & $18,696$ & $9,600$ & $6,492,834$\\ 
    ML1m & $22,795$ & $10,000$ & $6,678,784$\\
    Amazon & $2,194,798$ & $2,000$ & $6,496,236$ \\
    Yahoo & $42,105$ & $15,000$ & $6,496,563$
    \end{tabular}\label{tab:snapshots}
\end{table}
\section{Assortativity Analysis in Graph Streams}\label{sec:metric}
The tendency of vertices to connect to similar vertices with respect to one of their quantitative/qualitative attributes is called assortativity/homophily ~\cite{newman2002assortative}. In addition to connectivity insights (our primary goal in this paper), assortativity provides information about the dynamic behavior and robustness of the graph~\cite{trajanovski2013robustness, d2012robustness}. For instance, degree disassortative complex networks compared to degree assortative networks exhibit higher epidemiological threshold leading to easier immunization, however assortative networks get higher resilience to systemic risk by degree-targeted immunization policies~\cite{d2012robustness} (see ~\cite{noldus2015assortativity} for a complete survey). Assortativity is usually studied with respect to vertex degrees. 
A previous study~\cite{leung2007weighted} has shown that studying the assortativity by considering just the degree does not completely uncover the organizational patterns in the structure of graphs. Leung and Chau ~\cite{leung2007weighted} have introduced the weighted assortativity coefficient to measure the tendency of having a high-weighted edge between vertices with similar degrees. However, we are interested in strength assortativity, i.e. measuring the tendency of having an edge between vertices with similar strength, particularly in butterflies. In this section, we first establish the requirements for an effective measurement of strength assortativity that can accommodate analysis of mesoscopic, bipartite, and temporal structures; next, we introduce a new metric for strength mixing patterns called strength assortativity localization factor.

The assortativity coefficient ($r$)~\cite{newman2002assortative} is a common metric for assortativity~\cite{zou2019complex, vasques2020transitivity, van2021random}. Assuming that we are interested in quantifying the tendency of vertices to connect to each other based on the  similarity of their attribute $K$, $r$ is computed as the pearson correlation of $K$ of linked vertices and lies in the range $-1\leq r\leq 1$. Positive (negative) $r$ signals (dis)assortativity and $r=0$ denotes random mixing. Another approach to study assortativity is to compute the average $K$ of nearest-neighbors for each vertex and then aggregating the values by restricting the class of vertices with $K=k$. We denote it as $\langle$$K_n$$\rangle$ which is a function of $K$. An increasing (decreasing) $\langle$$K_n$$\rangle$ signals  (dis)assortativity. This can be inferred by checking the sign of the slope of a linear fit for log-log plot of $\langle$$K_n$$\rangle$ as a function of $K$. In the following, we investigate the effectiveness of $r$ and $\langle$$K_n$$\rangle$ in quantifying the strength assortativity of butterflies. 

We consider the evolution of two distributions over sequential graph snapshots: (1) $Pr(\delta)$, the probability distribution of strength difference for connected butterfly vertices which is computed as $Pr(\delta)=\frac{F(\delta)}{\Sigma F(\delta)}$, where $F(\delta)$ is the number of edges with strength difference $\delta$ and the sum runs over the range of $\delta$ values, and (2) $Pr(S_i)$, the probability distribution of strength for butterfly i-vertices which is computed as $Pr(S_i)=\frac{F(S_i)}{\Sigma F({S_i})}$, where $F(S_i)$ is the number of butterfly i-vertices with strength $S_i$. The same notations stand for j-vertices and $Pr(S_j)$.

As a running example in this section, we use the real-world graph stream Epinions and pick $20$ equally-distanced points in the timeline of burst arrival. At each point ($N_b$), we calculate $r$ for the strengths of linked butterfly vertices in the corresponding graph snapshot $G_{N_b}$ (Figure~\ref{fig:epinionr}.a). We also consider $Pr(\delta)$ at two points corresponding to the arrival of $92$ (Figure~\ref{fig:epinionr}.d) and $437$ bursts (Figure~\ref{fig:epinionr}.g). At $N_b=92$, the probability that a butterfly edge has strength difference below the average strength difference $\mu_{\delta}$ is $Pr(\delta \leq \mu_{\delta})=0.67$. However, the assortativity coefficient is $r=0.007$ suggesting no (dis)assortativity (i.e. random connection of butterfly vertices with no tendency to connect to (dis)similar vertices). Also, at $N_b=437$, we observe that majority of butterfly edges fall in the region behind $\mu_{\delta}$ with probability $Pr(\delta \leq \mu_{\delta})=0.71$, while $r=-0.17$ suggests strength disassortativity. 

The reason behind this confusing behavior of $r$ is its bias toward the distribution of strength of i- and j-vertices with respect to their average. To clarify, we consider $Pr(S_i)$ and $Pr(S_j)$ at these two time points (Figure~\ref{fig:epinionr}.e,f,h,i). At $N_b=92$, the probability that a butterfly i(j)-vertex has strength less than or equal to the average strength of butterfly i(j)-vertices $\mu_i(\mu_{j})$ is almost equal to the probability that a butterfly i(j)-vertex has strength greater than the average strength of butterfly i(j)-vertices ($Pr(S_i\leq\mu_{i})=0.57$ and $Pr(S_j\leq\mu_{j})=0.54$). Therefore, many strength deviations from the mean strength, particularly for j-vertices, would be zero, making the coefficient an insignificant value close to zero ($r=0.007$). At $N_b=437$, a large majority of butterfly i(j)-vertices have strength above the average strength of butterfly i(j)-vertices ($Pr(S_i>\mu_{i})=0.9$, $Pr(S_j>\mu_{j})=0.8$), therefore their high deviations from the mean lowers the coefficient. In summary, the assortativity coefficient reflects the global correlation between $Pr(S_i)$ and $Pr(S_j)$ (two separate distributions).  The assortativity coefficient fails to capture the pairwise correlations between strength of connected i- and j-vertices forming butterflies.
\begin{figure*}[t]
    \subfigure[]{\includegraphics[width=0.32\textwidth]{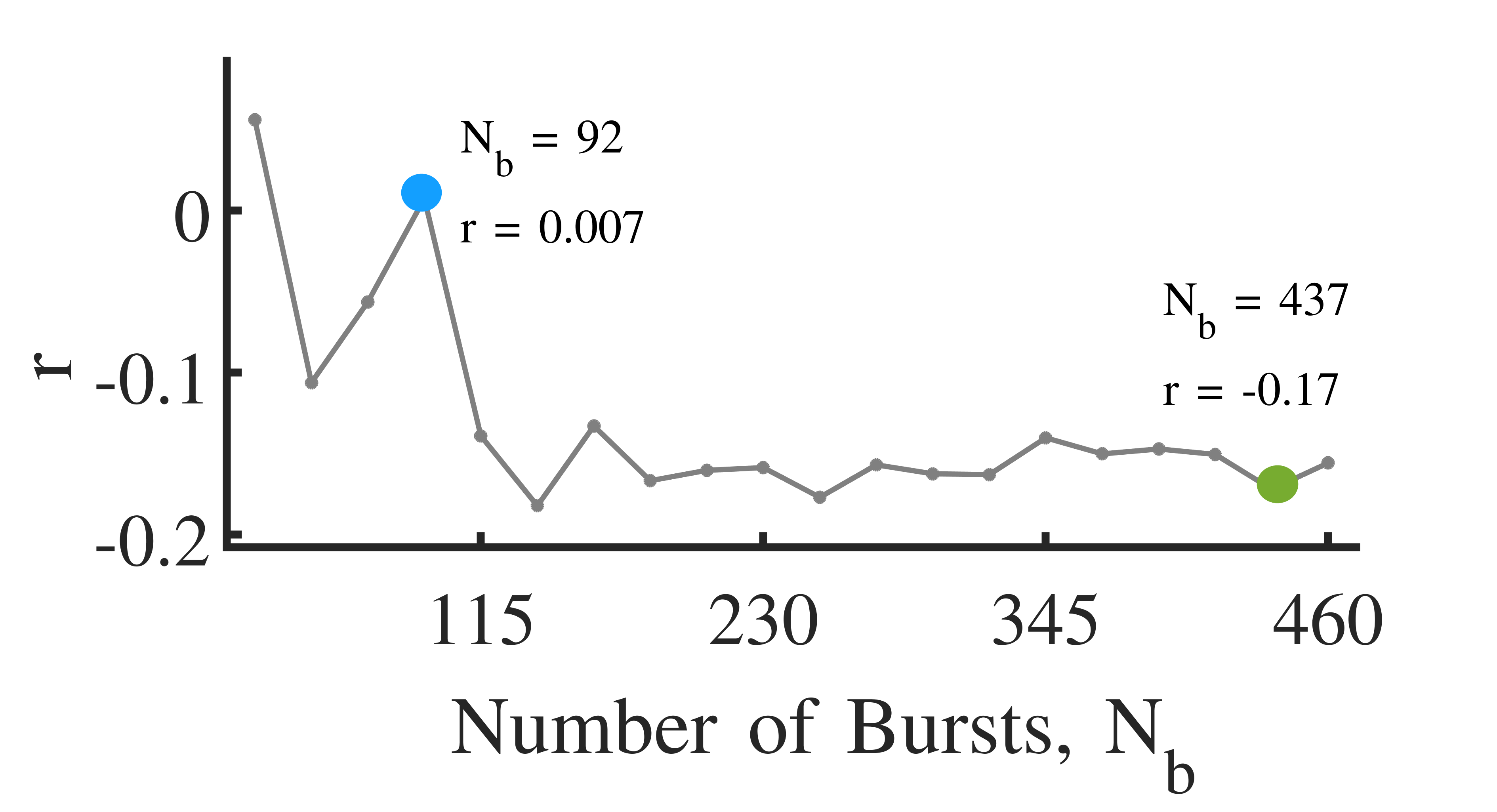}}
    \subfigure[$N_b=92$]{\includegraphics[width=0.32\textwidth]{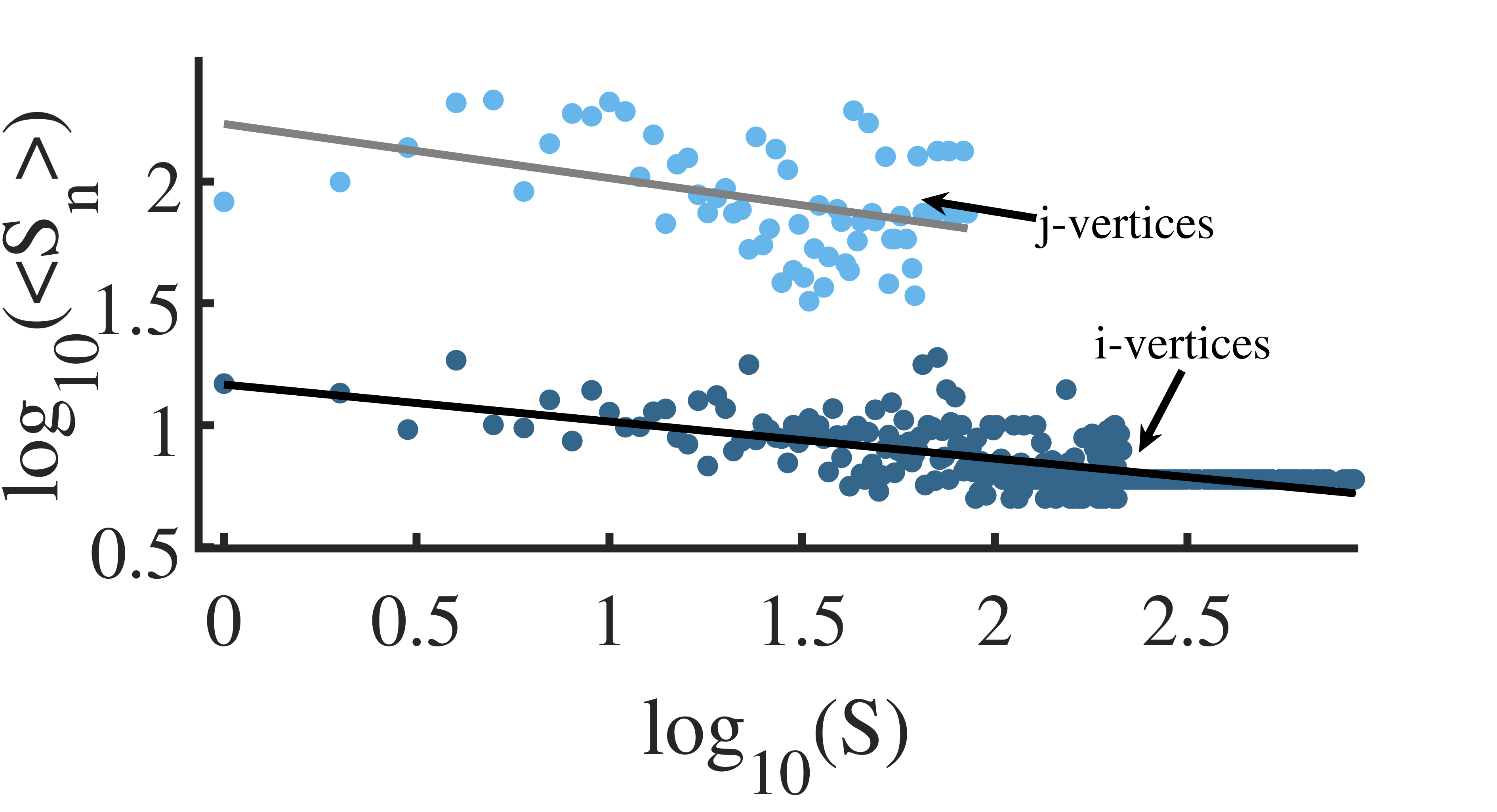}} 
    \subfigure[$N_b=437$]{\includegraphics[width=0.32\textwidth]{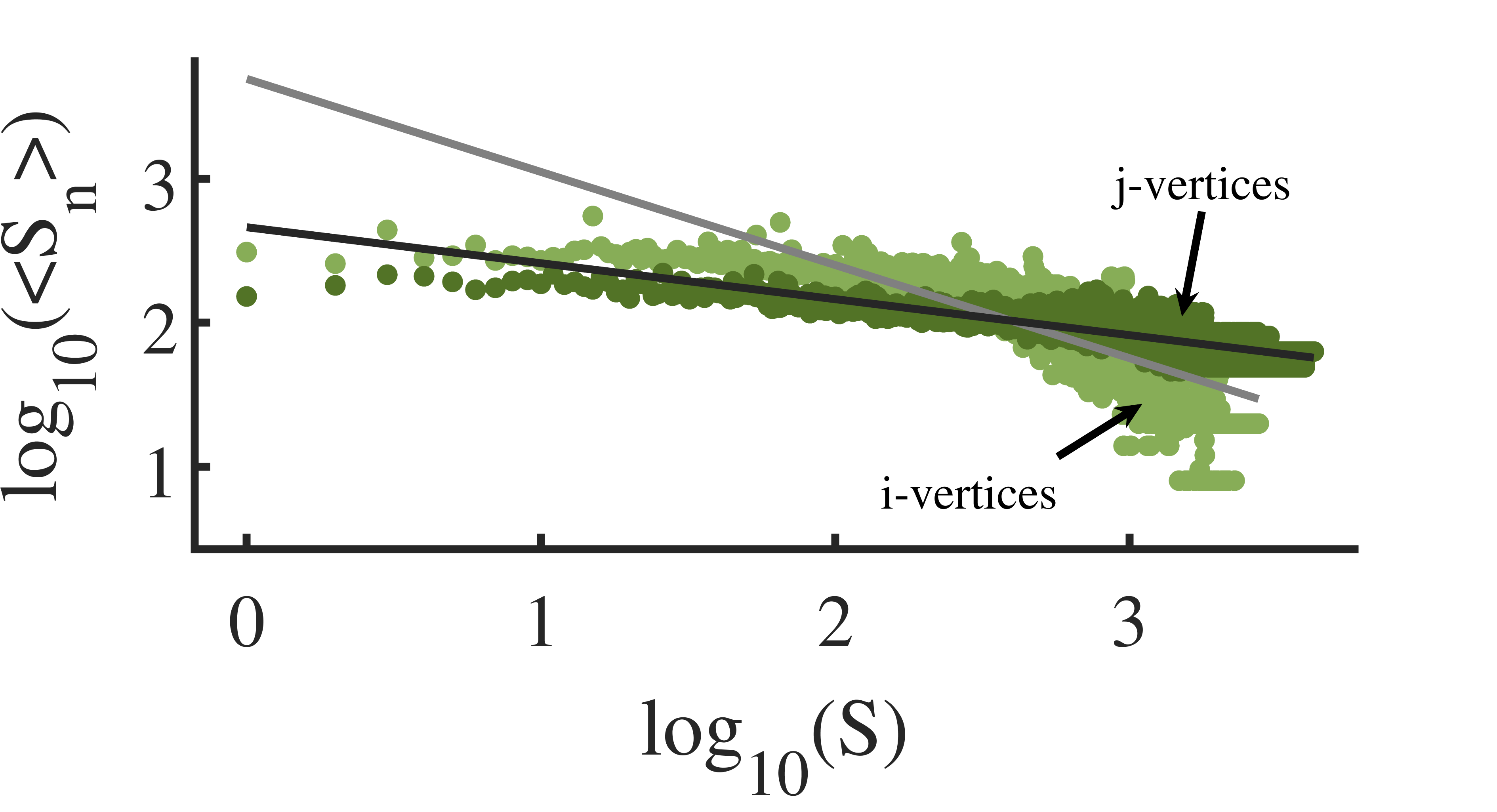}}
     
    \subfigure[$Pr(\delta)$ at $N_b=92$]{\includegraphics[width=0.32\textwidth]{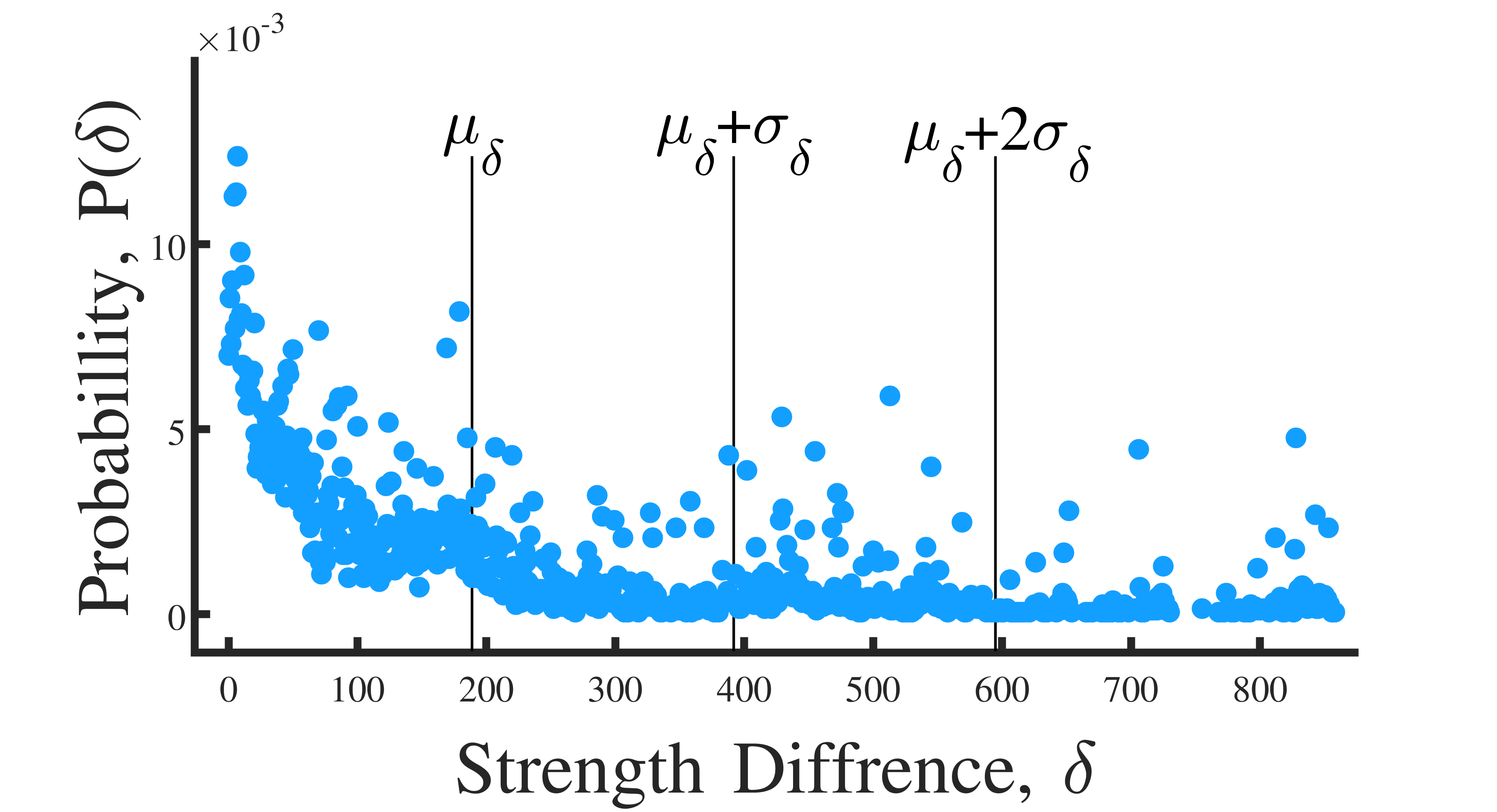}} 
    \subfigure[$Pr(S_i)$ at $N_b=92$]{\includegraphics[width=0.32\textwidth]{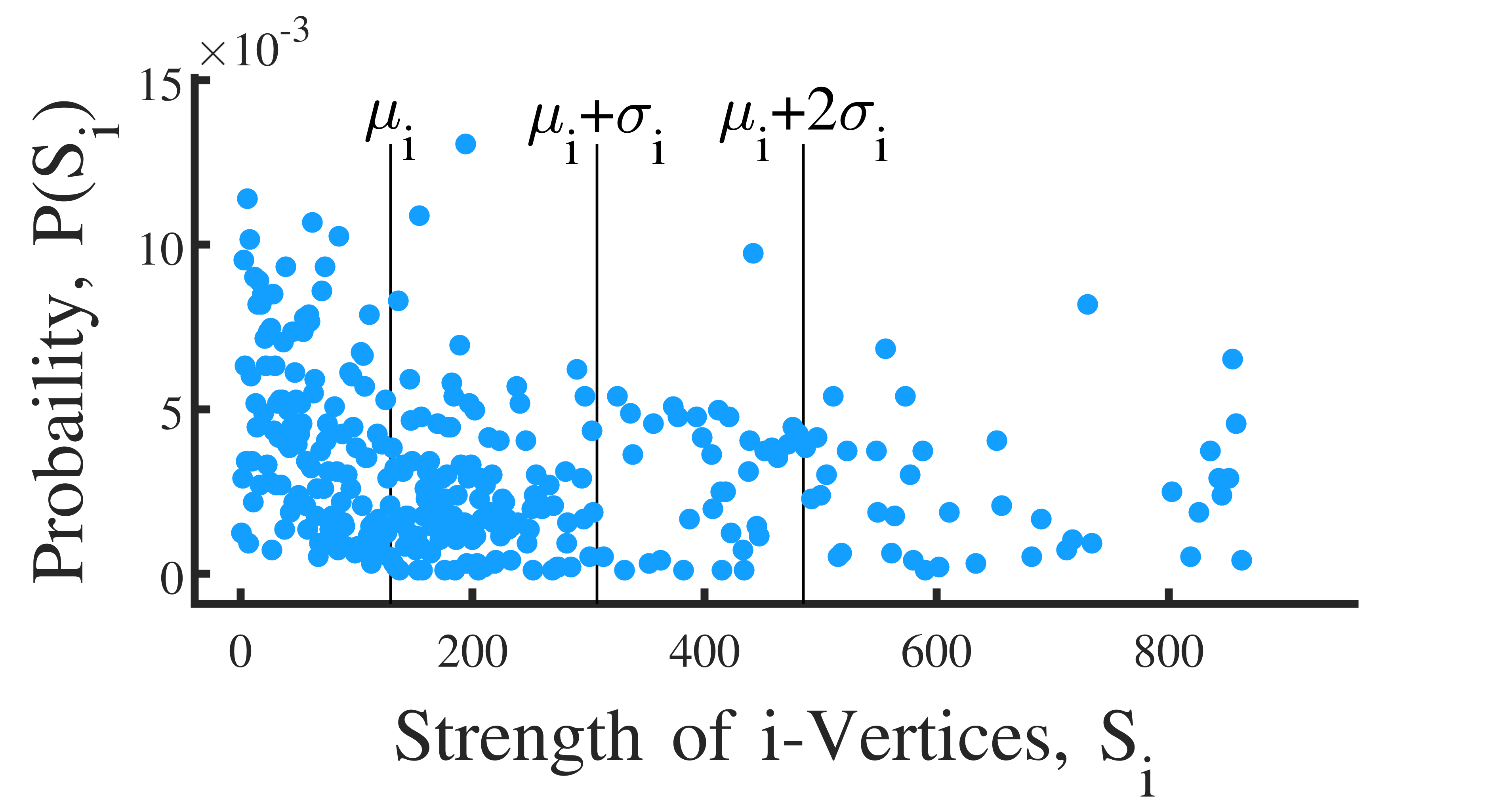}} 
    \subfigure[$Pr(S_j)$ at $N_b=92$]{\includegraphics[width=0.32\textwidth]{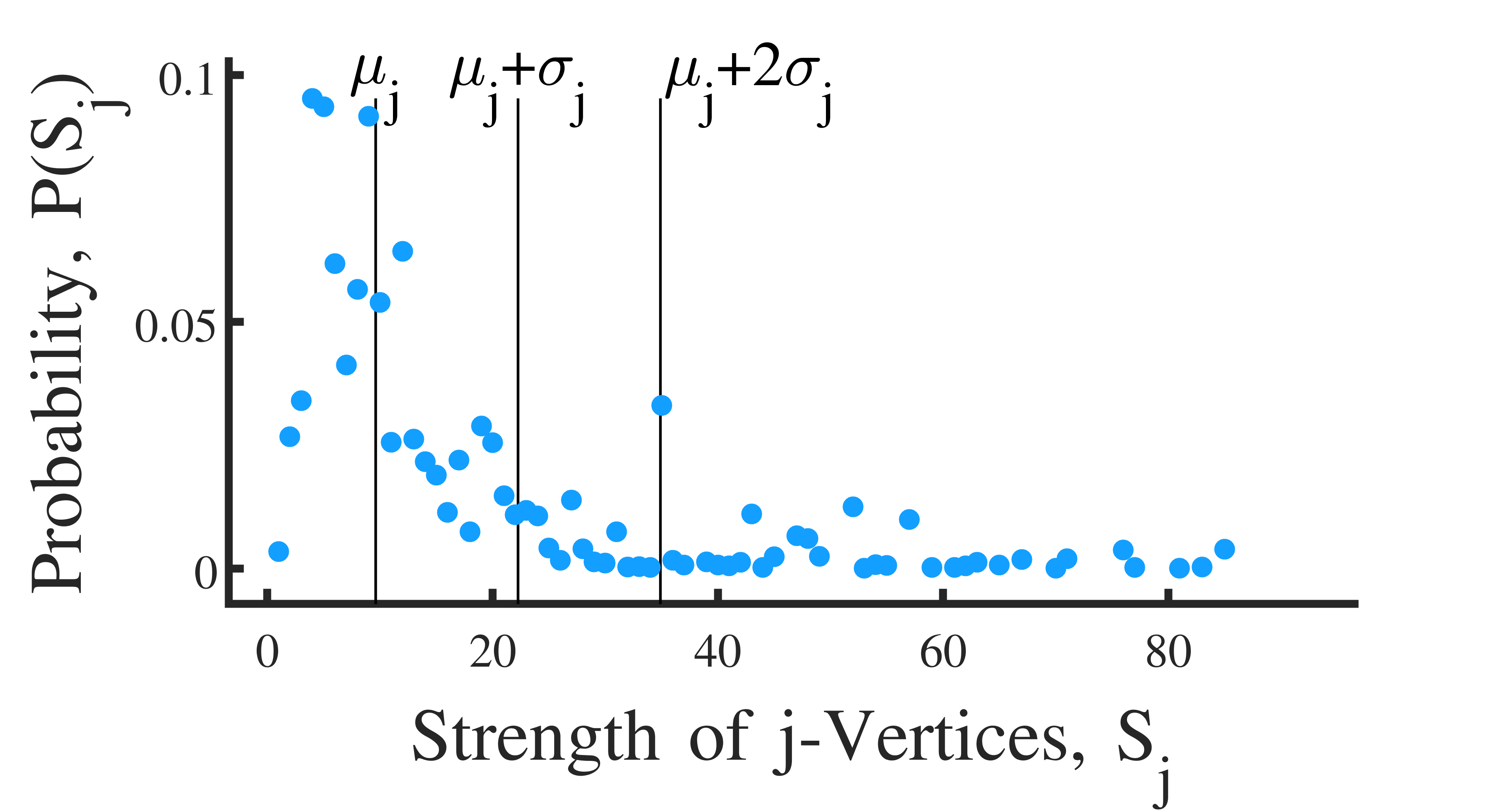}} 
    
    \subfigure[$Pr(\delta)$ at $N_b=437$]{\includegraphics[width=0.32\textwidth]{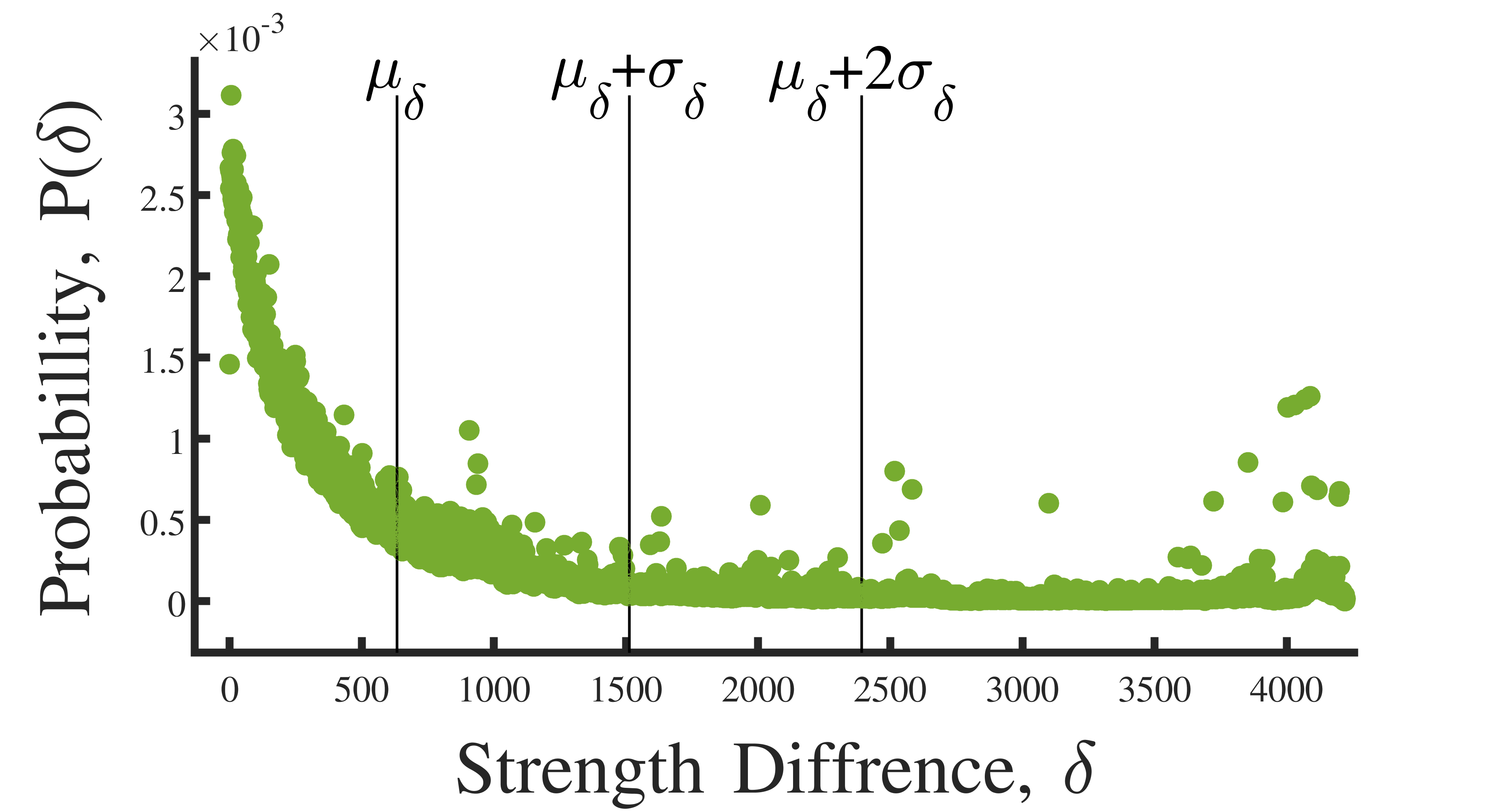}} 
    \subfigure[$Pr(S_i)$ at $N_b=437$]{\includegraphics[width=0.32\textwidth]{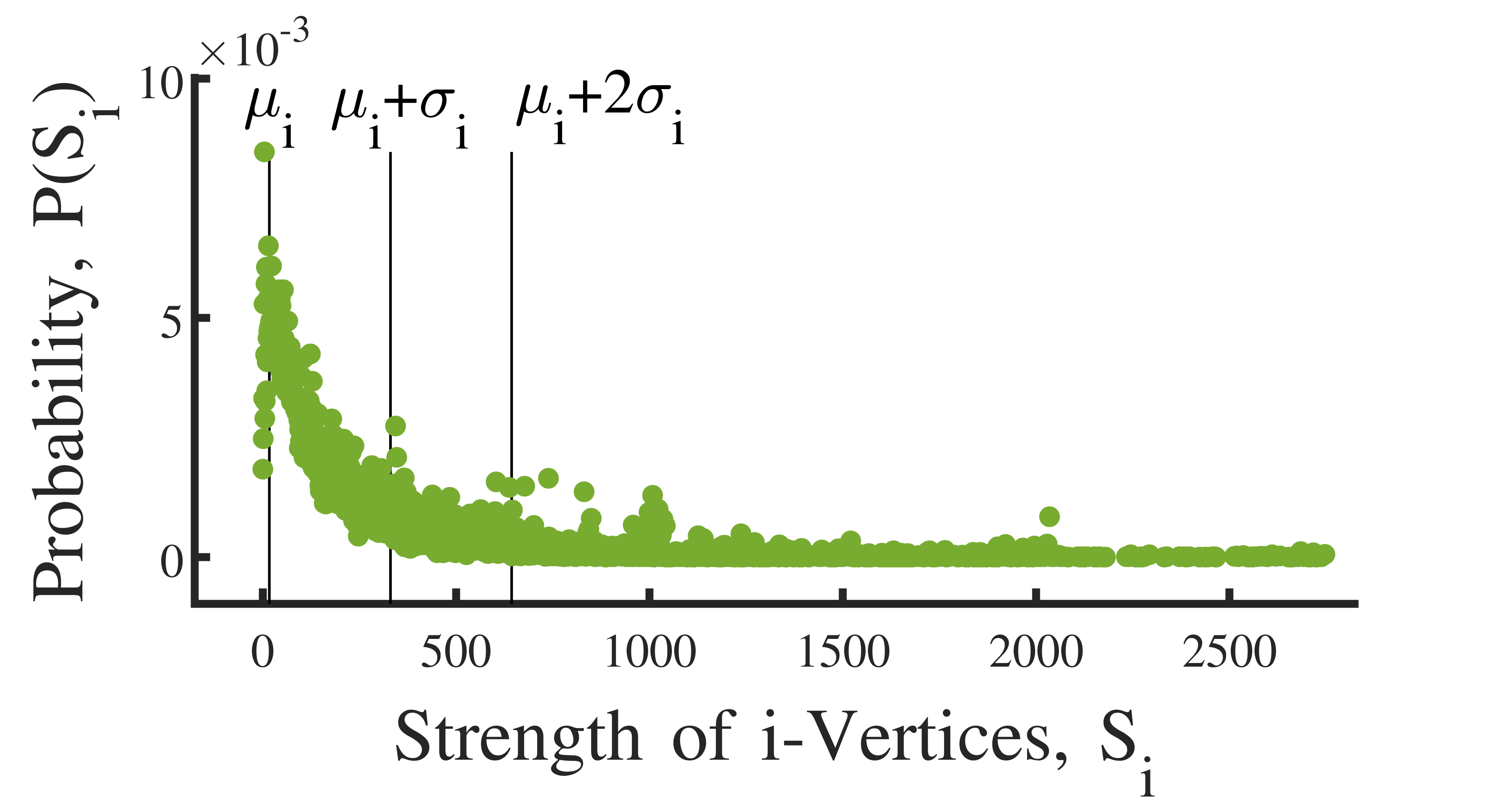}} 
    \subfigure[$Pr(S_j)$ at $N_b=437$]{\includegraphics[width=0.32\textwidth]{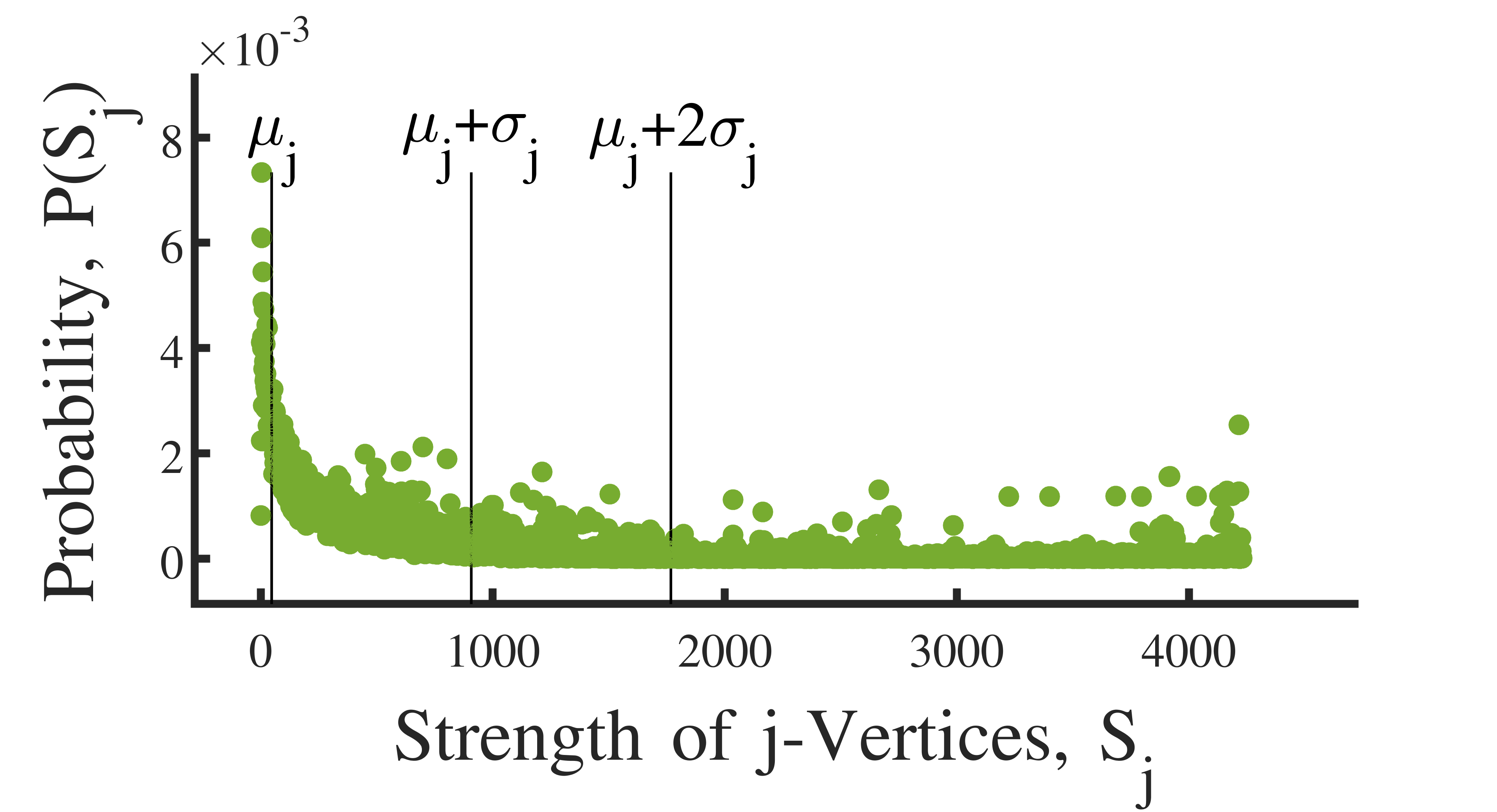}} 
    \caption{(a) Assortativity coefficient over timeline of burst arrival in Epinions stream. Nearest-neighbor average strength of vertices with strength S at (b) $N_b=92$ and (c) $N_b=437$. Distribution of strength differences of connected butterfly vertices  at (d) $N_b=92$ and (g) $N_b=437$. Distribution of strength of butterfly (e,h) i-vertices and (f,i) j-vertices at $N_b=92$ and $N_b=437$, respectively.}
    \label{fig:epinionr}
\end{figure*}

We examine the neighborhood-based approach for studying assortativity. Figures~\ref{fig:epinionr}.b-c, show the nearest-neighbor average strength of vertices with strength $S$~\cite{pastor2001dynamical}. At first glance, the decreasing trend suggests strength disassortativity: the higher the strength of a vertex, the lower the average strength of its neighbors and vice versa. However, we should consider the skewed $Pr(S_i)$ and $Pr(S_j)$ with high-strength vertices. Suppose that low-strength vertices are connected to many low-strength vertices and one high-strength vertex. In this case, the average strength of neighbors for these low-strength vertices would be high although the majority of neighbors have similar low strengths. That is, relatively few vertices having high strengths (because of many connections and/or connections with high weights) skew the average strength of their low-strength neighbors and hence mislead the assortativity interpretation. This again highlights the issue of measuring strength assortativity in graphs with broad and skewed strength distribution. 

To conclude, using conventional assortativity metrics is not reliable for analyzing the strength assortativity of butterflies in bipartite streaming graphs since (1) $r$ is vertex-centric and reflects the global strength correlations rather than pairwise strength correlations. Particularly, in case of butterflies that each vertex contributes duplicate values because of two adjacent edges, vertices with strength equal/close to the mean (zero $S-\mu$), decrease the overall correlation, no matter what the strength of their neighbor is, (2) $r$ is designed for unipartite graphs and using it in bipartite graphs can bias the outcome by the strength distributions of i- and j-vertices, and (3) The neighborhood-based approach can be misleading in case of graphs with broad and skewed strength distributions since high-strength vertices have outlier impacts and make the interpretation difficult.

\subsection{The Proposed Strength Assortativity Quantification Approach}~\label{subsec:measure}

Informed by the above discussion, an appropriate measure for the tendency of vertices to connect to vertices with similar strength that is applicable to butterfly edges should satisfy the following properties: (1) it should directly reflect the probability distribution of strength differences rather than the global correlations in the distribution of strengths; (2) it should not be designed based on neighbor information since in case of skewed distribution of strengths, it would be biased by the outlier vertices; and (3) it should enable comparison of strength assortativity for sequential graph snapshots in the same stream as well as comparison of strength assortativity of graph snapshots in different graph streams.

Our goal is to quantify and compare the distribution of strength differences in low dimension to enable temporal analysis over sequential graph snapshots of streams. A common approach for comparing distributions (usually degree distributions) is Kolmogrov-Smirnov test. However, this is sensitive to the distribution range and is not ideal for analyzing sequential graph snapshots and different graph  streams. The Degree Distribution Quantification and Comparison (DDQC) approach~\cite{aliakbary2014quantification} quantifies the degree distribution of a graph based on $4\times2^\beta$ regions in the degree distribution and uses this quantification for comparison. The regions are determined in two steps: first, the degree distribution is divided into four regions covering the intervals between five subsequent points: min(degree), $\mu -\alpha\sigma$, $\mu$, $\mu +\alpha\sigma$, and $max(degree)$, where $\mu$ is the mean degree and $\sigma$ is the standard deviation of degrees and $\alpha$ is a configurable parameter. Next, each region is divided into $2^\beta$ equal sub-regions, where $\beta$ is the second configurable parameter. Given these regions in the probability distribution, a vector is constructed with  $4\times2^\beta$ elements each representing the summation of probabilities in a corresponding region.

We consider the probability distribution of strength difference of connected butterfly vertices $Pr(\delta)$ given a graph snapshot $G_{N_b}$. Using burst-based snapshots enables fair comparison of different graphs with different temporal characteristics.  Inspired by the DDQC approach, we divide $Pr(\delta)$ into four regions based on the mean and standard deviation of $\delta$s ($\mu_\delta$ and $\sigma_\delta$, see Figure~\ref{fig:epinionr}.d,g). As long as the first region covers the low $\delta$s, the number/coverage of other regions for the tail of right-skewed distribution is not impactful in mixing pattern analysis. Accordingly, we summarize the probability distribution as an embedding vector $F$ with four elements ($\Sigma_{i=1,..,4}F_i=1$). Each element corresponds to a region as below: 
\begin{equation}
   F_1=\Sigma Pr(\delta) \text{, } \forall \delta\leq \mu_\delta
\end{equation}
\begin{equation}
   F_2=\Sigma Pr(\delta) \text{, } \forall \mu_\delta<\delta\leq\mu_\delta+\sigma_\delta
\end{equation}
\begin{equation}
   F_3=\Sigma Pr(\delta) \text{, } \forall \mu_\delta+\sigma_\delta<\delta\leq \mu_\delta+2\sigma_\delta
\end{equation}
\begin{equation}
   F_4=\Sigma Pr(\delta) \text{, } \forall \delta>\mu_\delta+2\sigma_\delta
\end{equation}

The vector $F$ provides fine-grained information. Additionally, to express the strength assortativity as an scalar for simple network inference in temporal analyses, we define the \textit{strength assortativity localization factor} as $r^s=F_1-0.5$ to track the localization of $\delta$s ($F$) on the region behind mean ($F_1$). $r^s$  lies in the range $[-0.5,0.5]$. A (negative) positive $r^s$ highlights strength (dis)assortativity and a zero value corresponds to random strength mixing. $r=0.5$ denotes perfect strength assortativity and $r=-0.5$ denotes perfect strength disassortativity.
\section{Butterfly Emergence Patterns}\label{sec:observations}
In this section, we analyze the butterfly emergence patterns in streaming graphs. We start by
 analyzing the real-world streaming graphs to identify the characteristic mixing patterns of butterflies (\S~\ref{subsec:realObservations}). We then briefly survey the related works on graph patterns and growth models (\S~\ref{subsec:relatedworks}). Next, we study synthetic streaming graphs based on seminal structure growth models to explain the generative processes underlying the observed patterns (\S~\ref{subsec:syntheticObservations}). Finally, we discuss our findings (\S~\ref{subsec:discussobservation}).
\subsection{Analysis of Real-world Graph Streams}\label{subsec:realObservations}
Figure~\ref{fig:butterflycountReal} shows the growth of butterfly count in real-world streams. To quantify this growth, we define the butterfly rate of each graph snapshot as the number of butterflies ($\bowtie$) in the graph normalized by the number of edges ($|E|$), $\bowtie/|E|$, and compute the average butterfly rate (plus/minus the standard deviation) over the sequential snapshots.  
As provided in  
Figure~\ref{fig:butterflycountReal}, the average butterfly rate is greater than $1$ in all streams, indicating that the number of butterflies grows superlinearly with respect to the number of edges in the sequential graph snapshots. In some graphs the superlinearity starts after some time.
\begin{figure*}[h]
    \centering
  \subfigure[Ciao, $8.9$$\pm$$1.8$]{\includegraphics[width=0.24\textwidth]{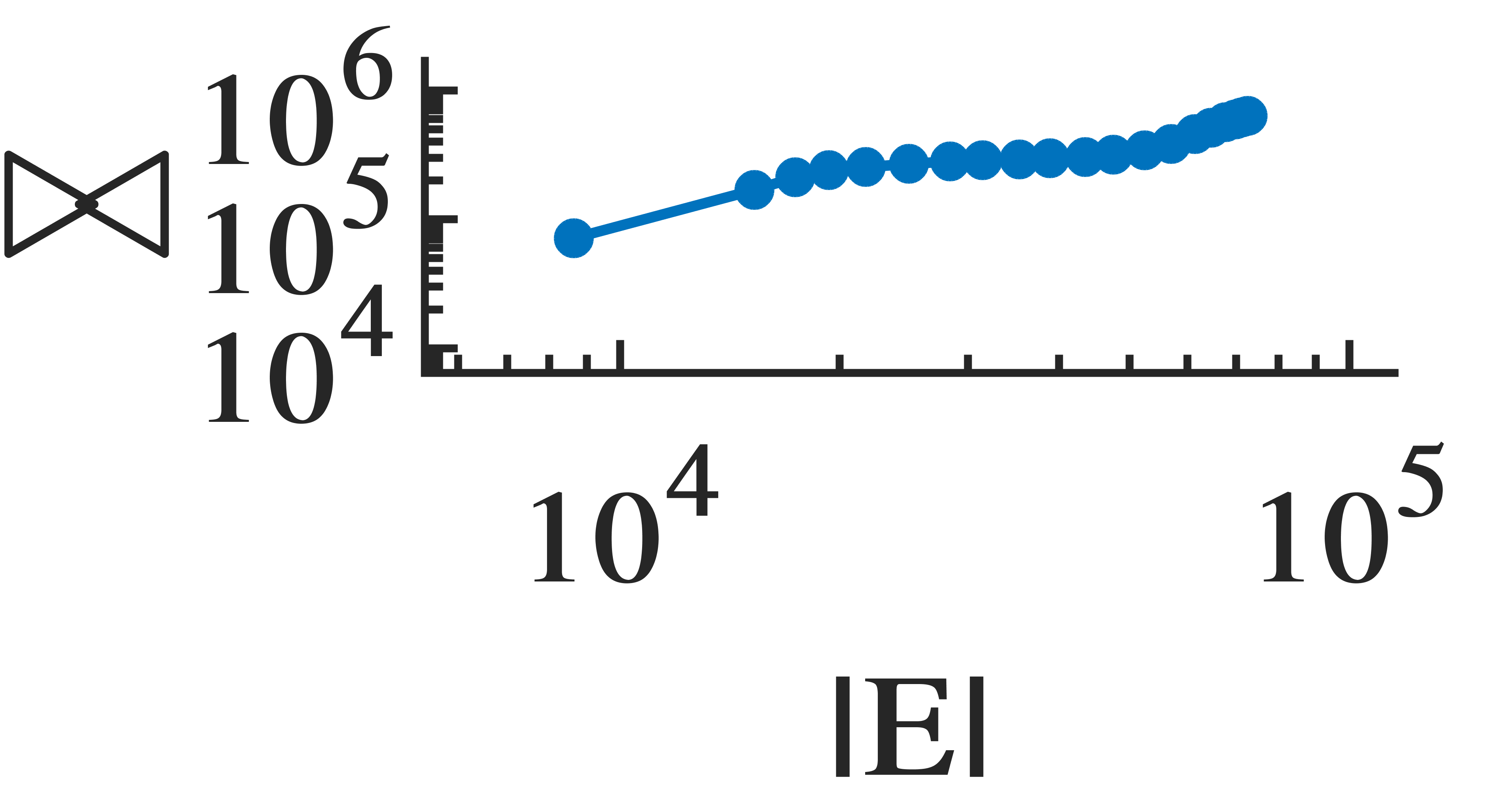}}
  \subfigure[Epinions, $10.6$$\pm$$7.6$]{\includegraphics[width=0.24\textwidth]{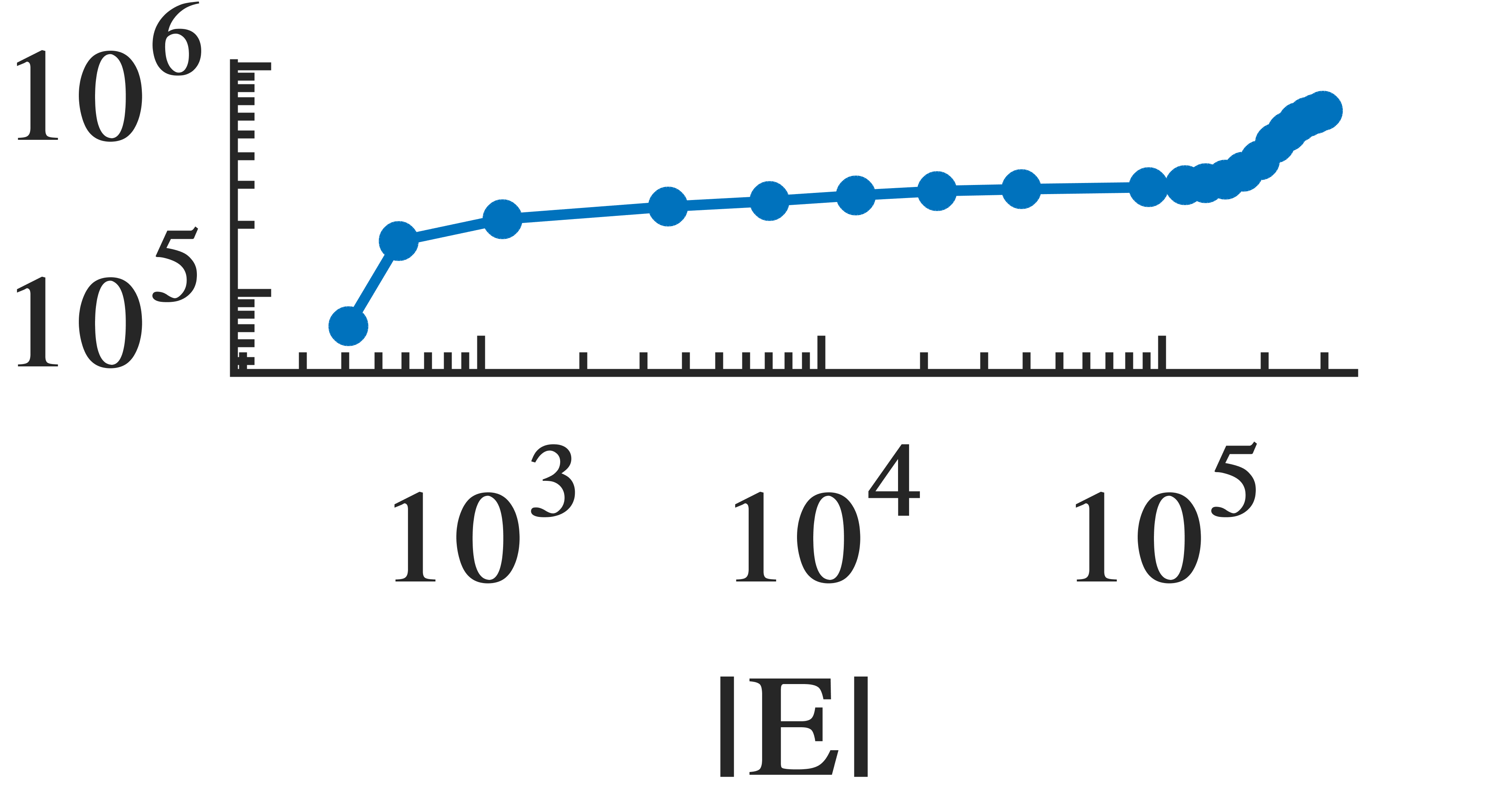}}
  \subfigure[WikiLens, $138$$\pm$$77$]{\includegraphics[width=0.24\textwidth]{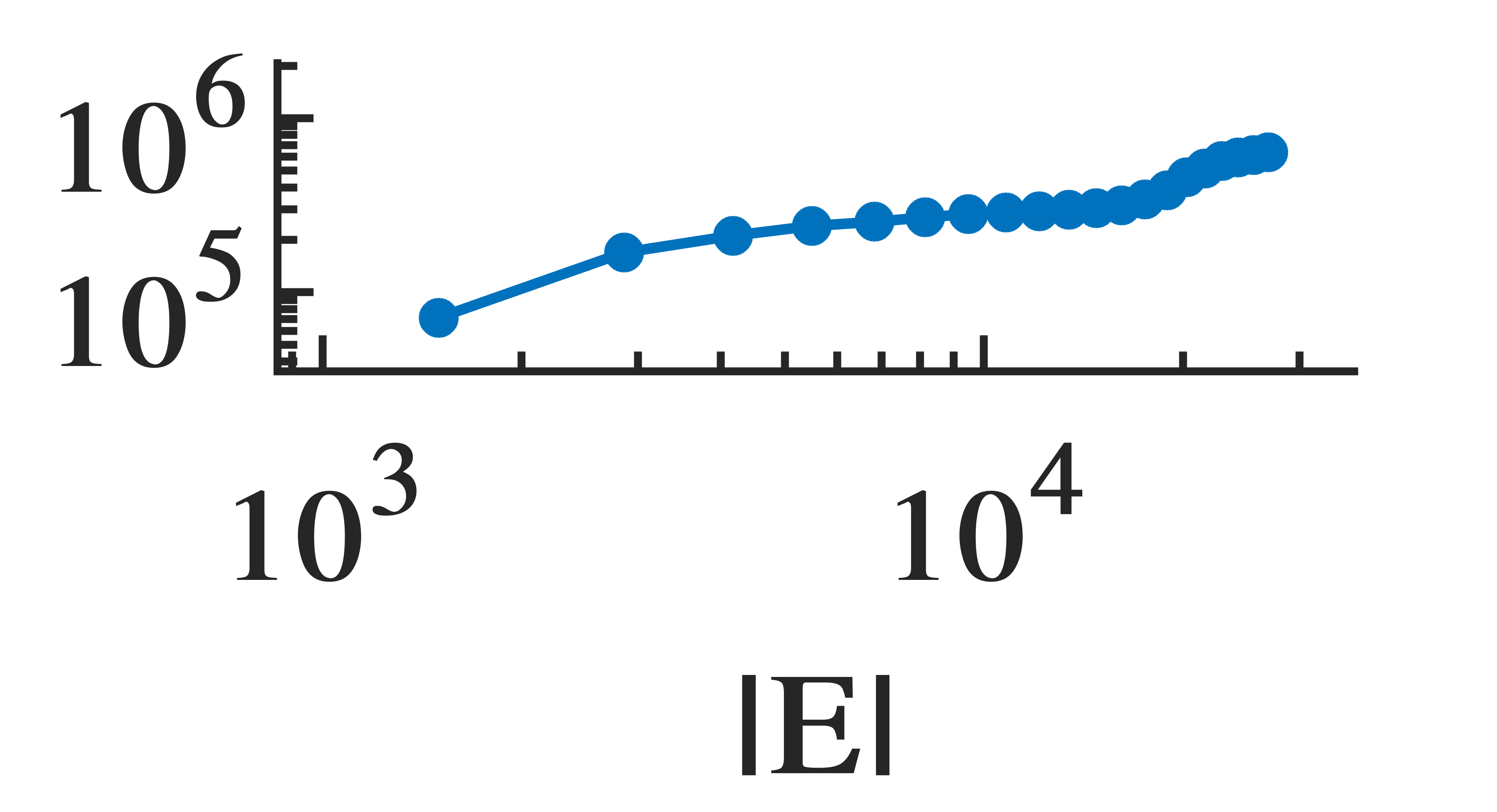}}
  \subfigure[ML100k, $171.4$$\pm$$99.7$]{\includegraphics[width=0.24\textwidth]{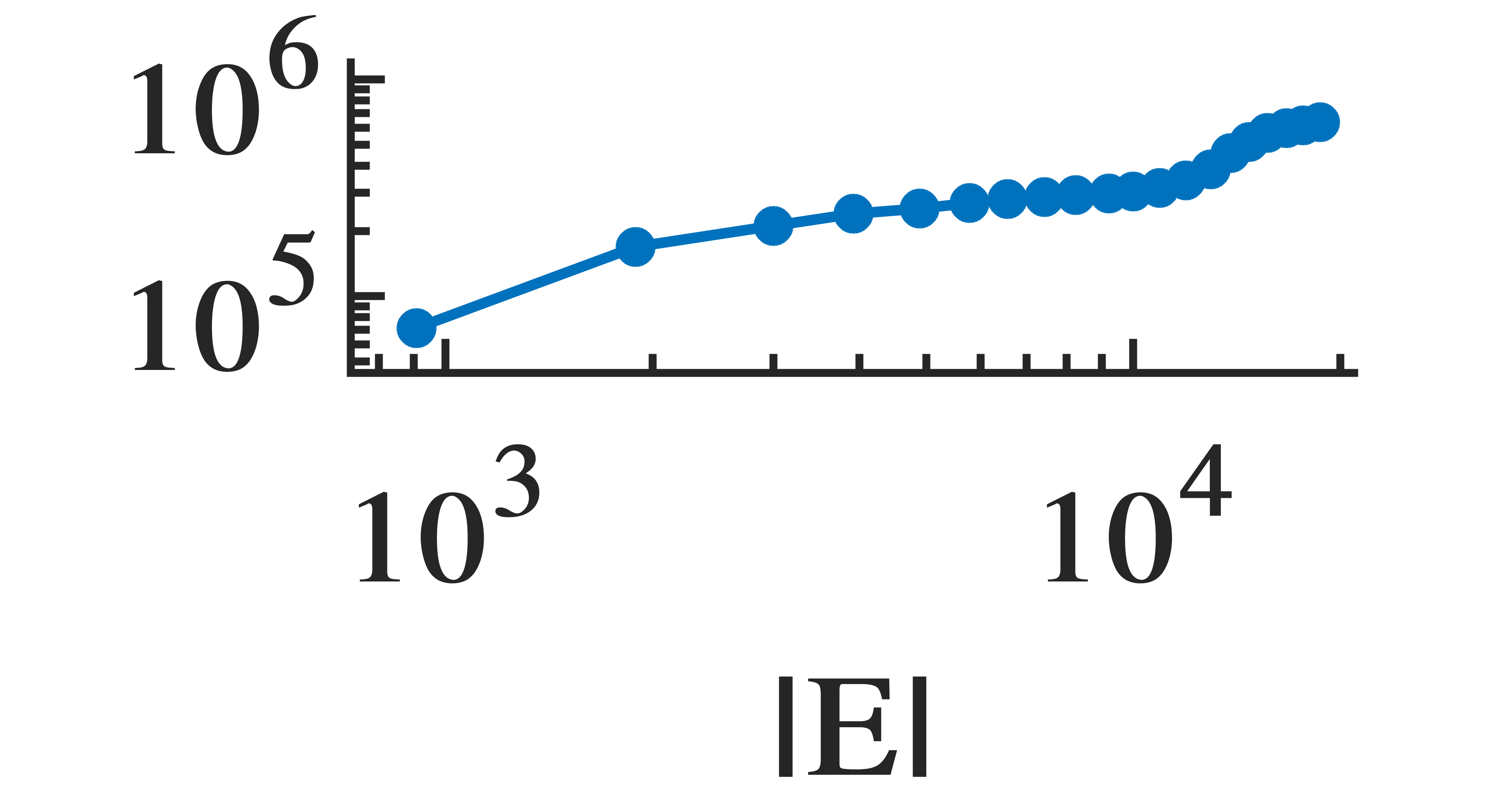}}
  \subfigure[ML1m, $142.1$$\pm$$79.4$]{\includegraphics[width=0.24\textwidth]{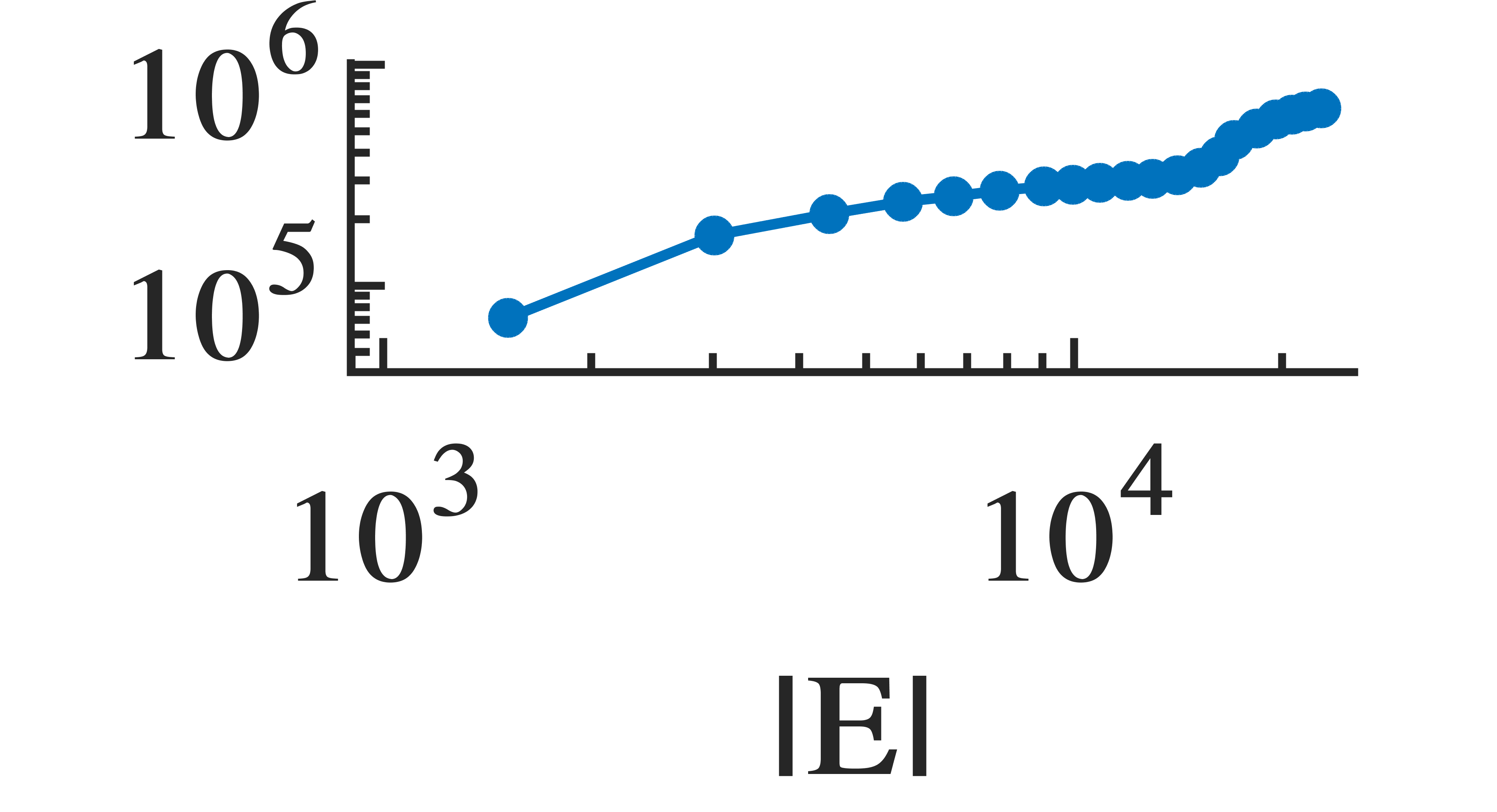}}
  \subfigure[Amazon, $0.9$$\pm$$1$]{\includegraphics[width=0.24\textwidth]{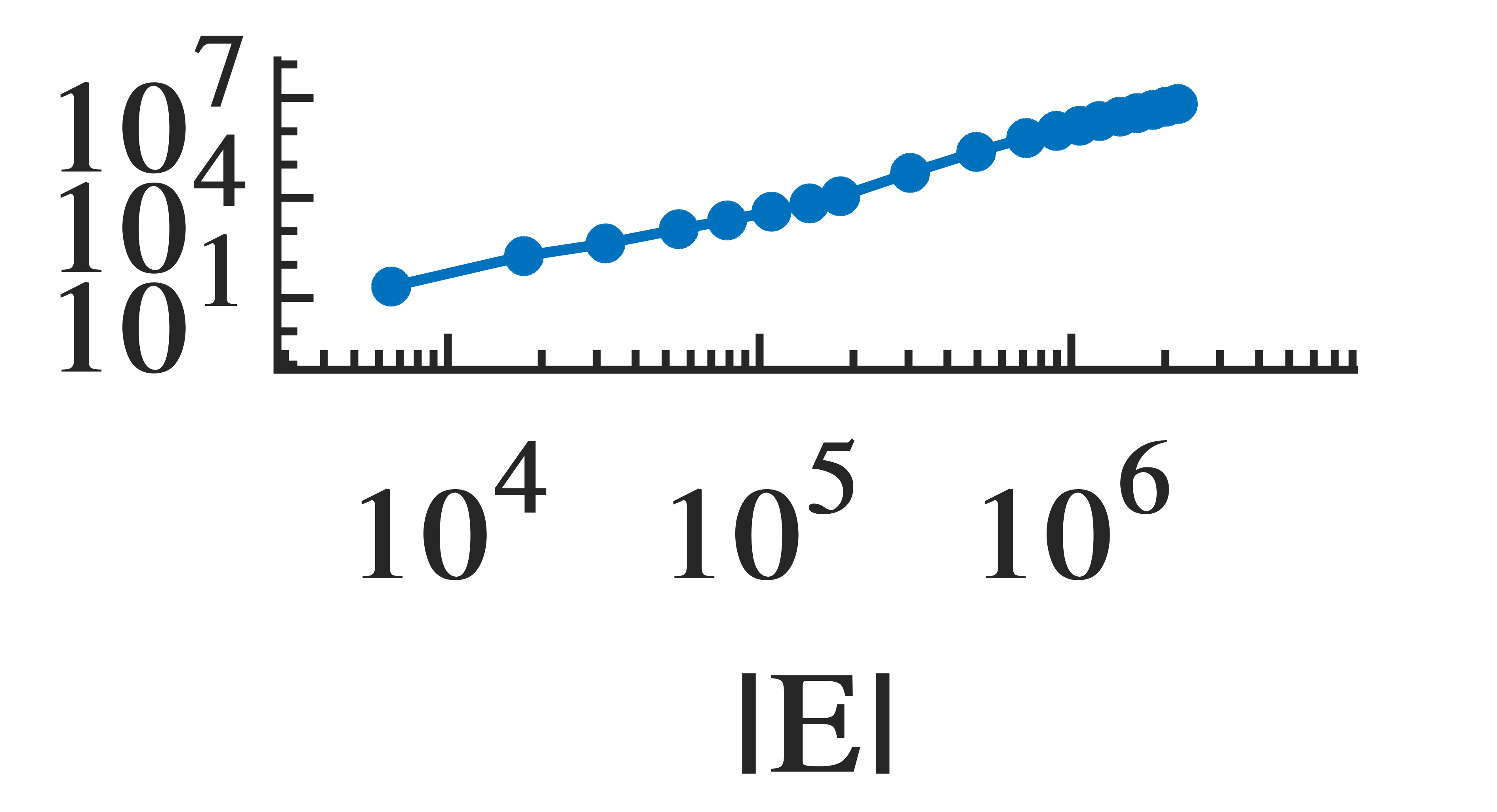}}
  \subfigure[Yahoo, $73$$\pm$$52.1$]{\includegraphics[width=0.24\textwidth]{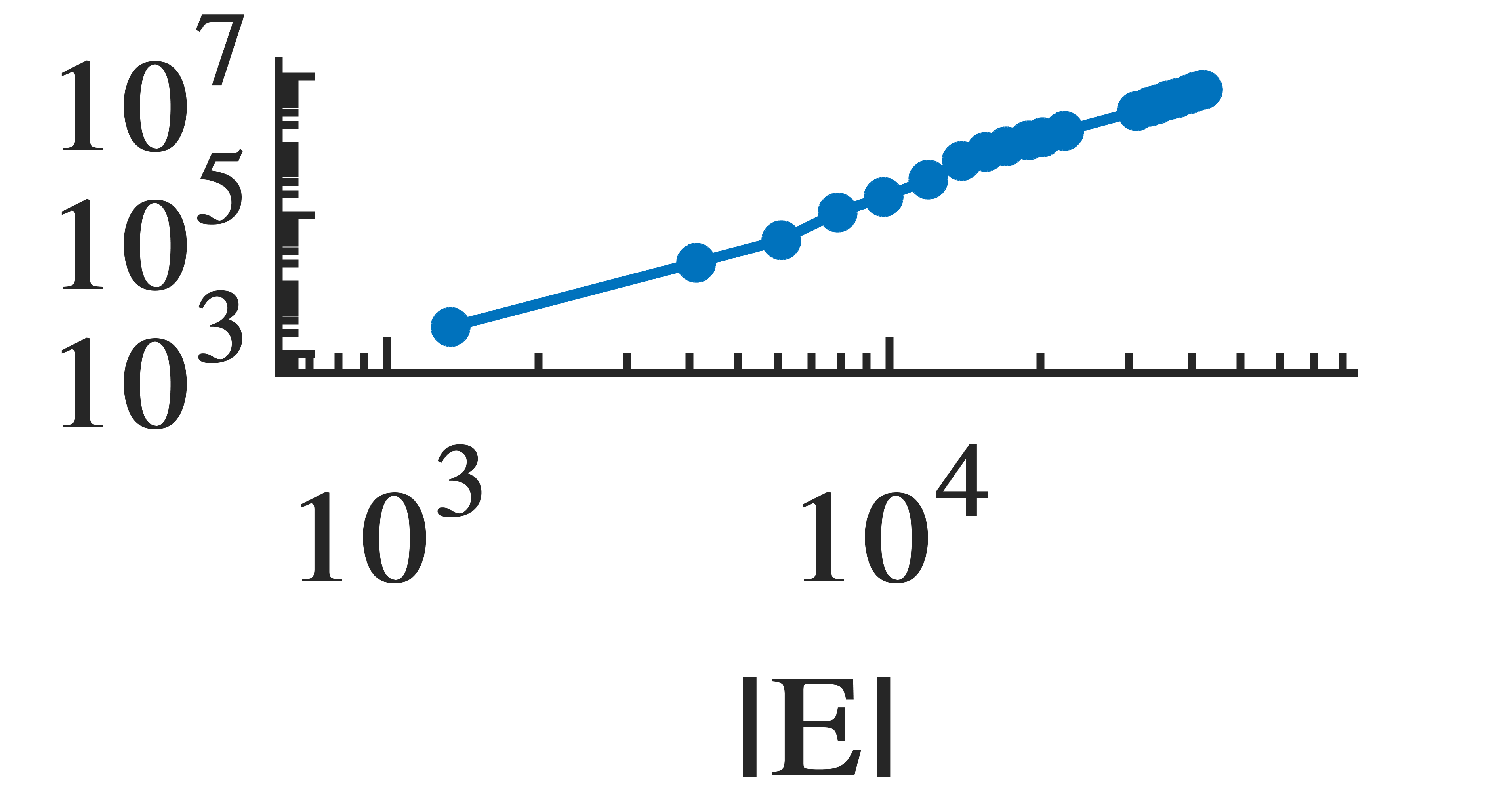}}
   \caption{Butterfly count versus edge count in real-world streams with various average butterfly rates.} 
    \label{fig:butterflycountReal}
\end{figure*}

We compute $Pr(\delta)$, the probability distribution of strength-differences of connected butterfly vertices, for the graph snapshots in the streams. We embed each probability distribution in a vector $F$ as explained in \S~\ref{subsec:measure}. Figures~\ref{fig:Frealgraphs} and \ref{fig:randrs} demonstrate the evolution of $F$ elements and their corresponding strength assortativity localization factor ($r^s$) over the timeline of burst arrivals. We observe that in all graph streams, butterfly edges have strength-difference less than equal to the average strength-difference ($\mu_\delta$) with probability $Pr(\delta$$\leq$$\mu_\delta)$$\approx$$0.7$ ($F$ is localized on $F_1$). The tail of $Pr(\delta)$ for all graphs is heavier in the region $[\mu _\delta$, $\mu_\delta$$+$$\sigma_\delta]$ with probability of $Pr(\mu_\delta$$<$$\delta$$\leq$$\mu_\delta$$+$$\sigma_\delta)$$\approx$$0.25$ (according to $F_2$ values) and gets lighter at the end. This demonstrates that the majority of butterfly edges are formed by vertices with similar strengths at all time points. Also, the strength assortativity localization factor is $0.15$$\leq$$r^s$$\leq$$0.2$ in all graphs at almost all time points (Figure~\ref{fig:randrs}). 
\begin{figure*}[h]
    \centering
    \subfigure[Ciao]{\includegraphics[width=0.24\textwidth]{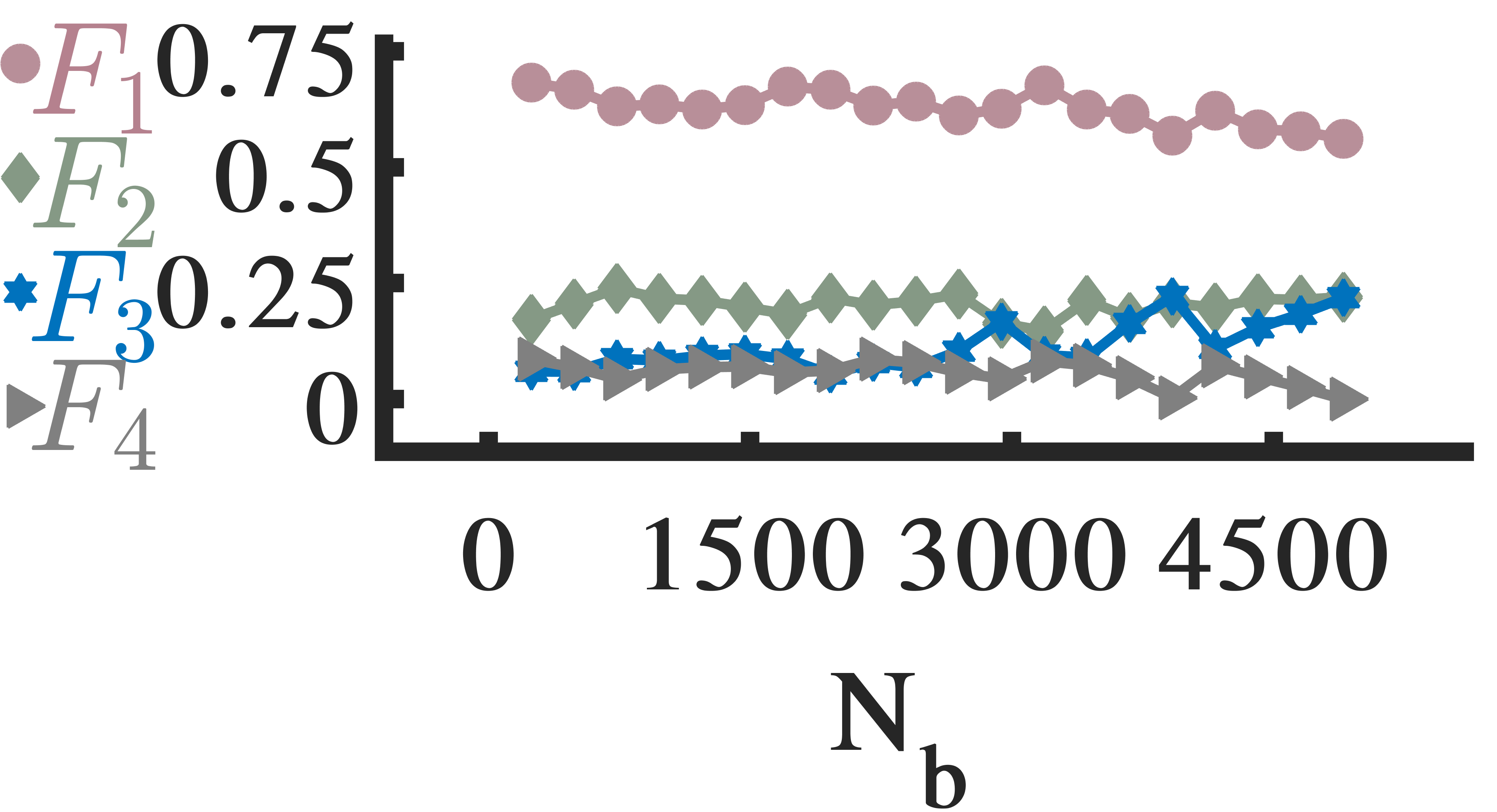}}
     \subfigure[Epinions]{\includegraphics[width=0.24\textwidth]{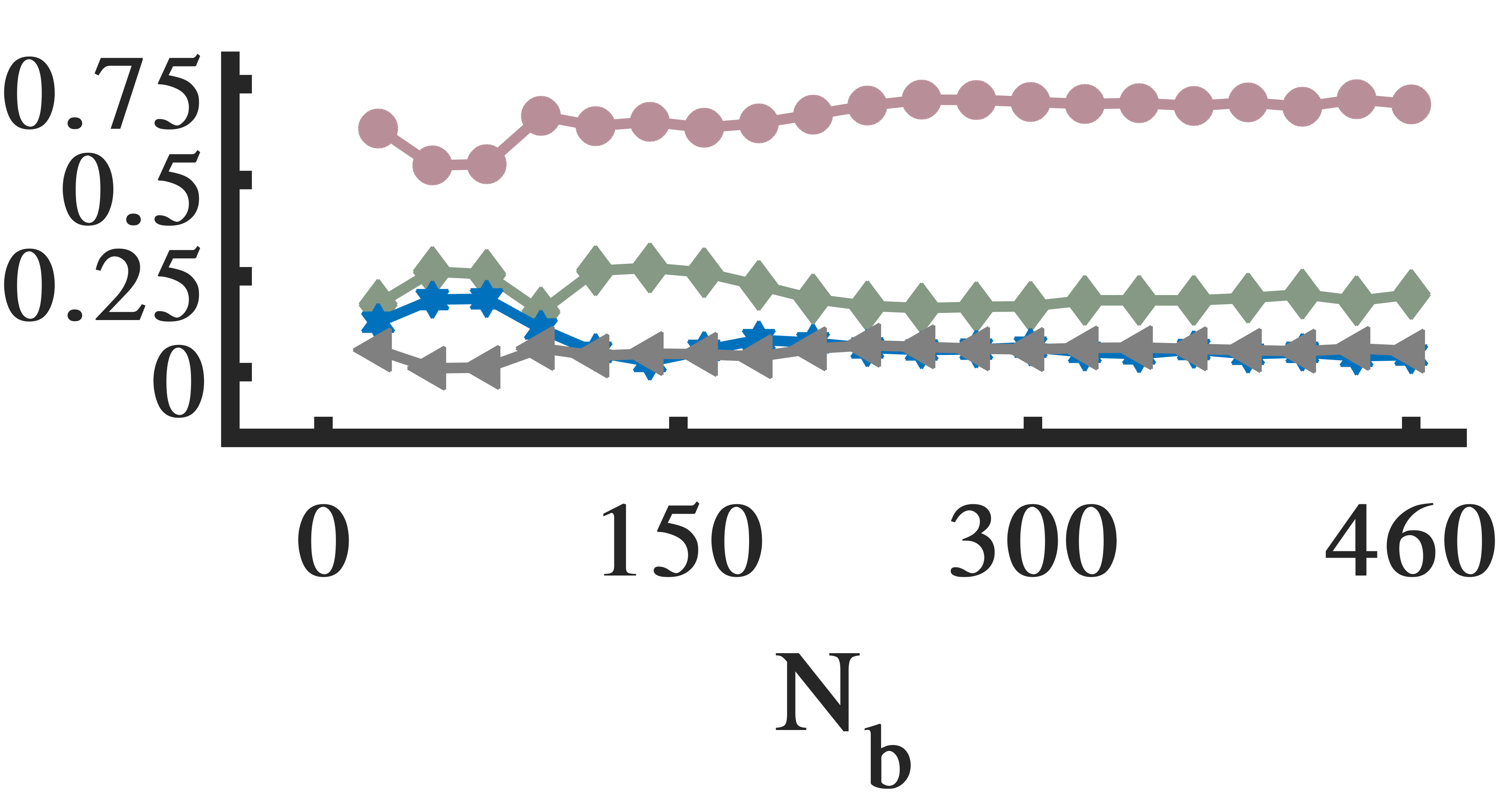}}
  \subfigure[WikiLens]{\includegraphics[width=0.24\textwidth]{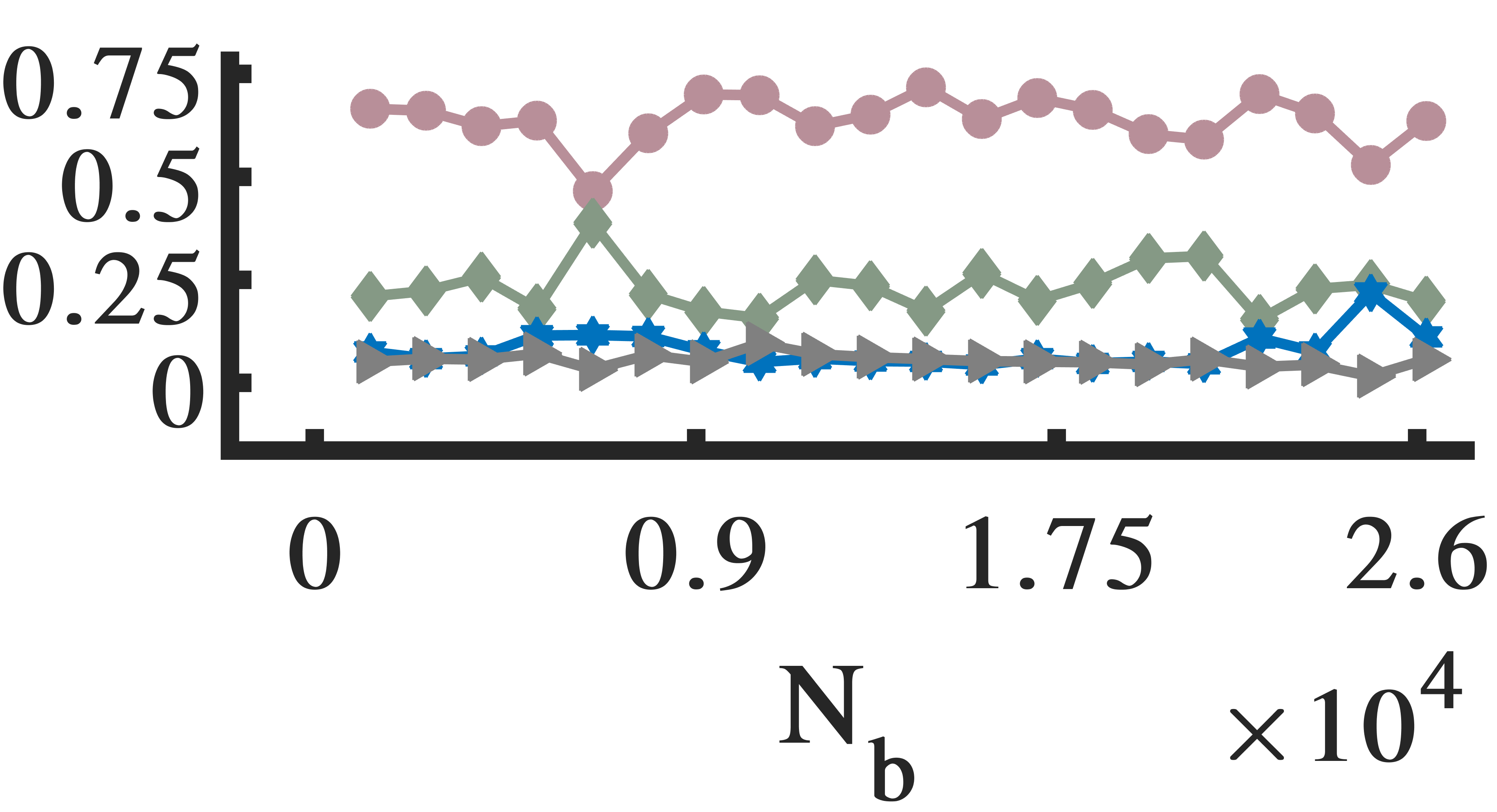}}
  \subfigure[ML100k]{\includegraphics[width=0.24\textwidth]{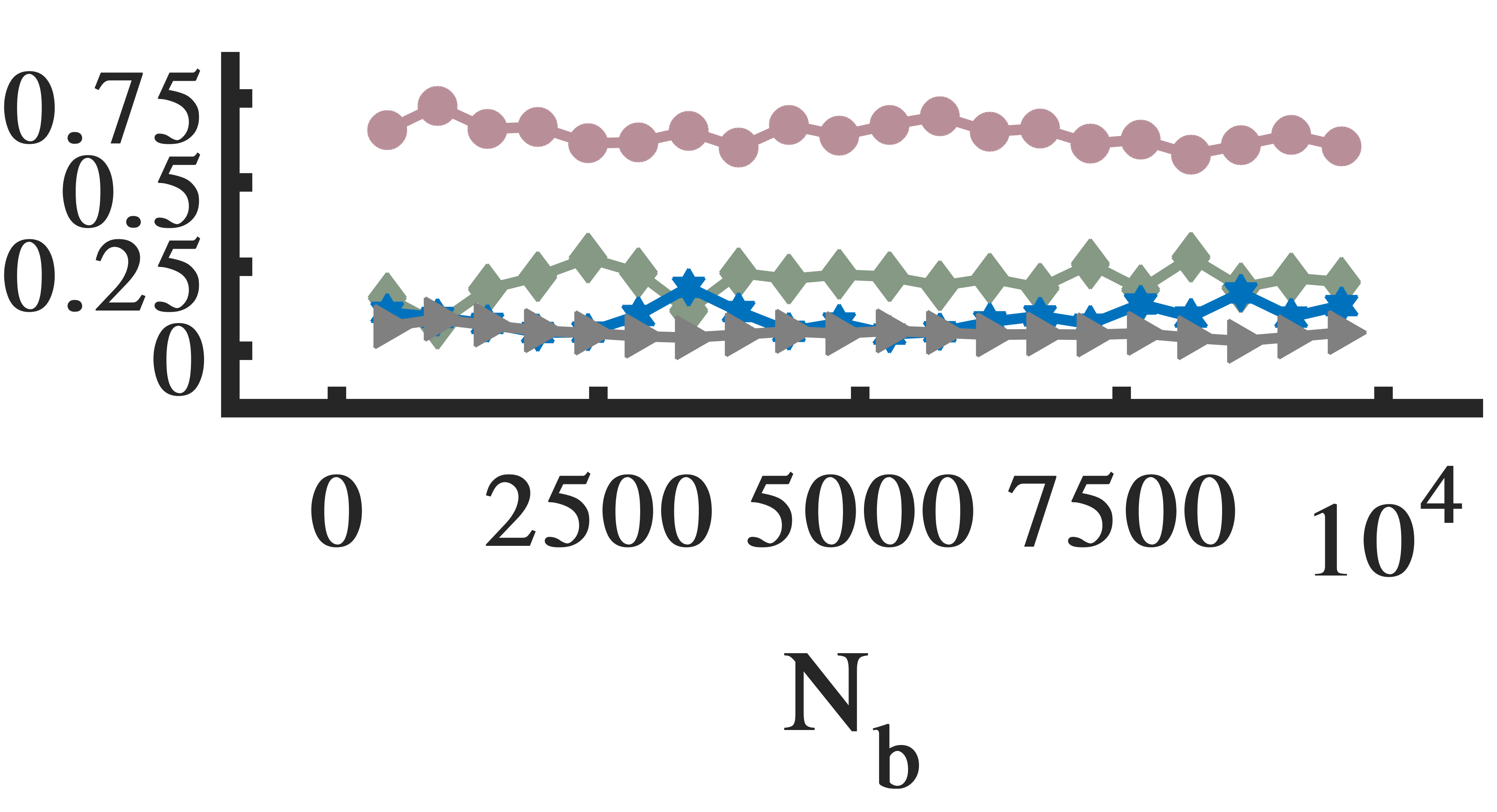}}
  \subfigure[ML1m]{\includegraphics[width=0.24\textwidth]{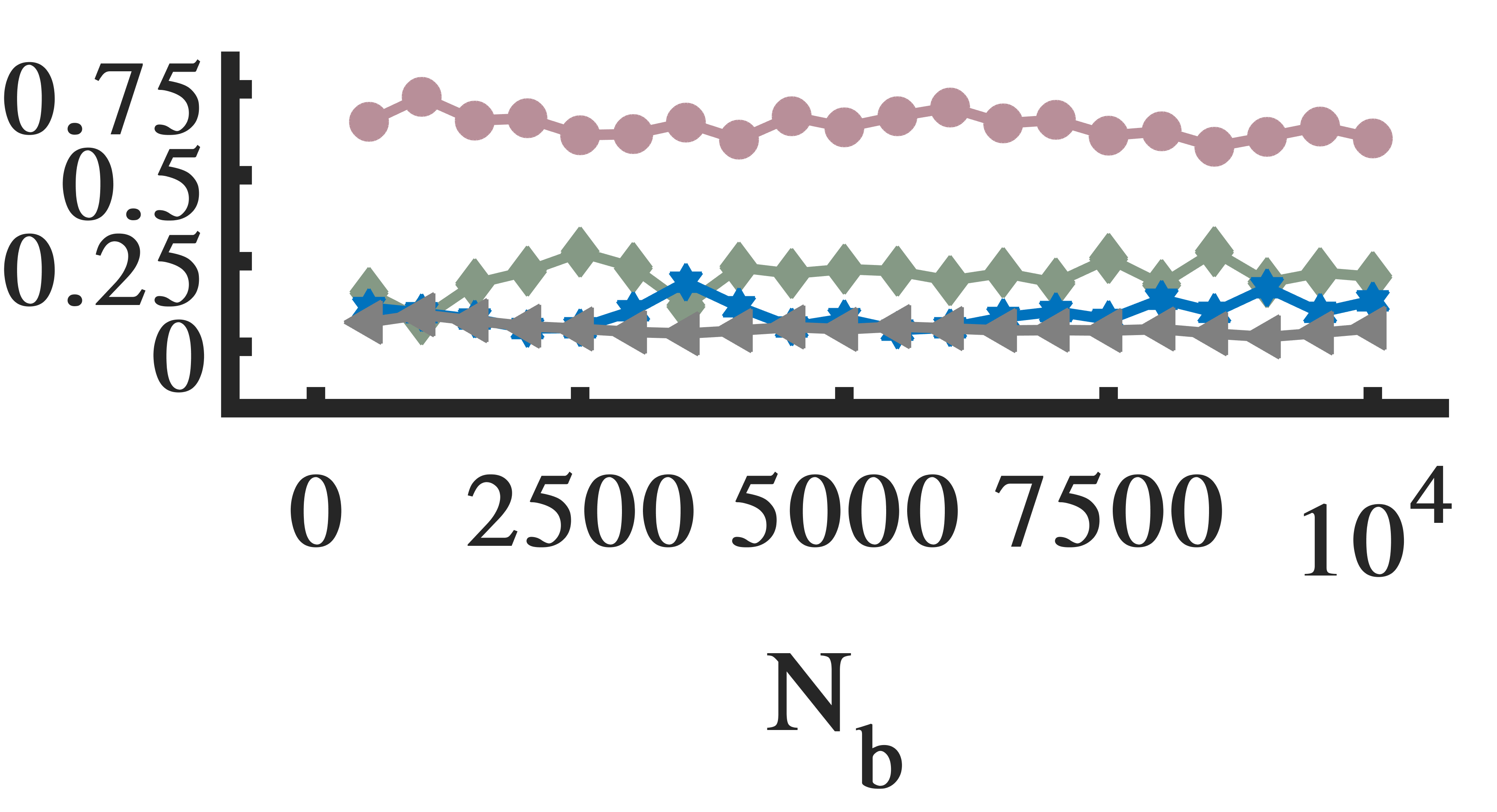}}
  \subfigure[Amazon]{\includegraphics[width=0.24\textwidth]{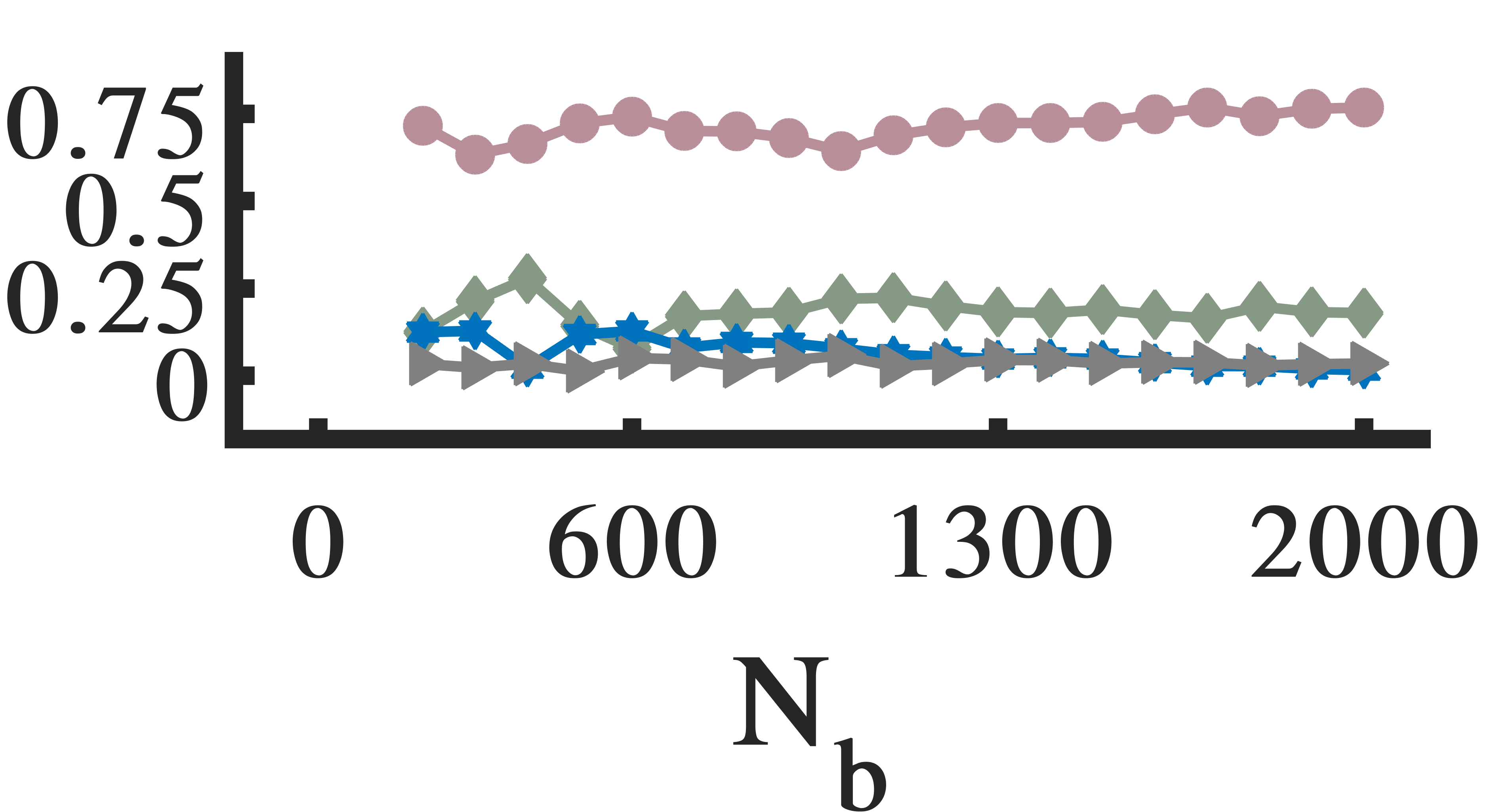}}
\subfigure[Yahoo]{\includegraphics[width=0.24\textwidth]{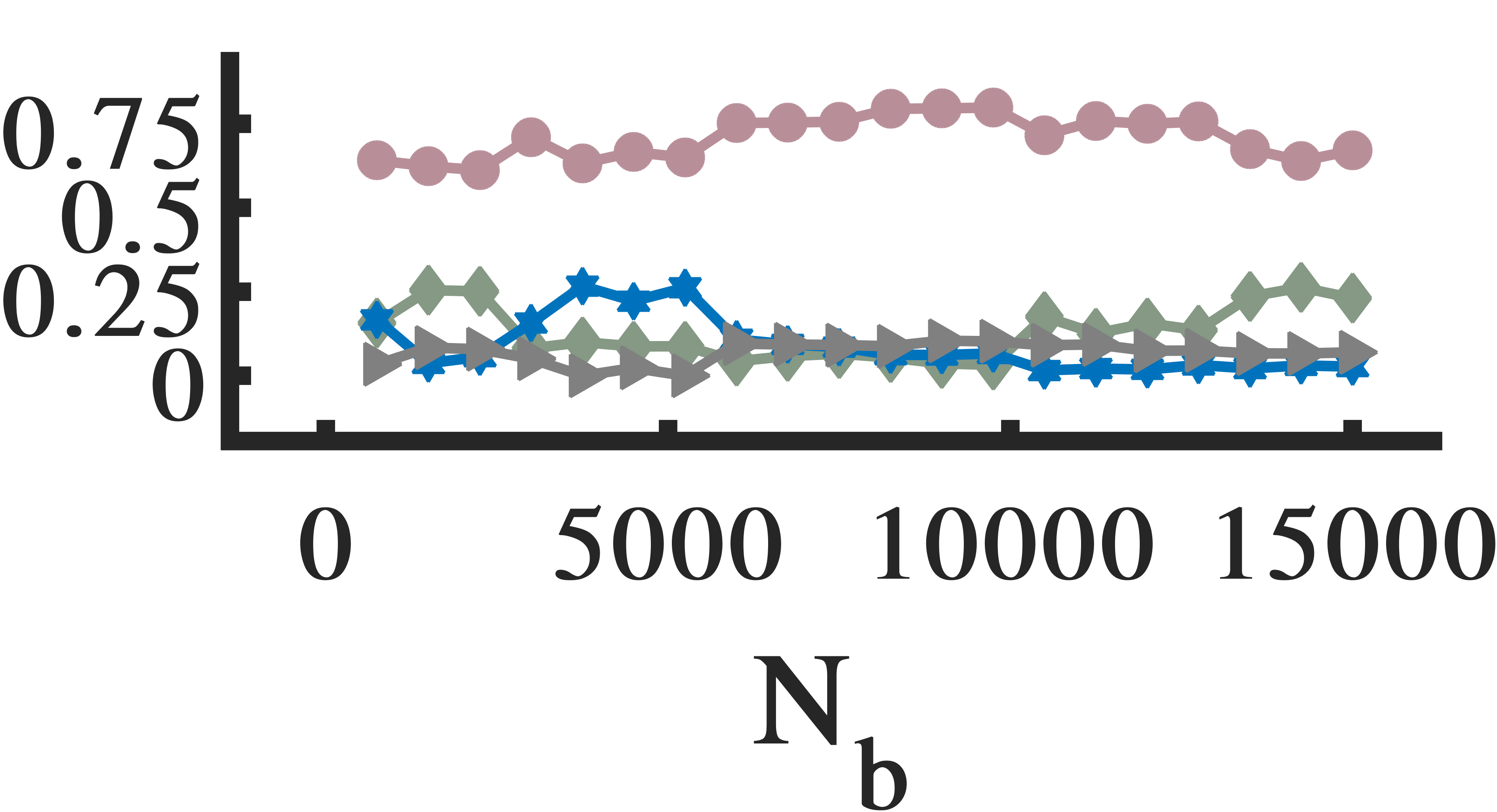}}
    \caption{F elements over the timeline of burst arrivals in real-world streams.}
    \label{fig:Frealgraphs}
\end{figure*}

\begin{figure*}[h]
    \centering
  \subfigure[Ciao]{\includegraphics[width=0.24\textwidth]{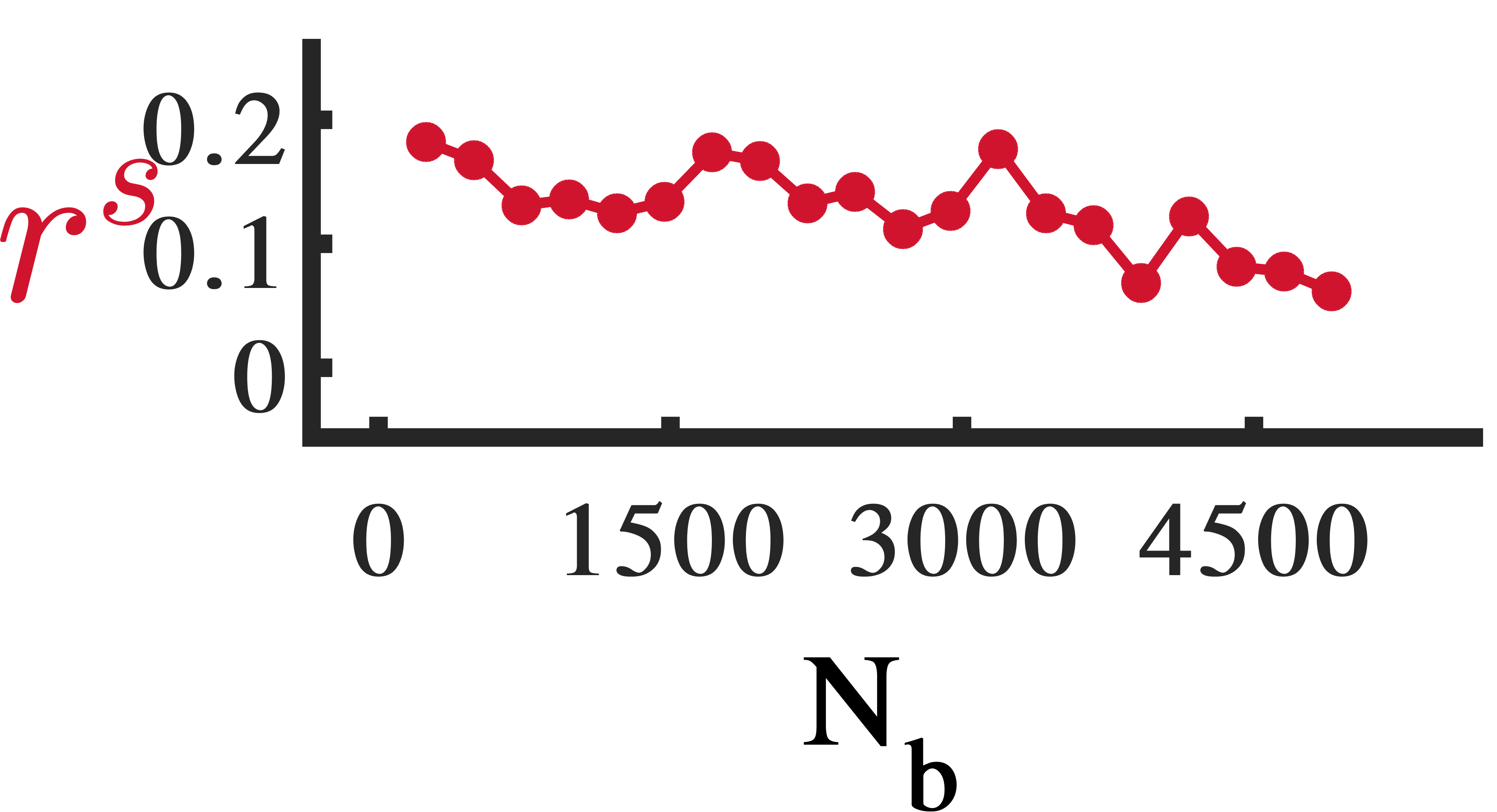}}
  \subfigure[Epinions]{\includegraphics[width=0.24\textwidth]{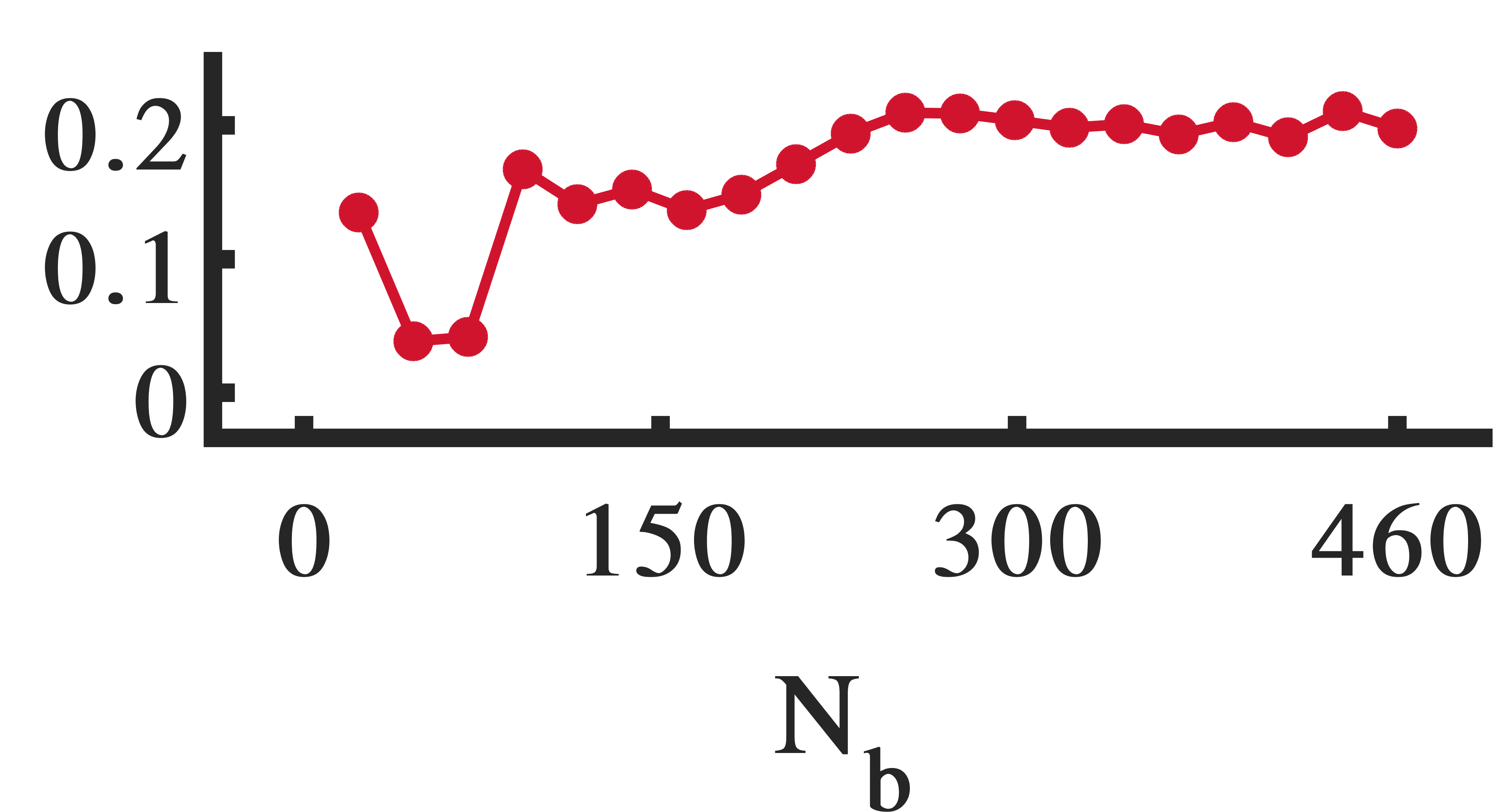}}
  \subfigure[WikiLens]{\includegraphics[width=0.24\textwidth]{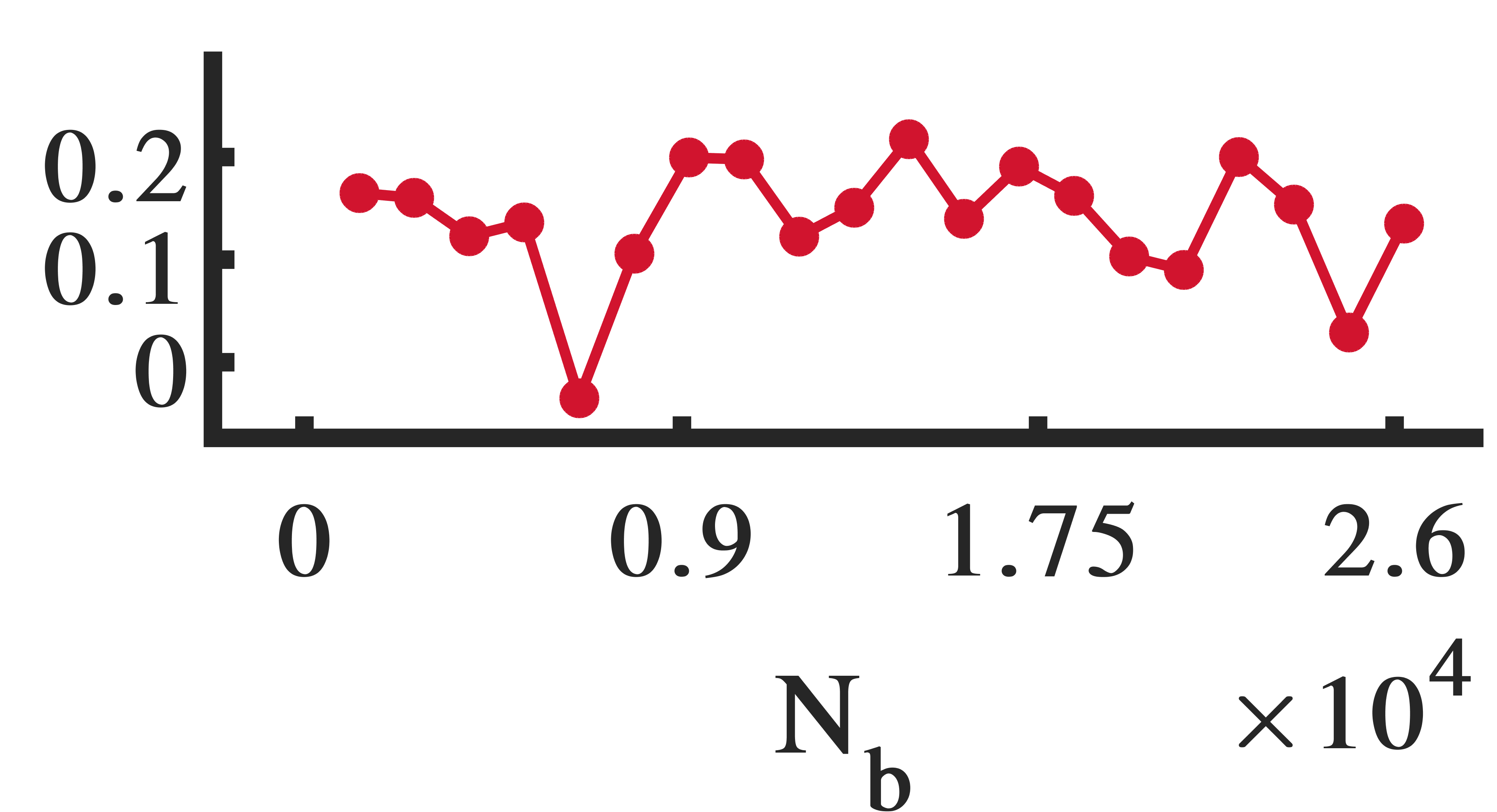}}
  \subfigure[ML100k]{\includegraphics[width=0.24\textwidth]{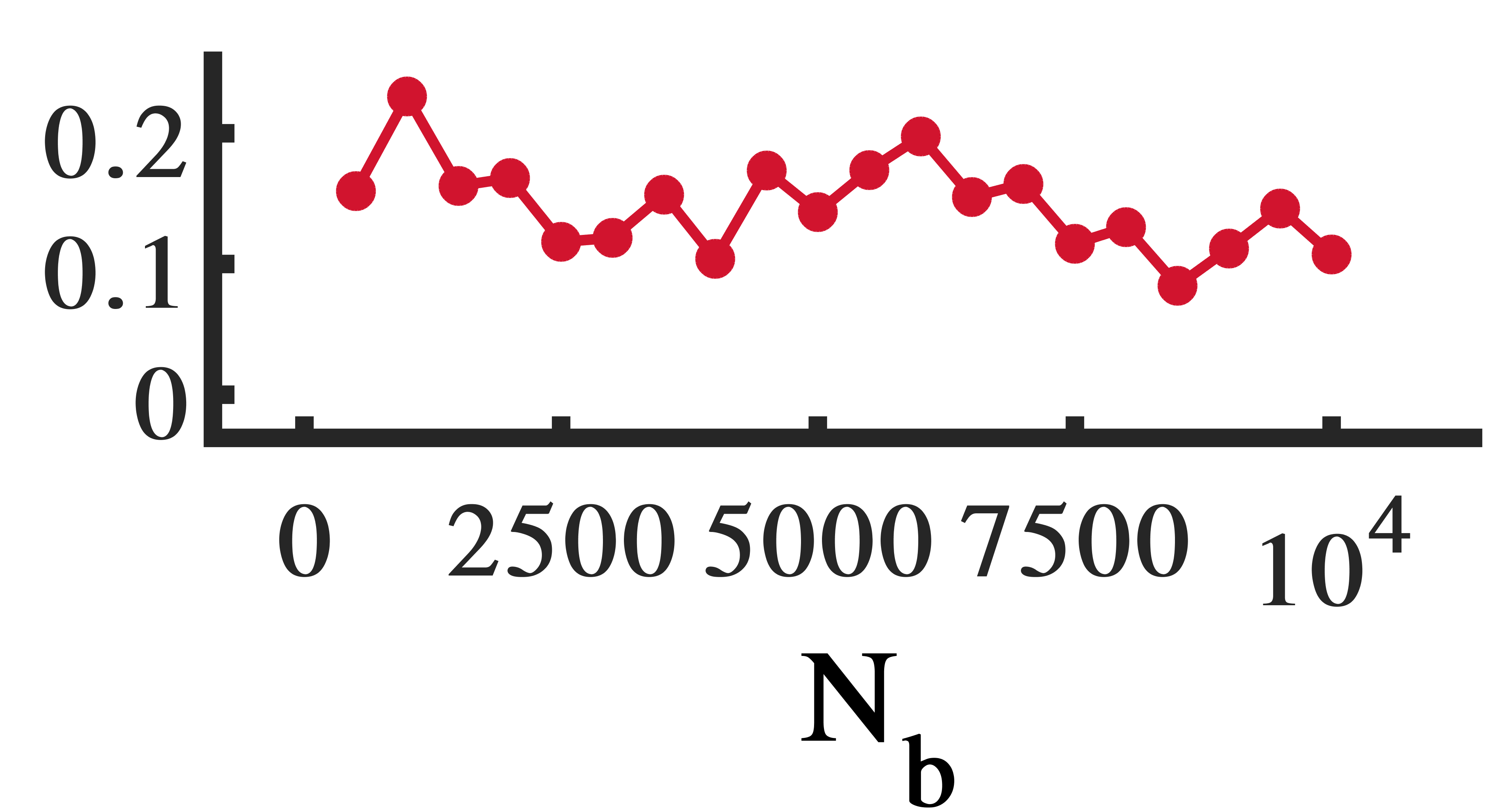}}
  \subfigure[ML1m]{\includegraphics[width=0.24\textwidth]{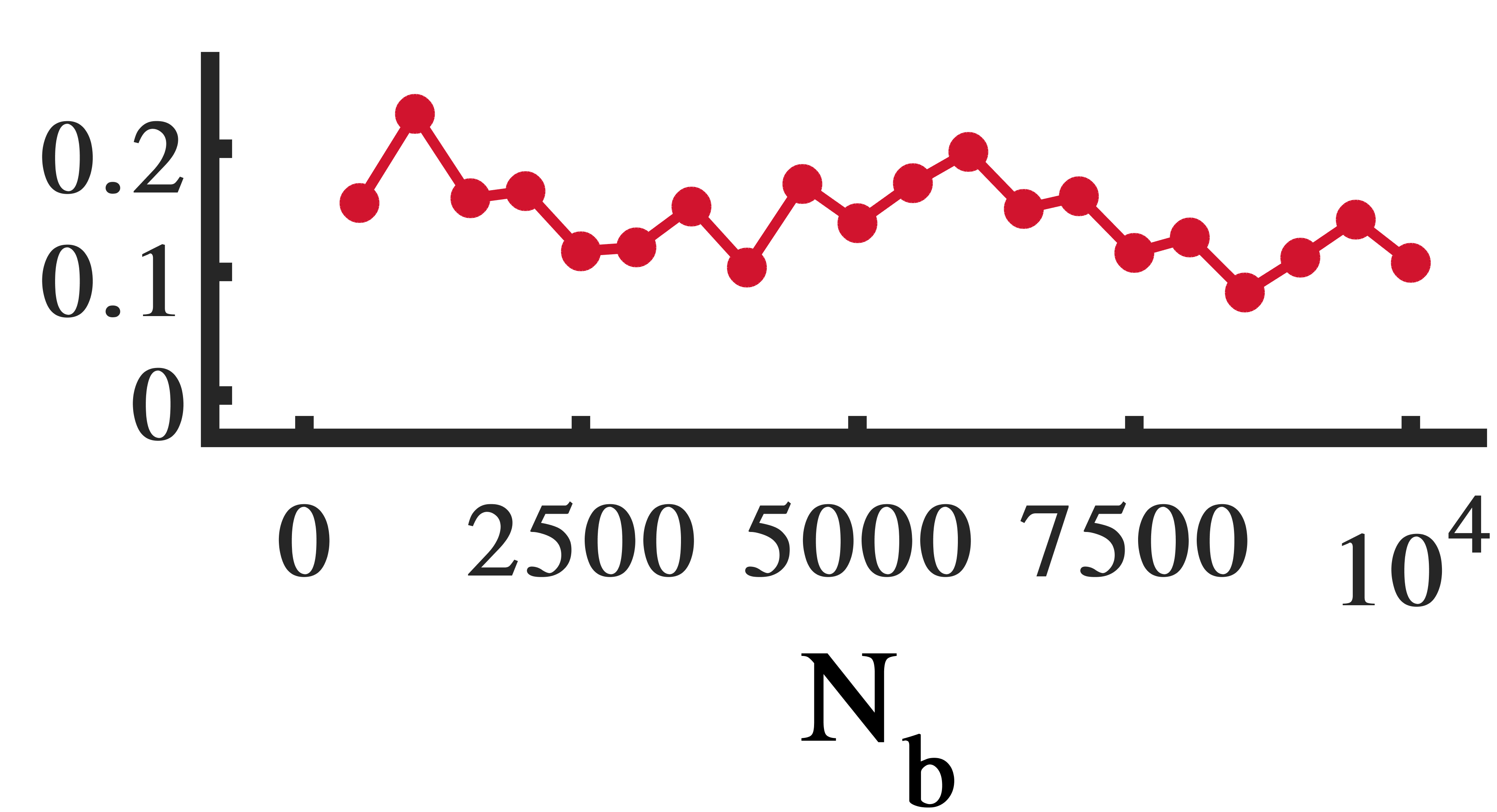}}
  \subfigure[Amazon]{\includegraphics[width=0.24\textwidth]{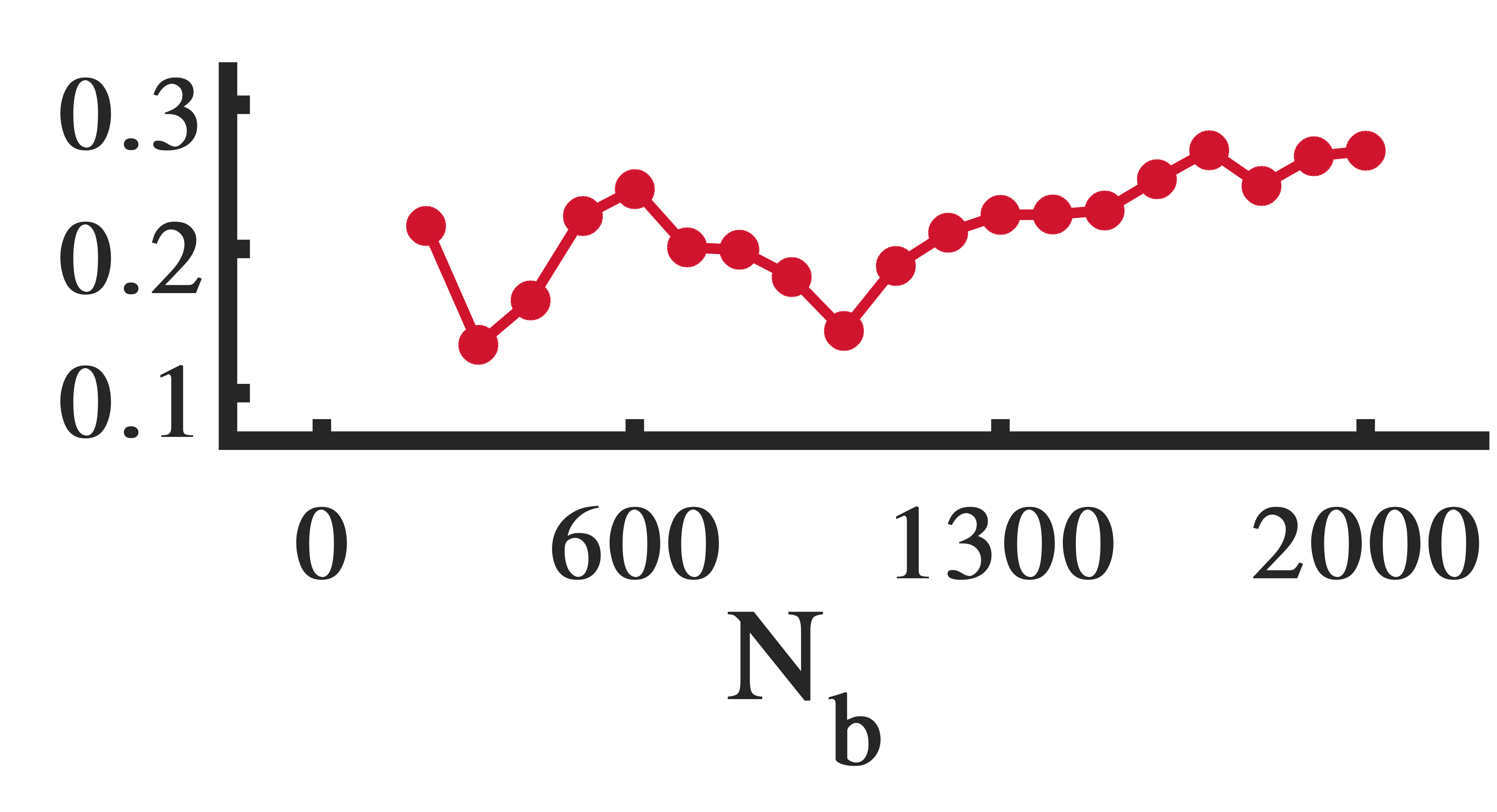}}
  \subfigure[Yahoo]{\includegraphics[width=0.24\textwidth]{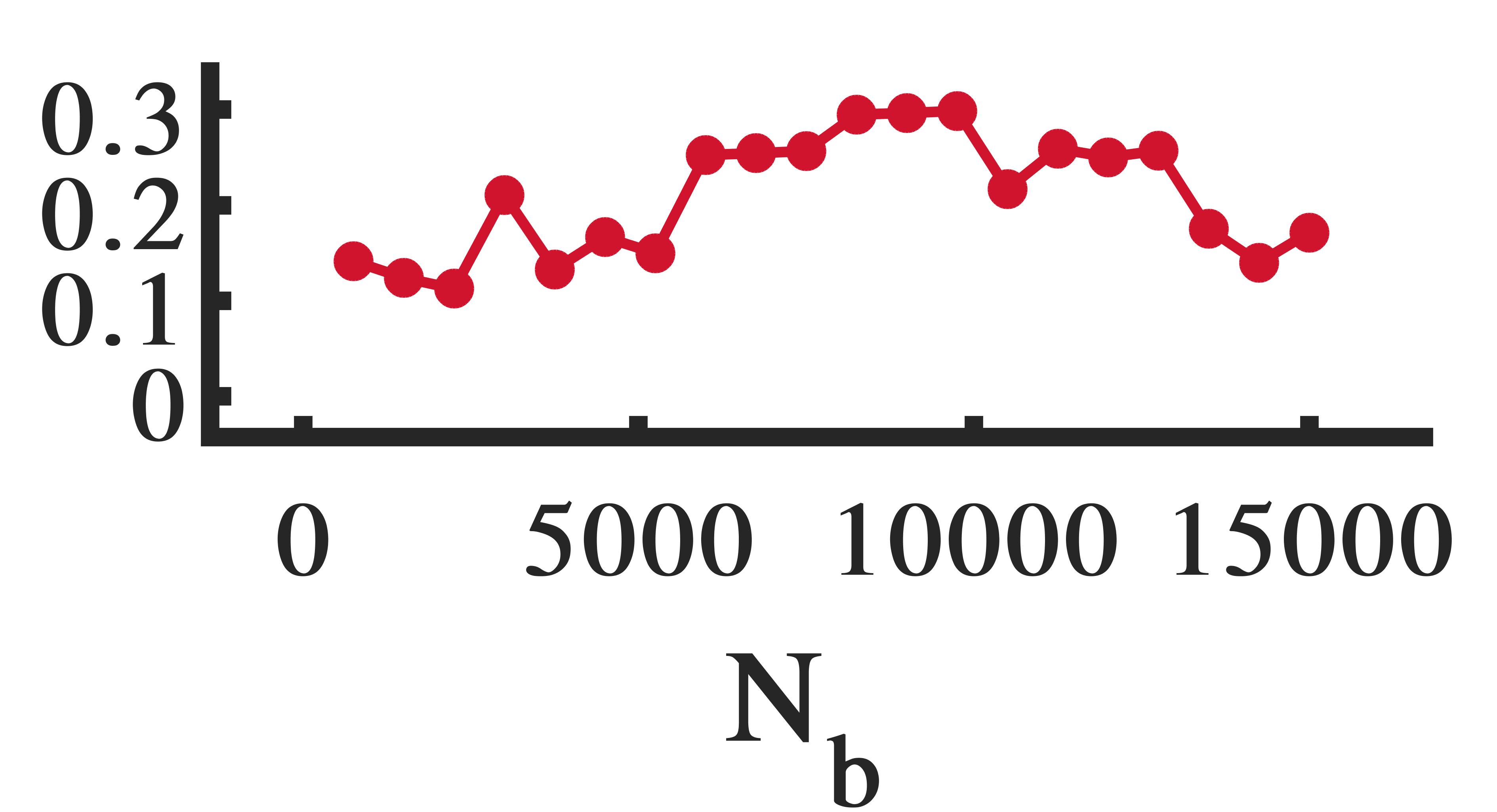}}
    \caption{Strength assortativity localization factor ($r^s$) of butterflies over the timeline of burst arrivals in real-world streams.}
    \label{fig:randrs}
\end{figure*}
\begin{figure*}[]
    \centering
  \subfigure[Ciao]{\includegraphics[width=0.24\textwidth]{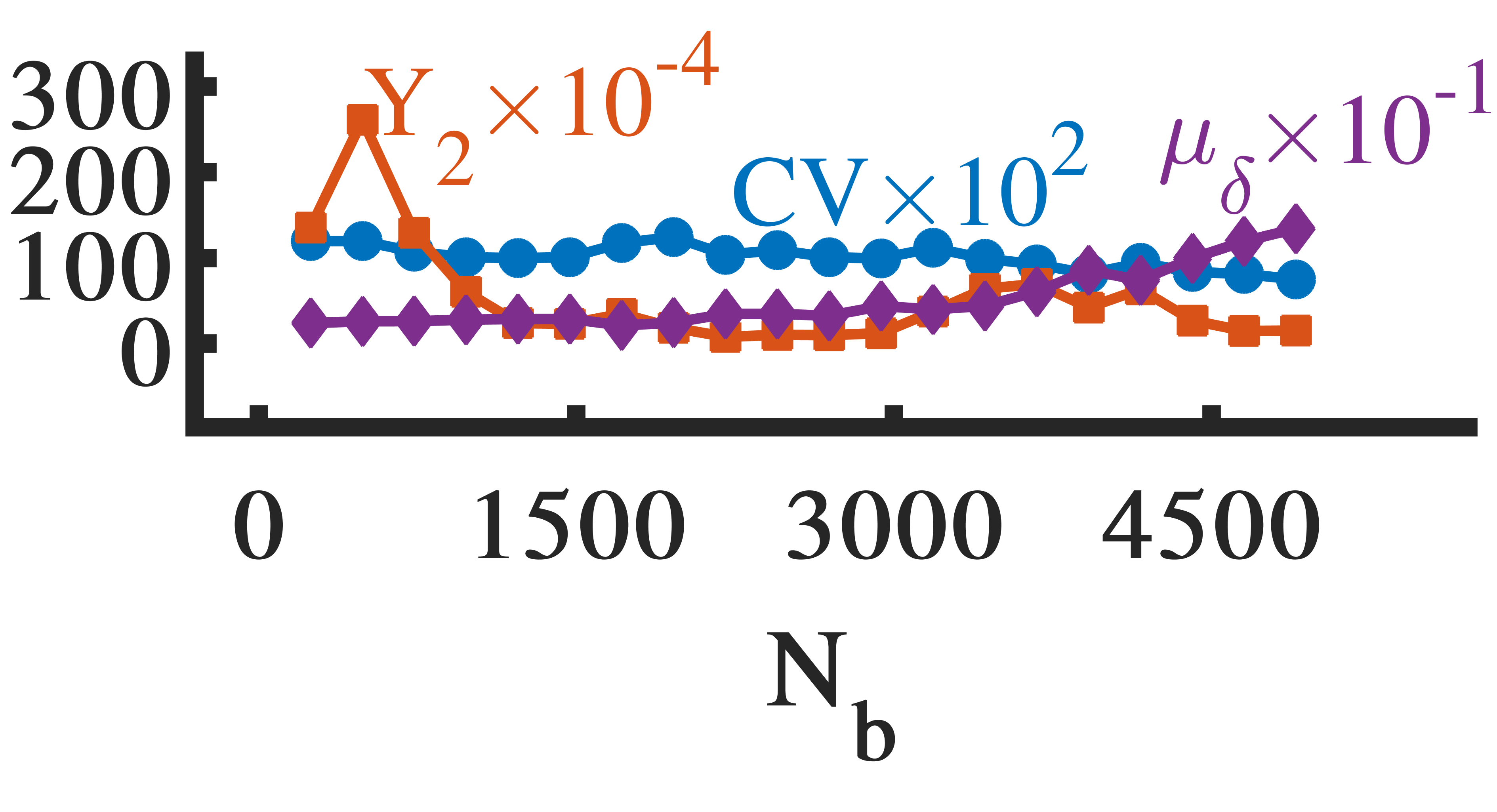}}
  \subfigure[Epinions]{\includegraphics[width=0.24\textwidth]{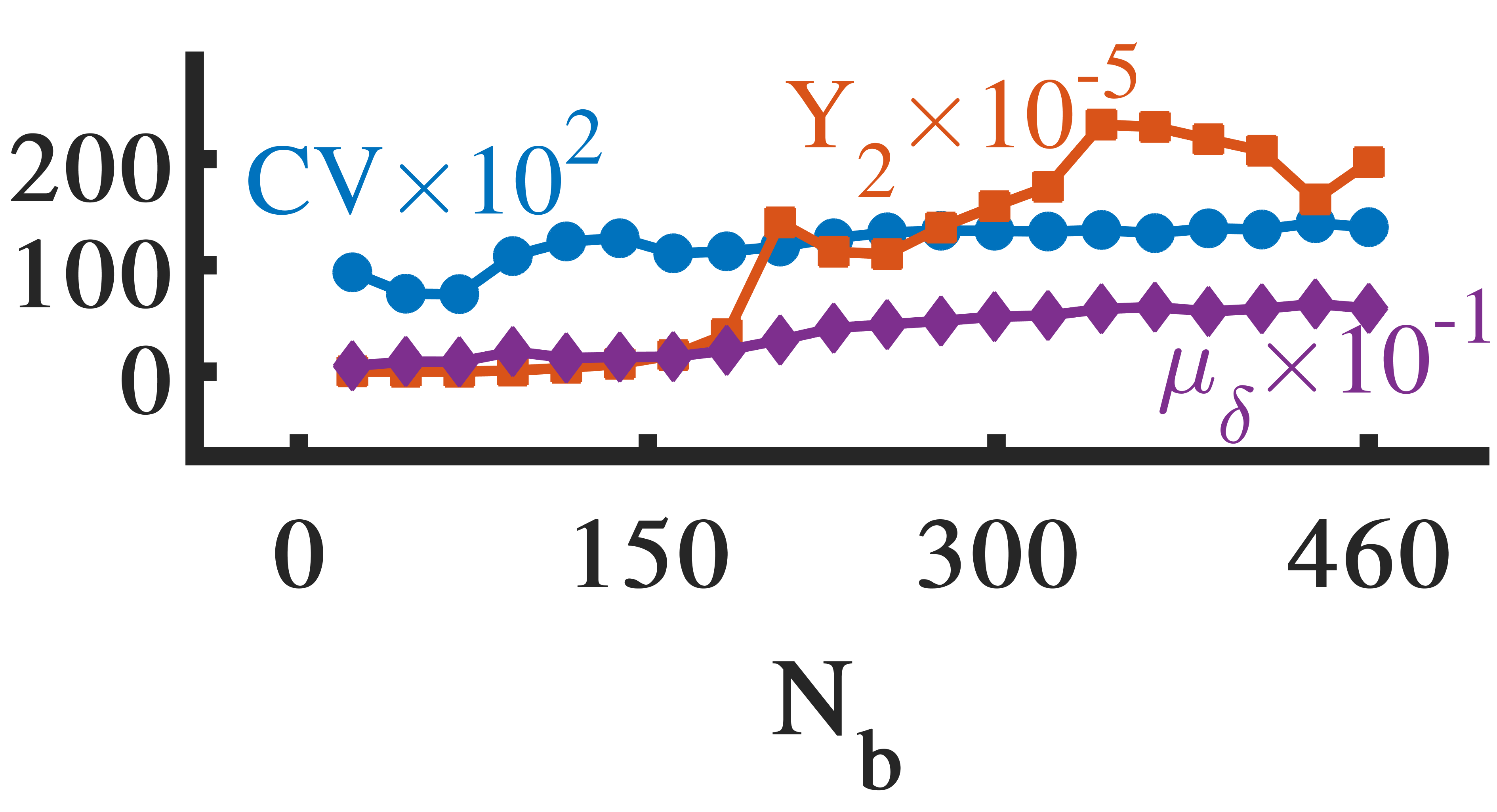}}
  \subfigure[WikiLens]{\includegraphics[width=0.24\textwidth]{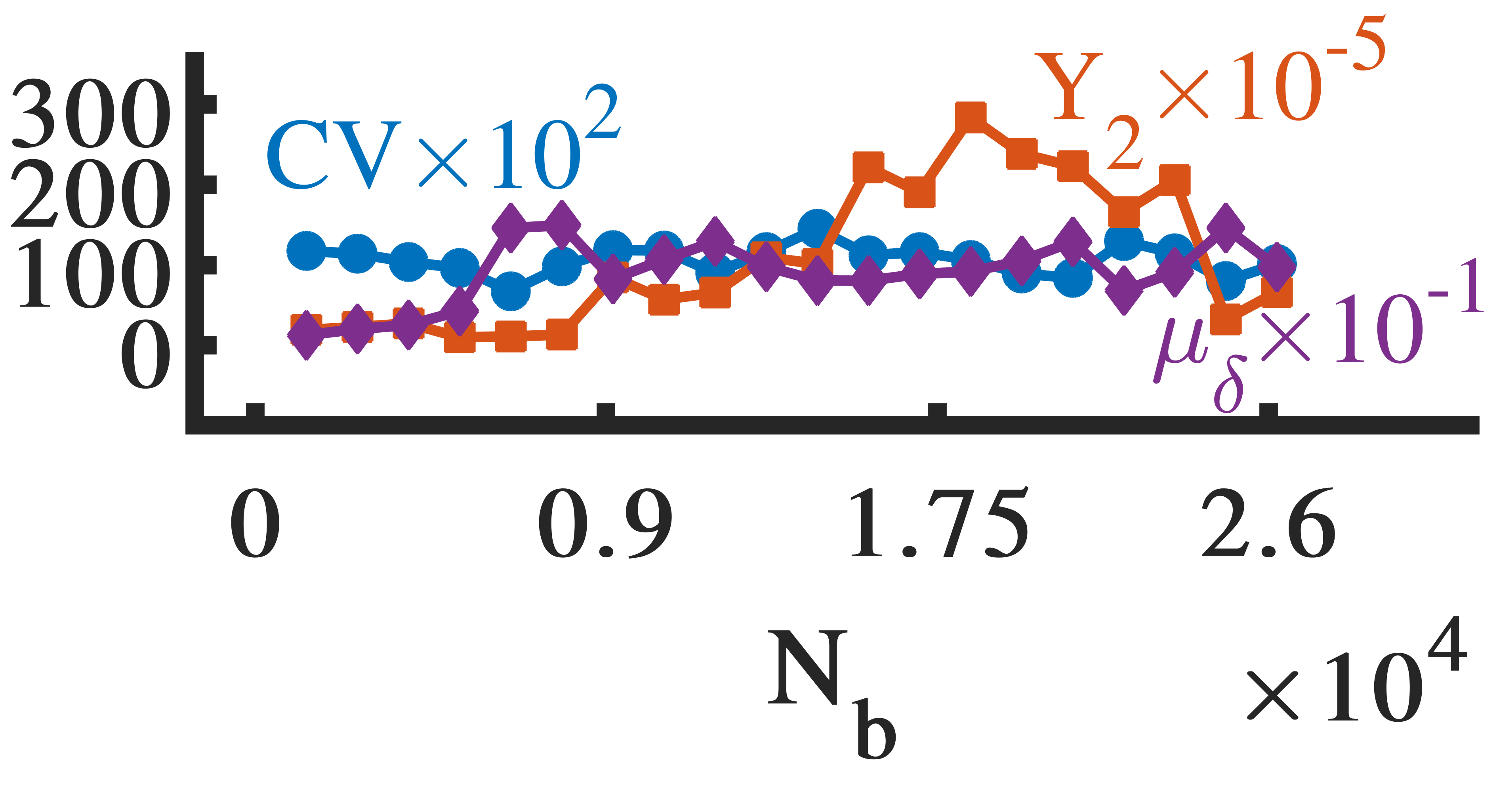}}
  \subfigure[ML100k]{\includegraphics[width=0.24\textwidth]{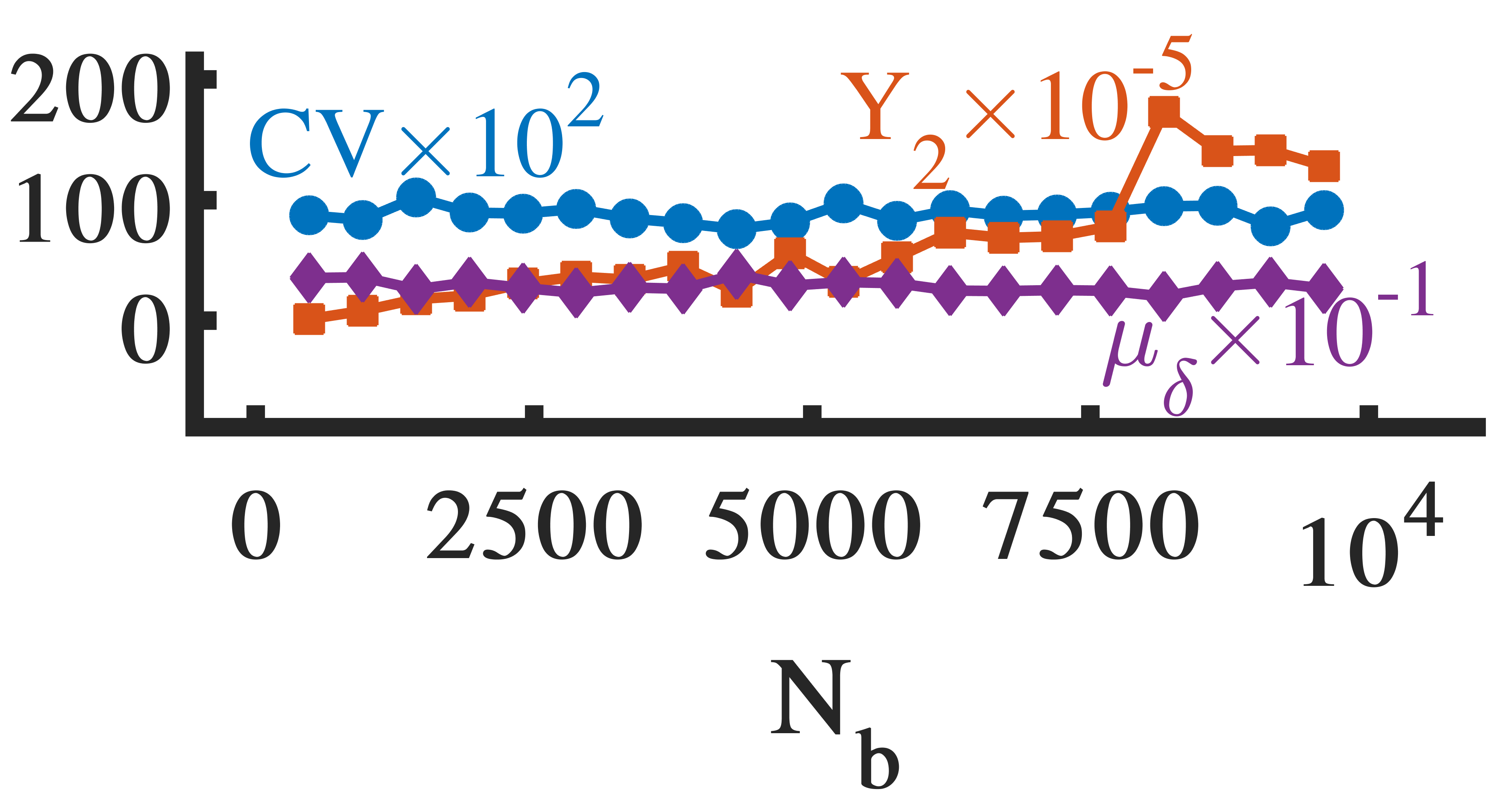}}
  \subfigure[ML1m]{\includegraphics[width=0.24\textwidth]{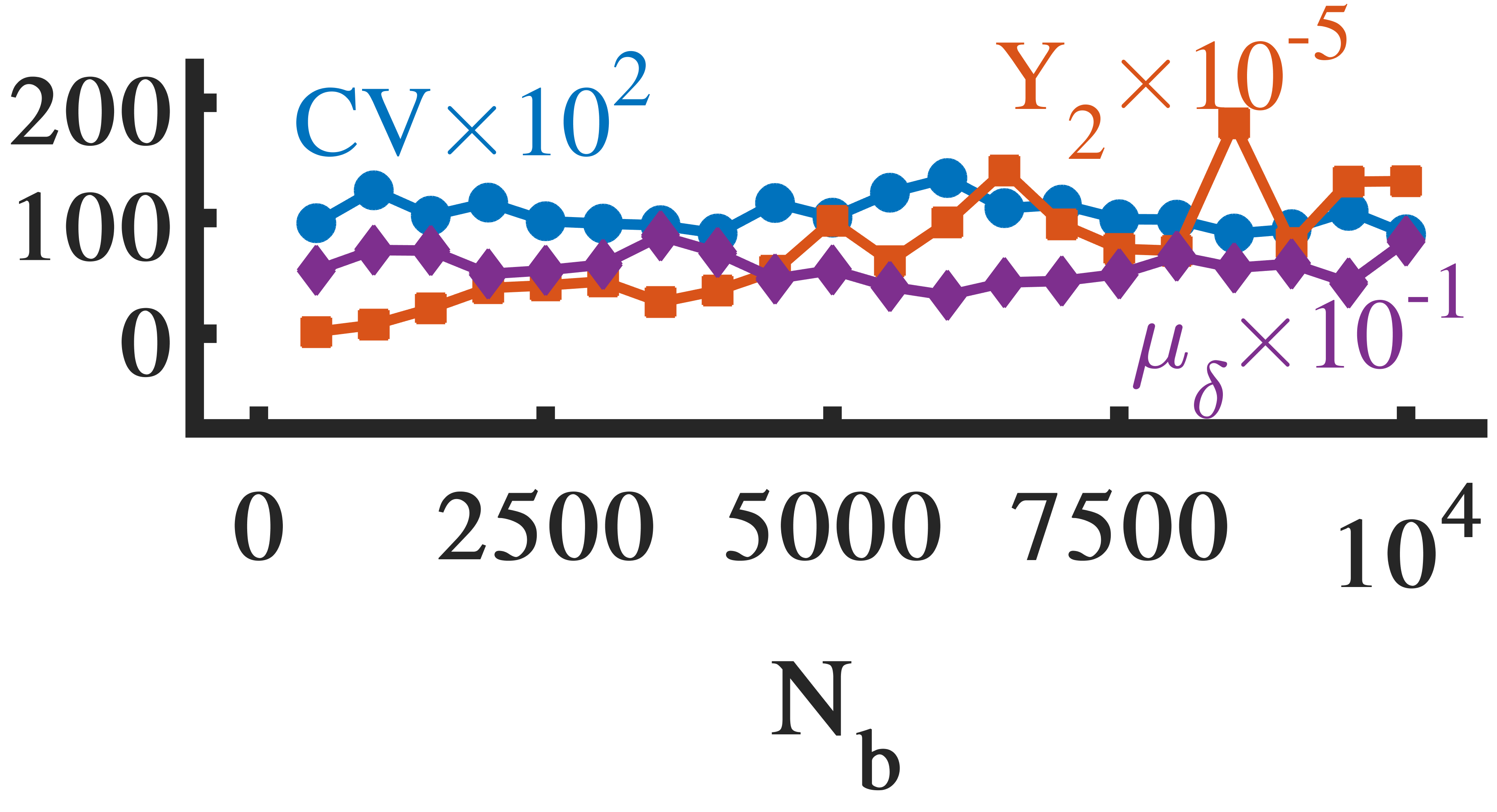}}
  \subfigure[Amazon]{\includegraphics[width=0.24\textwidth]{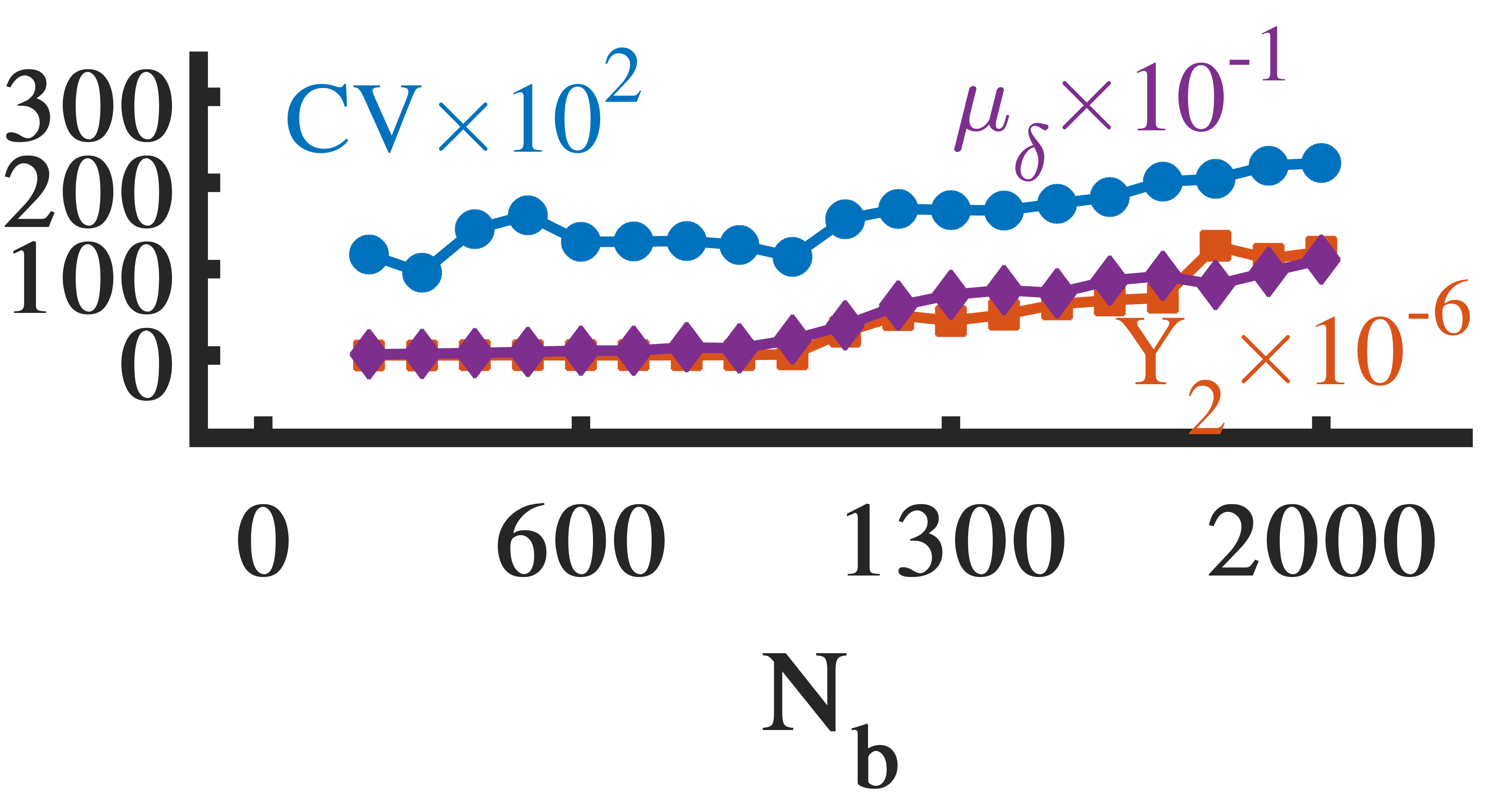}}
  \subfigure[Yahoo]{\includegraphics[width=0.24\textwidth]{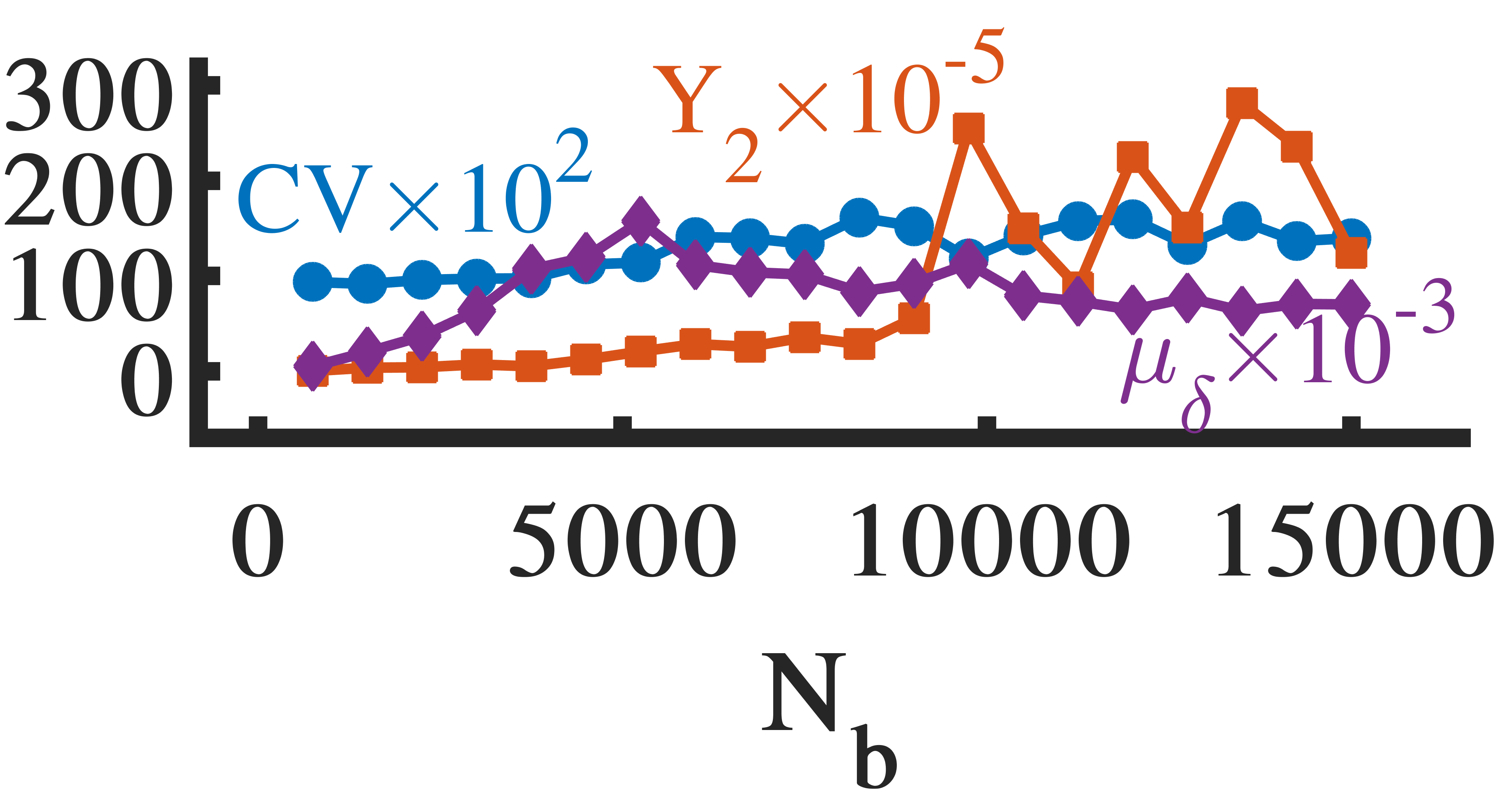}}
    \caption{Coefficient of variation (circles), excess kurtosis (squares), and mean (diamonds) of butterfly strength-differences over the timeline of burst arrivals in real-world streams.} 
    \label{fig:deltastatsReal}
\end{figure*}
\begin{figure*}[h]
    \centering
    \subfigure[$Pr(S_i)$, Ciao]{\includegraphics[width=0.24\textwidth]{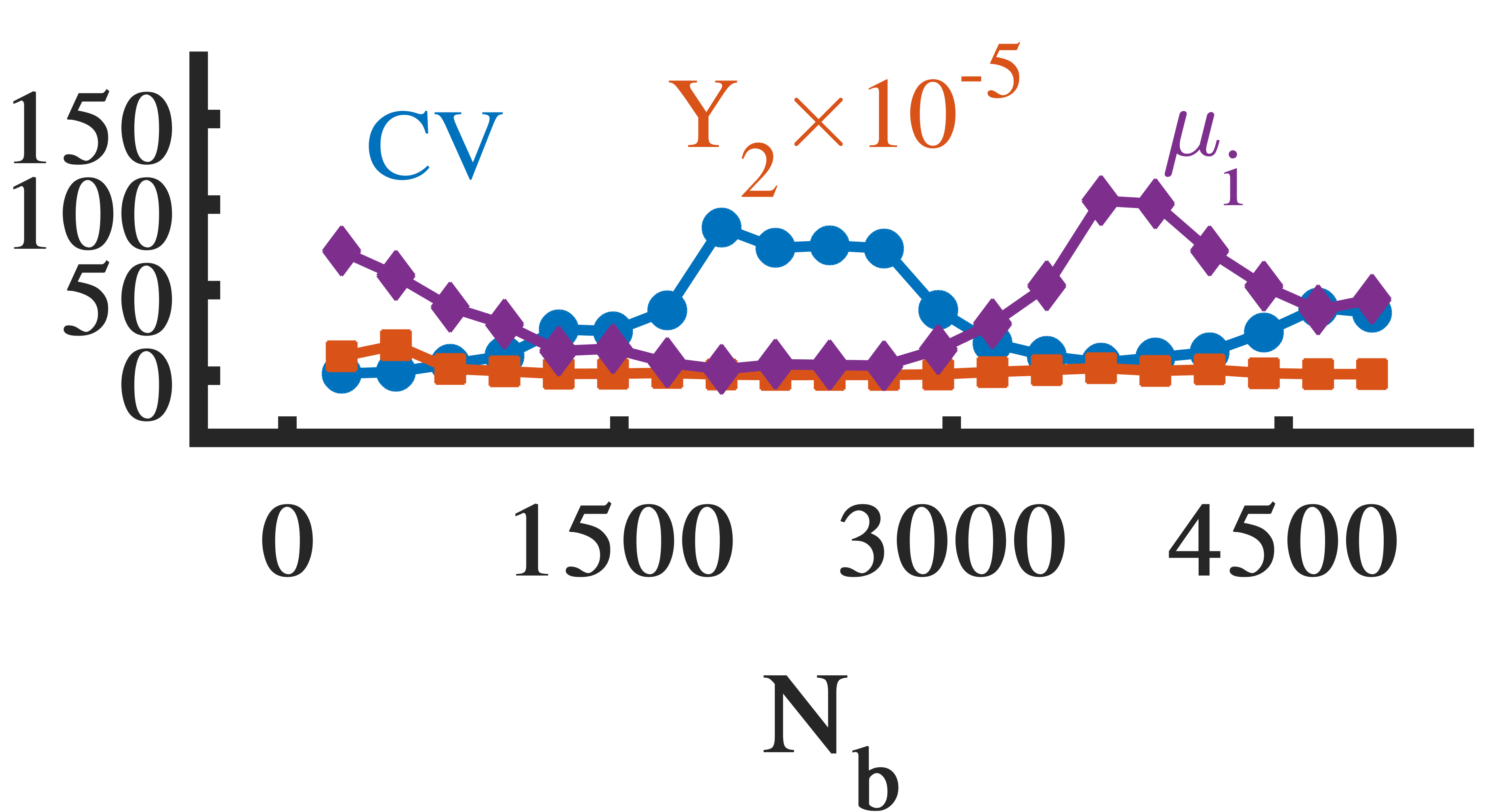}}
    \subfigure[$Pr(S_i)$, Epinions]{\includegraphics[width=0.24\textwidth]{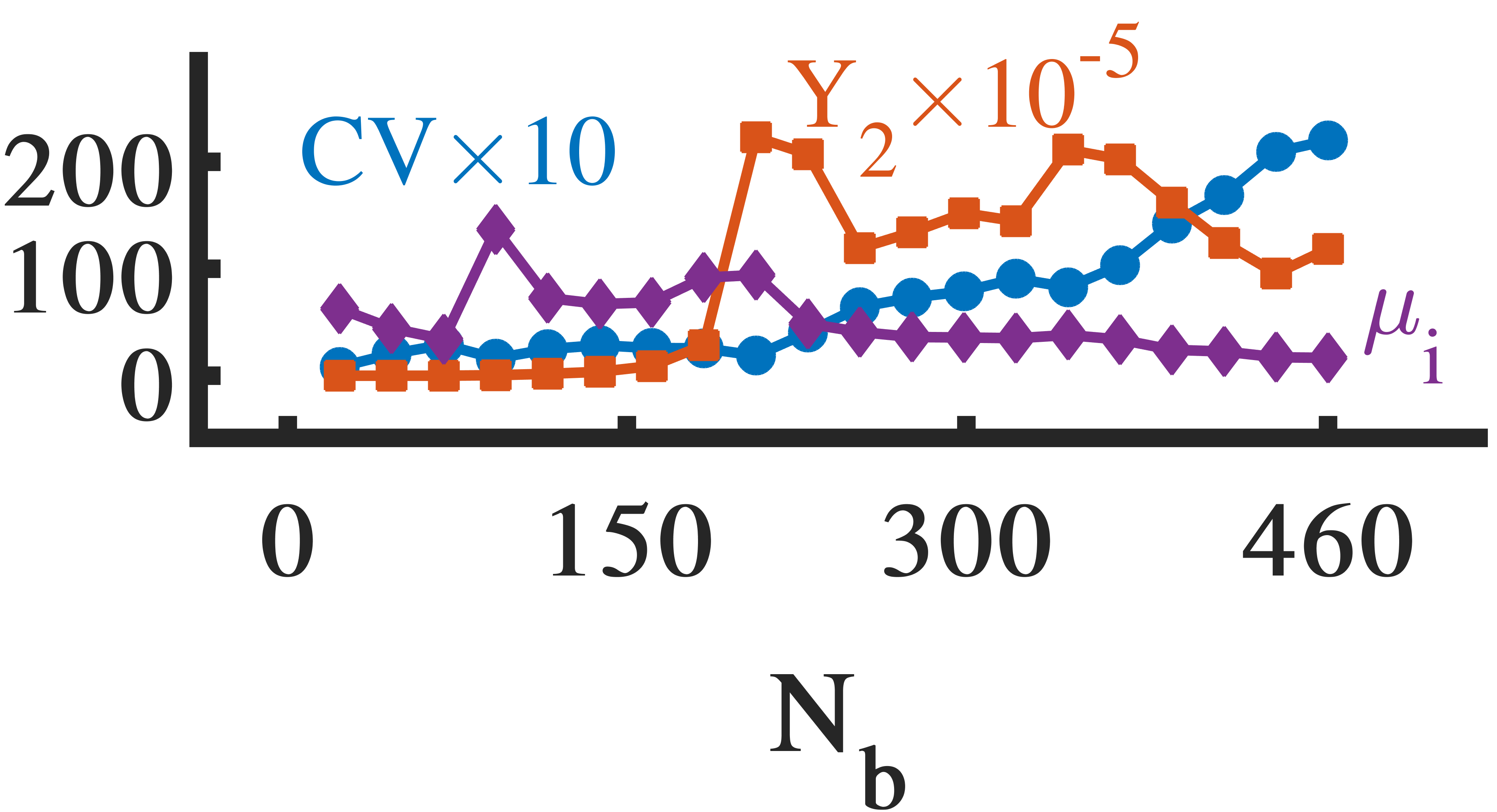}}
    \subfigure[$Pr(S_i)$, WikiLens]{\includegraphics[width=0.24\textwidth]{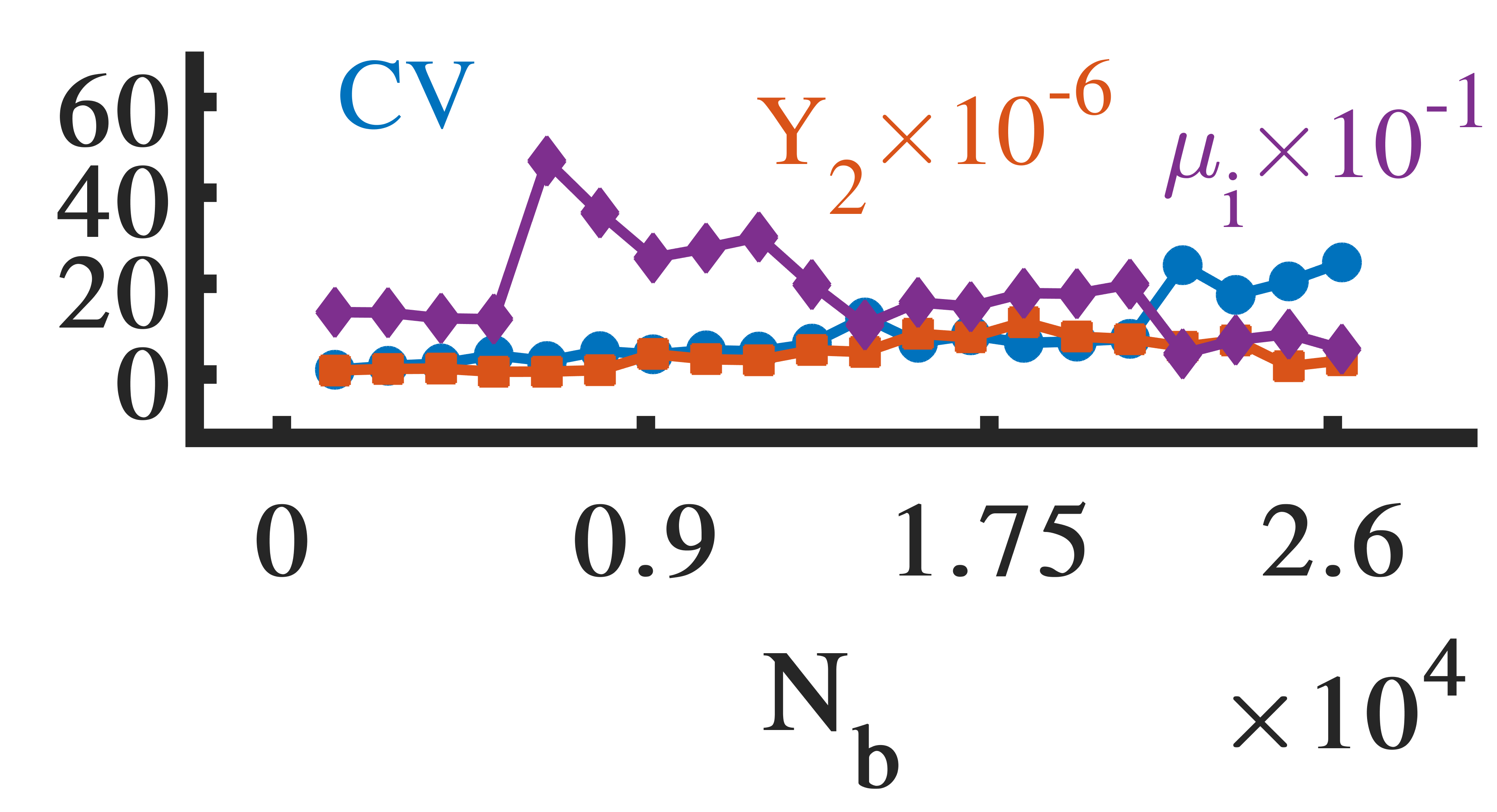}}
    \subfigure[$Pr(S_i)$, ML100k]{\includegraphics[width=0.24\textwidth]{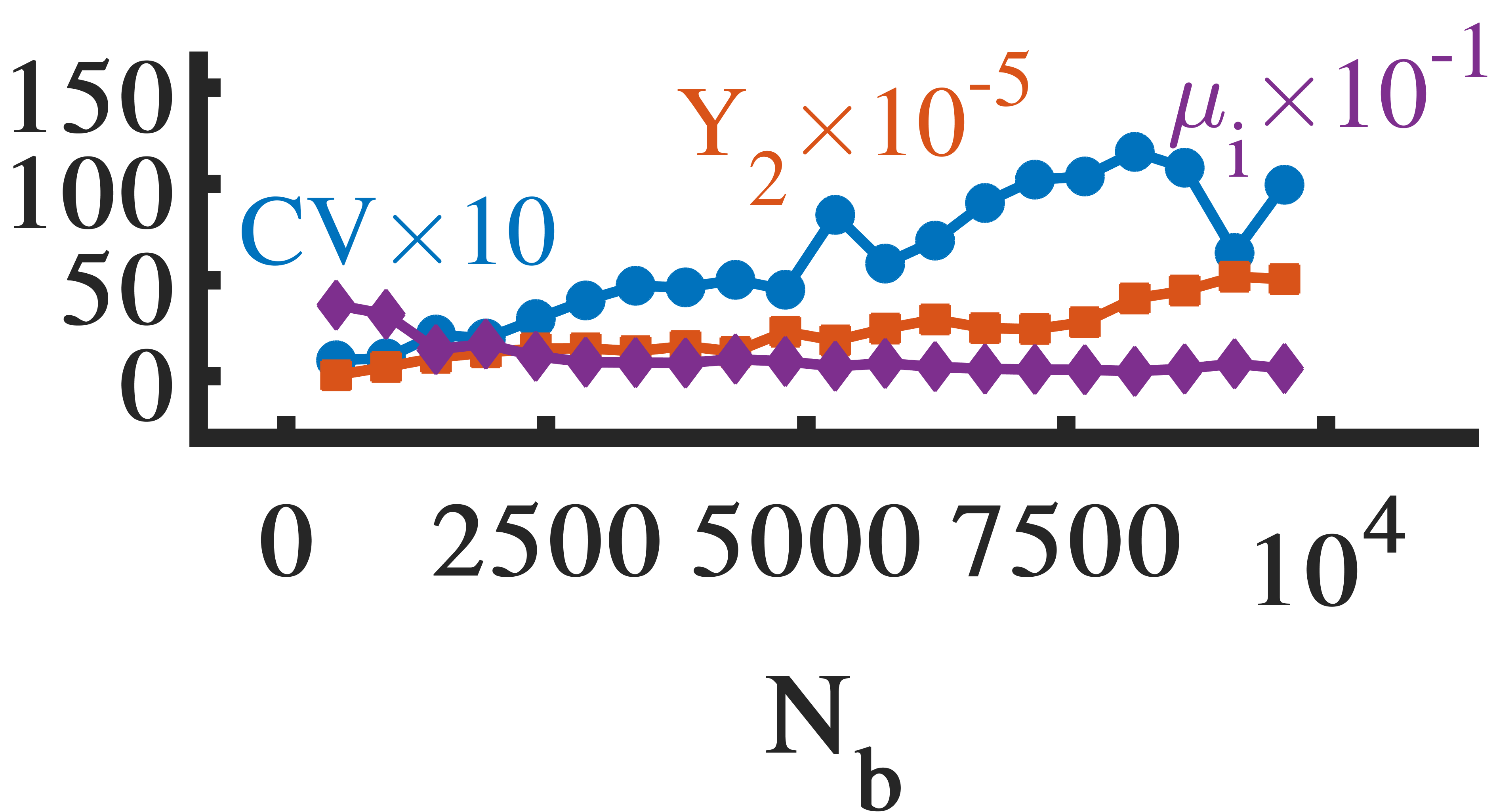}}
    \subfigure[$Pr(S_i)$, ML1m]{\includegraphics[width=0.24\textwidth]{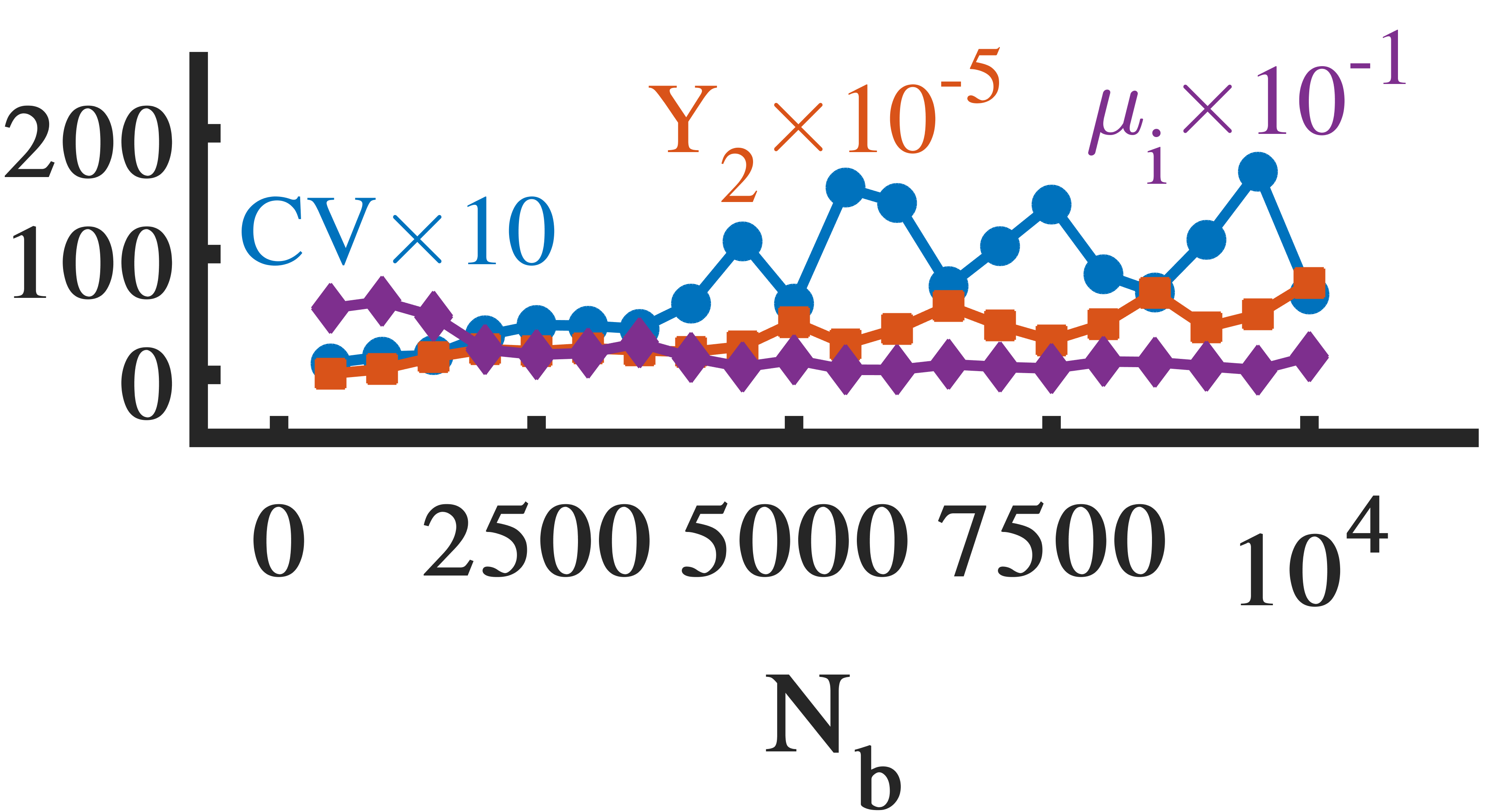}}
    \subfigure[$Pr(S_i)$, Amazon]{\includegraphics[width=0.24\textwidth]{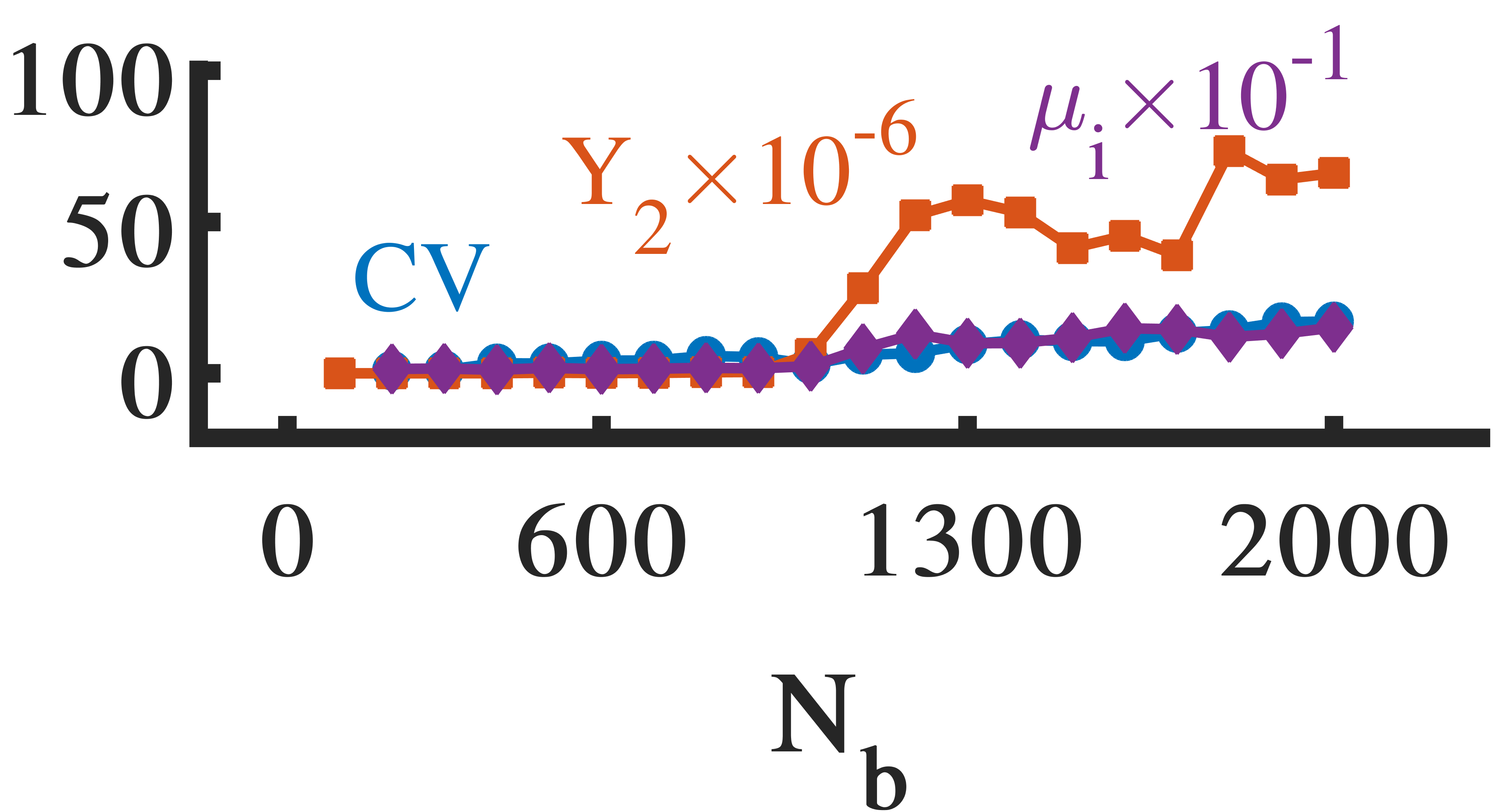}}
    \subfigure[$Pr(S_i)$, Yahoo]{\includegraphics[width=0.24\textwidth]{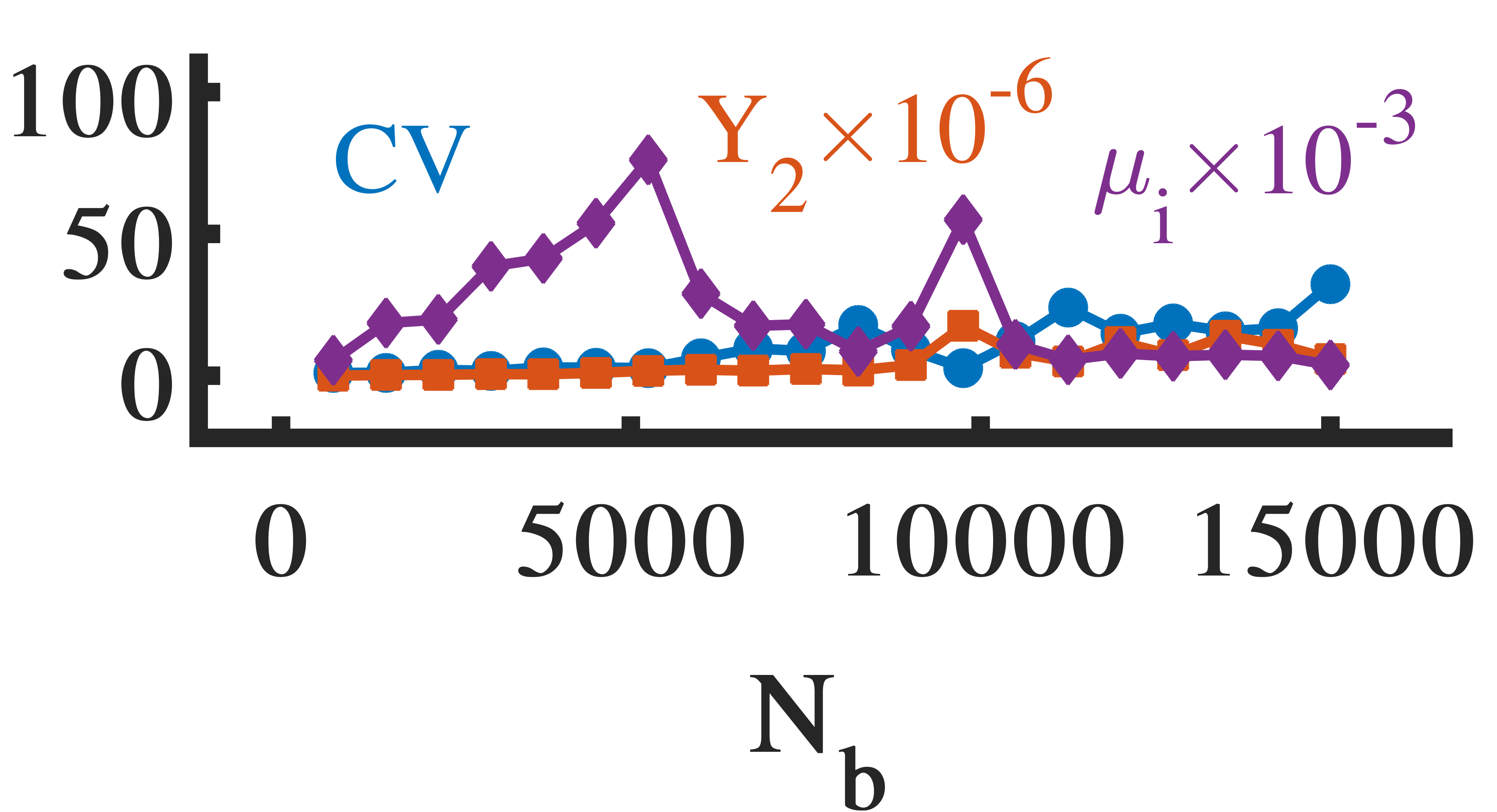}}
    
    \subfigure[$Pr(S_j)$, Ciao]{\includegraphics[width=0.24\textwidth]{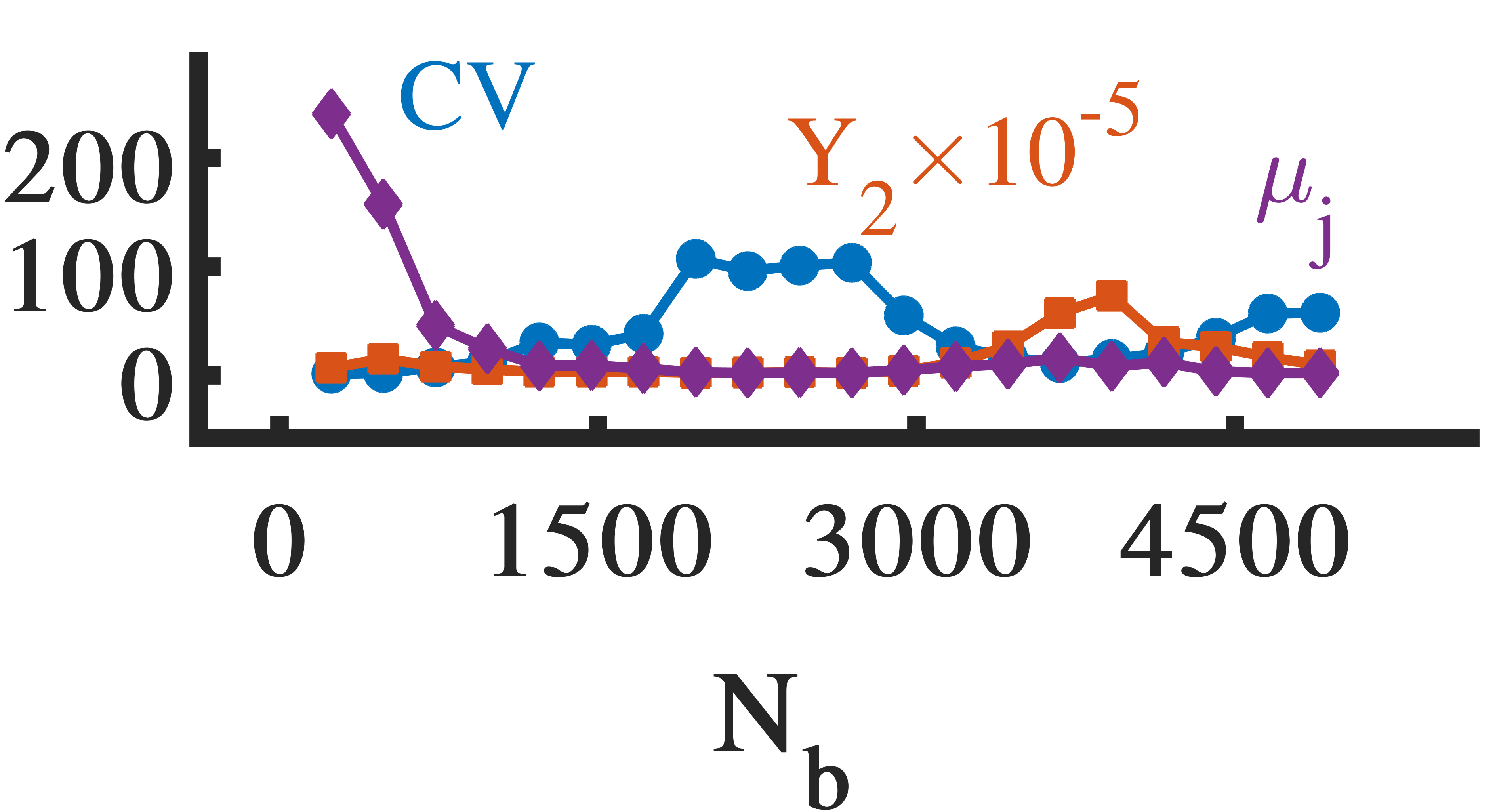}}
    \subfigure[$Pr(S_j)$, Epinions]{\includegraphics[width=0.24\textwidth]{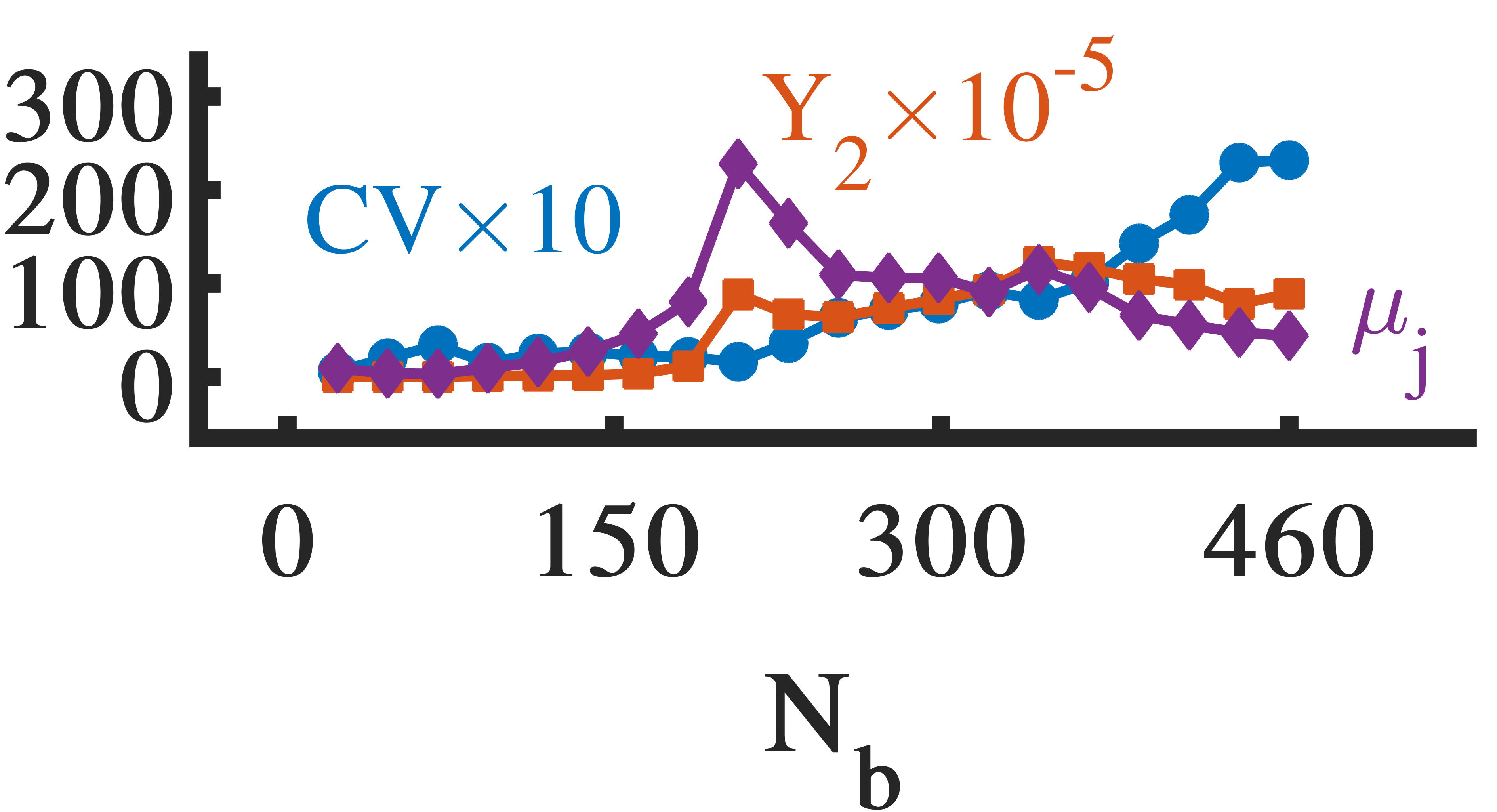}}
    \subfigure[$Pr(S_j)$, WikiLens]{\includegraphics[width=0.24\textwidth]{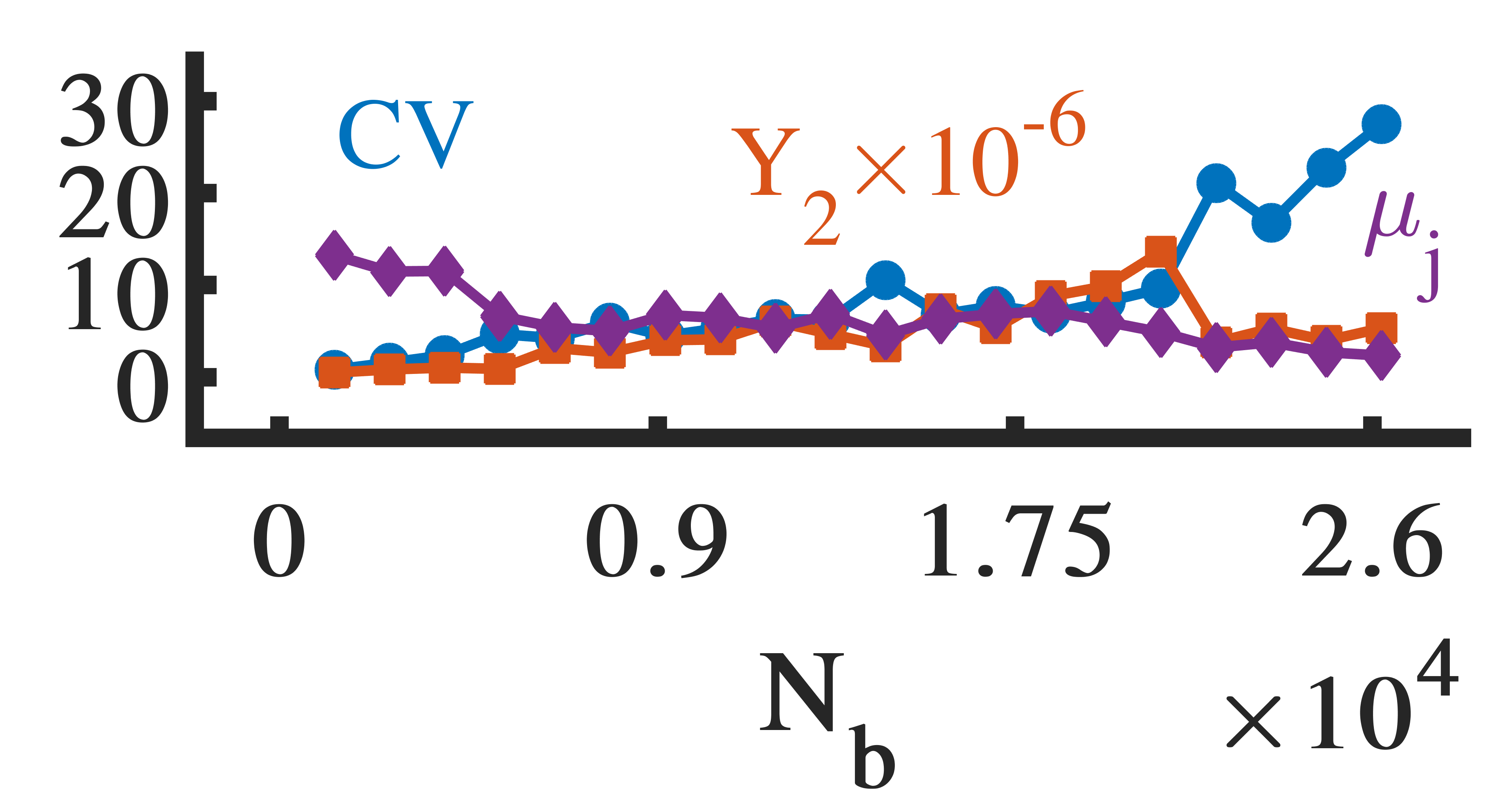}}
    \subfigure[$Pr(S_j)$, ML100k]{\includegraphics[width=0.24\textwidth]{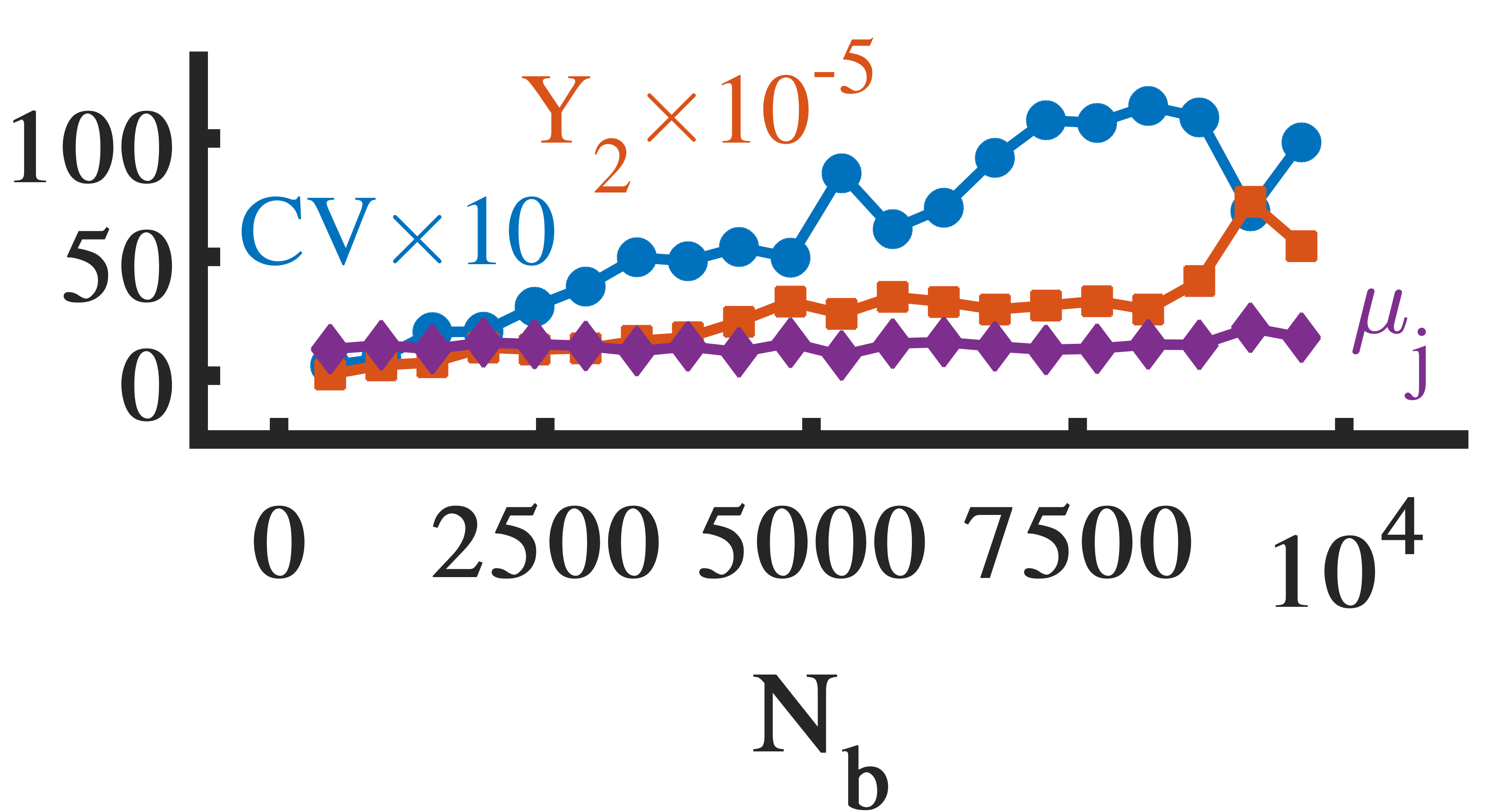}}
    \subfigure[$Pr(S_j)$, ML1m]{\includegraphics[width=0.24\textwidth]{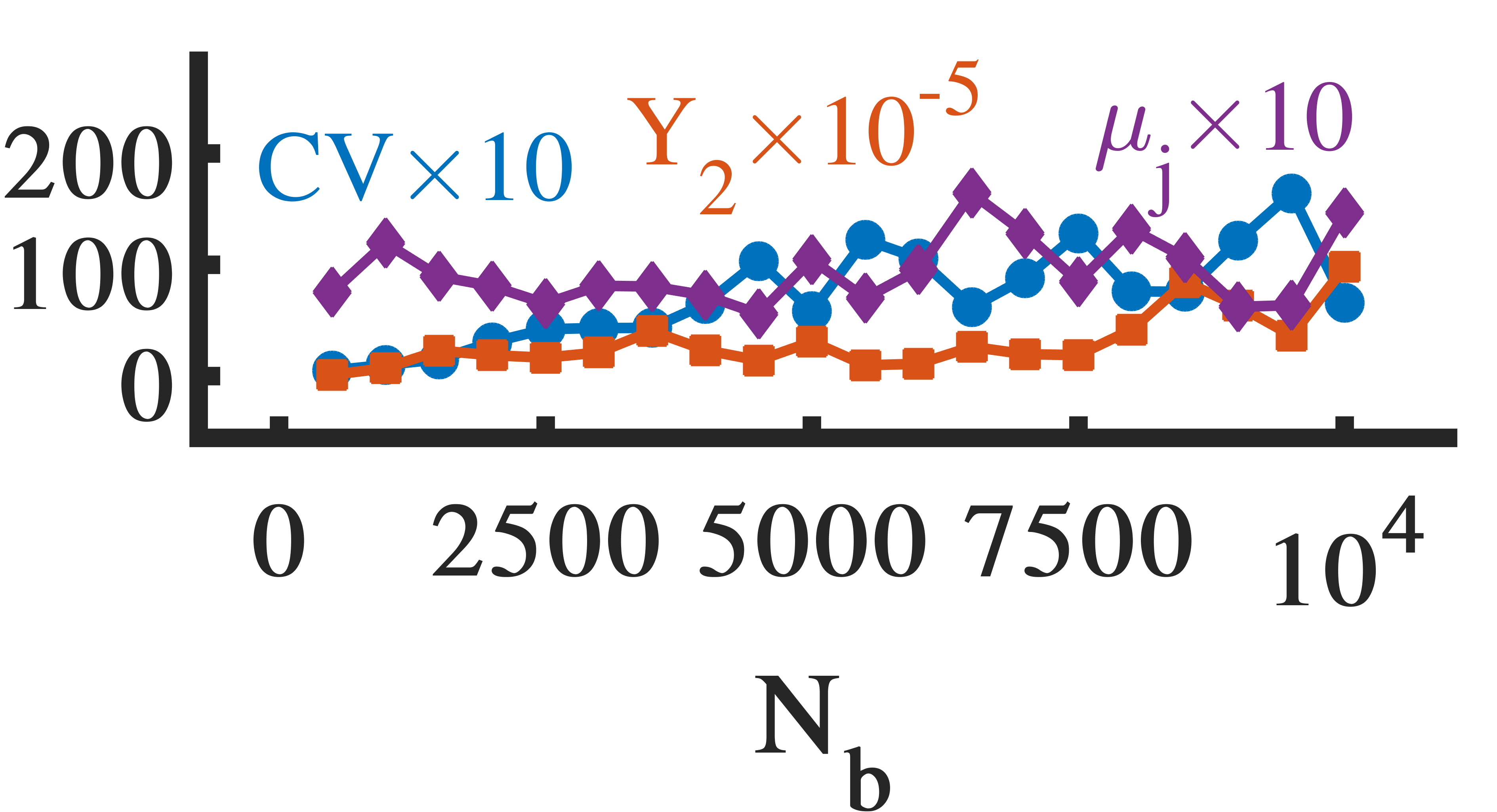}}
    \subfigure[$Pr(S_j)$, Amazon]{\includegraphics[width=0.24\textwidth]{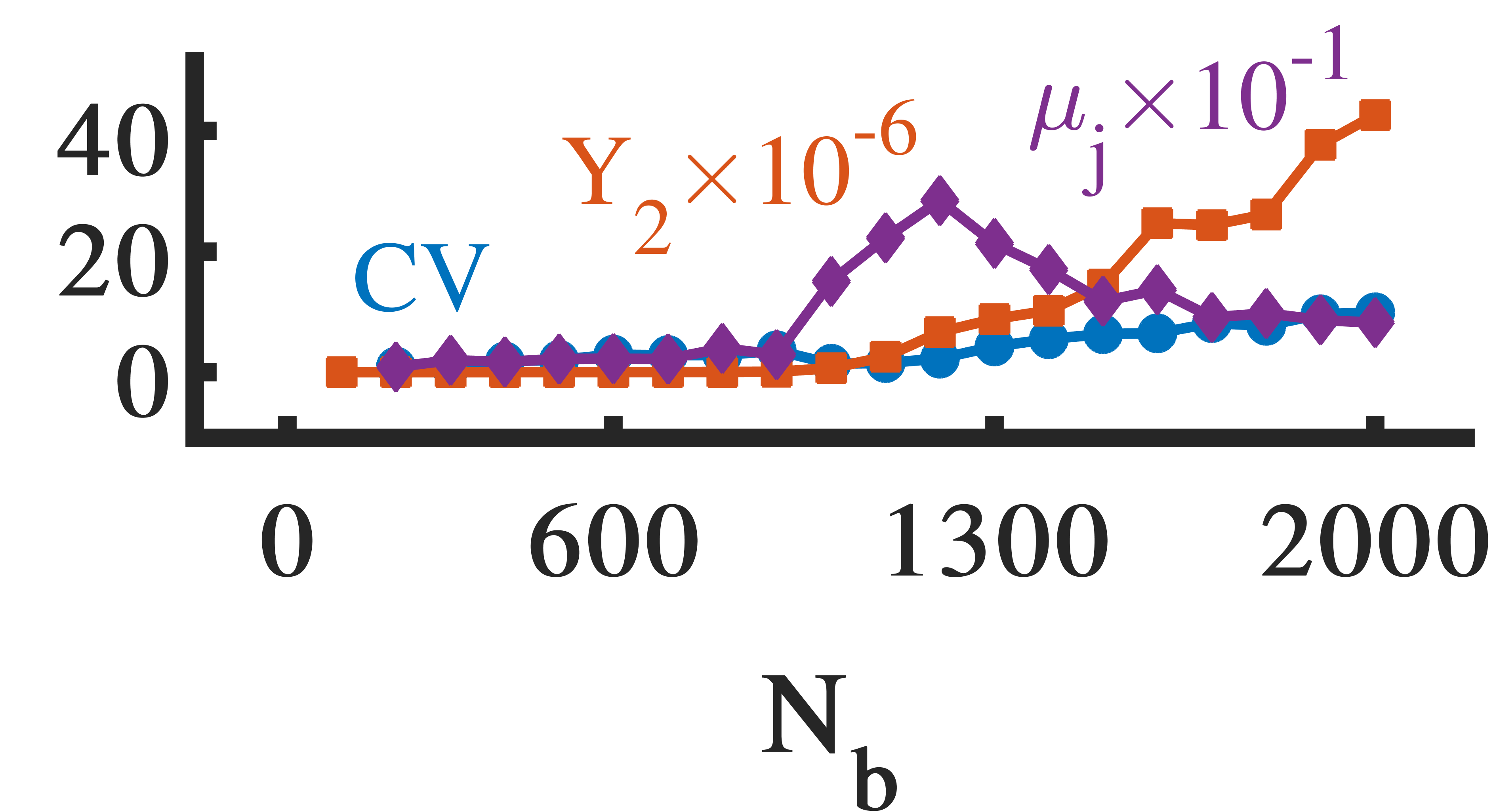}}
    \subfigure[$Pr(S_j)$, Yahoo]{\includegraphics[width=0.24\textwidth]{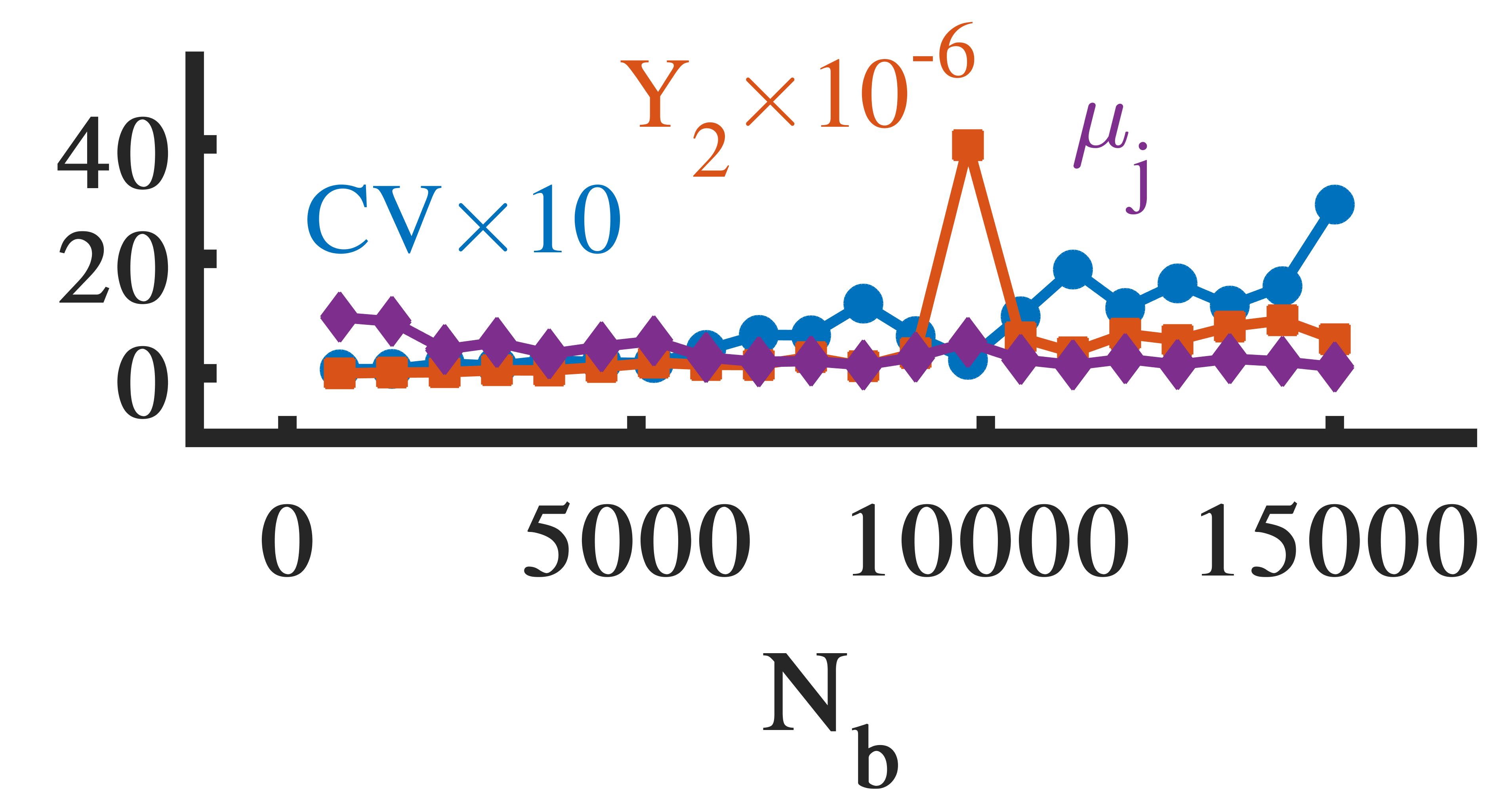}}
    \caption{Coefficient of variation (circles), excess kurtosis (squares), and mean (diamonds) of strengths of butterfly (a-g) i-vertices and (h-n) j-vertices over the timeline of burst arrivals in real-world streams.}
    \label{fig:sstatsReal}
\end{figure*}
Figure~\ref{fig:deltastatsReal} shows the evolution of three statistical quantities for $Pr(\delta)$: (1) mean $\mu_\delta$, (2) coefficient of variation $CV$$=$$\sigma_\delta/\mu_\delta$, and (3) excess kurtosis $Y_2$$=$$(N^{-1}\Sigma_{\delta_i} (\delta_i - \mu_\delta)^4/\sigma_\delta^4)$$-$$3$. Figure~\ref{fig:sstatsReal} shows the evolution of the same quantities for $Pr(S_i)$ and $Pr(S_j)$. $CV$, also known as relative standard deviation (RSD), enables measuring the degree of variation (dispersion) over distributions with different mean values. A high-variance distribution has $CV$$>$$1$ and a low-variance distribution has $CV$$<$$1$. Distributions such as exponential distribution with equal mean and standard deviation have $CV$$=$$1$. The excess kurtosis $Y_2$ enables measuring the heaviness of the tail of distribution relative to a normal distribution (which has $Y_2$$=$$3$). A heavy-tailed distribution has a positive $Y_2$ (called a \textit{leptokurtic} distribution) and a light-tailed distribution has a negative $Y_2$ (called a \textit{platykurtic} distribution). Distributions such as family of normal distributions have zero $Y_2$ (called \textit{mesokurtic}). We observe that the mean and standard deviation of strength-differences are equal to each other and evolve synchronously (Figure~\ref{fig:deltastatsReal} -- $CV$$\approx$$1$ for sequential $G_{N_b}$). On the other hand, the tail of right-skewed $Pr(\delta)$ gets heavier and the distribution gets broader (Figure~\ref{fig:deltastatsReal} -- $Y_2$ increases). In Ciao, the tail gets lighter initially and then gets heavier. Moreover, We observe that all of the graphs have right-skewed $Pr(S)$ which gets broader and more skewed over time with the tail of strength distribution becomes heavier/longer over time (Figure~\ref{fig:sstatsReal} -- $Y_2$ and $CV$$>$$1$ increase).
These observations make the steady behavior of strength assortativity more interesting; despite the fact that new high-strength vertices form butterflies and $Pr(\delta)$ gets broader, the relative standard deviation of $\delta$s does not change significantly and the strength assortativity localization factor $r^s$ remains steadily positive. This 
implies that these graphs obey non-trivial mixing patterns. Understanding these mixing patterns relies on studying the connection micro-mechanics responsible for generating the graphs. Therefore, in the following, we study the growth mechanisms using synthetic graphs and investigate their properties with respect to the observed patterns.
\subsection{Related Works on Modeling of Graph Patterns - A Brief Survey}\label{subsec:relatedworks}
The graph models providing micro-mechanics or high-order generative process of graph structure are generally deemed as the explanation for the patterns observed in real-world graphs~\cite{arnold2021likelihood, leskovec2008microscopic, drobyshevskiy2019random}. We review the graph patterns and the mechanisms used in different graph models to explain the patterns. Our goal is to identify the potential \emph{local rules}~\cite{vazquez2003growing} (attachment mechanisms based on connections to vertices and their neighborhoods) that explain the observed patterns in previous section. Therefore, we focus on growth models that lead to skewed distributions, degree correlation, and emergence of large numbers of cliques. In the next subsection, we analyze the synthetic streams based on candidate mechanisms to check whether they explain the observed realistic patterns.
Our work deviates from the graph models optimized based on downstream analytics task (e.g. link prediction and node classification) and randomizing of the structures in which vertices are associated with feature vectors~\cite{zhou2020data, wang2021bipartite} and also null (feature-driven) graph models~\cite{van2021random} generating prescribed patterns/features such as degreewise metamorphis coefficient~\cite{aksoy2017measuring}, degree distribution~\cite{aksoy2017measuring, newman2001random, fosdick2018configuring, stanton2012constructing}, subgraph distribution~\cite{wegner2011random}, and k-core sequence~\cite{van2021random}. We refer to~\cite{drobyshevskiy2019random, chakrabarti2006graph}  and references therein for surveys of these graph models.

\textbf{Graph Patterns.} Graph patterns characterize a microscopic, mesoscopic, or macroscopic property of a graph (depending on the granularity of the reporting pattern, i.e. vertices/edges, neighborhoods and motifs, or the entire topology) and can be viewed as either static or dynamic (depending on the underlying graph being a static snapshot or an evolving structure). Examples of static patterns include small diameter accompanied by high clustering coefficient~\cite{watts1998collective}, degree (anti)correlation~\cite{newman2002assortative}, community structure~\cite{girvan2002community}, and power-laws (PL) such as degree distribution PL~\cite{barabasi1999emergence}, weight PL~\cite{mcglohon2008weighted}, and snapshot/vertex strength PL~\cite{mcglohon2008weighted, barrat2004weighted}. 
Examples of dynamic patterns include gelling points~\cite{mcglohon2008weighted}, increasing average degree, shrinking/controlled diameter, edge densification~\cite{leskovec2005graphs, leskovec2007graph, fischer2014dynamic}, and bursty weight addition~\cite{mcglohon2008weighted}. Table~\ref{tab:patterns} provides instances of different patterns partitioned across dynamism and granularity. 
\begin{table}[]
\caption{Graph Patterns. CC and PL refer to clustering coefficient and power law, respectively.}\label{tab:patterns}
    \centering
    \begin{tabular}{p{5cm} p{2cm} p{1.7cm} p{3.2cm}}
        Patterns & Granularity & Dynamism & References\\ \hline
        Small diameter \& High CC & Macroscopic & Static & \cite{watts1998collective}\\
        Degree Correlation & Microscopic & Static & \cite{newman2002assortative}\\
        Community structure & Mesoscopic & Static & \cite{girvan2002community} \\
        Degree Distribution PL & Microscopic & Static & \cite{barabasi1999emergence}\\
        Weight PL & Microscopic & Static & \cite{mcglohon2008weighted} \\
        Snapshot/Vertex Strength PL & Mesoscopic & Static & \cite{mcglohon2008weighted, barrat2004weighted}\\
        Gelling points & Macroscopic & Dynamic & \cite{mcglohon2008weighted} \\
        Increasing Average Degree & Microscopic & Dynamic & \cite{leskovec2005graphs, leskovec2007graph} \\
        Shrinking/Controlled Diameter & Macroscopic & Dynamic & \cite{leskovec2005graphs, leskovec2007graph, fischer2014dynamic} \\
        Edge Densification & Microscopic & Dynamic & \cite{leskovec2005graphs, leskovec2007graph} \\
        Bursty Weight Addition & Microscopic & Dynamic & \cite{mcglohon2008weighted} \\
    \end{tabular}
\end{table}

\textbf{Growth Models.} Two well-studied network growth mechanisms, preferential attachment and copying, form the basis of many generative models ~\cite{do2020structural, cooper2003general,bianconi2011competition, akoglu2009rtg, arnold2021likelihood, guillaume2006bipartite, kumar2000stochastic, krapivsky2005network, hadian2016roll, park2021lineageba} and are widely adopted in development of graph management approaches~\cite{huang2016leopard, chang2012exact,ma2019linc, jin2010computing, liu2008towards, mondal2012managing, yang2012towards}.  Both mechanisms are commonly applied in graph models based on the conception of adding a new vertex at each time step during an iterative process. Preferential attachment leads to skewed distributions, while copying mechanism leads to degree correlation~\cite{vazquez2003growing} and emergence of large numbers of cliques when applied explicitly~\cite{kumar2000stochastic} or implicitly and among other mechanisms~\cite{leskovec2005graphs}. In the following, we review these mechanisms and their alternatives and extensions.

Barabasi-Albert model~\cite{barabasi1999emergence} starts with a small clique with $m_0$ vertices and applies the preferential attachment by connecting the new vertex to $m$$\leq$$m_0$ existing vertices selected randomly with probability proportional to their degrees. The preferential attachment rule has also been extended to strength-driven preferential attachment (SPA) where each new vertex is connected to $m$ existing vertices randomly selected with probability proportional to their strength~\cite{barrat2004modeling,barrat2004weighted,leung2007weighted}. It has been shown that preferential attachment is induced by the following microscopic mechanisms. All of these mechanisms imply that the probability that a vertex receives a new edge is proportional to its degree, therefore they amount to preferential attachment and lead to scale-free structures~\cite{albert2002statistical}. 
\begin{itemize}
    \item Copying~\cite{kleinberg1999web, kumar2000stochastic}: at every time step, a new vertex is connected to a constant number of vertices and the end point of each new edge is a randomly selected vertex with probability $p$ or a neighbor of a prototype vertex with probability $1$$-$$p$.
    \item Edge redirection~\cite{krapivsky2001degree}: at every time step, a new vertex is added and a directed edge from the new vertex to a randomly selected vertex is created with probability $1-p$, or the edge is redirected to the ancestor of the randomly selected vertex with probability $p$.
    \item Random walks~\cite{vazquez2000knowing}: at every time step, a new vertex is connected a random vertex and the vertices reachable from it through breadth-first traversal with probability $p$ until no new target is found.
    \item Attaching to edges~\cite{dorogovtsev2001scaling}: at every time step, a new vertex is connected to two connected vertices. 
\end{itemize}

The original version of copying, mentioned above, copies a neighbor of a randomly selected vertex with some probability at each time step. Other works have modified it as the following. 
\begin{itemize}
    \item Butterfly model~\cite{mcglohon2008weighted} mixes copying and random walk mechanisms: at every time step,  with probability $p_{host}$ a new vertex picks a random vertex called host and with probability $p_{link}$ forms edges with the vertices reachable from the host through a probabilistic random walk with traversal probability $p_{step}$. This model exhibits shrinking diameter, stabilized next-largest weakly connected component size, and edge densification.
    \item Growing network model with copying~\cite{krapivsky2005network} connects the new vertex to a randomly selected vertex as well as its neighbors which leads to sparse ultrasmall graphs with logarithmic growth of the average degree wrt the number of vertices while the diameter equals 2.
    \item Duplication divergence model~\cite{vazquez2003growing} removes the copied neighbors with some probability leading to power law decay of clustering coefficient as a function of degree.
    \item Nearest neighbors model~\cite{vazquez2003growing} connects the new vertex to one randomly selected vertex and copies one neighbor with some probability leading to clustering coefficient power law and correlation between average neighbor degree and vertex degree.
    \item Forest Fire model~\cite{leskovec2005graphs, leskovec2007graph} applies the copying process by recursively, connecting each new vertex to a randomly selected vertex and certain numbers of its randomly selected out- and in-neighbors with forward probability $p$ and backward probability $p_b$. This process leads to heavy-tailed in- and out-degree distributions due to an implicit preferential attachment, community structures due to neighbor copying mechanism~\cite{barrat2004weighted}, edge densification due to many internal-edge establishments, and shrinking diameter due to shortcut-edge establishments.
\end{itemize}

\textbf{Our Work.} We identify/explain mesoscopic dynamic patterns in weighted bipartite streaming graphs. To apply the local rules in which new vertices connect to target vertices and their neighbors, it is important to decide when and how many target vertices are selected, how to select the target vertices, and how to copy their neighbors. As we will discuss in the next section, in our growth model, we perform a preferential random walk (PRW) with a dynamic and randomized length bounded to a parameterized range and combined BFS and DFS traversals. We use this PRW as a backbone including the target i- and j-vertices. That is, the number of target vertices in each iteration is non-static and the target vertices are selected via strength-driven preferential method. Given a new edge between $v_i$ and $v_j$, we connect each target i-vertex $u_i$ to $v_j$ and we copy each $u_i$'s neighbor as a neighbor of $v_i$, with a parameterized probability $\rho$. We also connect each $u_i$ to a j-vertex $z_j$ selected uniformly at random. We perform this procedure for target j-vertices as well. This procedure is done during a burst addition mechanism in which a batch of new isolated edges with a randomized size $m$, bounded to a parameter value, are added to the graph. This burst addition adds burst per new vertex and per iteration. We further optimize the burst addition mechanism with realistic assignment of weights and timestamps to the edges to ensure realistic connection iterations. We also enforce a time-based filtering on the neighborhoods by using a sliding window over the computational graph.

\subsection{Analysis of Microscopic Growth Mechanisms}\label{subsec:syntheticObservations}
As we discussed above, preferential attachment leads to skewed distributions, and copying mechanism (particularly with the implementation scheme of Forest Fire model) leads to degree correlation and emergence of cliques and edge densification as well. Therefore, in an attempt to explain the origins of the observed patterns in real-world streams, we investigate the properties of synthetic graph streams generated by these local rules. We synthesize weighted bipartite streaming graphs such that the graph structure grows according to the Forest Fire (FF) and strength preferential attachment (SPA) models. To this end, we create directed graphs via the growth models and treat the source vertices as the i-vertices and destination vertices as the j-vertices. For the timestamp assignment, we use the time step at which new vertices are connected to existing vertices and for the weight assignments, we use random integers in the range $[1,5]$ (the same weight scale as in real-world streams). In Forest Fire model, when the backward-burning probability $p_b$ is fixed and the forward-burning probability $p$ increases, the graphs become denser and more clique-like with low diameter~\cite{leskovec2007graph}. Therefore, we generated graphs with fixed $p_b=0.3$ and $p=0.15$ (sparse region), $0.4$ (transition region), and $0.7$ (dense region). Our experiments show that most of the edges are burned (visited) after checking the neighbors of the ambassador vertex, therefore we do not check further edges to reduce computations and also allow addition of new external links beside the internal densification.  In the SPA model, we use $m\in\{10,50,100\}$ since the average degree of vertices in real-world streams are mostly below $100$ (Table~\ref{tab:graphs}). We use the same analytical approach as for real-world streams to investigate the emergence patterns of butterflies in the synthetic streams quantitatively (by checking the growth patterns of butterfly count) and qualitatively (by checking the assortativity patterns of butterflies and the confounding distributions). 

As shown in Figure~\ref{fig:butterflycountSynthetic}, 
the butterfly count has a slow growth in FF streams and a speedy growth in SPA streams. The average butterfly rate is less than $1$ in FF streams and higher than $500$ in SPA streams. That is, the growth of butterfly count  with respect to the number of edges in the sequential graph snapshots is sub-linear in FF streams and extremely super-linear in SPA streams.

Figure~\ref{fig:rsandFforestfire} shows the evolution of $r^s$ and the corresponding $F$ elements for the three parameter regions in FF streams. As the graph grows, in the transition region, the assortativity level fluctuates, and, in the sparse and dense regions, it changes trivially ($0.02$). We check the statistics of the corresponding  $Pr(\delta)$s in Figure~\ref{fig:sdifstatsforestfire} and that of $Pr(S_i)$s and $Pr(S_j)$s in Figure~\ref{fig:sstatsforestfire}. 
Although in FF stream butterflies emerge such that the range of $Pr(S_i)$ and $Pr(S_j)$ get broader and more skewed over time, strength-difference of butterfly edges retain the same distribution. $Pr(\delta)$ remains unchanged with a low dispersion ($CV$$<$$1$) as the graph grows in each region since the mean, and standard deviation are fixed and the tail changes slightly over  time. Therefore, it is not surprising that $r^s$ is stable. Moreover, in the denser graphs with more butterflies, assortativity patterns vanish ($r^s$$\xrightarrow{}$$0$ as $p$ increase).

Figure~\ref{fig:rsandFSPA} shows the evolution of $r^s$ and the corresponding $F$ elements in SPA streams. For small values of $m$, there is no assortativity pattern. For $m=100$, the graph snapshots display weak assortativity ($0.05$$\leq$$r^s$$<$$0.1$). We check the statistics of the corresponding  $Pr(\delta)$s in Figure~\ref{fig:sdifstatsSPA} and that of $Pr(S_i)$s and $Pr(S_j)$s in Figure~\ref{fig:sstatsSPA}. As the graph grows, for all values of $m$, diversity of the strength of i- and j-vertices does not change significantly (small change in $\mu_i$, $\mu_j$, and  corresponding $CV$ and $Y_2$). The strength-differences continuously follow a skewed distribution with a short tail as most $\delta$s remain around the mean ($F_1+F_2$$\approx$$0.85$ and $CV$$<$$1$) and the skewness does not grow to very high numbers.
\begin{figure*}[ht]
    \centering
  \subfigure[FF, $p$$=$$0.15$, $0.04$$\pm$$0$]{\includegraphics[width=0.24\textwidth]{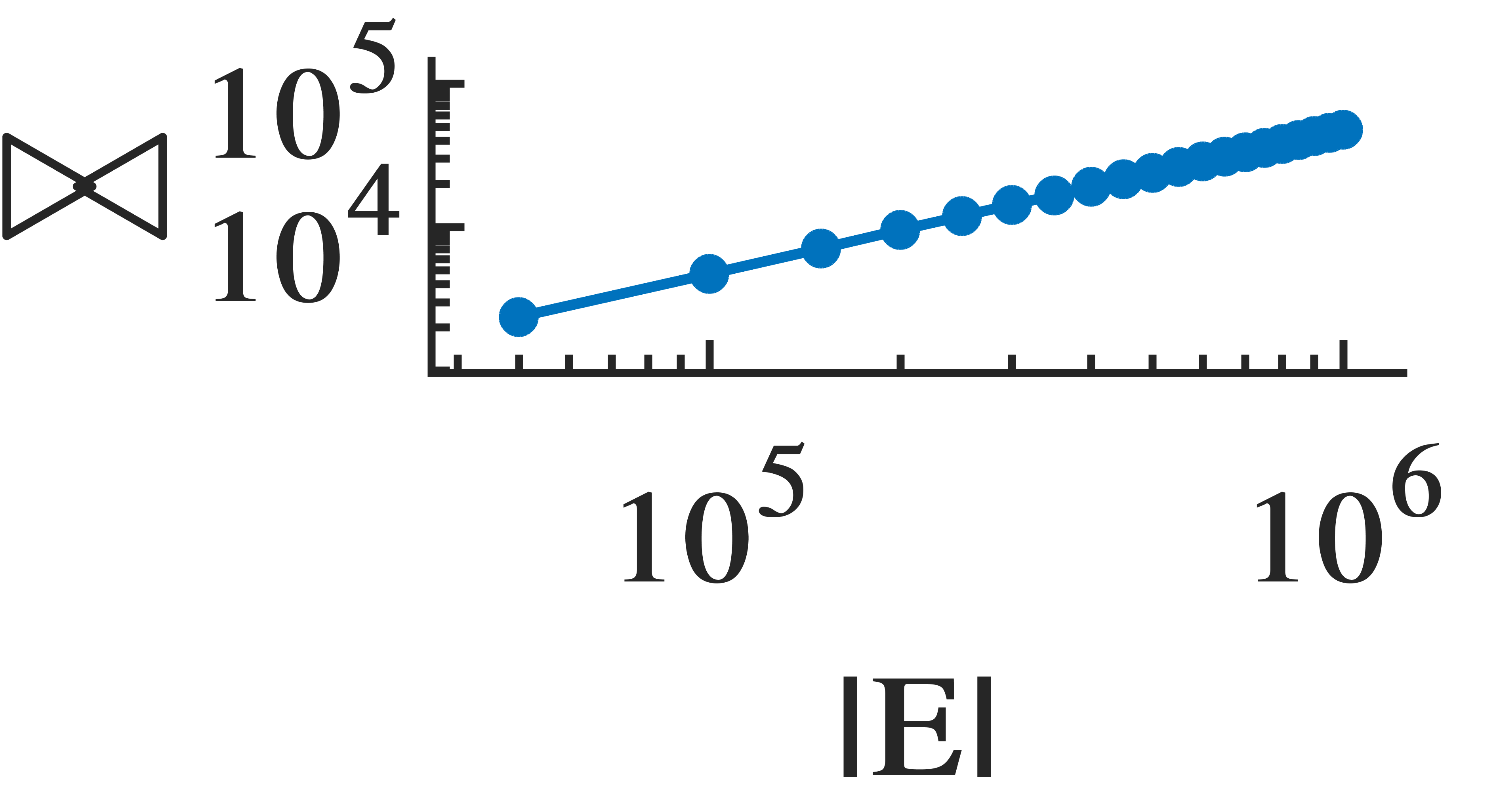}}
  \subfigure[FF, $p$$=$$0.4$, $0.08$$\pm$$0$]{\includegraphics[width=0.24\textwidth]{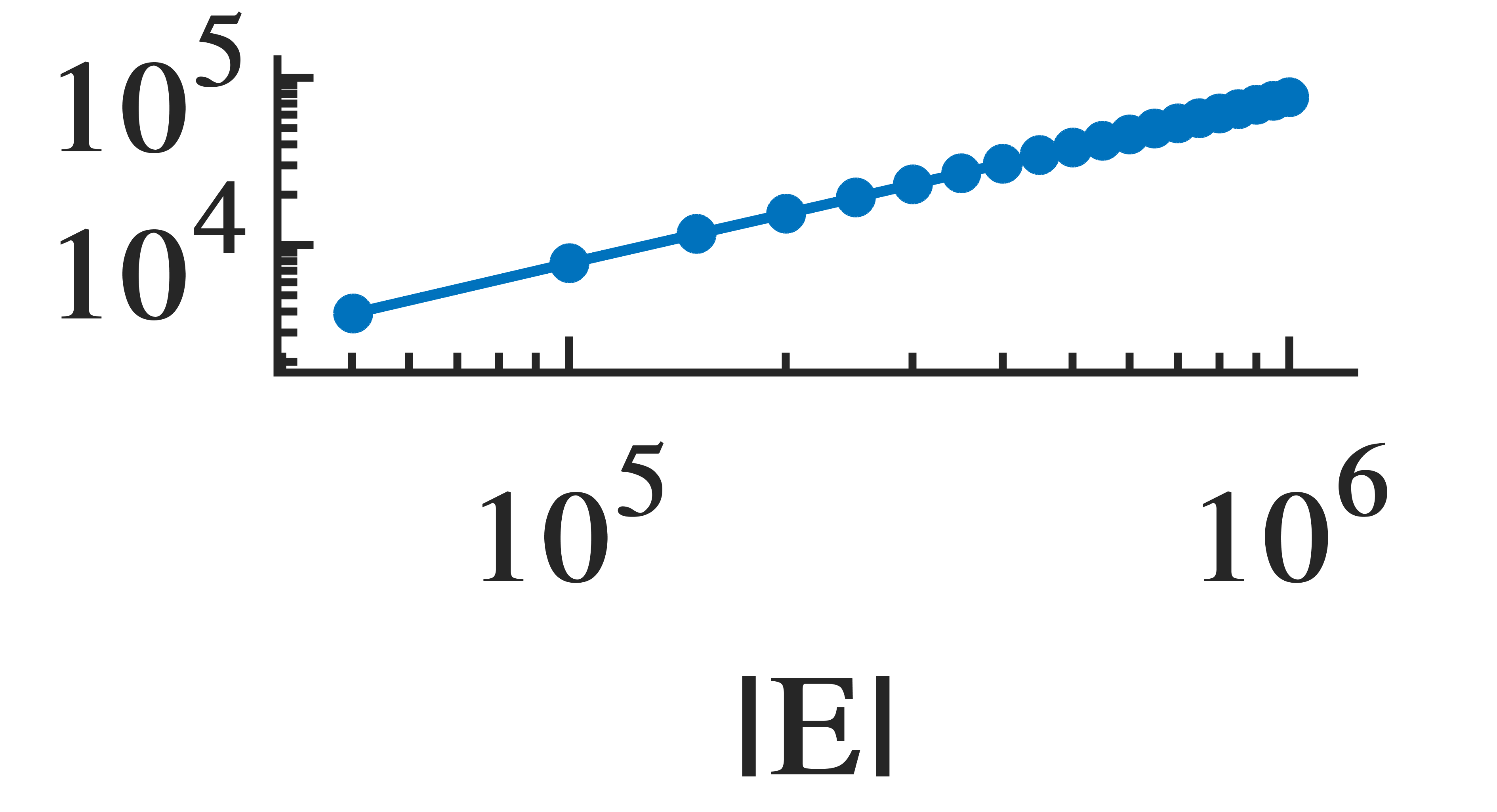}}
  \subfigure[FF, $p$$=$$0.7$, $0.09$$\pm$$0$]{\includegraphics[width=0.24\textwidth]{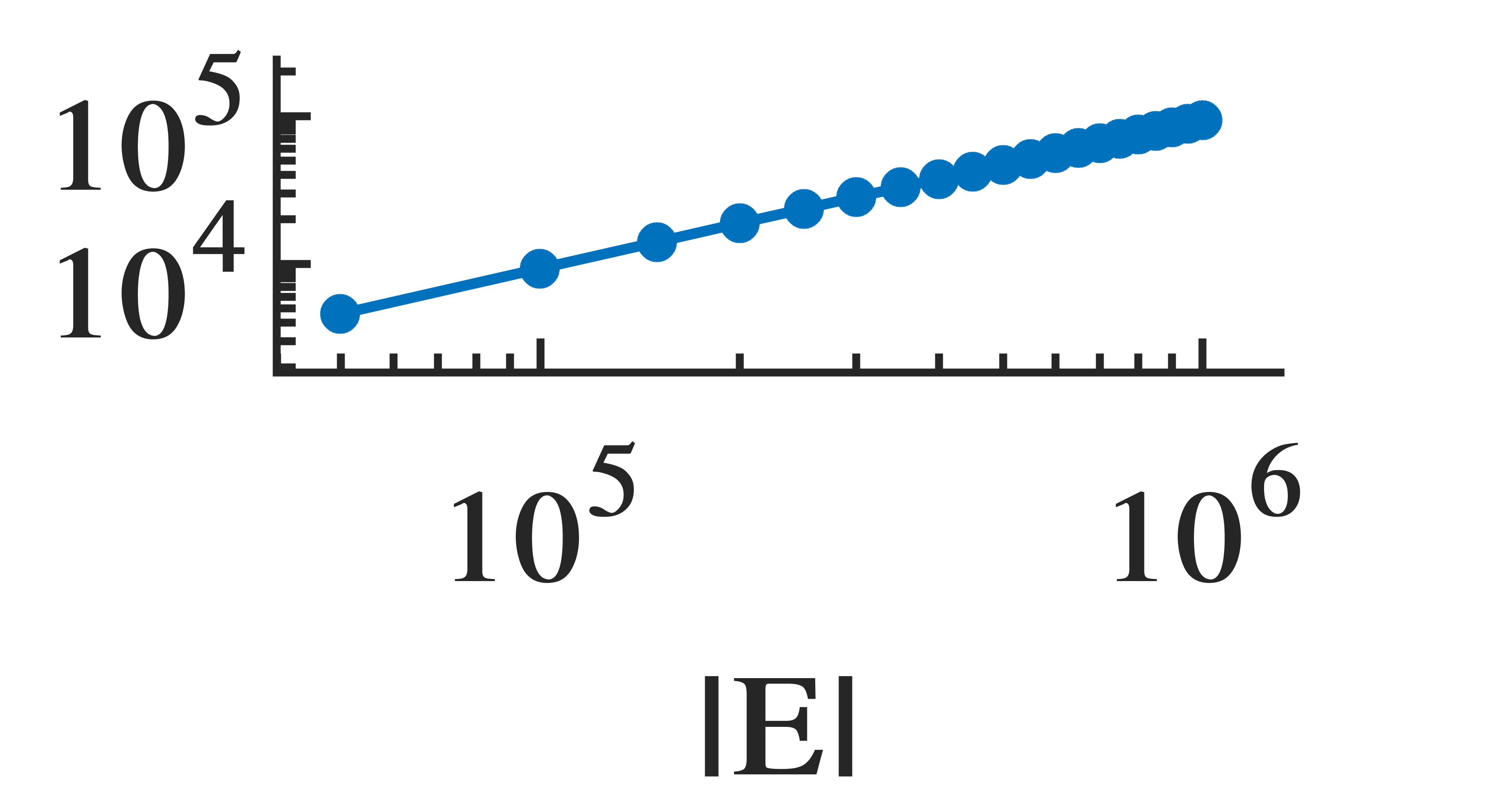}}
  
  \subfigure[SPA, $m$$=$$10$, $556.6$$\pm$$24.5$]{\includegraphics[width=0.24\textwidth]{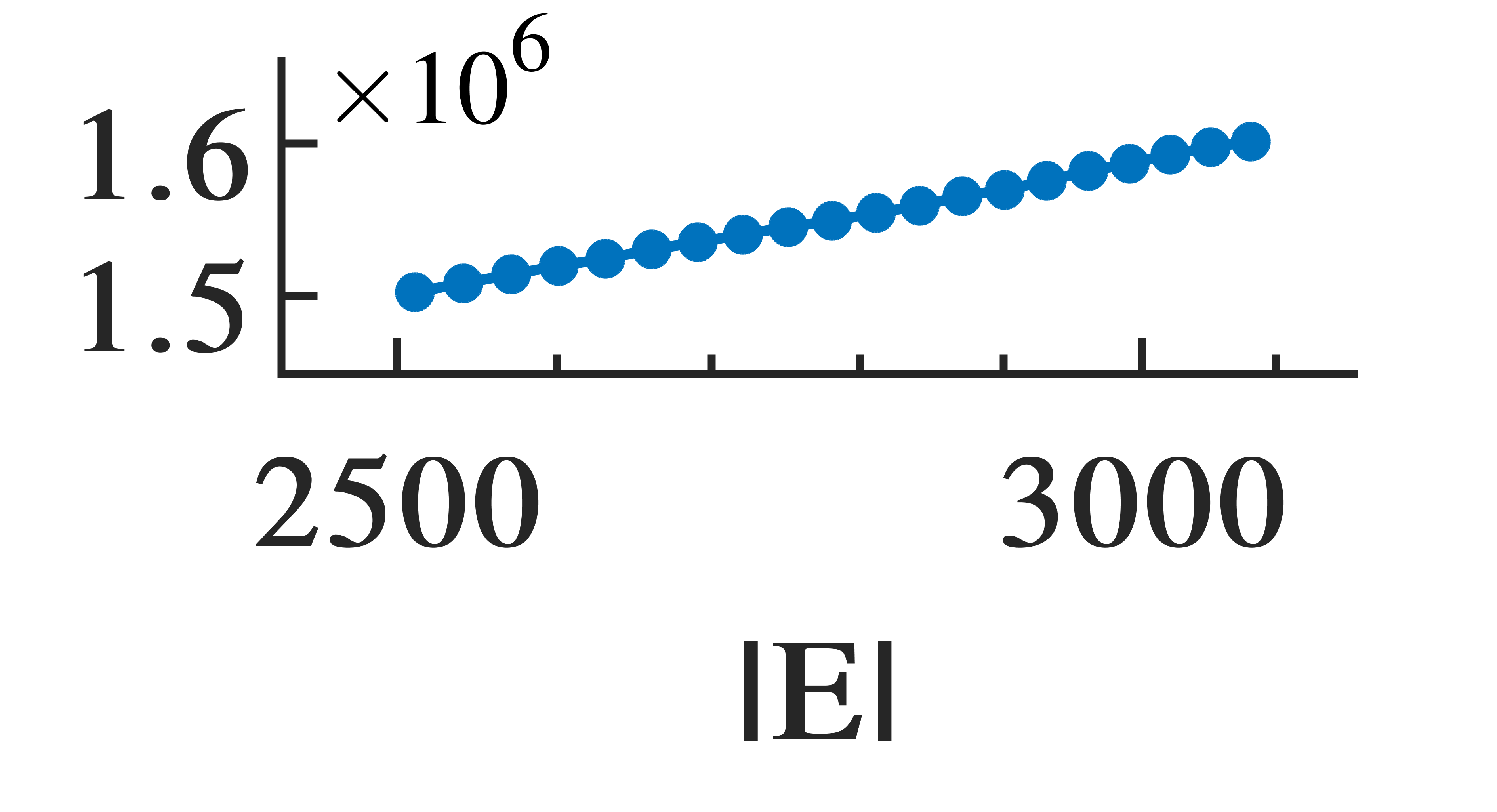}}
  \subfigure[SPA, $m$$=$$50$, $558$$\pm$$22.4$]{\includegraphics[width=0.24\textwidth]{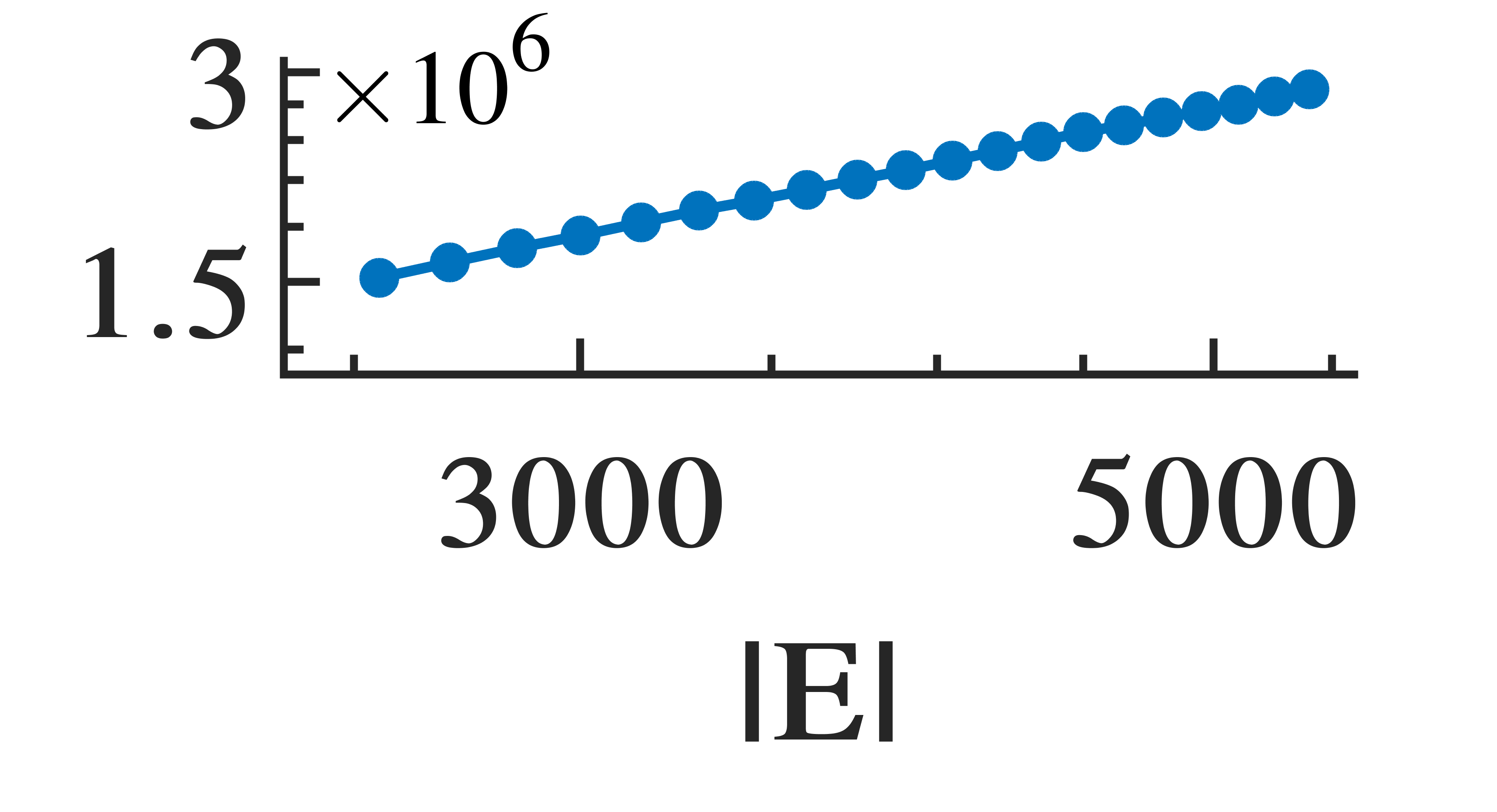}}
  \subfigure[SPA, $m$$=$$100$, $548.1$$\pm$$17$]{\includegraphics[width=0.24\textwidth]{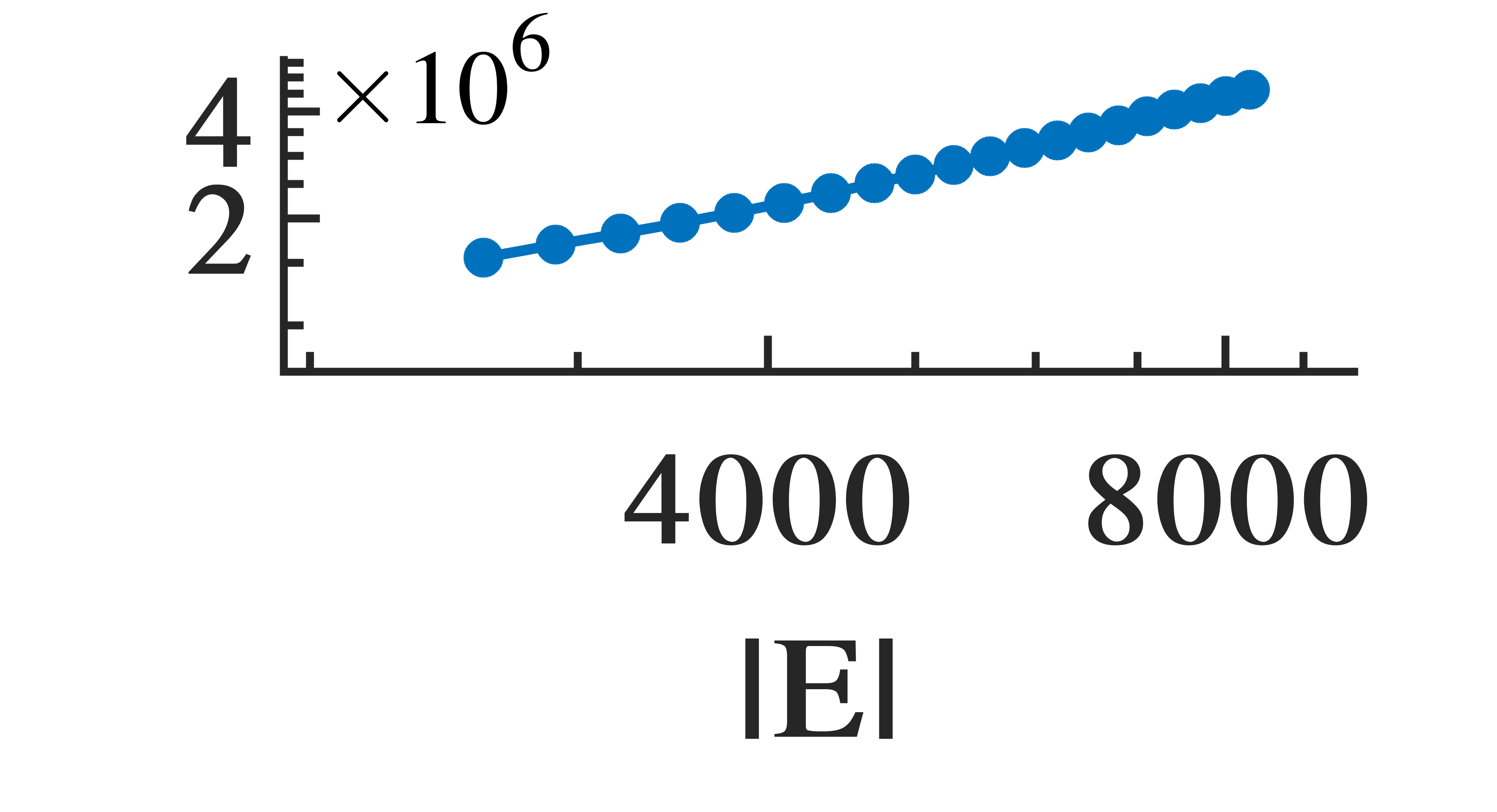}}
    \caption{Butterfly count versus edge count in FF and SPA streams with various average butterfly rates.} 
    \label{fig:butterflycountSynthetic}
\end{figure*}
\begin{figure*}[h]
    \centering
  \subfigure[$p=0.15$]{\includegraphics[width=0.24\textwidth]{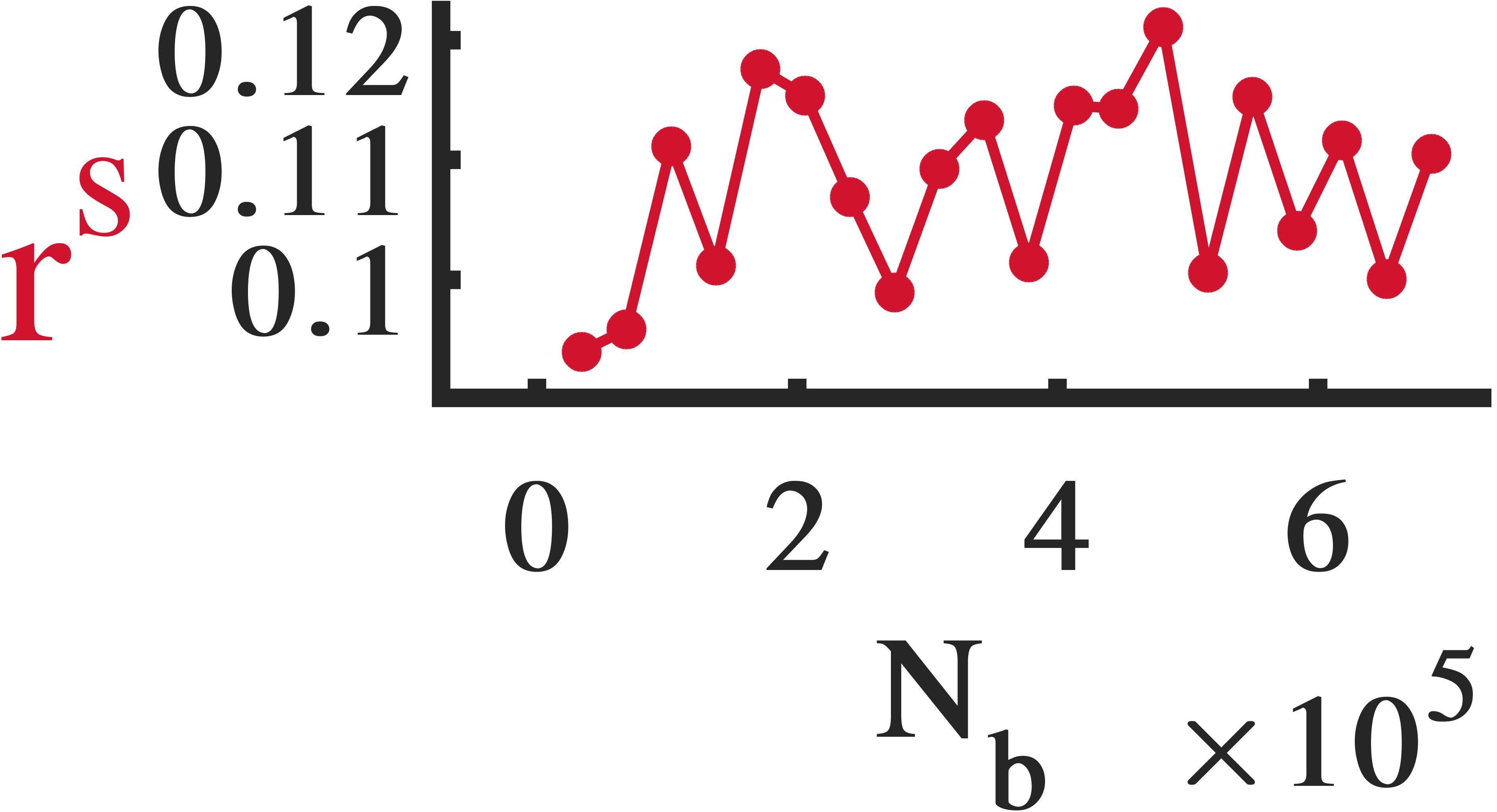}}
   \subfigure[$p=0.4$]{\includegraphics[width=0.24\textwidth]{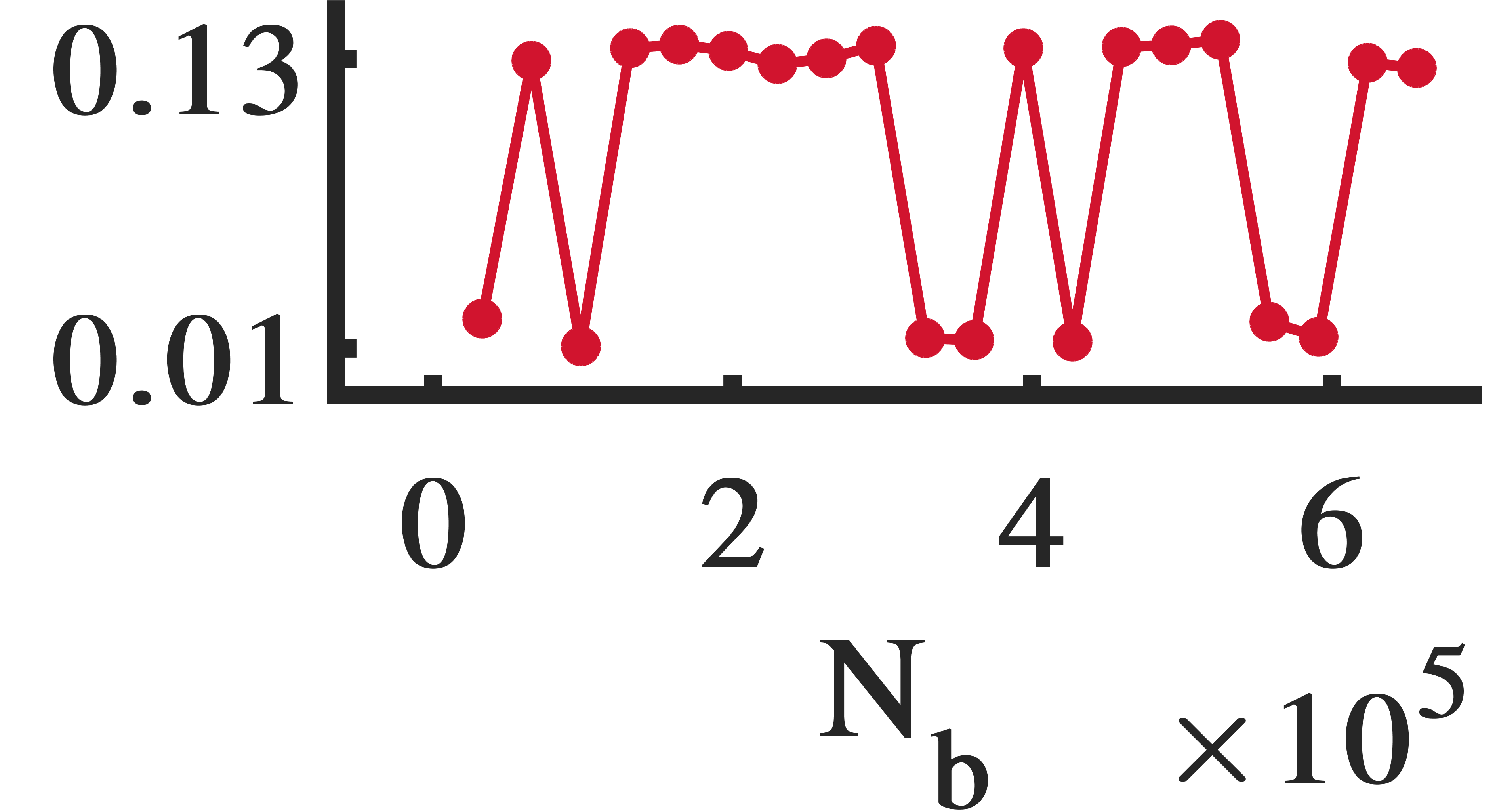}}
    \subfigure[$p=0.7$]{\includegraphics[width=0.24\textwidth]{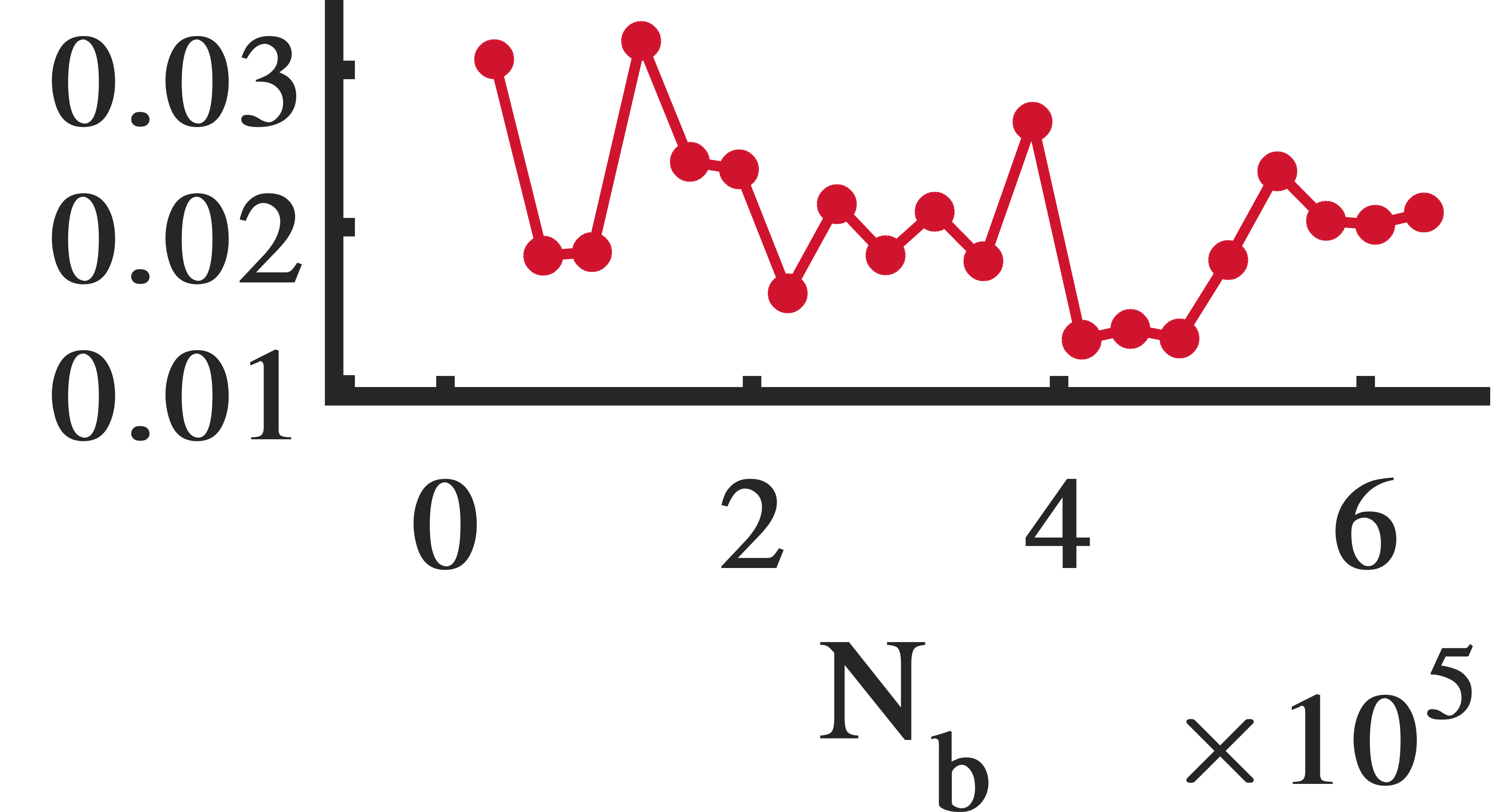}}
    
   \subfigure[$p=0.15$]{\includegraphics[width=0.24\textwidth]{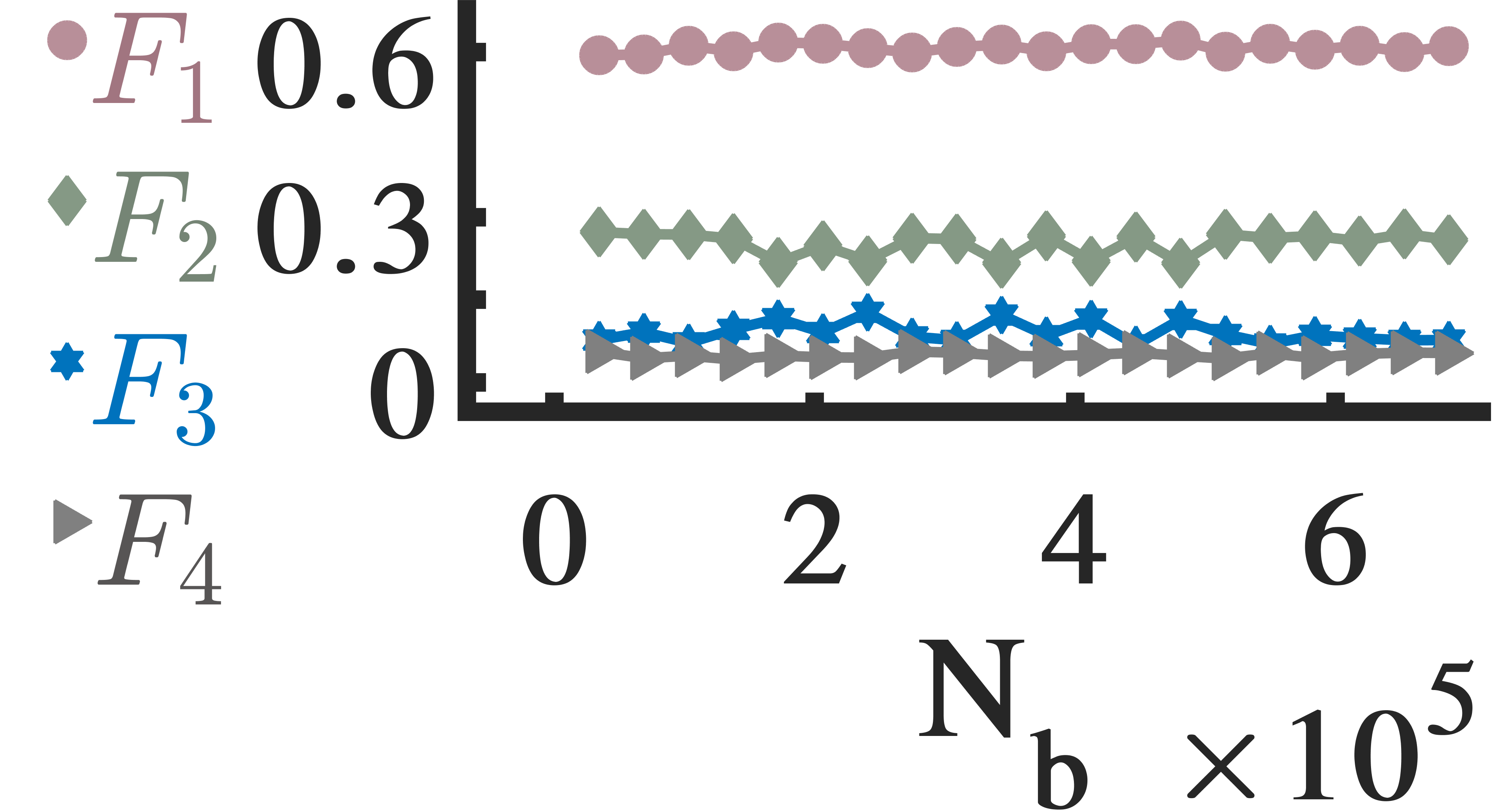}}
    \subfigure[$p=0.4$]{\includegraphics[width=0.24\textwidth]{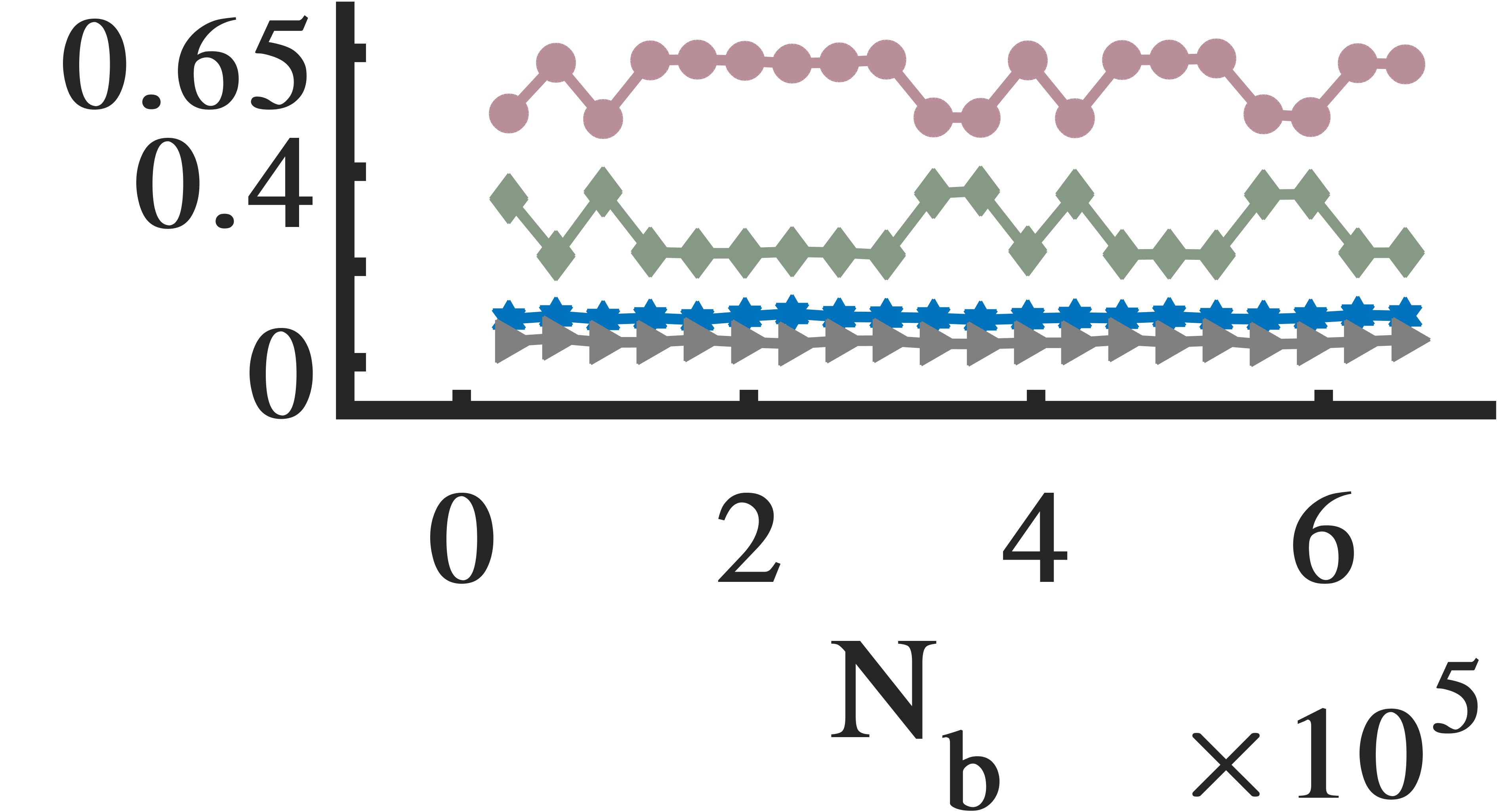}}
    \subfigure[$p=0.7$]{\includegraphics[width=0.24\textwidth]{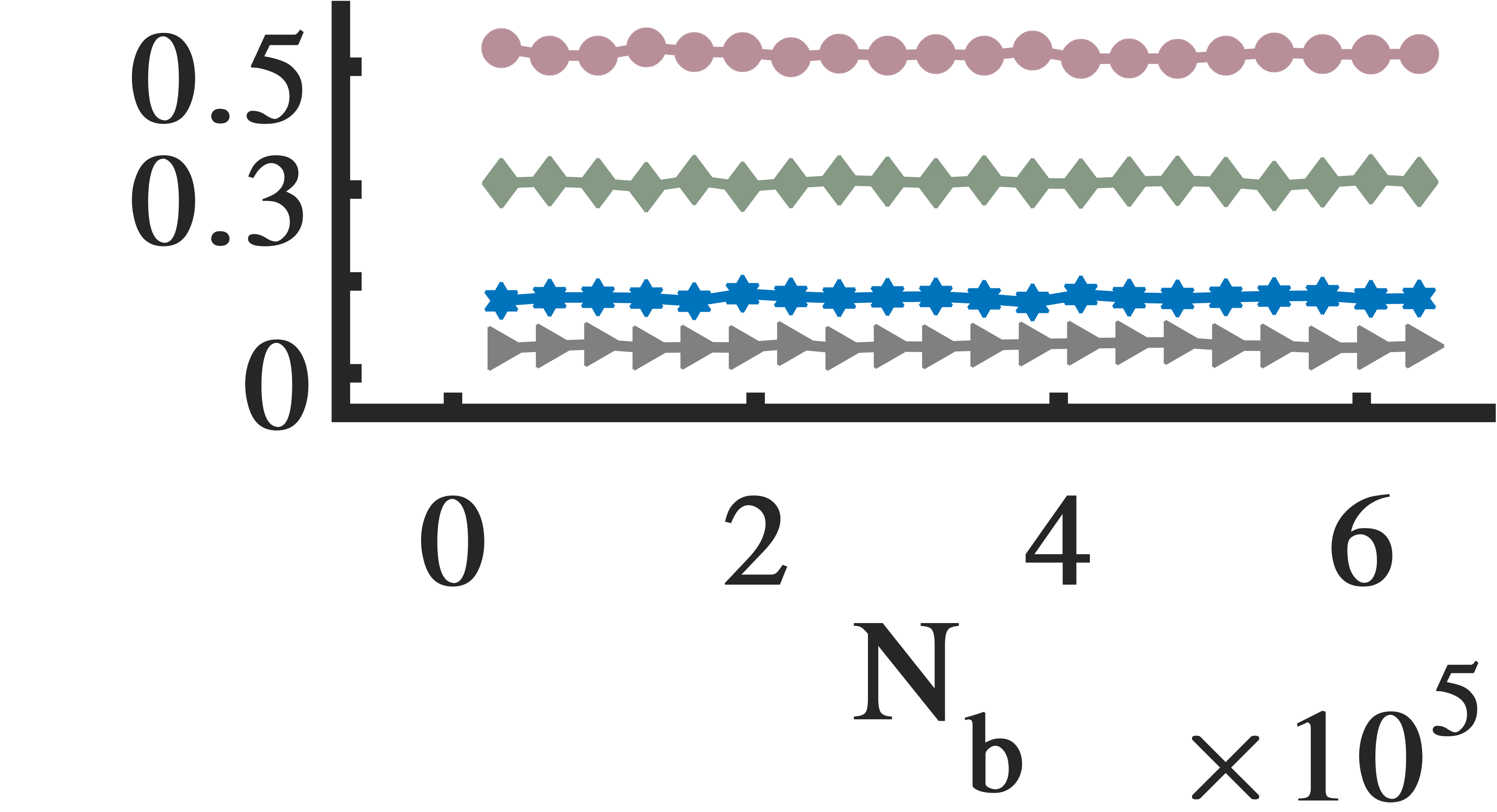}}
    \caption{(a-c) strength assortativity localization factor and  and (d-f) corresponding F elements of butterflies over the timeline of burst arrivals in FF streams with $p_b=0.3$, $p=0.15,0.4,0.7$.}
    \label{fig:rsandFforestfire}
\end{figure*}
\begin{figure*}[h]
    \centering
    \subfigure[$p=0.15$, $p_b=0.3$]{\includegraphics[width=0.24\textwidth]{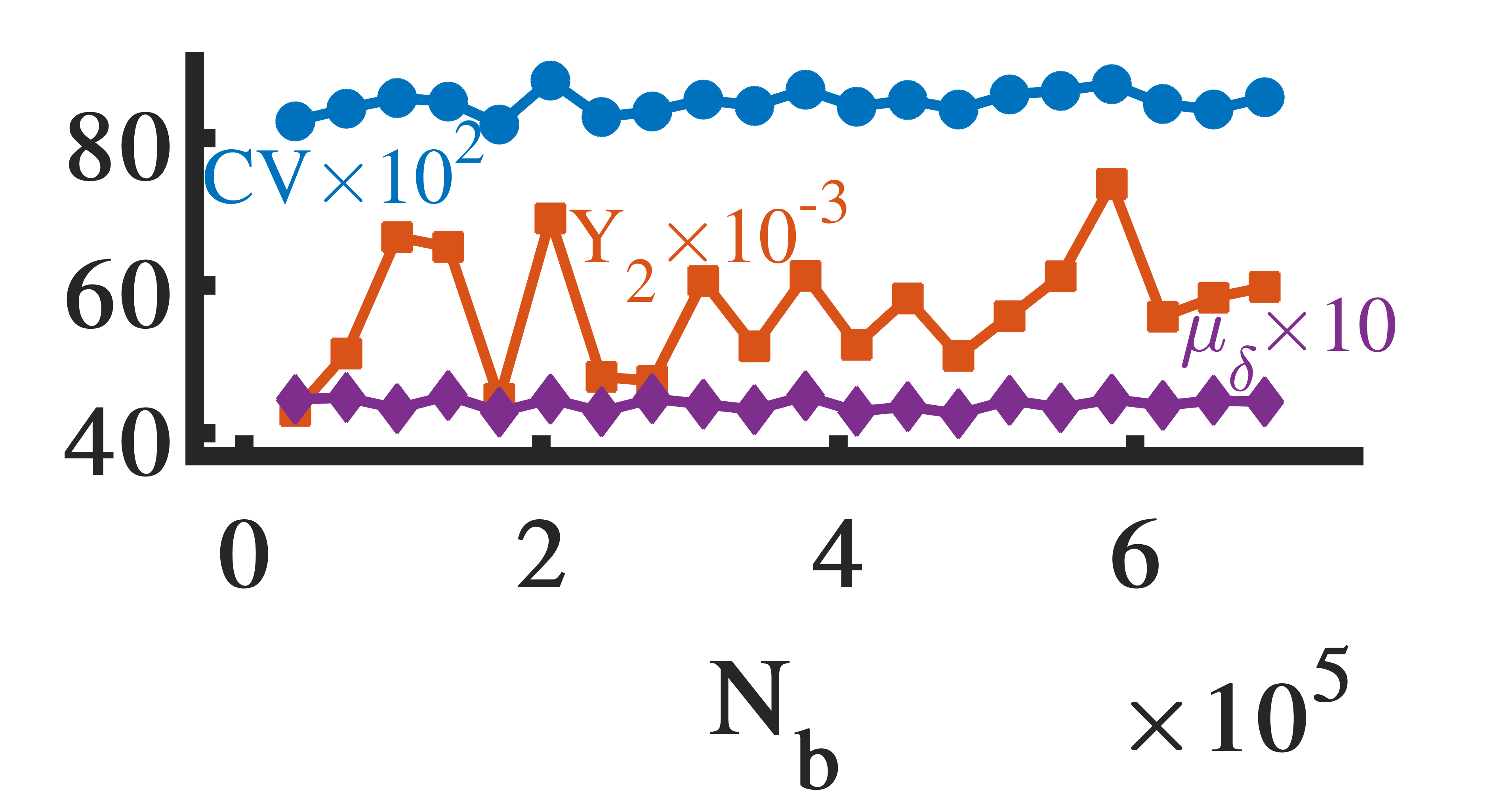}}
    \subfigure[$p=0.4$, $p_b=0.3$]{\includegraphics[width=0.24\textwidth]{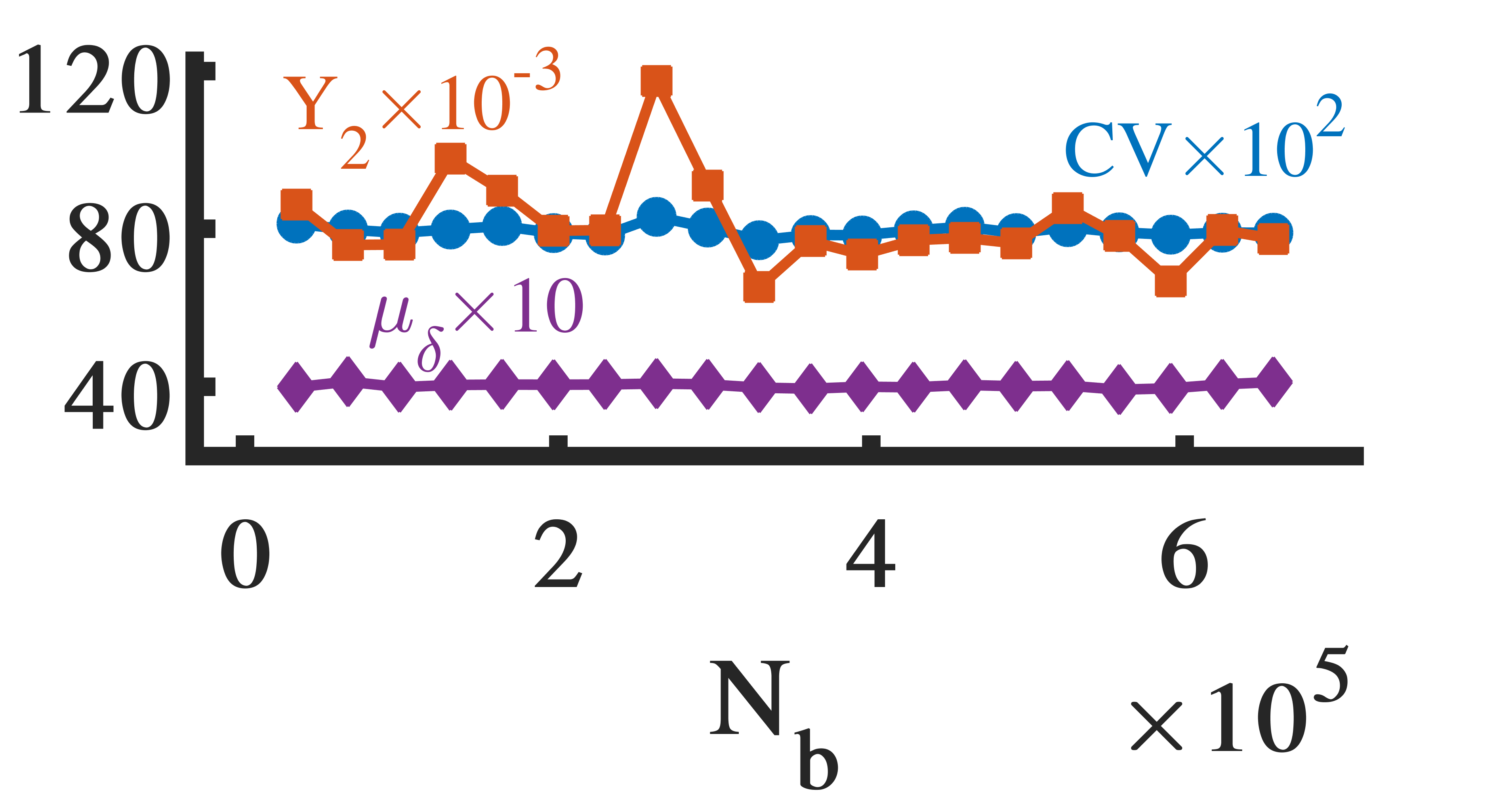}}
    \subfigure[$p=0.7$, $p_b=0.3$]{\includegraphics[width=0.24\textwidth]{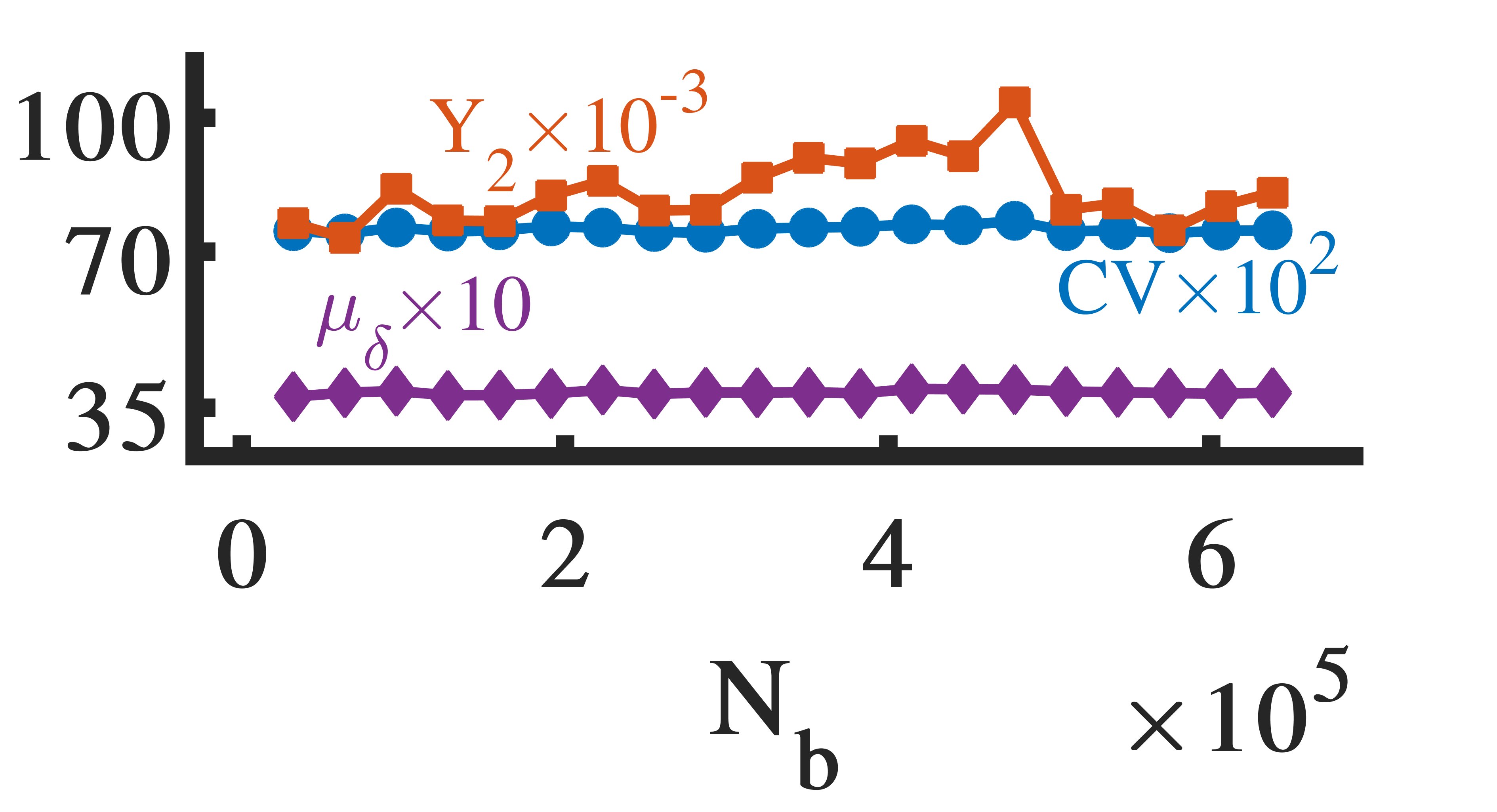}}
    \caption{Coefficient of variation (circles), excess kurtosis (squares), and mean (diamonds) of butterfly strength-differences over the timeline of burst arrivals in FF streams with $p_b=0.3$, $p=0.15,0.4,0.7$.}
    \label{fig:sdifstatsforestfire}
\end{figure*}
\begin{figure*}[h]
    \centering
    \subfigure[$Pr(S_i)$, $p$$=$$0.15$]{\includegraphics[width=0.24\textwidth]{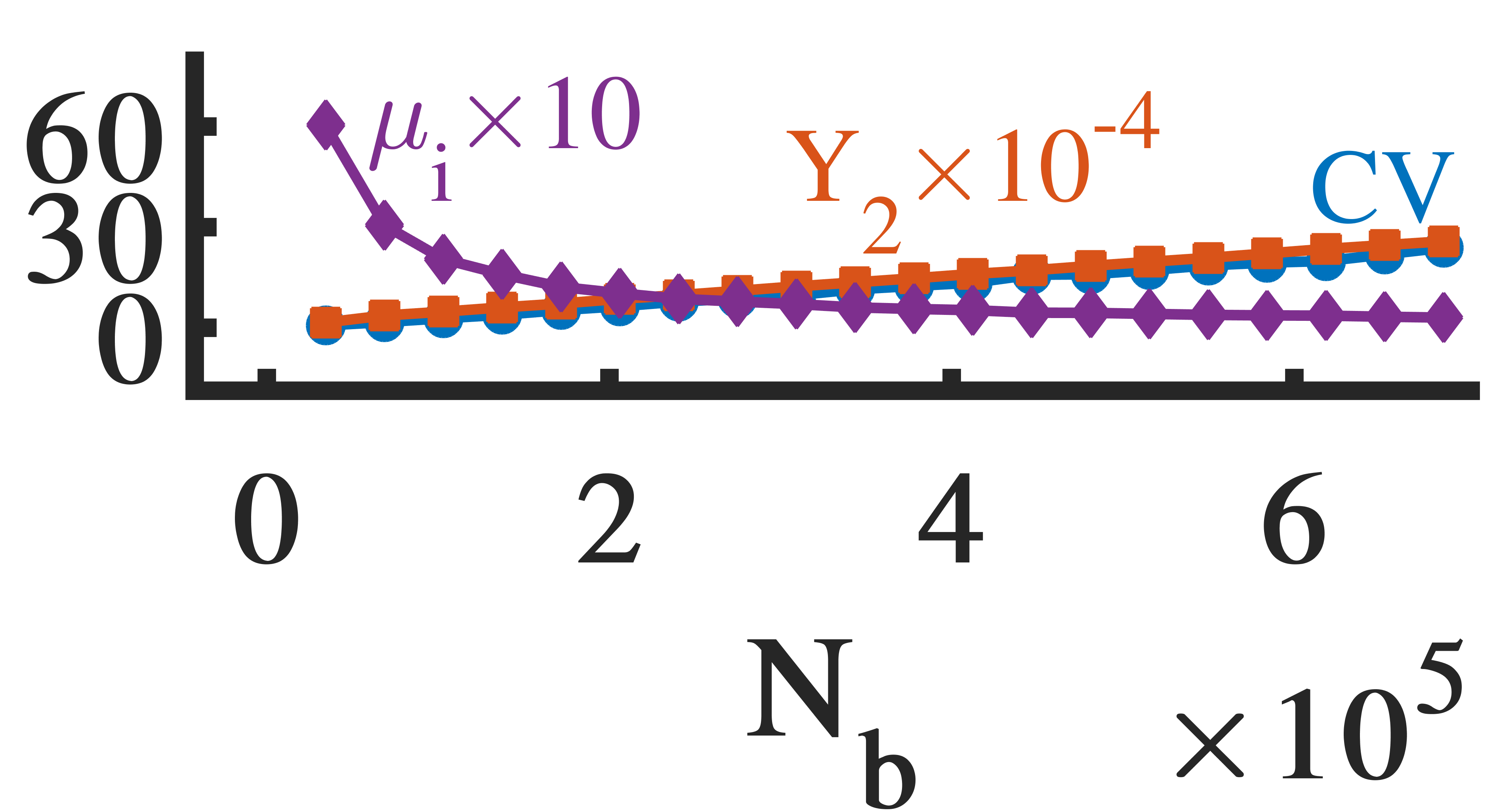}}
    \subfigure[$Pr(S_i)$, $p$$=$$0.4$]{\includegraphics[width=0.24\textwidth]{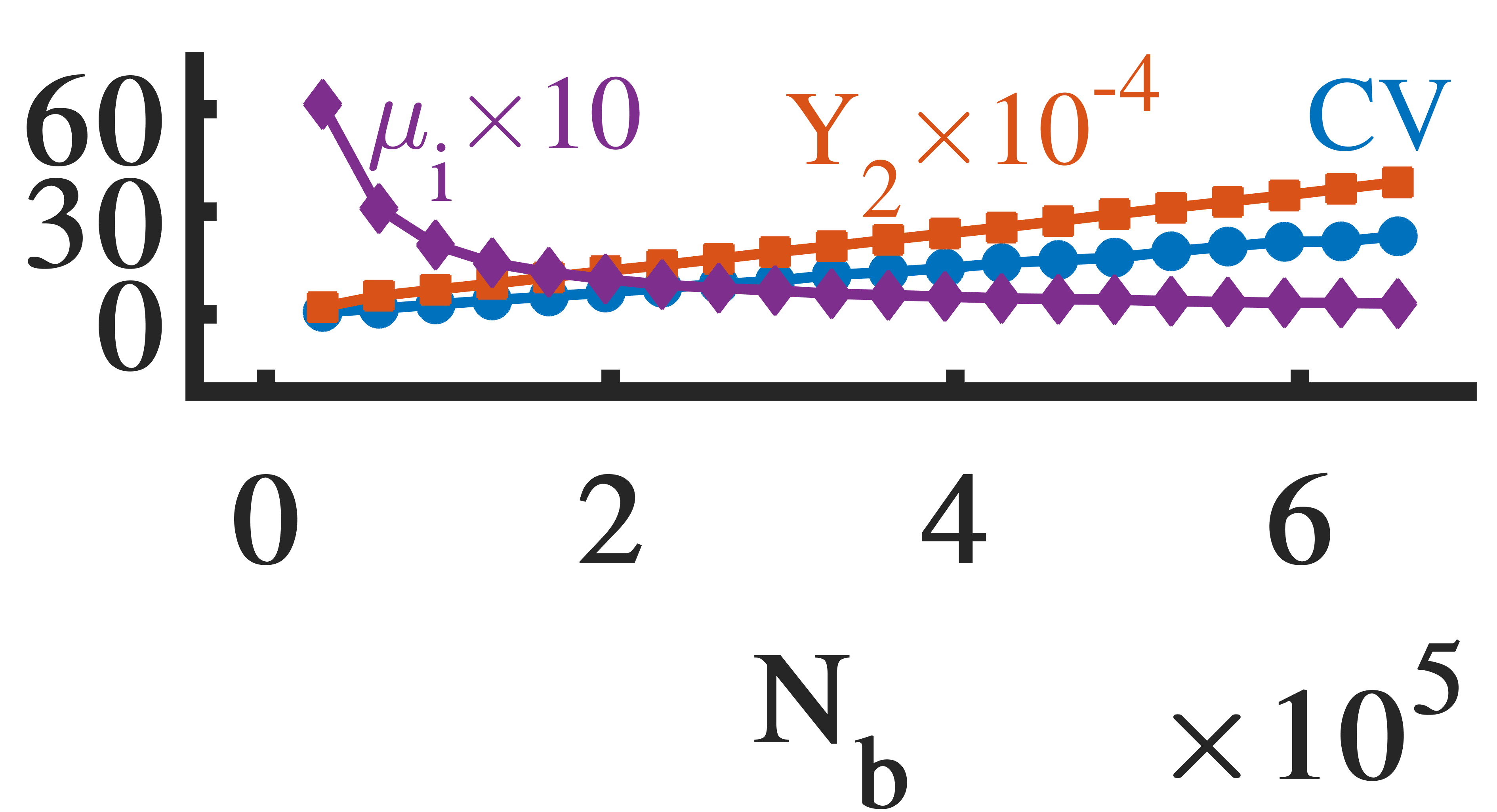}}
    \subfigure[$Pr(S_i)$, $p$$=$$0.7$]{\includegraphics[width=0.24\textwidth]{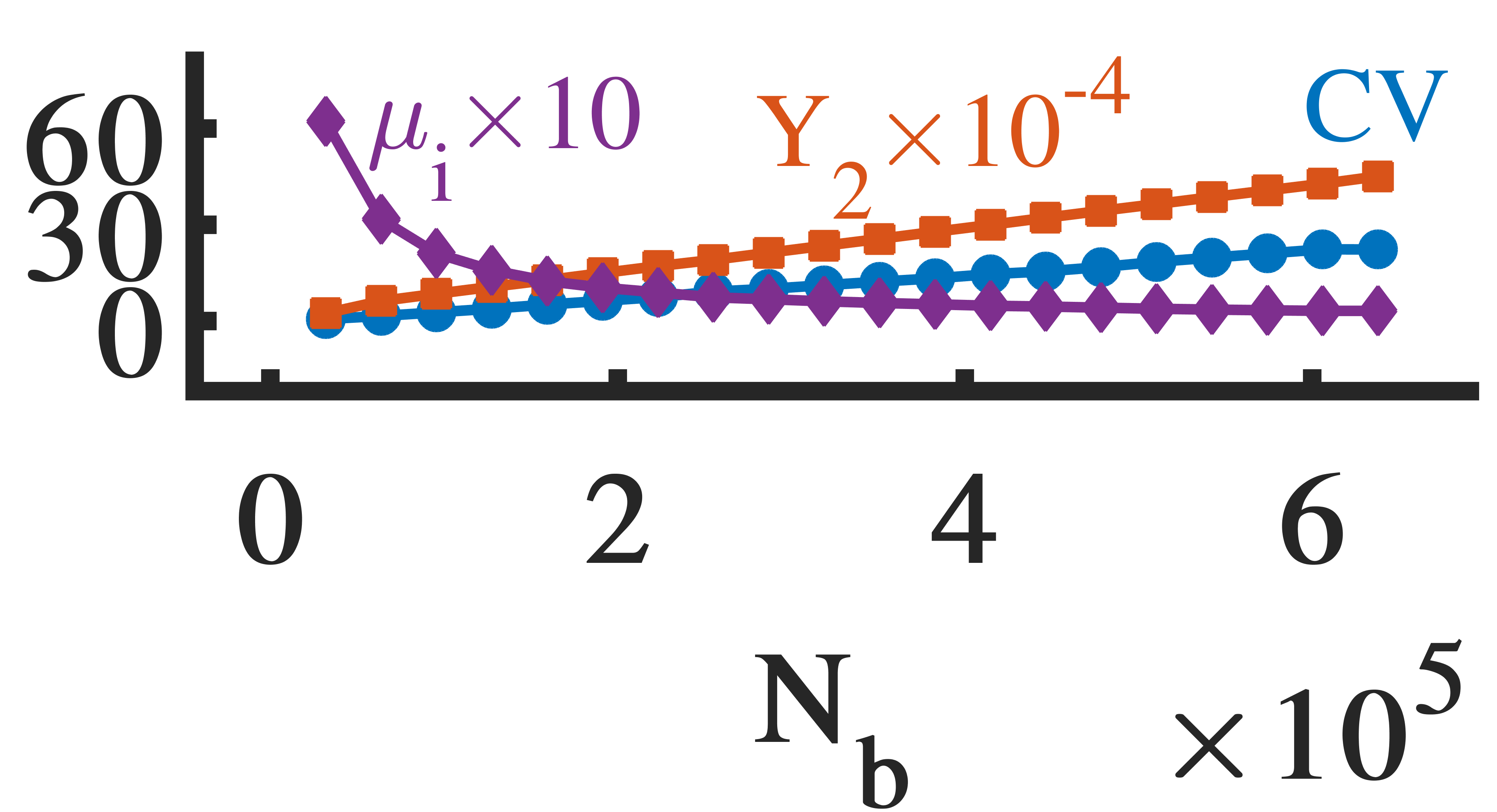}}
    
    \subfigure[$Pr(S_j)$, $p$$=$$0.15$]{\includegraphics[width=0.24\textwidth]{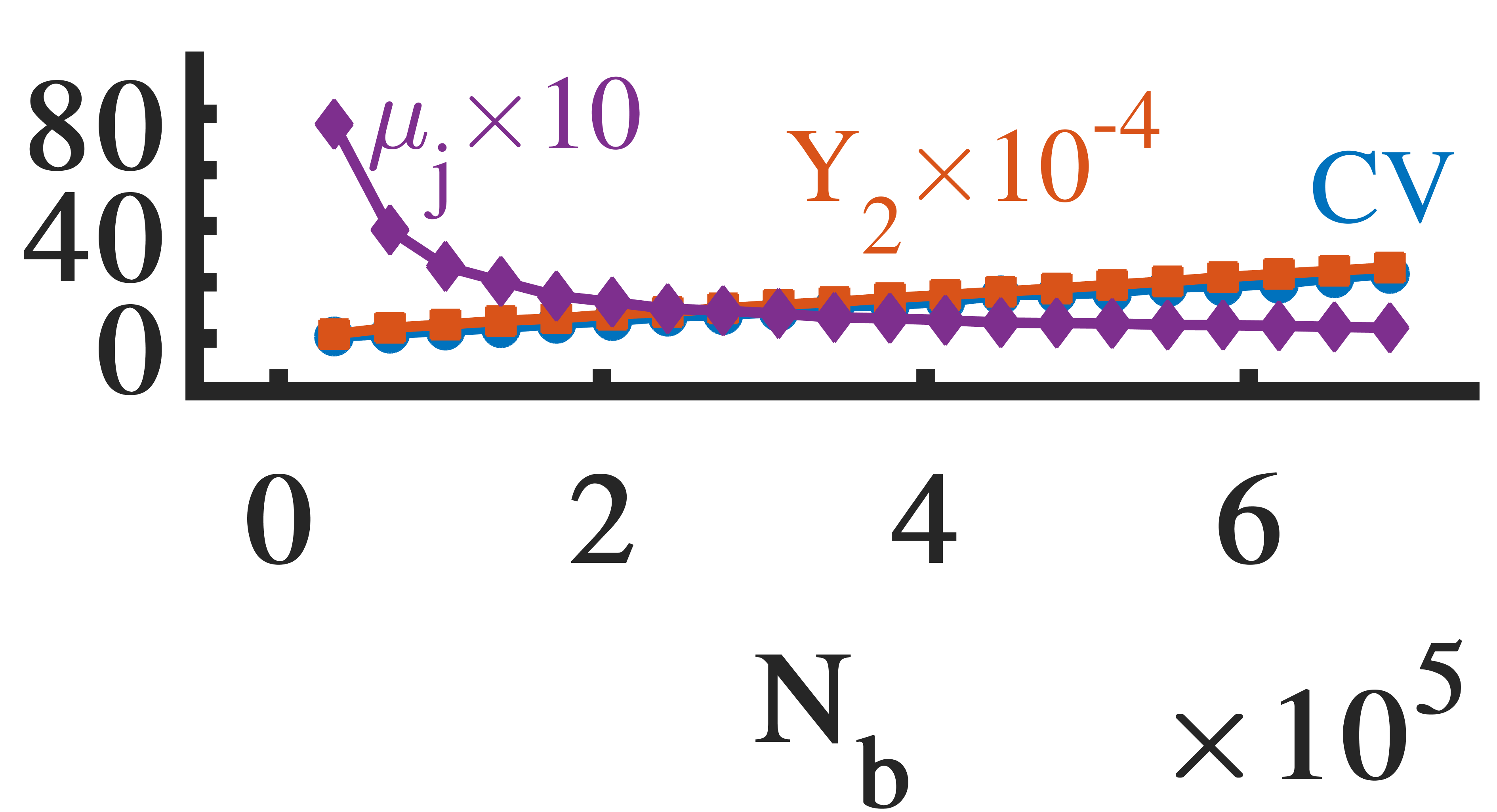}}
    \subfigure[$Pr(S_j)$, $p$$=$$0.4$]{\includegraphics[width=0.24\textwidth]{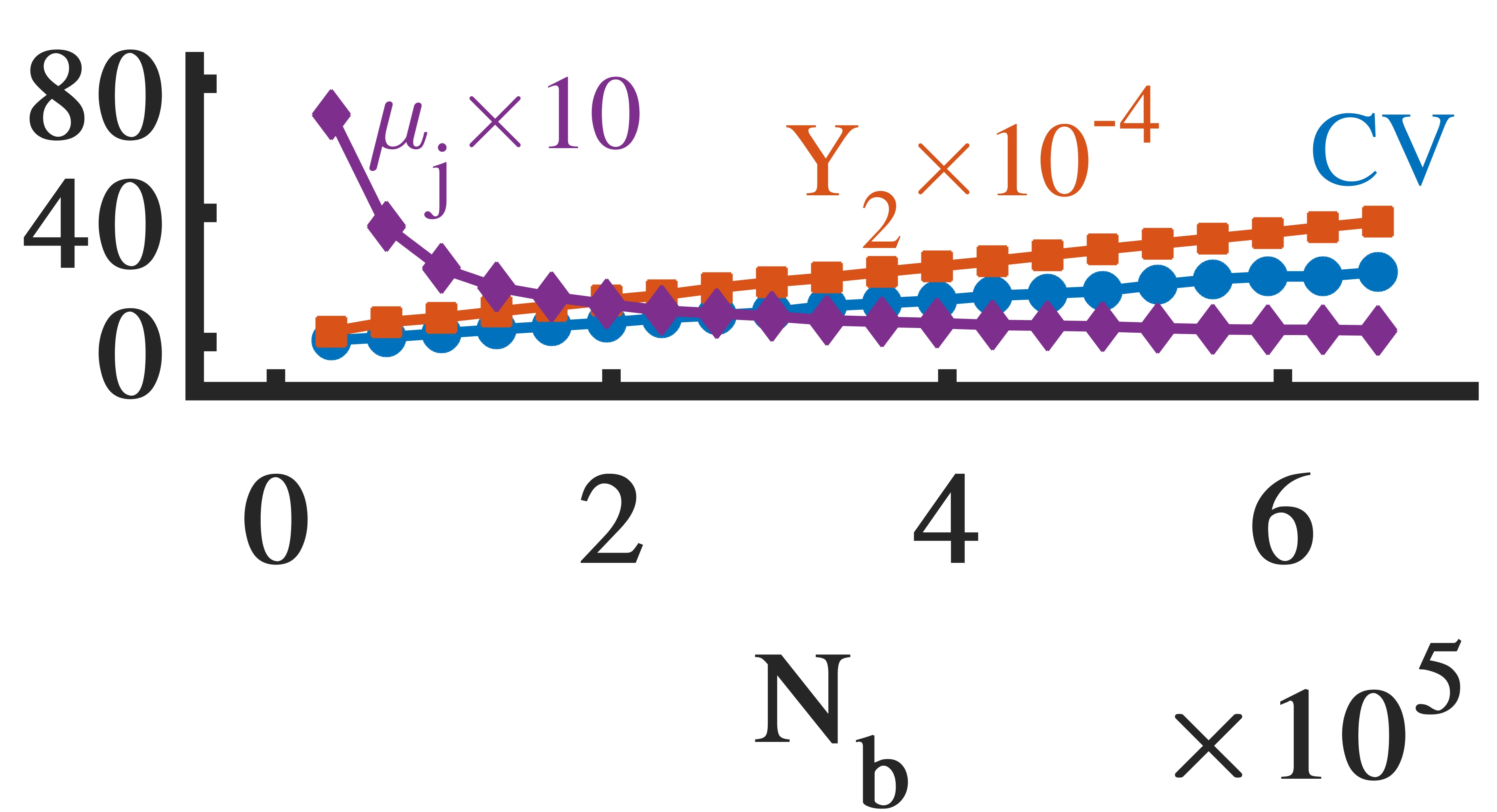}}
    \subfigure[$Pr(S_j)$, $p$$=$$0.7$]{\includegraphics[width=0.24\textwidth]{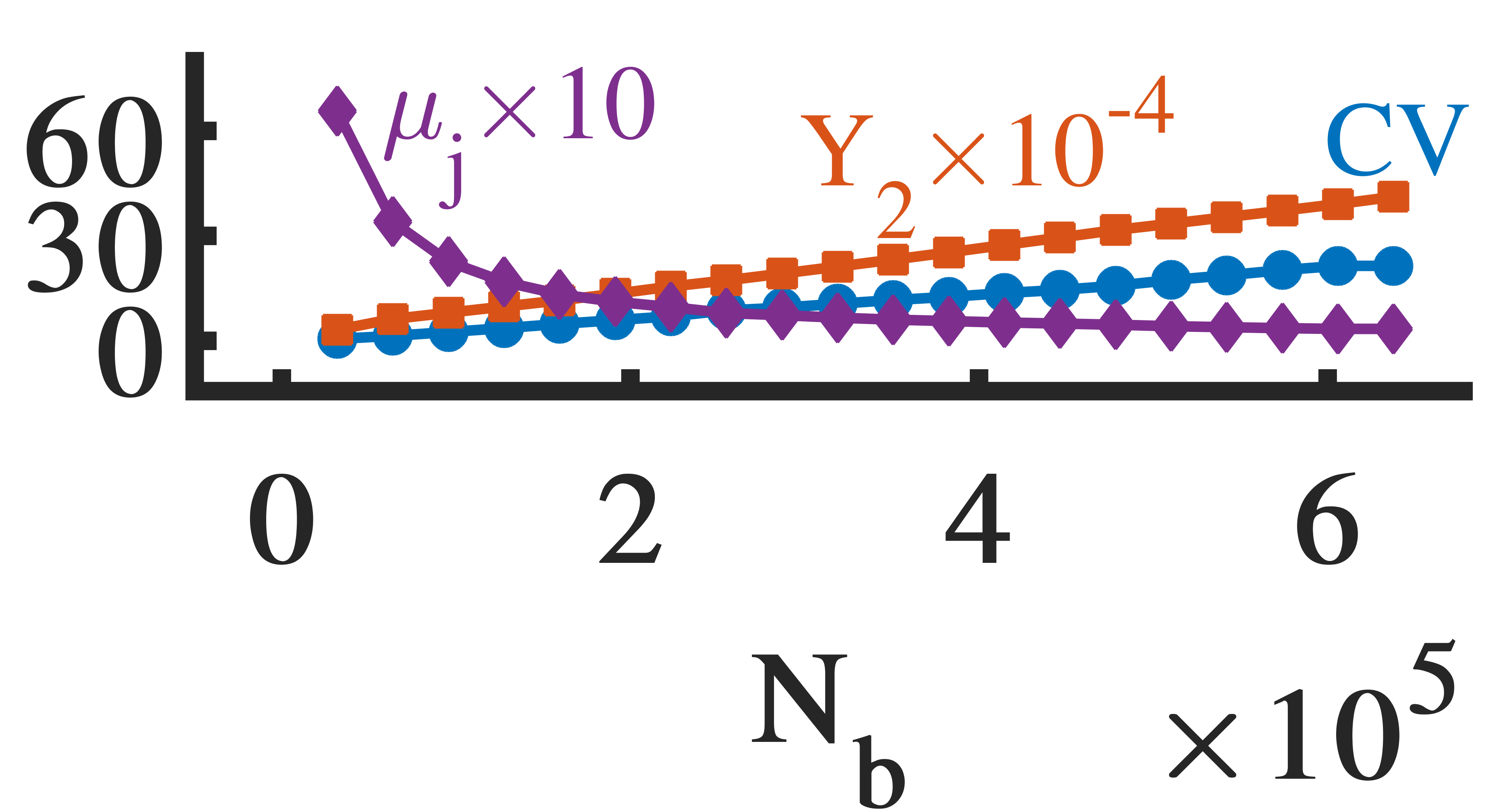}}
    \caption{Coefficient of variation (circles), excess kurtosis (squares), and mean (diamonds) of strengths of butterfly (a-c) i-vertices and (d-f) j-vertices over the timeline of burst arrivals in FF streams with $p_b=0.3$, $p=0.15,0.4,0.7$.}
    \label{fig:sstatsforestfire}
\end{figure*}
\begin{figure*}[h]
    \centering
  \subfigure[$m$$=$$10$]{\includegraphics[width=0.24\textwidth]{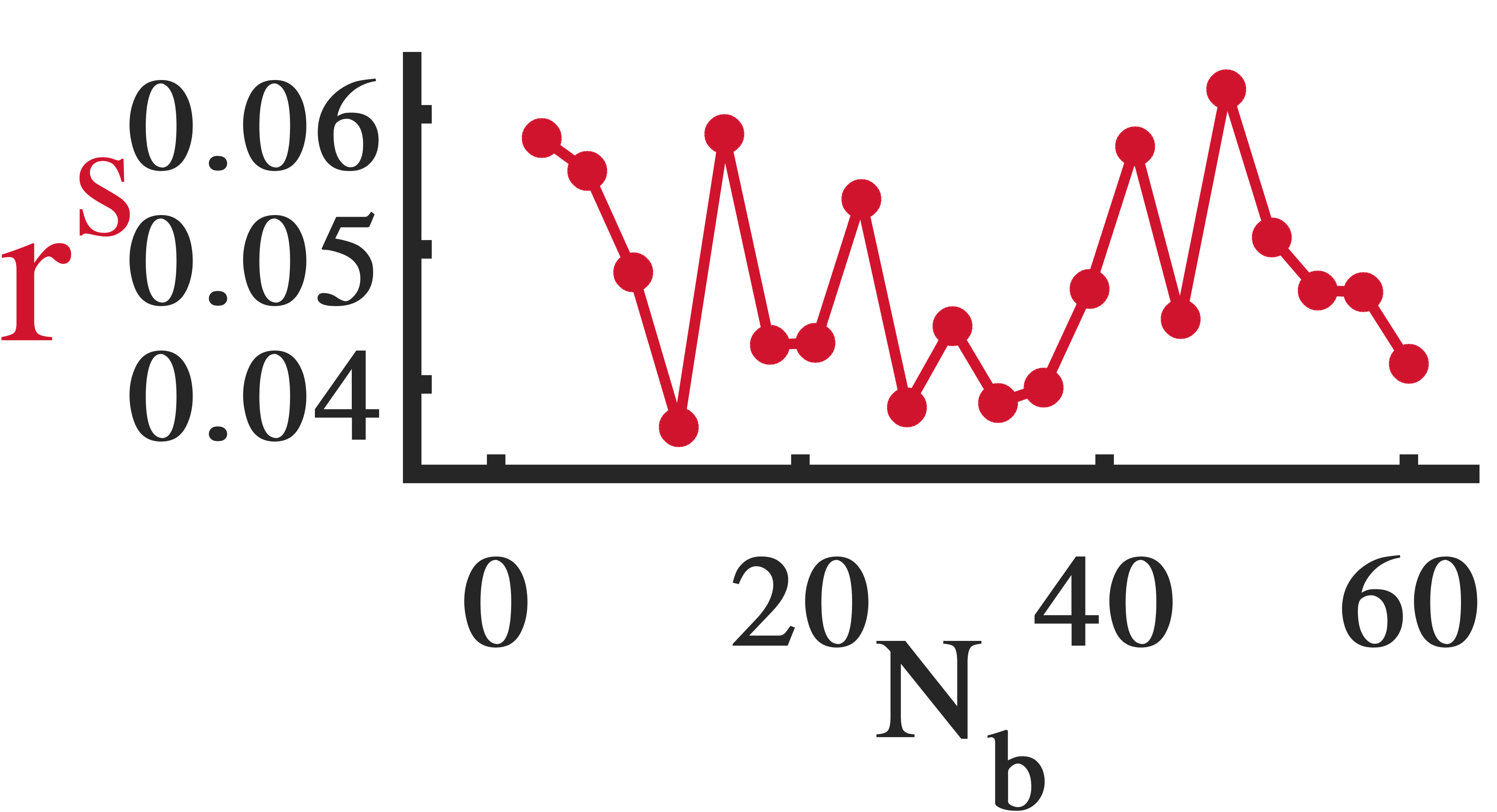}}
   \subfigure[$m$$=$$50$]{\includegraphics[width=0.24\textwidth]{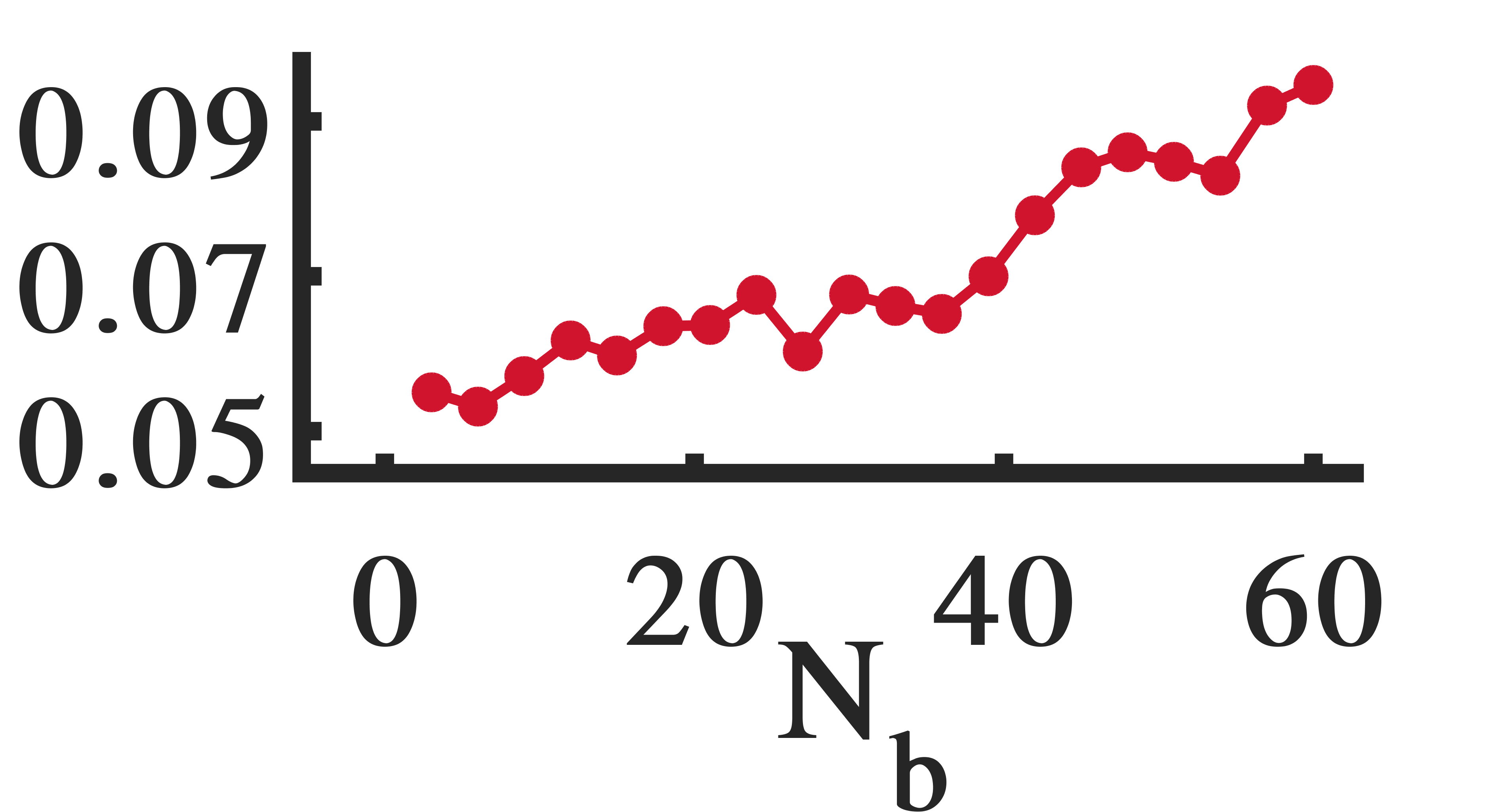}}
    \subfigure[$m$$=$$100$]{\includegraphics[width=0.24\textwidth]{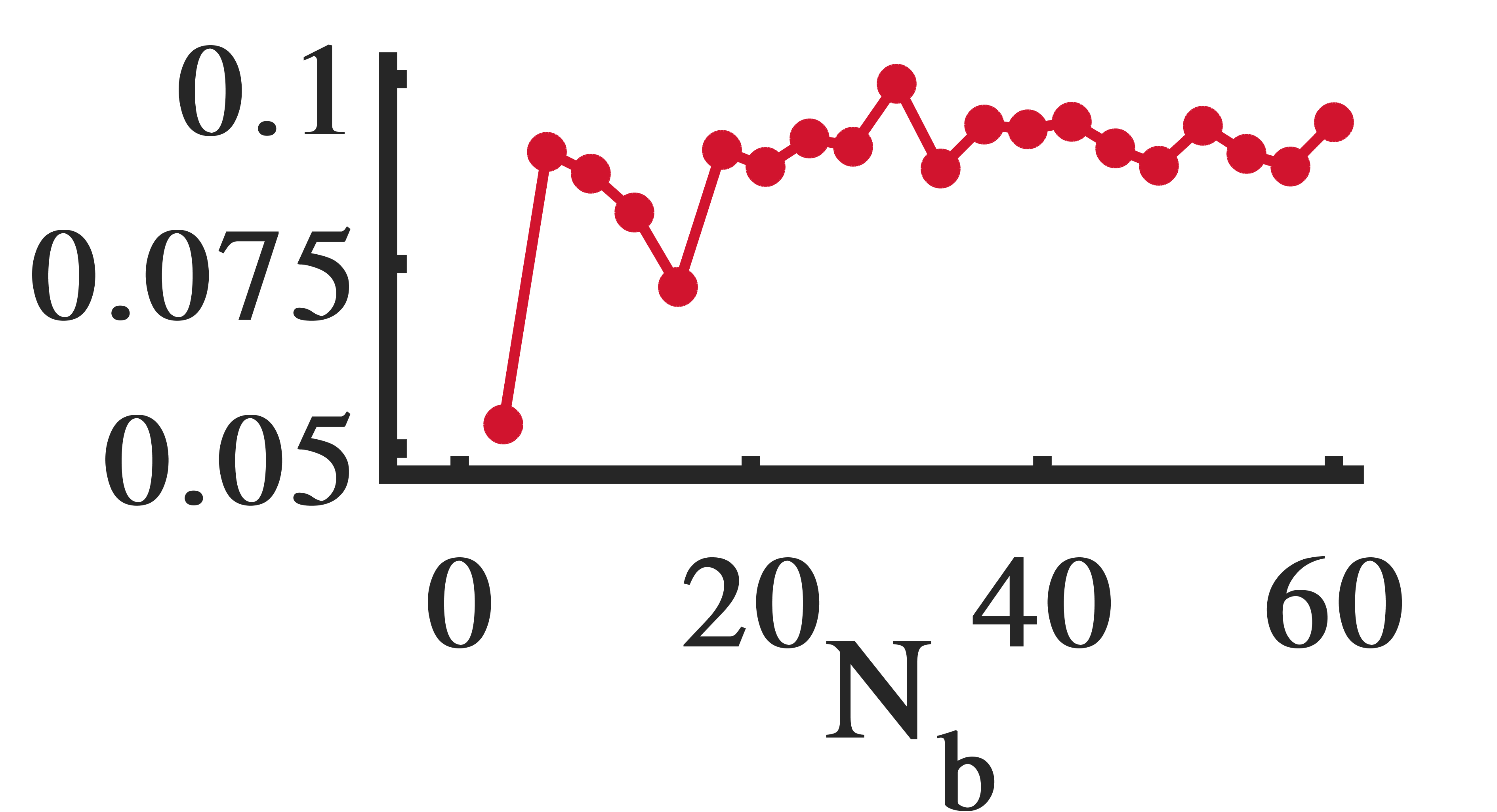}}
    
   \subfigure[$m$$=$$10$]{\includegraphics[width=0.24\textwidth]{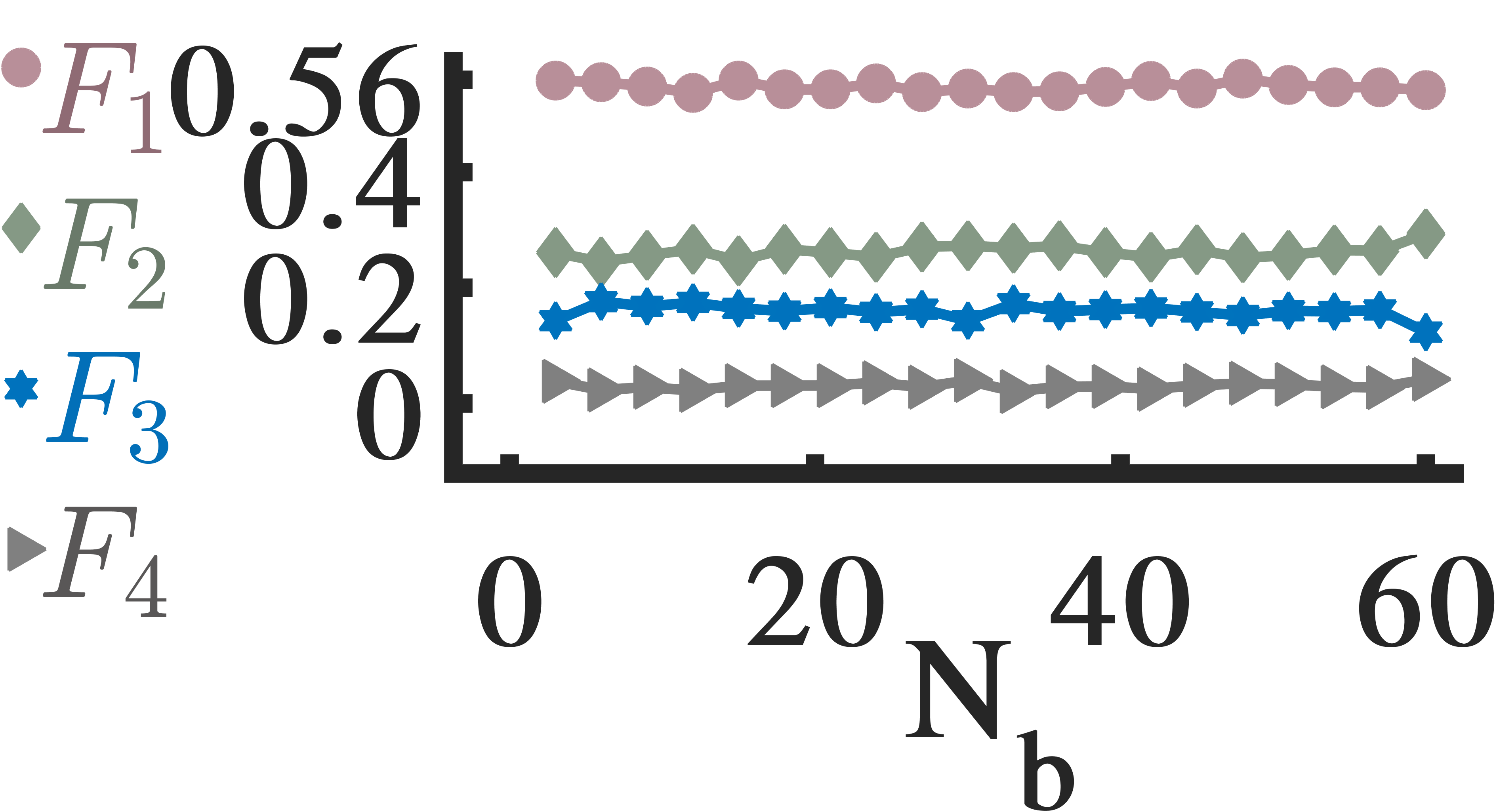}}
    \subfigure[$m$$=$$50$]{\includegraphics[width=0.24\textwidth]{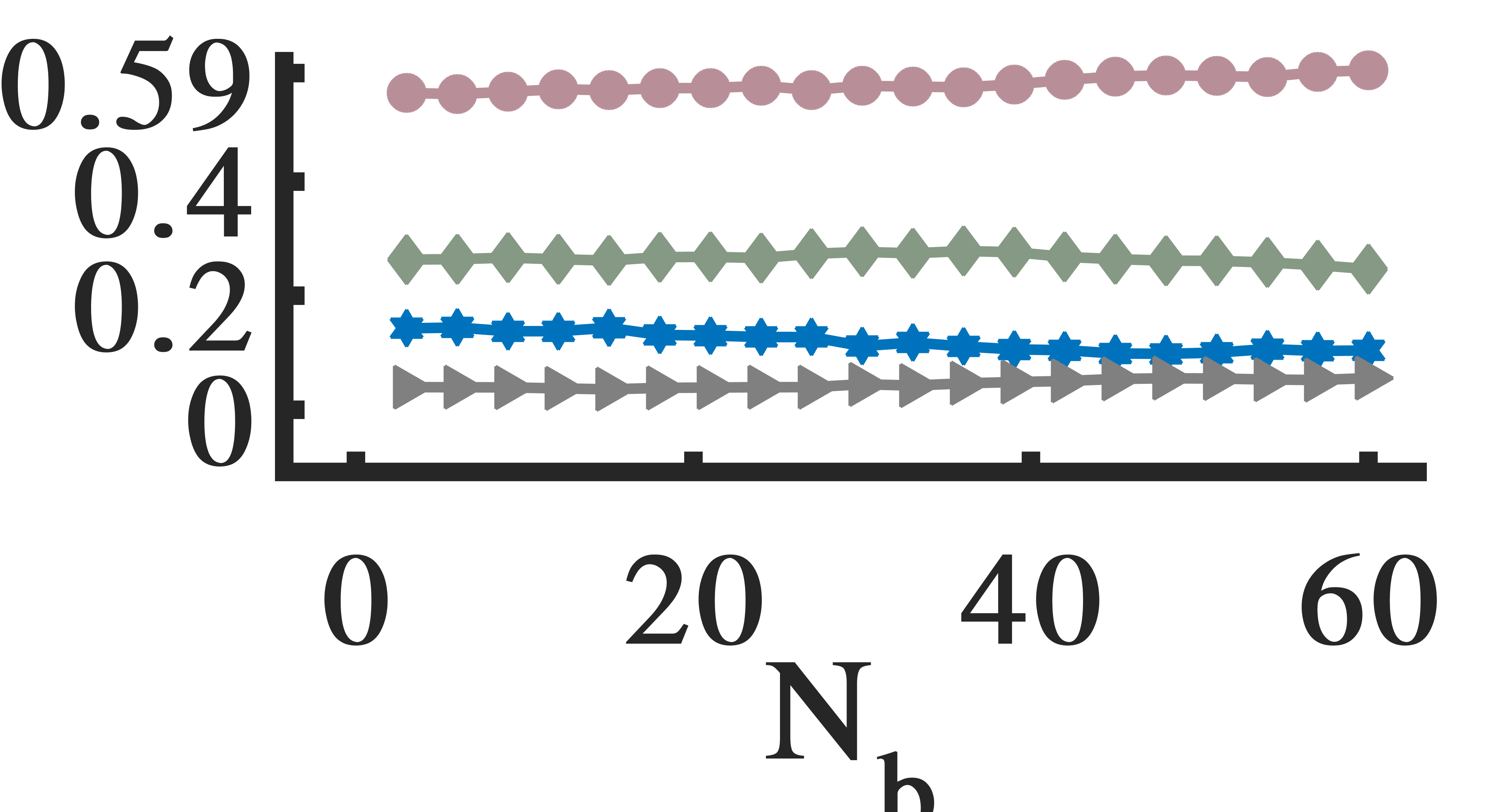}}
    \subfigure[$m$$=$$100$]{\includegraphics[width=0.24\textwidth]{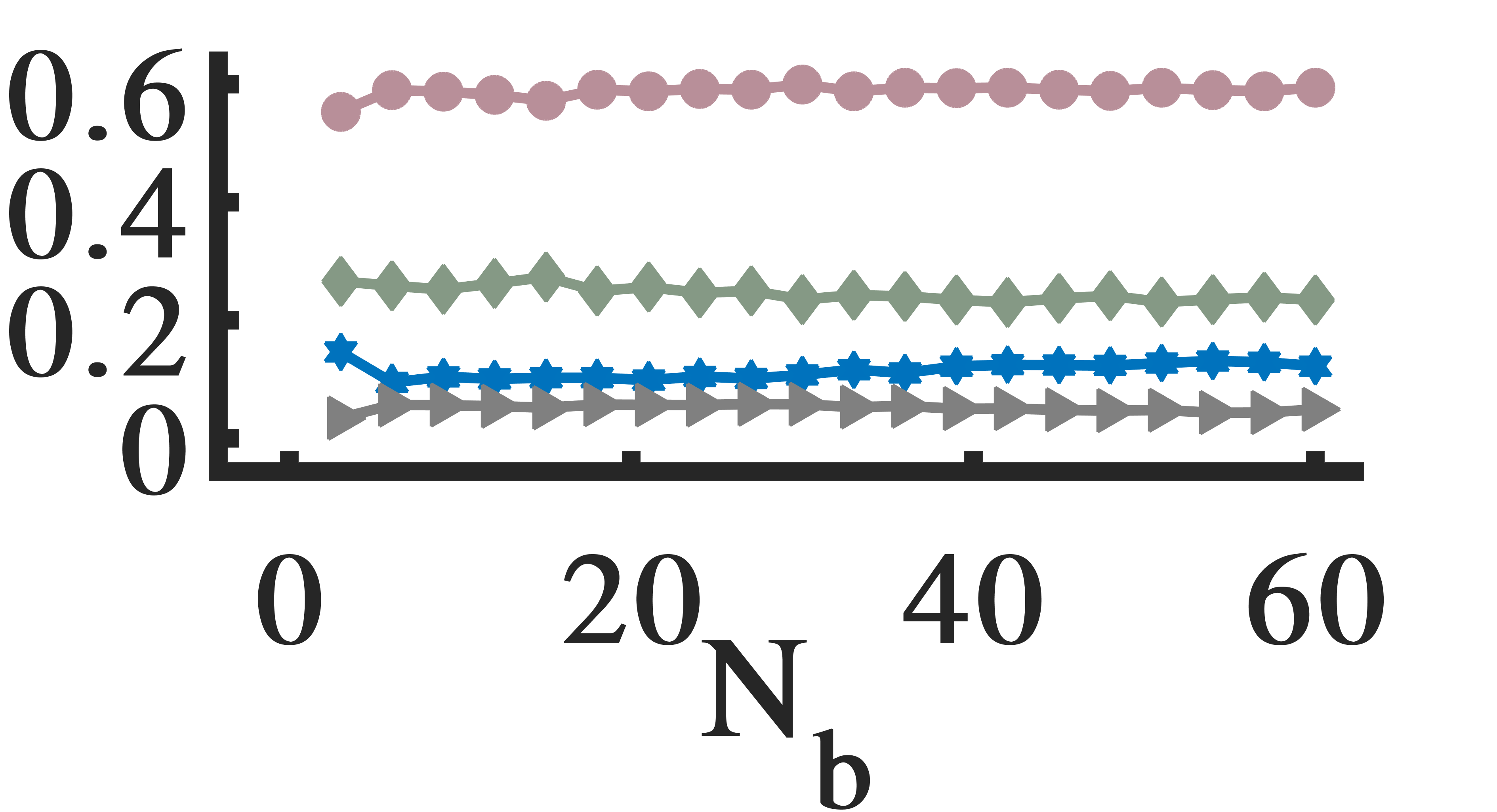}}
    \caption{(a-c) strength assortativity localization factor $r^s$  and (d-f) corresponding F elements of butterflies over the timeline of burst arrivals in SPA model with $m=10,50,100$.}
    \label{fig:rsandFSPA}
\end{figure*}
\begin{figure*}[h]
    \centering
    \subfigure[$m$$=$$10$]{\includegraphics[width=0.24\textwidth]{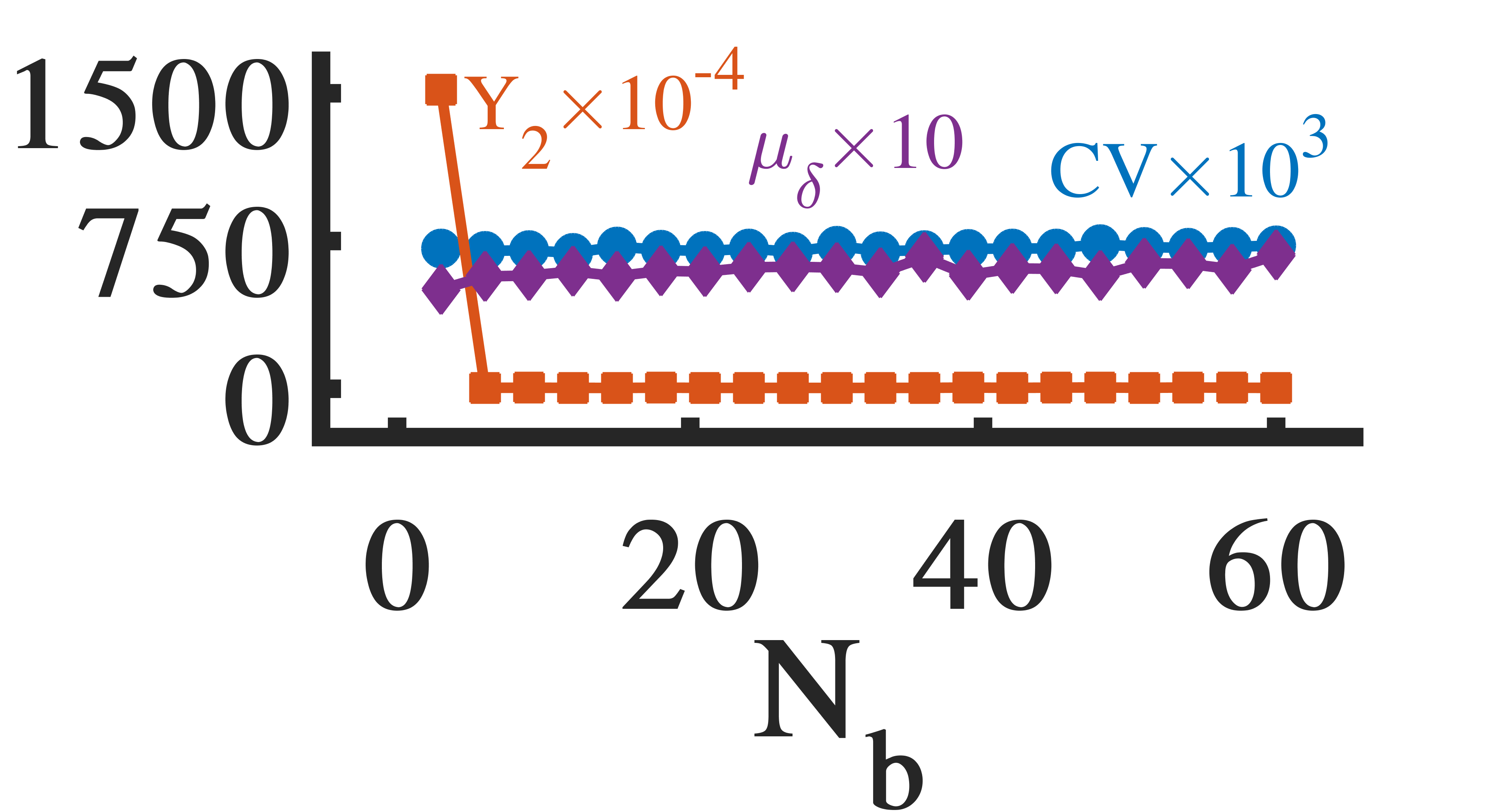}}
    \subfigure[$m$$=$$50$]{\includegraphics[width=0.24\textwidth]{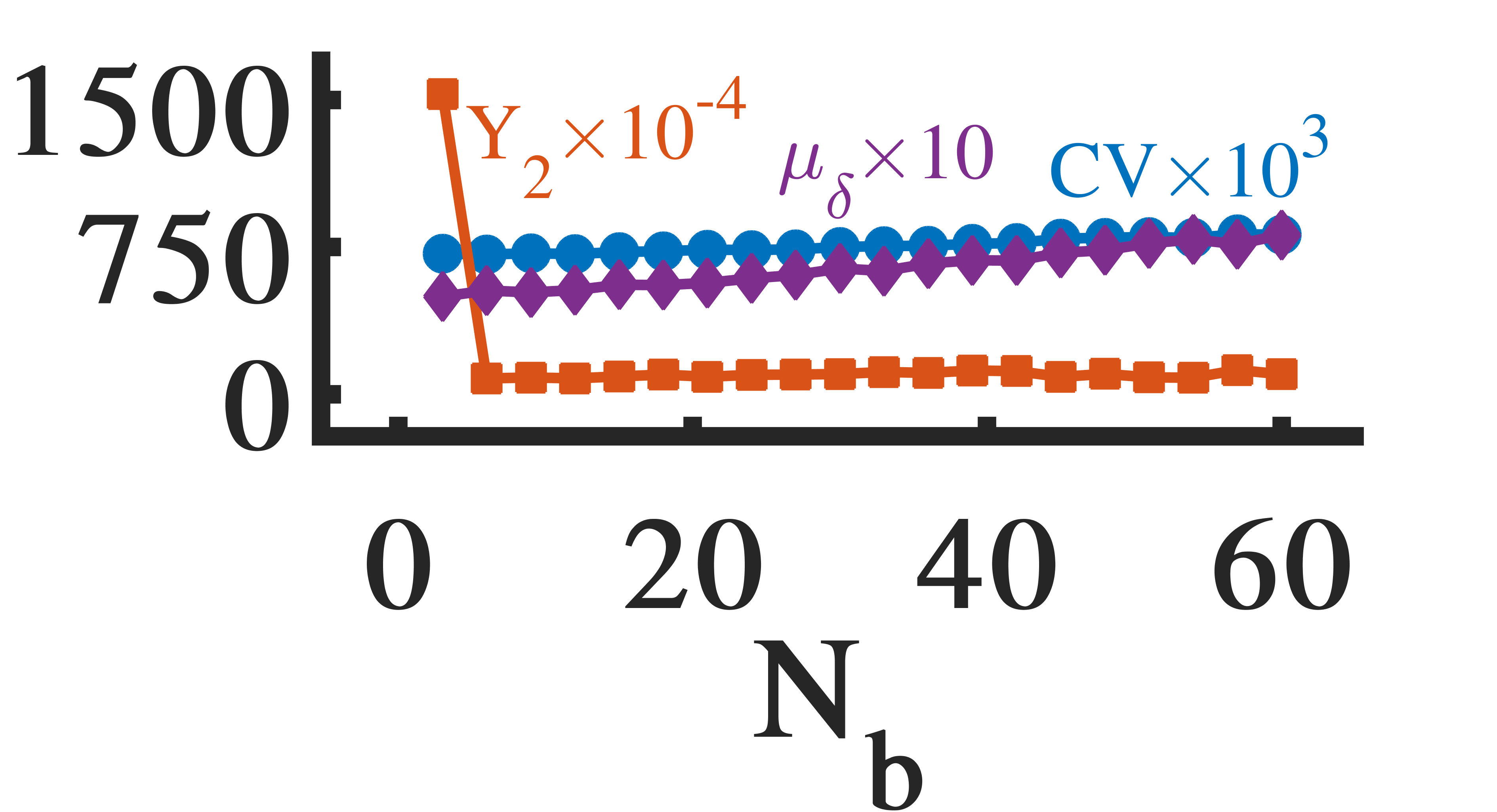}}
    \subfigure[$m$$=$$100$]{\includegraphics[width=0.24\textwidth]{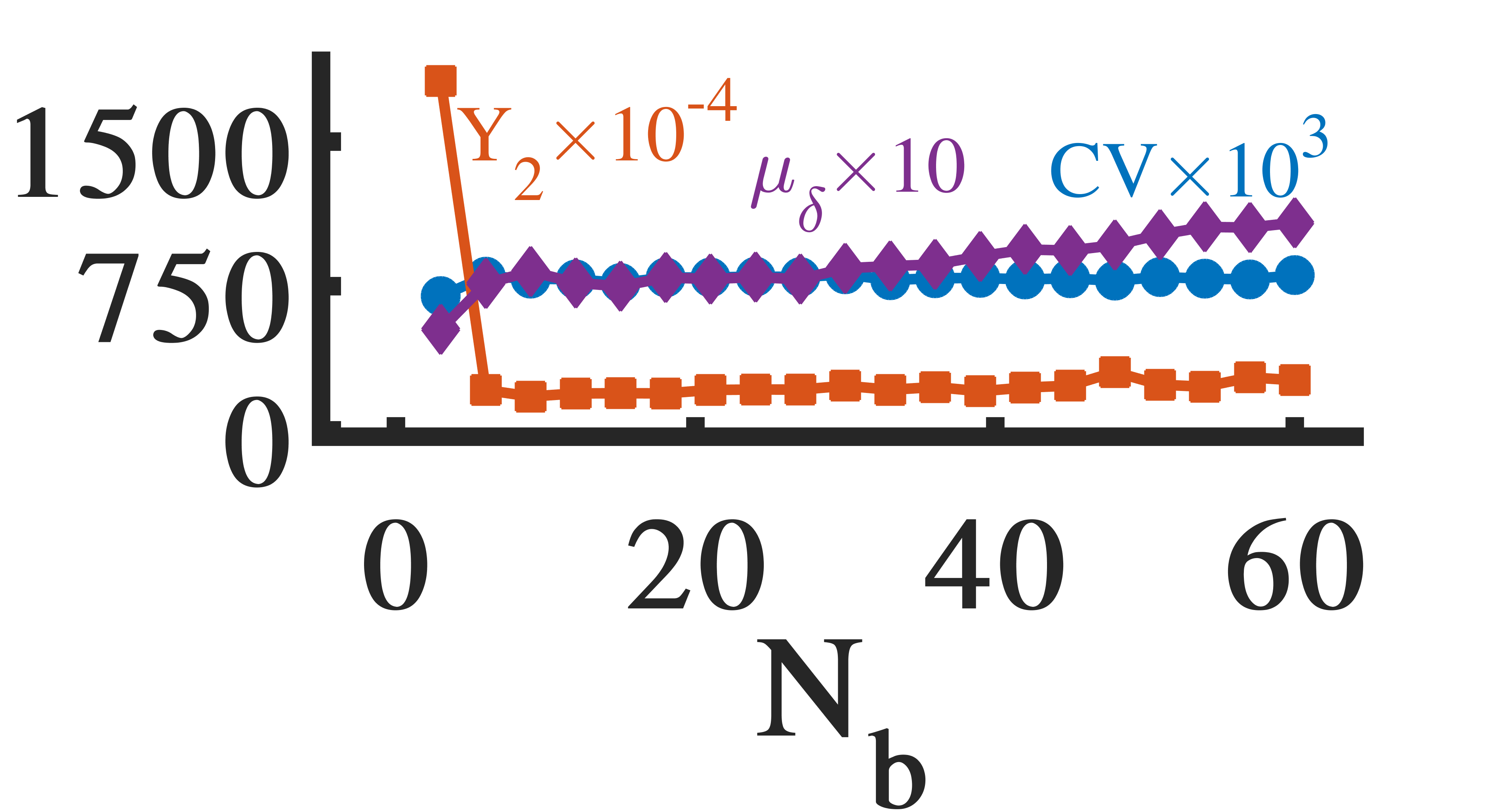}}
    \caption{Coefficient of variation (circles), excess kurtosis (squares), and mean (diamonds) of butterfly strength-differences over the timeline of burst arrivals in SPA model with $m=10,50,100$.}
    \label{fig:sdifstatsSPA}
\end{figure*}
\begin{figure*}[h]
    \centering
    \subfigure[$Pr(S_i)$, $m$$=$$10$]{\includegraphics[width=0.24\textwidth]{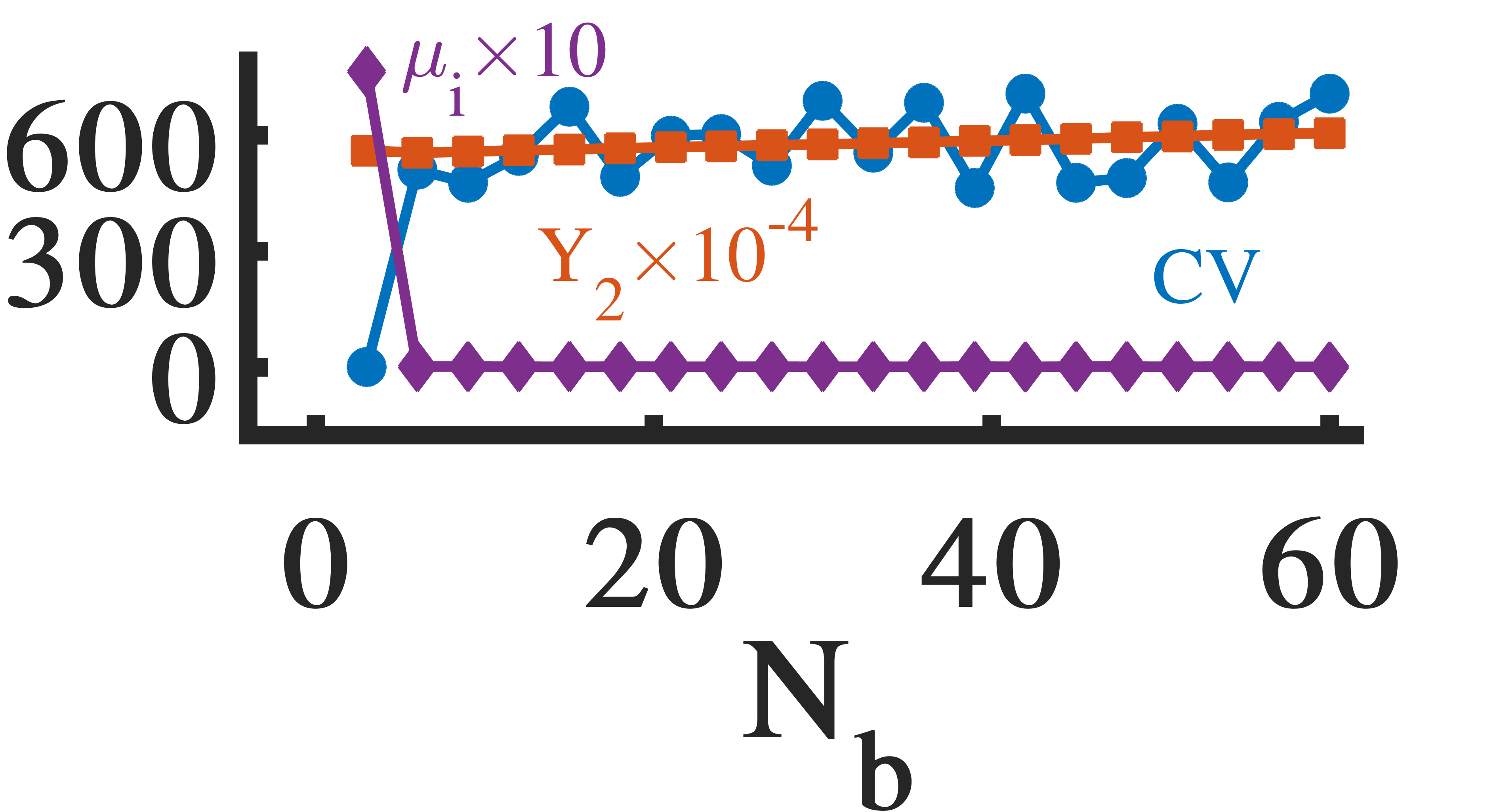}}
    \subfigure[$Pr(S_i)$, $m$$=$$50$]{\includegraphics[width=0.24\textwidth]{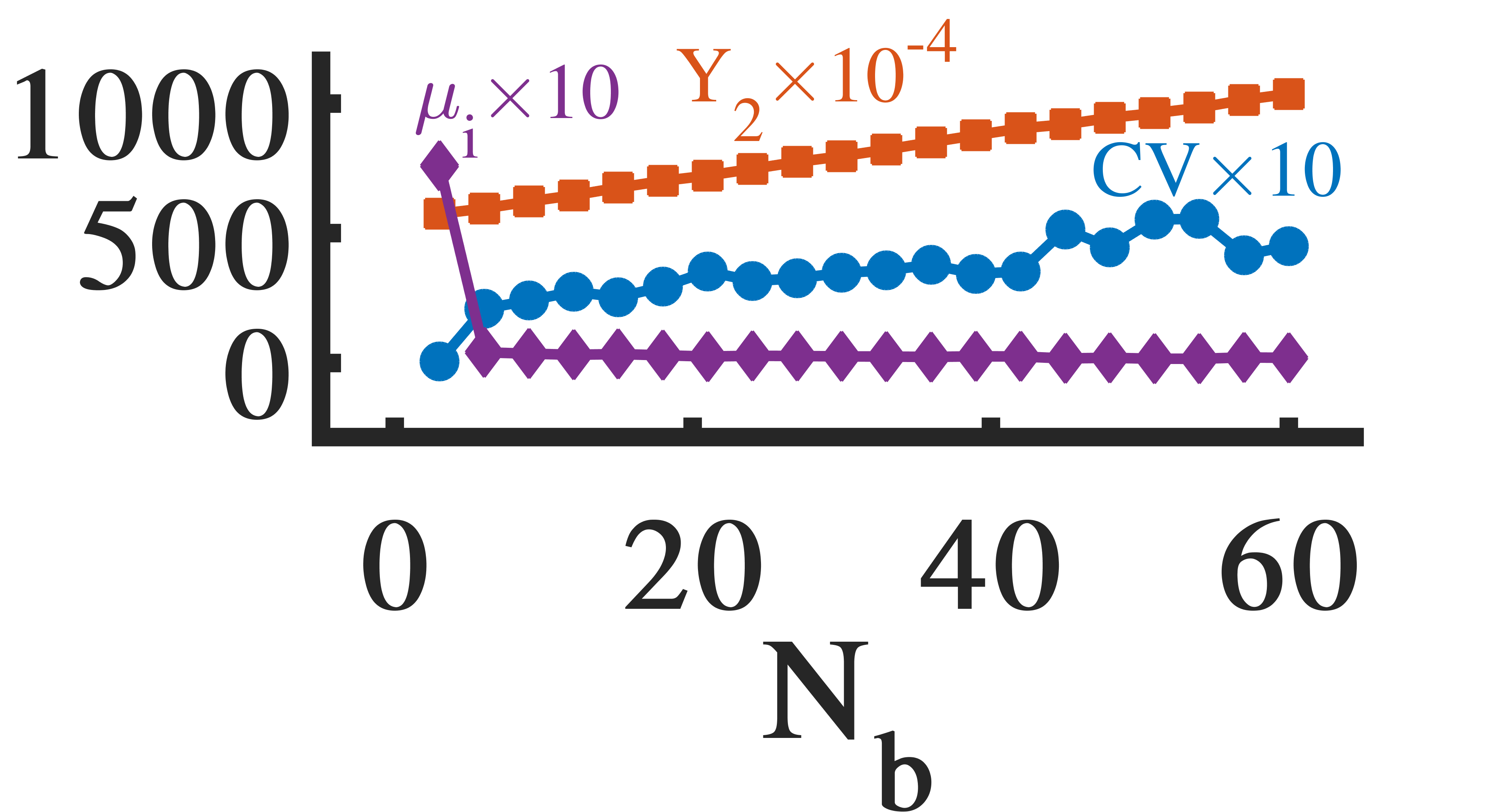}}
    \subfigure[$Pr(S_i)$, $m$$=$$100$]{\includegraphics[width=0.24\textwidth]{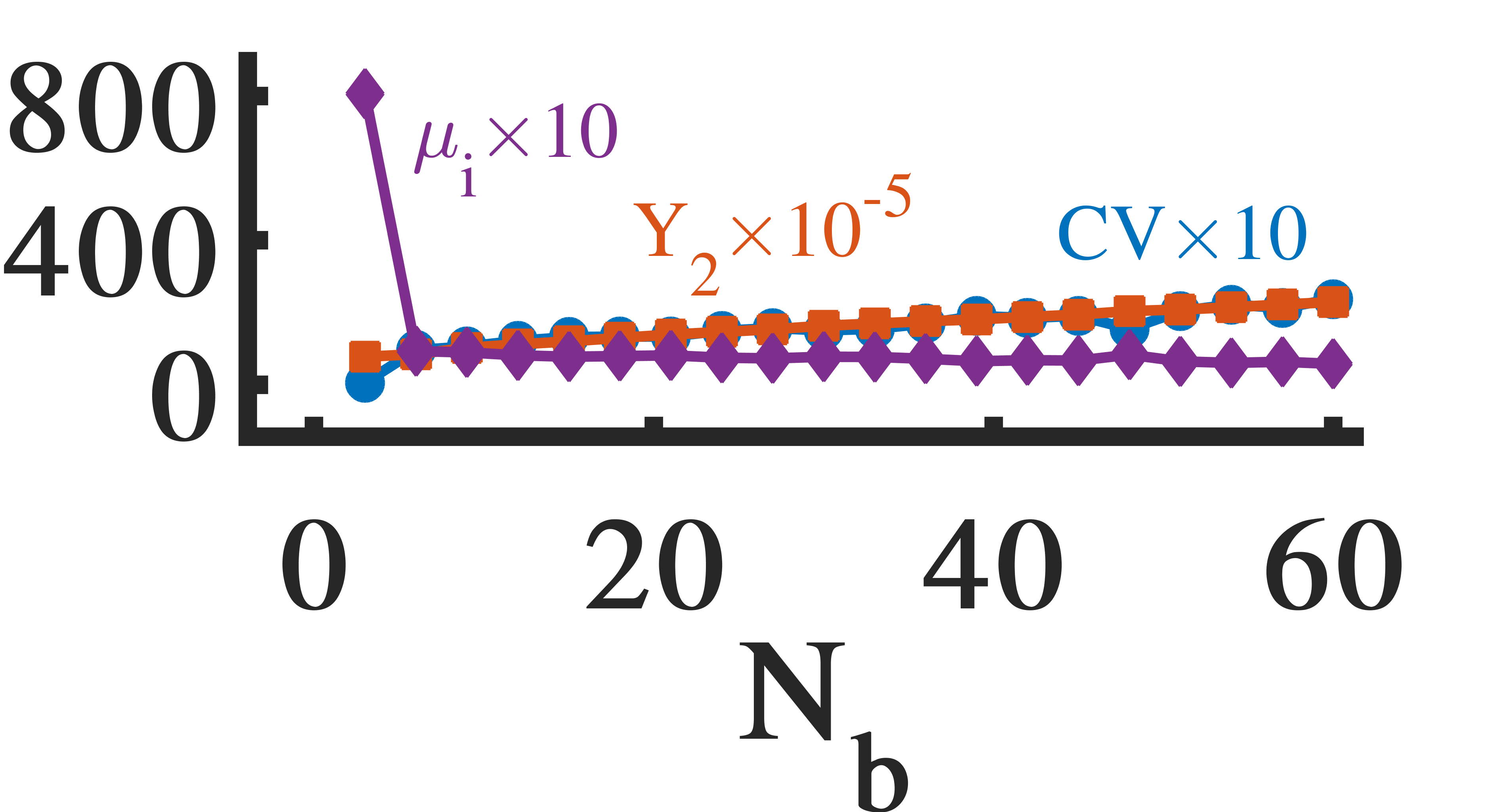}}
    
    \subfigure[$Pr(S_j)$, $m$$=$$10$]{\includegraphics[width=0.24\textwidth]{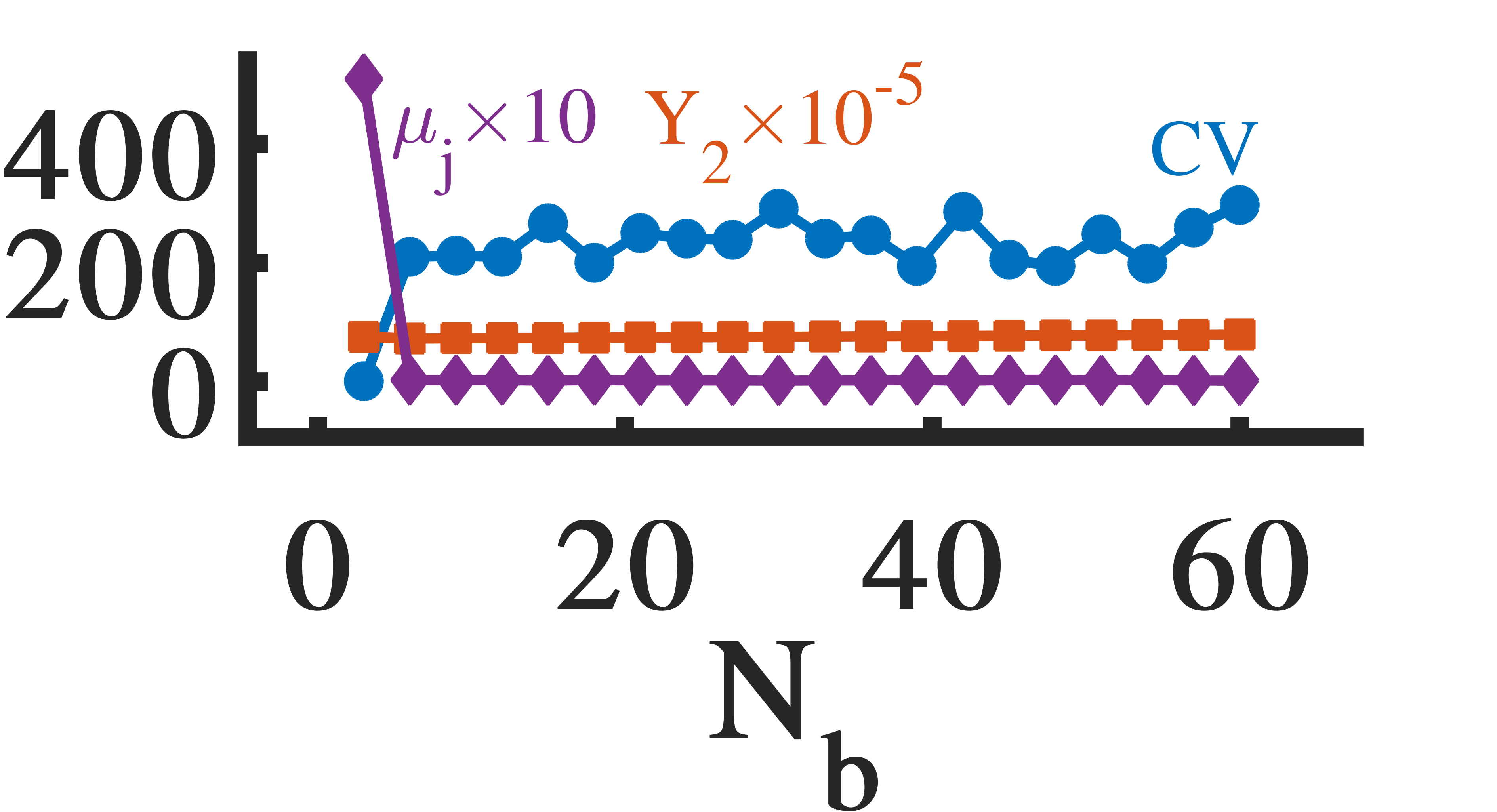}}
    \subfigure[$Pr(S_j)$, $m$$=$$50$]{\includegraphics[width=0.24\textwidth]{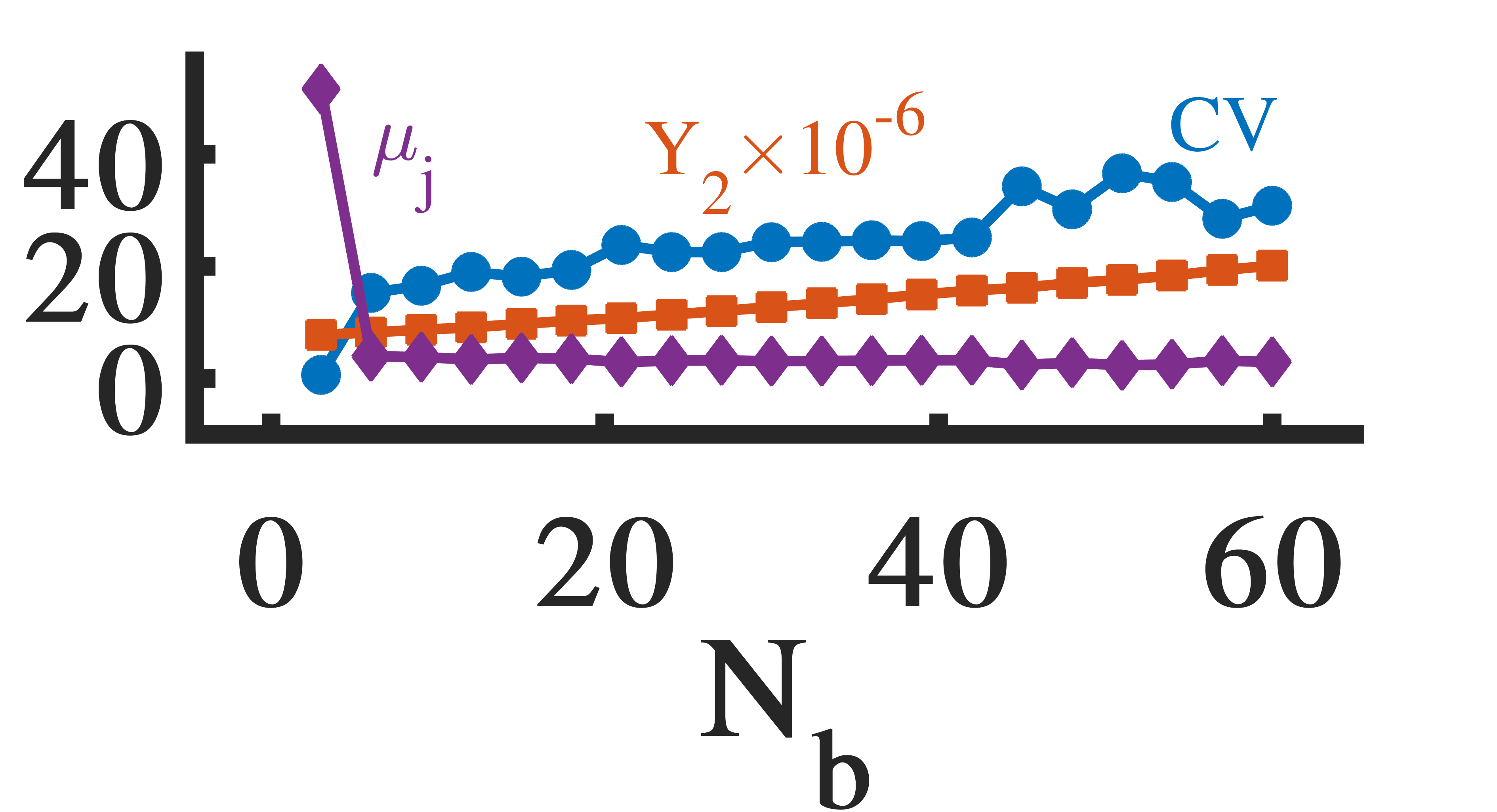}}
    \subfigure[$Pr(S_j)$, $m$$=$$100$]{\includegraphics[width=0.24\textwidth]{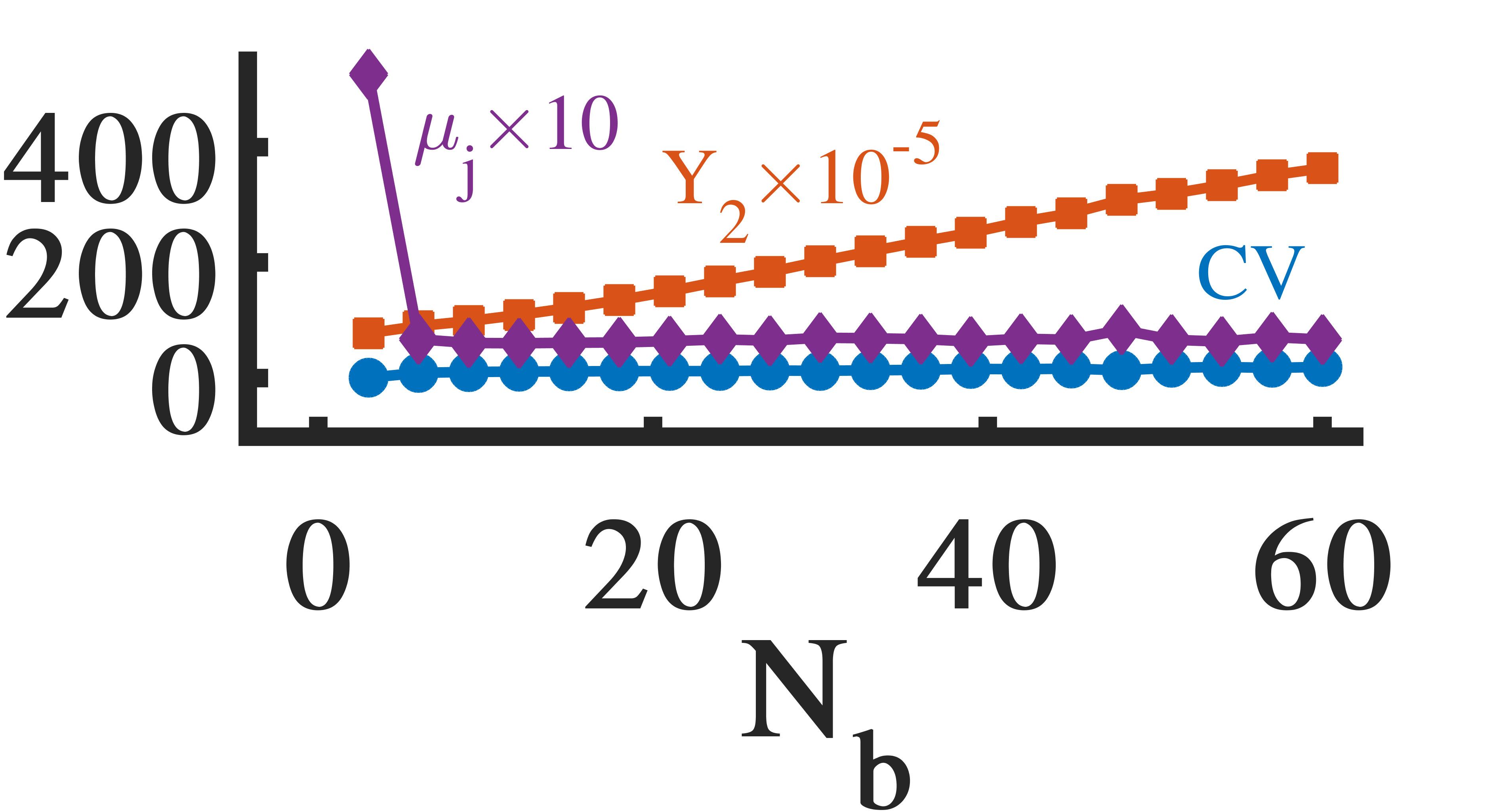}}
    \caption{Coefficient of variation (circles), excess kurtosis (squares), and mean (diamonds) of strengths of butterfly (a-c) i-vertices and (d-f) j-vertices over the timeline of burst arrivals in SPA model with $m=10,50,100$.}
    \label{fig:sstatsSPA}
\end{figure*}
\subsection{Discussion}\label{subsec:discussobservation}
We observe  the following concurrent mixing patterns held in the real-world streams as butterflies emerge over time:
\begin{enumerate}
    \item[\hypertarget{\th$_1$}{\th$_1$}] Butterfly densification -- The number of butterflies grows over time and at each time point it is a super-linear function of the number of edges. This is in line with a previous study~\cite{sheshbolouki2021sgrapp}. 
    
    \item[\hypertarget{\th$_1$}{\th$_2$}]  Strength diversification -- $Pr(S)$ of butterflies is initially mesokurtic and gets more right-skewed as the right tail grows  heavier/longer ($Y_2$ starts from $0$ and rises to extremely high values). The dispersion of strengths increases over time ($CV$$>$$1$ increases) as the standard deviation increases and the mean decreases.
    
    \item[\hypertarget{\th$_1$}{\th$_3$}] Steady strength assortativity -- The strength assortativity localization factor $r^s$ is fixed at a positive value over time due to the fixed-shaped yet growing distribution of strength-differences of butterflies. $Pr(\delta)$ is initially mesokurtic and gets more right-skewed as the right tail grows  heavier/longer ($Y_2$ starts from $0$ and rises to extremely high values). However, the dispersion of strength-differences does not change ($CV$$\approx$$1$) due to synchronous evolution of mean and standard deviation. Also, the proportion of $\delta$s in different regions of $Pr(\delta)$ is constant (stable $F$ elements). Therefore, the shape of the distribution is stable although the range expands.
\end{enumerate}
The co-occurrence of these patterns is counter-intuitive and interesting. As the stream and the number of butterflies grow rapidly, we observe that diversity of strengths for butterfly vertices increases and strong (high-strength) vertices get stronger and obtain weak neighbors with the increasing of variance of strength differences. Therefore, we expect an increasing trend of disassortativity. However, we observe that the majority of butterfly edges are formed by vertices with similar strength and this assortativity remains at a fixed level regardless of stream size or butterfly count. We refer to this phenomenon as \textit{scale-invariant strength assortativity of streaming butterflies}, which is originated by the three parallel mixing patterns of butterflies.

To explain the data-driven semantics of the observed patterns in the domain of user-item rating streams, we relate them to the following graph concepts:
\begin{itemize}
    \item Burstiness -- User-item interactions can be viewed as human-initiated events which introduce two levels of burstiness: individual-level and group-level. The former relates to the interactions of each user with several items at each time point or sequential time points with negligible differences. Such bursty interactions leads to formation of many wedges incident to each user/item. The latter relates to the concurrent interactions of several users at each time point. Such bursty interactions lead to merging the individual-level wedges and densification of butterflies. The significance of this continuous burstiness can change over time due to different circumstances leading to peak hours. For instance, Alibaba has reported that customer purchase activities during a heavy period in 2017 resulted in generation of 320 PB of log data in a six hour period~\cite{ozsu2019big}. 
    
    \item Strong-get-stronger -- Online platforms utilize filters such as trends, best sellers, mostly viewed, hot/top categories, newly added, as well as  timely promotions, point collection rewarding strategies, and (advertised) recommendations. These systematically lead to the increasing  popularity and visibility of the items with most interactions and encouraging the users to interact more and become more active. Such interaction mechanics are similar to the rich-gets-richer argument, where the richness denotes the vertex strength. The butterflies are formed incident to such highly connected and high strength users/items (strong vertices) leading to butterfly densification and diversification of strengths.
    
    \item Core-periphery -- Popular items attract the active users and in another view, active users mostly engage with trending items or make items trending/popular. This is similar to the mesoscale phenomenon core-periphery~\cite{holme2005core, csermely2013structure} also called rich club~\cite{zhou2004rich, colizza2006detecting} stating that high-degree vertices tend to connect to each other and create a core attracting the new connections. Such core sets of vertices with high degrees/strengths in user-item streams create numerous edges between strong users and items with high butterfly support leading to assortativity patterns of butterfly vertices.
\end{itemize}
 
We now discuss the properties of FF and SPA streams with respect to the observed patterns. FF streams follow \hyperlink{\th$_2$}{\th$_2$} but not \hyperlink{\th$_1$}{\th$_1$} and \hyperlink{\th$_3$}{\th$_3$}. In FF streams, as the new vertices attach to random vertices (ambassadors) and reach high-degree vertices (hubs) through copying the neighbors of ambassador, new butterflies emerge and the diversity of strength of butterfly vertices increases. When the probability of neighbor copying $p$ is low (sparse regions), the new vertex establishes fewer connections, therefore, the probability of connecting to the high-strength hubs is lower and also the strength of the new vertex remains low. As a result, many edges have low strength-difference, $Pr(\delta)$ is broader, and strength assortativity localization factor is positive. On the other hand, when $p$ is high, although the number of butterflies is higher, the connections are established between pairs of vertices with both low and high strength-difference since the ambassadors and their neighbors are selected uniformly at random. As a result, $Pr(\delta)$ has a lower variance and the strength assortativity displays randomness. Moreover, the butterfly count is a sub-linear function of the number of edges over time even when the graph displays edge densification (in dense regions). 
SPA streams follow \hyperlink{\th$_1$}{\th$_1$} but not \hyperlink{\th$_2$}{\th$_2$} and \hyperlink{\th$_3$}{\th$_3$}. In SPA streams, as new low-strength vertices attach to $m$ vertices with the highest strengths (strong vertices), many butterflies are formed around the high-strength vertices with a rate much higher than that of real-world streams. When the number of connections per new vertex $m$ is higher, the probability of attachment to low-strength vertices is higher since the number of strong vertices is limited, therefore the number of edges among low-strength vertices increases. As a result, the graphs display weak strength assortativity when average degree is high. When $m$ is low, the number of edges with high strength-difference is higher compared to the case with high $m$, although they don't exceed edges with low strength-difference. The diversity of $Pr(S)$ and $Pr(\delta)$ does not increase significantly in either cases. 

To summarize, FF streams with  implicit degree-driven preferential attachment and neighbor copying yield graphs with increasing diversity of strengths of butterfly vertices, however the quantity of butterflies and their mixing schemes do not preserve realistic patterns. SPA streams with pure strength-driven preferential attachment lead to graphs with rapidly growing butterfly density, however the mixing patterns do not match realistic patterns. This highlights the essence of a growth model which has both strength-driven preferential attachment and neighbor copying flavors to ensure a balanced butterfly densification and incremental strength diversity. Further considerations regarding the integration of these two mechanisms with other effective mechanisms are also required to create realistic streams. 
We resolve this in the next section.
\section{\texorpdfstring{The proposed Streaming Growth Model: \MakeLowercase{s}G\MakeLowercase{row}}{sGrow}}\label{sec:model}
In the previous section, we identified burstiness, strong-gets-stronger, and core-periphery as the semantic concepts explaining the butterfly emergence patterns and also the strength preferential attachment and neighbor copying as the microscopic mechanisms explaining the butterfly densification and strength diversification. Now, we integrate these concepts and growth mechanisms with further mechanisms that we introduce in the body of our streaming growth model, called \emph{sGrow}, to explain the co-occurrence of all three realistic emergence patterns of butterflies in streaming graphs such that all four-vertex graphlets emerge in the graph, the sgrs are realistic (i.e. preserve streaming data characteristics), and the stream properties are configurable. Figure~\ref{fig:mechanisms} illustrates a summary of introduced microscopic mechanisms and the techniques for implementing them. 
\tikzset{arrow/.style = {thick,black,->,>=stealth,}}
\tikzset{arrow2/.style = {thick,dotted,->,>=stealth,}}
\tikzset{nearnodes/.style={node distance=0.2cm}}
\tikzset{farnodes/.style={node distance=2cm}}
\tikzset{model1/.style = {rectangle, text width=5.5cm, minimum height=1cm,text centered, draw=black, fill=green!10,}}
\tikzset{model2/.style = {rectangle, text width=5cm, minimum height=1cm,text centered, draw=black, fill=yellow!25,}}
\begin{figure}[h]
\centering
\includegraphics[width=0.9\textwidth]{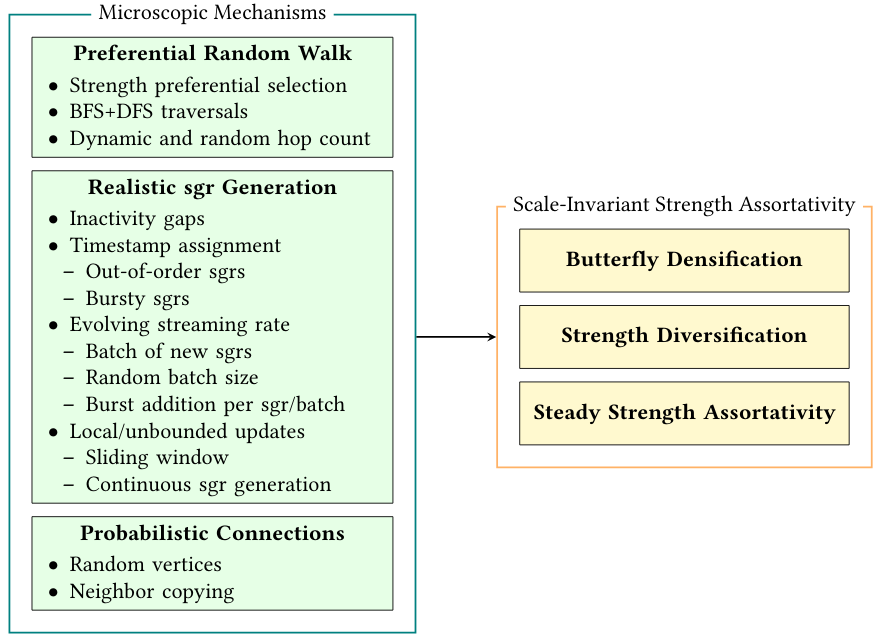}
\caption{Introduced microscopic mechanisms and techniques for explaining butterfly emergence patterns.} \label{fig:mechanisms}
\end{figure}
\subsection{Overview}\label{subsec:model}
We provide an overview of\emph{sGrow} (Algorithm~\ref{alg:model} - Figure~\ref{fig:sgrow}). We use a sliding window mechanism to generate a sequence of \hyperlink{sgr}{sgr}s that constitute the synthetic \hyperlink{wbsg}{weighted bipartite streaming graph}. In the following $G$ refers to the computational graph snapshot formed by the sgrs within the sliding window. The output stream is a sequence of sgrs denoted as  $\Re$. The time step $t$, is the computational time point used for controlling the sliding window and the timestamp $\tau$ is the sgr's time-label which follows the timestamp scale of an initial graph snapshot. We use five-point scale $[1,5]$ (similar to that of real-world streams) to generate weights. 

$G$ and $\Re$ are initiated with an initial graph snapshot $G_0 = (V_0, E_0)$. The sliding window's beginning border $W^b$ is set to the first timestamp in $G_0$ (lines~\ref{model:initstart}-\ref{model:initend} - Figure~\ref{fig:sgrow}.a). At each time step $t$, $m$ (a random number in $ [0,M)$, where $M$ is a parameter) new \hyperlink{sgr}{sgr}s $r^{l=1,..,m}$$=$$\langle v_i^l,v_j^l,\omega_{ij}^l,\tau \rangle$ with new vertices are created and added to $\Re$ and $G$ (Figure~\ref{fig:sgrow}.b,c). The shared timestamp is one plus the last timestamp in $G$ and the weights are random integers $\omega_{ij}^l$$\in$$[1,5]$ (line \ref{model:addmsgrs}). To connect these new isolated edges to the rest of sgrs, the following procedure is followed. A random integer $\omega$$\in$$[-1,5]$ is generated based on which, one of the following three operations is performed for each of the $m$ new sgrs in parallel (lines~\ref{model:w}-\ref{model:otherend}). Additions and removals happen as described in \S~\ref{subsec:addremoveedge}.

\begin{itemize}
    \item $\omega$$=$$-1$, the connection between $v_i^l$ and $v_j^l$ is removed from $\Re$ and $G$ (line~\ref{model:removeedge}). This edge removal introduces isolated vertices if the vertices do not acquire neighbors from the current or next batch ($v_i^{12}$, $v_j^{12}$, and $v_i^{13}$ in Figure~\ref{fig:sgrow}.b). Since real-world streaming graphs are dominated by edge additions, we give a lower probability to removals.

    \item $\omega$$=$$0$, nothing happens (line~\ref{model:noop}). This no-operation introduces a gap of inactivity between streaming records to form \hyperlink{burst}{burst}s and also introduces isolated edges. If the current sgr $r^l$ is not connected to the subsequent sgrs in current batch ($r^{l+1,..,m}$) or the following batches, it will remain isolated ($r^{22}$ in Figure~\ref{fig:sgrow}.c).

    \item $\omega$$>$$0$, a j-vertex $u_j^0$ in $G$ ($v_j^{01}$ in Figure~\ref{fig:sgrow}.b and $v_j^{14}$ in Figure~\ref{fig:sgrow}.c) is randomly selected  via \textit{Strength Preferential Selection} (\hyperlink{SPS}{SPS}) (\S~\ref{subsec:SPS}). Next, a \textit{Preferential Random Walk} (\hyperlink{PRW}{PRW}) starting from $u_j^0$ is performed in $G$ (\S~\ref{subsubsec:PRW}). The number of hops in PRW (i.e. walk length) is a random integer $L$ in the parameter range $[L_{min},L_{max}]$.  Using this PRW as a backbone, \hyperlink{burst}{burst}s of new sgrs between $r^l$ and the rest of sgrs in $G$ and $\Re$ are established (\S~\ref{subsec:burst}). 
    After adding the last edge  with weight $\omega'$, the timestamp $\tau$ is incremented as a function of $\omega'$: $\tau$$=$$\tau$$+$$|(\omega'-5)(\omega'-4)(\omega'-3)$$/$$2|$. This function creates a timestamp interval as soon as generation of a sgr with low weight (i.e. $\omega'\leq2$), therefore it helps characterize the burstiness of the stream (lines~\ref{model:otherstart}-\ref{model:otherend}). This is based on an observation in real-world streams that the last sgr in each burst has a low weight. It is noteworthy that there are two level of burstiness: (1) the bursts initiated by each $r^l$ as described here, and (2) the bursts created concurrently by $m$ sgrs $r^{l=1,..,m}$ which can be assumed as multiple generative sources.
\end{itemize}  
\begin{figure}[t]

\subfigure[$t$$=$$0$, $W^b=r^{01}.\tau$]{\includegraphics[width=0.38\textwidth]{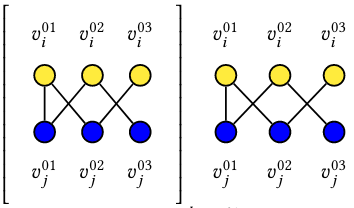}

}

\subfigure[$t$$=$$1$, $W^b=r^{01}.\tau+\beta$, $L$$=$$1$, $m$$=$$4$ ($r^{11}$  incurs no-op, $r^{12}$ \& $r^{13}$ incur removal, $r^{14}$ adds burst.)]{\includegraphics[width=0.78\textwidth]{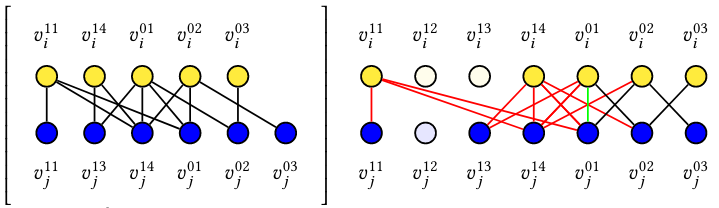}}

\subfigure[$t$$=$$2$, $W^b=r^{01}.\tau+2\beta$, $L$$=$$2$,  $m$$=$$2$ ($r^{21}$  adds burst, $r^{22}$ incurs no-op.)]{\includegraphics[width=0.9\textwidth]{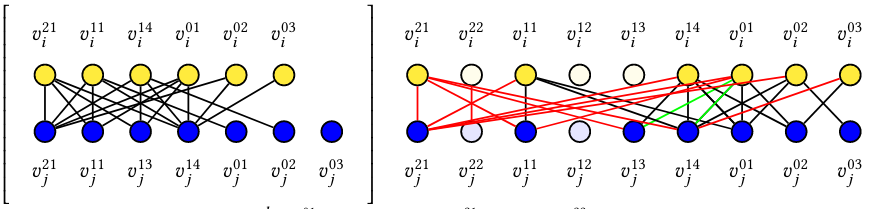}}

\caption{(left) The computational graph $G$ (right) and the stream $\Re$ at the end of time steps $t$$=$$0,1,2$ with $\beta$$=$$2$, $\rho$$=$$0.4$, $M$$=$$5$, and $L$$\in$$[1,2]$. New edges are in red and PRW edges are in green. (a) At $t$$=$$0$, the stream and the computational graph are initiated with $G_0$. (b) At $t$$=$$1$, a batch of 4 new sgrs ($r^{l=11,12,13,14}$) with same timestamps is created which incur no-op, removal, removal, and burst addition, respectively. $v_j^{13}$  acquires two neighbors after the removal of its connection to $v_i^{13}$. (c) At $t$$=$$2$, a batch of 2 new sgrs ($r^{l=21,22}$) with same timestamps is created which incur burst addition and no-op, respectively. $r^{22}$ does not acquire any neighbor and consequently is removed from $G$.  For simple illustration, we do not depict timestamps and weights in this figure and assume that the edges in $G_0$ expire from the window as their timestamps are below the $W^b$.}\label{fig:sgrow}
\end{figure}
The aforementioned procedure takes place for all  $m$ new sgrs in parallel\footnote{We implemented this algorithm in a single machine shared-memory architecture; a distributed version is doable but not considered in this paper.}.  After that, any newly added vertex with less than two neighbors is removed from  $G$ (line~\ref{model:removeisolated} -  $v_i^{22}$ and $v_j^{22}$ in Figure~\ref{fig:sgrow}.c). This removal of new isolated vertices is done to retain a connected computational graph, yet the old vertices whose adjacent edges are discarded by the window (as described below) become isolated in $G$ ($v_j^{03}$ in Figure~\ref{fig:sgrow}.c). Also, the stream may hold isolated vertices/edges ($v_i^{22}$-$v_j^{22}$, $v_i^{12}$, $v_j^{12}$, and $v_i^{13}$ in Figure~\ref{fig:sgrow}.c). Next, the timestamp is incremented by one (line~\ref{model:stamp}). The window slides as $W^b$ is incremented by $\beta$; the edges with timestamps out of the sliding window are removed from the graph after each $\beta$ time steps; and the time step is reset to zero (lines~\ref{model:slidestart}-\ref{model:slideend}). This sliding window mechanism is used to avoid pure preferential attachment to old vertices in global scale and create time-sensitive and local connections leading to emergence of young hubs (high degree vertices) to support the observations of real-world stream analysis in~\cite{sheshbolouki2021sgrapp}. The generation process happens continuously and $\Re$ streams-out as the sgrs are generated. This process can be restricted to continue until a desired number of sgrs $S$ are generated ($|\Re|=S$) and then return the stream $\Re$.
\begin{algorithm}[h]\caption{sGrow}\label{alg:model}
  \DontPrintSemicolon
    \KwData{$G_0$: an initial graph}
    \KwInput{$\rho$:  connection probability, $M$: maximum number of new edges, 
    $\beta$: slide parameter,  
    $[L_{min},L_{max}]$: range of PRW's length
    }
    \KwOutput{$\Re$,  sequence of streaming graph records}
    
     $G \gets G_0=(V_0,E_0)$ \label{model:initstart} \tcp*{computational graph}
    $\Re\gets E_0$   \tcp*{sequence of sgrs}
    $\tau \gets 1+$ last timestamp in $G_0$  \tcp*{timestamp}
    $t\gets0$           \tcp*{time step}
    $W^b\gets$  \label{model:initend} first timestamp in $G_0$  \tcp*{sliding window's beginning border} 
    
    \While{true}{
    $t \gets t+1$
   \\ Add $m\in[0,M)$ new sgrs $r^{l=1,..,m}$ to $\Re$ and $G$\label{model:addmsgrs}\\
    
    \For{\textbf{each} $r^{l=1,..,m}=\langle v_i^l, v_j^l, \omega_{ij}^l, \tau \rangle$}{
        $\omega\gets$ a random integer in $[-1,5]$\label{model:w}
        
        \Switch{$\omega$}{
            \Case{-1}{
                Remove $v_i^l-v_j^l$ from $\Re$ and $G$.\label{model:removeedge}
            }
            \Case{0}{
                No operation\label{model:noop}
            }
            \Other{
                $u_j^0$$\gets$\hyperlink{SPS}{SPS($V_j$)} \label{model:otherstart}
                \\$L\gets $ a random integer in $[L_{min},L_{max}]$
                \\$(PRW_i,PRW_j) \gets$ \hyperlink{PRW}{PRW($u_j^0$, $false$, $G$, $L$)}
                \\ \hyperlink{addBurst}{$addBurst(v_i^l,v_j^l,PRW_i,G,\Re,\rho)$}
                \\ \hyperlink{addBurst}{$addBurst(v_i^l,v_j^l,PRW_j,G,\Re,\rho)$}
                
                $\tau\gets\tau +|\frac{(\omega'-5)(\omega'-4)(\omega'-3)}{2}|$\label{model:otherend}
            }
            
        }
    }
    Remove any newly added vertex $v_i^l$ and $v_j^l$ with less than 2 neighbors from $G$\label{model:removeisolated}
    \\$\tau \gets \tau+1$ \label{model:stamp}
    \\$W^b\gets W^b+\beta$ \label{model:slidestart}
    
    \If{$t=\beta$}{
        Remove any edge with timestamp less than $W^b$ from $G$
        \\$t\gets 0$\label{model:slideend}
    }
    }
\end{algorithm}

\subsection{Data Structures}\label{subsec:addremoveedge}
We describe the object-oriented data structures and basic graph/stream operators used in the algorithms. A vertex is an object with three attributes: ID, Strength, and timestamp $\tau$. A new i(j)-vertex $v_i$($v_j$) is assigned an integer ID equal to the current number of i(j)-vertices, a strength initialized to zero, and a timestamp equal to that of the edge by which this vertex is added. We use dot notation to refer to attributes of an object, e.g. $v_i.ID$ denotes the ID of the vertex $v_i$. An edge/\hyperlink{sgr}{sgr} between vertices $v_i$ and $v_j$ is an object with four attributes: 
i-vertex (object $v_i$), j-vertex (object $v_j$), timestamp (integer $\tau$), and weight (integer $\omega$).  The connections of graph $G$ are stored by two hash-map data structures to map each vertex ID to the hash-set of its immediate neighbors: \textit{iNeighbors}$=\{ (v_j.ID:N_i(v_j))\}$ and \textit{jNeighbors}$=\{ (v_i.ID:N_j(v_i))\}$. We use hash-sets since we don't store multiple edges between two vertices in the computational graph and we use hash-map for fast access to the neighborhoods. The sgrs in $\Re$ are stored in a vector that retains the edges in the order of their additions which include out-of-order sgrs wrt timestamps as we explain in \S~\ref{subsec:burst}. 
When a new edge is added/removed to/from the graph or stream these data structures are updated accordingly and also the strengths of the vertices at the either ends of the edge are incremented/decremented by the weight of the edge.
\subsection{Strength Preferential Selection}\label{subsec:SPS}
We describe  function $SPS(V)$ given in Algorithm~\ref{alg:sps}, which is invoked in Algorithms~\ref{alg:model} and \ref{alg:prw}. This function  selects a random vertex in the set $V$ according to \textit{strength preferential probability} $\Lambda_v=\frac{v.strength}{\Sigma_{v'\in V}v'.strength}$~\cite{barrat2004modeling,barrat2004weighted}. Vertices in $V$ are concurrently added to a list with multiplicity equal to their strength (lines~\ref{sps:addvstart}-\ref{sps:addvend}). Next, the list is shuffled (line~\ref{sps:shuffle}) and a random element $v_0$ in the list is selected as the output vertex (lines~\ref{sps:selectvstart}-\ref{sps:selectvend}).
\subsection{Preferential Random Walk}\label{subsubsec:PRW}
We describe function $PRW(starter,isI,G,L)$ given in Algorithm~\ref{alg:prw}. This function performs a random walk  with $L$ hops on a graph $G$. It starts from a $starter$ vertex whose type determines a boolean flag $isI$ (true when $starter$ is an i-vertex and false otherwise). At each hop, a neighbor of the $starter$ vertex  ($u_i\in N_i(starter)$, $u_j\in N_j(starter)$) is selected via strength preferential selection (invoking $SPS(N_i(starter))$, $SPS(N_j(starter))$ -- Algorithm~\ref{alg:sps}). The selected neighbor ($u_i$, $u_j$) is added to a hash set of unique vertices ($PRW_i$, $PRW_j$) and is set as the $starter$ vertex. The starter flag is accordingly set and the hop counter is incremented by one (lines~\ref{prw:ihopstart}-\ref{prw:ihopend} and \ref{prw:jhopstart}-\ref{prw:jhopend}). The next hop starts with the last added vertex. When the current selected neighbor is already in the hash set, if it is the last element, the walk continues to another neighbor of that vertex (in depth traversal) and if it is one of the previously selected vertices other than the last element, the walk continues in breadth traversal. Therefore, the walk is a combination of BFS and DFS with random preferential selection. 

\subsection{Burst Addition}\label{subsec:burst}
We describe function $addBurst(v_i^l, v_j^l, PRW_i, G,\Re, \rho)$ given in Algorithm~\ref{alg:addburst}. This function adds \hyperlink{burst}{burst}s of sgrs to $G$ and $\Re$ based on the i-vertices in the $PRW_i$ and new vertices $v_i^l$ and $v_j^l$ given a probability parameter $\rho$ (Figure~\ref{fig:addburst}.a). We follow the same procedure to add bursts with respect to the j-vertices in the $PRW_j$. As illustrated in Figure~\ref{fig:addburst}, considering each i-vertex $u_i$ in the $PRW_i$, the following connections are established:
    \begin{itemize}
        \item[Step 1] An edge between $u_i$ and the newly added $v_j^l$ is formed with the timestamp of $v_j^l$ and a weight $\omega'\in [1,5]$ (lines~\ref{burst:addsgrstart}-\ref{burst:addsgrend} -- Figure~\ref{fig:addburst}.b). This edge connects the edge $v_i^l$-$v_j^l$ to the graph and also leads to emergence of $N_1+N_2$ caterpillars (solid 3-paths in Figure~\ref{fig:addburst}.e,f ), where $N_1$ and $N_2$ are the number of 1-hop (immediate) and 2-hop neighbors of $u_i$, respectively.  
    
        \item[Step 2] With probability $\rho$, an edge between $u_i$ and an existing j-vertex $z_j$, selected uniformly at random, is formed with timestamp $Min(u_i.\tau,z_j.\tau)$ and a weight $\omega'\in [1,5]$ (lines~\ref{burst:randomstart}-\ref{burst:randomend} -- Figure~\ref{fig:addburst}.c). 
        Using a timestamp other than the current timestamp ($v_j^l.\tau$) introduces out-of-order sgrs (late arrival) since this sgr has a timestamp less than ($v_j^l.\tau$). It also helps balance the burst sizes since the current timestamp is not assigned to all edges in the current time step. This probabilistic edge leads to converting the caterpillars between $u_i$ and $z_j$ into butterflies at the generation time of either vertices (closed 4-path in Figure~\ref{fig:addburst}.g).
    
        \item[Step 3] With probability $\rho$, an edge between the newly added $v_i^l$ and each $u_i$'s immediate j-vertex neighbor $n_j$ is concurrently formed with $n_j$'s timestamp $n_j.\tau$ and a  weight $\omega'\in [1,5]$ (lines~\ref{burst:copystart}-\ref{burst:copyend} -- Figure~\ref{fig:addburst}.d). In other words, each of the adjacent links of $u_i$ is copied with probability $\rho$. Since $n_j.\tau <v_i^l.\tau$, out-of-order sgrs join previous burst of sgrs with same timestamp (including the sgr incident to $n_j$s). These probabilistic edges lead to converting the $N_1$ caterpillars emerged in Step 1 into butterflies (closed 4-path in Figure~\ref{fig:addburst}.e). We run this step for generating streams with high number of sgrs per burst.
    \end{itemize}
    
    
\begin{figure}[t]\centering
\subfigure[]{\includegraphics[width=0.28\textwidth]{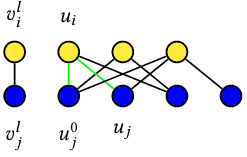}}%
\hspace{5mm}\subfigure[Step 1]{\includegraphics[width=0.28\textwidth]{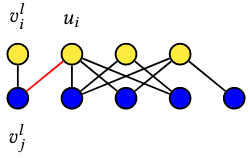}}
\hspace{5mm}\subfigure[Step 2]{\includegraphics[width=0.28\textwidth]{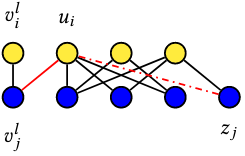}}

\subfigure[Step 3]{\includegraphics[width=0.28\textwidth]{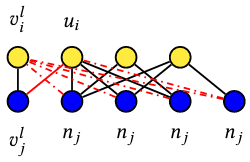}}
\hspace{5mm}\subfigure[]{\includegraphics[width=0.099\textwidth]{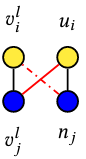}}
\hspace{5mm}\subfigure[]{\includegraphics[width=0.09\textwidth]{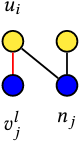}}
\hspace{5mm}\subfigure[]{\includegraphics[width=0.094\textwidth]{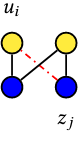}}
\caption{Four-vertex graphlets (e, f, g) and schematic burst addition steps (b, c, d) based on an i-vertex $u_i$ in the PRW starting from $u_j^0$ designated by green lines (a). New edges are colored in red and dashed lines denote probabilistic connections.}\label{fig:addburst}
\end{figure}
\begin{algorithm}[h]\caption{Strength Preferential Selection}
\label{alg:sps}
  \DontPrintSemicolon
   \SetKwProg{Fn}{Function}{}{}
   \Fn{\hypertarget{SPS}{SPS}($V$)}{
    list$\gets \emptyset$
    
    \For{\textbf{each} $v\in V$}{ 
        $S\gets v.strength$\label{sps:addvstart}
        
        \For{k=1 to S}{
            list.add($v$)\label{sps:addvend}
        }
    }
    shuffle(list)\label{sps:shuffle}
    \\index$\gets$ a random integer in $[0,size(list))$\label{sps:selectvstart}
    \\$v_o\gets list[index]$\label{sps:selectvend}
    \\Return $v_o$ 
}
\end{algorithm}
\begin{algorithm}[h]\caption{Preferential Random Walk}
\label{alg:prw}
  \DontPrintSemicolon
   \SetKwProg{Fn}{Function}{}{}
   \Fn{\hypertarget{PRW}{PRW}($starter, isI, G, L$)}{
    
    $h\gets 0$   \tcp*{hop counter}
    $(PRW_i,PRW_j) \gets \emptyset$  \tcp*{two hash-sets of unique i-vertices and j-vertices of the walk }
    
    \While{$h<L$}{
        \If{$isI$}{
            $u_j\gets$ \hyperlink{SPS}{$SPS(N_j(starter))$}\label{prw:ihopstart} 
            \\Add $u_j$ to $PRW_j$
            \\$starter\gets u_j$
            \\$isI\gets false$
            \\$h\gets h+1$\label{prw:ihopend}
        }
        \Else{
            $u_i\gets$ \hyperlink{SPS}{$SPS(N_i(starter))$}\label{prw:jhopstart} 
            \\Add $u_i$ to $PRW_i$
            \\$starter\gets u_i$
            \\$isI\gets true$
            \\$h\gets h+1$\label{prw:jhopend}
        }
    }
    Return $(PRW_i,PRW_j)$
}
\end{algorithm}
\begin{algorithm}[h]\caption{Add Burst}
\label{alg:addburst}
  \DontPrintSemicolon
   \SetKwProg{Fn}{Function}{}{}
   \Fn{\hypertarget{addBurst}{addBurst}$(v_i^l, v_j^l, PRW_i, G, \Re, \rho$)}{
   \For{\textbf{each} $u_i\in PRW_i$}{
                    $\omega' \gets$ a random integer in $[1,5]$\label{burst:addsgrstart}
                    \\Add a new sgr $\langle u_i, v_j^l, \omega', v_j^l.\tau  \rangle$ to $\Re$ and $G$\label{burst:addsgrend}\\
                    \If{coin($\rho$) is Head}{\label{burst:randomstart}
                        $z_j\gets$ Select a random j-vertex 
                        \\$\omega' \gets$ a random integer in $[1,5]$
                        \\Add a new sgr $\langle u_i, z_j, \omega', Min(u_i.\tau,z_j.\tau)  \rangle$ to $\Re$ and $G$ \label{burst:randomend}\\
                    }
                    \For(\tcp*[f]{in highly bursty streams} ){\textbf{each} $n_j\in N_j(u_i)$}{ \label{burst:copystart}
                        \If{coin($\rho$) is Head}{
                            $\omega' \gets$ a random integer in $[1,5]$
                            \\Add a new sgr $\langle v_i^l, n_j, \omega', n_j.\tau  \rangle$ to $\Re$ and $G$\label{burst:copyend}
                        }
                     
                    }
                }
}
\end{algorithm}
\section{Evaluations}\label{sec:evaluations}
In this section, we analyze the computational complexity of \emph{sGrow} and evaluate its performance from three perspectives:
\begin{itemize}
    \item Computational Complexity (\S~\ref{subsec:complexity}) -- We theoretically analyze the computational complexity of \emph{sGrow}.
    \item Pattern Reproducing (\S~\ref{subsec:patternreproducing}) --  We examine the ability of \emph{sGrow} to reproduce the realistic patterns under different levels of burstiness, initial graph snapshots $G_0$, and butterfly counts. We create streaming graphs with a prefix of $1000$ edges from real-world streams ($G_0$) and the rest of the stream is synthesized via \emph{sGrow} with various parameter configurations and different number of butterflies. We refer to the generated streams as S-\{$G_0$-name\}.
    \item Stress Testing (\S~\ref{subsec:stresstesting}) -- We examine the impact of introduced parameterized techniques on the effectiveness, efficiency, and burstiness of the generated stream. We also provide a reference guide for setting the parameters. We use a prefix of $1000$ edges from Amazon stream and synthesize the rest of the stream via \emph{sGrow} with different parameter configurations.
\end{itemize}  
Our experiments as well as the analysis in previous sections are conducted on a machine with $15.6$ GB native memory and Intel Core $i7-6770HQ CPU @ 2.60GHz * 8$ processor. We have implemented all algorithms in Java (OpenJDK version $11.0.11$). 

\subsection{Computational Complexity}\label{subsec:complexity}
In the following, we show that high degree/strength vertices and PRW hop count determine the computational expenses of \emph{sGrow}.
 \begin{theorem}
The worst case computational complexity of \emph{sGrow} in each window with graph $G$$=$$(V_i\cup V_j,E)$ and PRW parameter $L$ is  $\mathcal{O}(S_{max}+L(N_{max}^j+N_{max}^i))$, where $S_{max}$ is the maximum strength in $G$ and $N_{max}^i$ and $N_{max}^j$ are the maximum number of i-neighbors and j-neighbors for vertices in $V_i$ and $V_j$.
 \end{theorem}

 \begin{proof}
\emph{sGrow}'s computations at each window are dominated by burst additions as the initializations, sgr addition/removals, and window sliding take one unit of computation. The worst case computational complexity of burst additions is the following.

 \begin{equation}
     \mathcal{O}(\hyperlink{SPS}{SPS(.)})+\mathcal{O}(\hyperlink{PRW}{PRW(.)})+\mathcal{O}(2\hyperlink{addBurst}{addburst(.)})
 \end{equation}

 Let us assume that the maximum strength in $G$$=$$(V_i\cup V_j,E)$ is $S_{max}$, $L$ is the parameter for the PRW hop count, and the maximum number of i-neighbors and j-neighbors for vertices in $V_i$ and $V_j$ are $N_{max}^i$ and $N_{max}^j$. Accordingly, we would have the following complexities: $\mathcal{O}(\hyperlink{SPS}{SPS(V_j)})=S_{max}$ since the value assignments and corresponding operators, and the list shuffling take one unit of computations, and the outer for loop is parallel and the inner loop sequentially performs $S_{max}$ computational units. We have  $\mathcal{O}(\hyperlink{PRW}{PRW(u_j^0,false,G,L)})=(L/2+1)N_{max}^j+(L/2)N_{max}^i$ since $|PRW_j|=L/2+1$ and $|PRW_i|=L/2$. Also, $\mathcal{O}(\hyperlink{addBurst}{addburst(v_i^l,v_j^l,PRW_j,G,\Re,\rho)})= \mathcal{O}(\hyperlink{addBurst}{addburst(v_i^l,v_j^l,PRW_i,G,\Re,\rho)})=\mathcal{O}(1)$ since the value assignments, and probabilistic connections are done in $\mathcal{O}(1)$ and Step 3 is performed via a parallel loop. Therefore, the total complexity of burst additions would be $S_{max}+(L/2+1)N_{max}^j+(L/2)N_{max}^i+\mathcal{O}(1)$, i.e. high degree/strength vertices and PRW hop count determine the computational cost.
 \end{proof}

\subsection{Pattern Reproducing}\label{subsec:patternreproducing}
As provided in Table~\ref{tab:graphs}, Epinions, Amazon, and Ciao are bursty streams with average burst sizes of $b=27282$, $1753.7$, and $14.8$; Yahoo and ML1m are also bursty but with lower values $b=2.4$ and $2.2$;  ML100k and WikiLens have the lowest burstiness with $b=2$ and $1$, respectively. According to these burstiness profiles, we set the parameters to control the temporal distribution of sgrs. 
That is, to simulate a stream with high burstiness, we set $M$ and $[L_{min},L_{max}]$ to high values to increase the probability of creating high number of new edges at each timestep ($M$) and increase the burst size by generating backbone walks with more vertices ($[L_{min},L_{max}]$). 
To simulate a stream with low burstiness (S-ML100k and S-WikiLens), we do not perform the neighborhood copying (lines~\ref{burst:copystart}-\ref{burst:copyend} in Algorithm~\ref{alg:addburst} -- step 3 in \S~\ref{subsec:burst}). The default value of $\rho$ is $0.3$ and we further adjust it by decreasing (increasing) to push the burst size towards lower (higher) values. 
We set $\beta=5$ in all streams. The exact value of parameters are given in Table~\ref{tab:parameters}. All of the reported results in Figures~\ref{fig:butterflycountmodel}-\ref{fig:rssynthetic} are based on the same stream.
\begin{table*}[h]\caption{ Parameters}
\small \centering
    \begin{tabular}{p{1.5cm}p{0.5cm}p{0.5cm}p{0.5cm}p{1cm}}
          & $\rho$ & $M$ & $\beta$ & $[L_{min},L_{max}]$\\\hline
         S-Ciao & $0.3$ & $100$ & $5$ & $[2,3]$\\
         S-Epinions & $0.2$ & $300$ & $5$ & $[3,6]$ \\
         S-WikiLens & $0.4$ & $100$  & $5$ & $[1,3]$\\
         S-ML100k & $0.3$ & $100$ & $5$ & $[1,3]$\\
         S-ML1m & $0.3$ & $10$  & $5$ & $[1,2]$\\
         S-Amazon & $0.3$ & $50$ & $5$ & $[1,2]$\\
         S-yahoo & $0.3$ & $50$ & $5$ & $[2,3]$
    \end{tabular}\label{tab:parameters}
\end{table*}

In the following, we compare the synthetic streams generated by \emph{sGrow} and real-world streams with respect to the patterns observed in real-world streams qualitatively by checking whether the patterns hold and quantitatively by checking the error of $r^s$ and $F$. We observe that \emph{sGrow} streams obey the scale-invariant strength assortativity of butterflies since all the synthetic streams preserve the realistic mixing patterns \hyperlink{\th$_1$}{\th$_1$}, \hyperlink{\th$_2$}{\th$_2$}, and \hyperlink{\th$_3$}{\th$_3$} regardless of the initial graph snapshot, butterfly count, and burstiness level:
\begin{itemize}
    \item[\hyperlink{\th$_1$}{\th$_1$}] The number of butterflies grows over time (Figure~\ref{fig:butterflycountmodel}) with an average butterfly rate higher than $1$ in all the streams (Table~\ref{tab:butterflyrateModel}) indicating that the butterfly count grows super-linearly with respect to the number of edges .
    \item[\hyperlink{\th$_2$}{\th$_2$}] The mean, relative standard deviation, and tail skewness/heaviness of $Pr(S)$ increases over time (Figure~\ref{fig:sstatssynthetci}) showing that butterfly vertex strengths diversify over time.
    \item[\hyperlink{\th$_3$}{\th$_3$}] The synchronous evolution of mean and standard deviation accompanied by increasing skewness/heaviness of the tail of $Pr(\delta)$ (Figure~\ref{fig:statssynthetic}) plus the stable values of $F$ elements over time (Figure~\ref{fig:Fsyntheticgraphs}) demonstrate that $Pr(\delta)$ is fixed-shaped yet growing. Moreover, the strength assortativity localization factor changes trivially over time and is positive (Figure~\ref{fig:rssynthetic}) with $F_1$ values between $0.5$ and $0.7$ indicating the steady strength-assortativity of butterflies over time.
\end{itemize}
Table~\ref{tab:Ferrors} presents the mean absolute error of $r^s$ and $F$ elements in \emph{sGrow} streams with respect to that of real-world streams over the sequential burst-based graph snapshots (comparing Figures ~\ref{fig:randrs} and \ref{fig:rssynthetic}). We observe that in all synthetic streams the error is between $0.01$ and $0.1$. This indicates that \emph{sGrow} reproduces the similar strength difference distribution and strength assortativity of butterfly edges as in the real-world streams.
\begin{figure*}[h]
    \centering
  \subfigure[S-Ciao]{\includegraphics[width=0.24\textwidth]{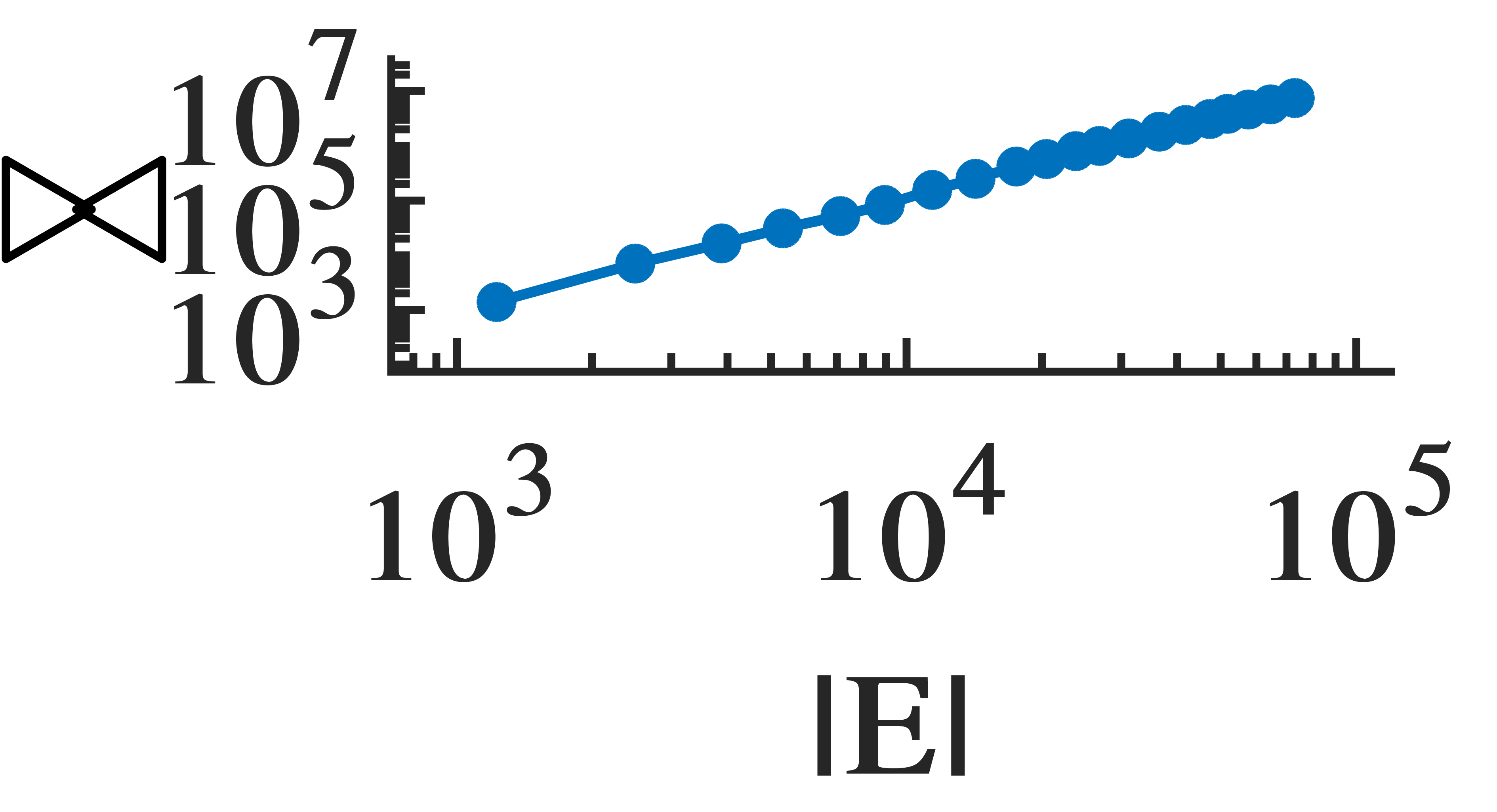}}
  \subfigure[S-Epinions]{\includegraphics[width=0.24\textwidth]{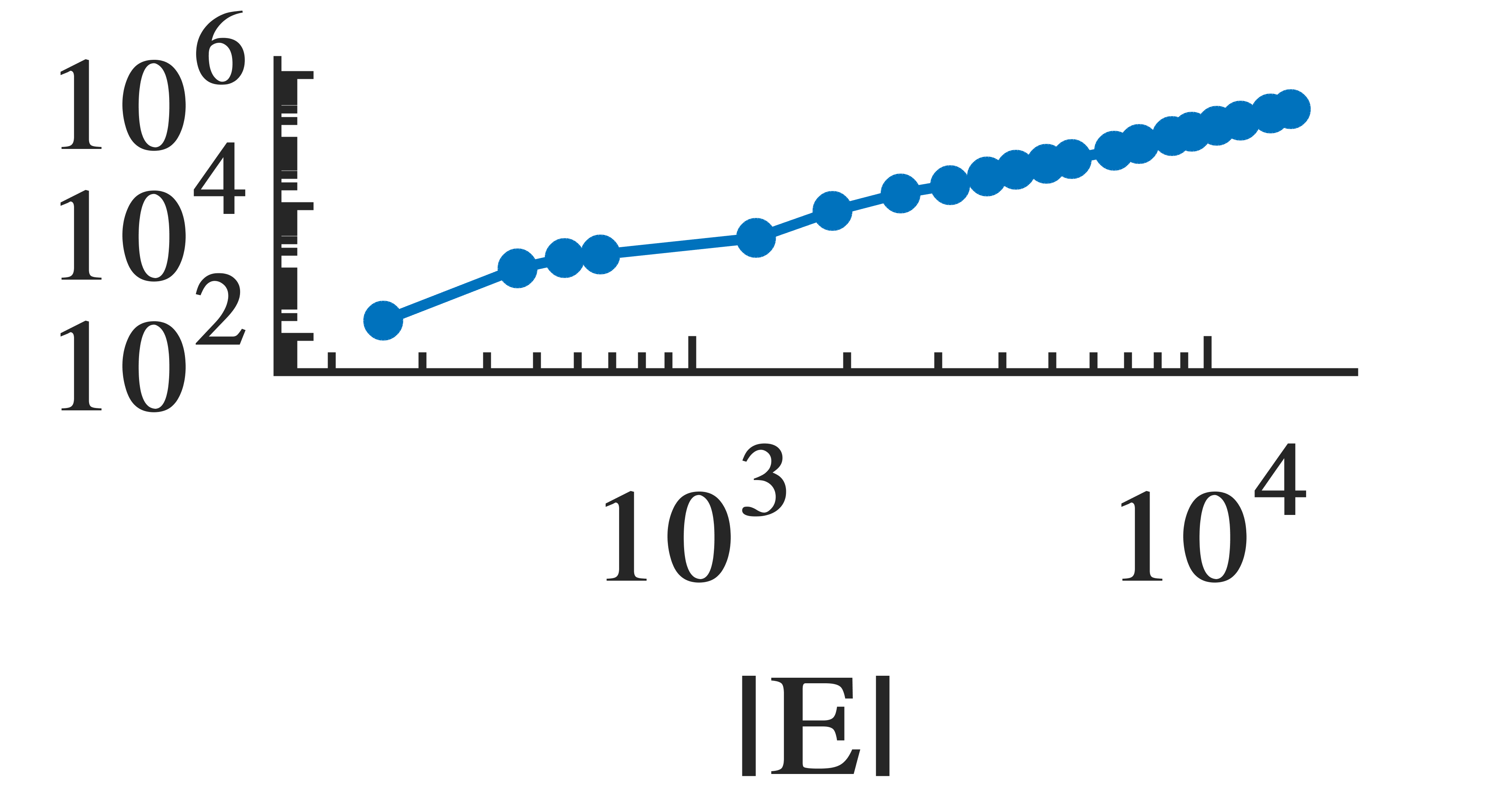}}
  \subfigure[S-WikiLens]{\includegraphics[width=0.24\textwidth]{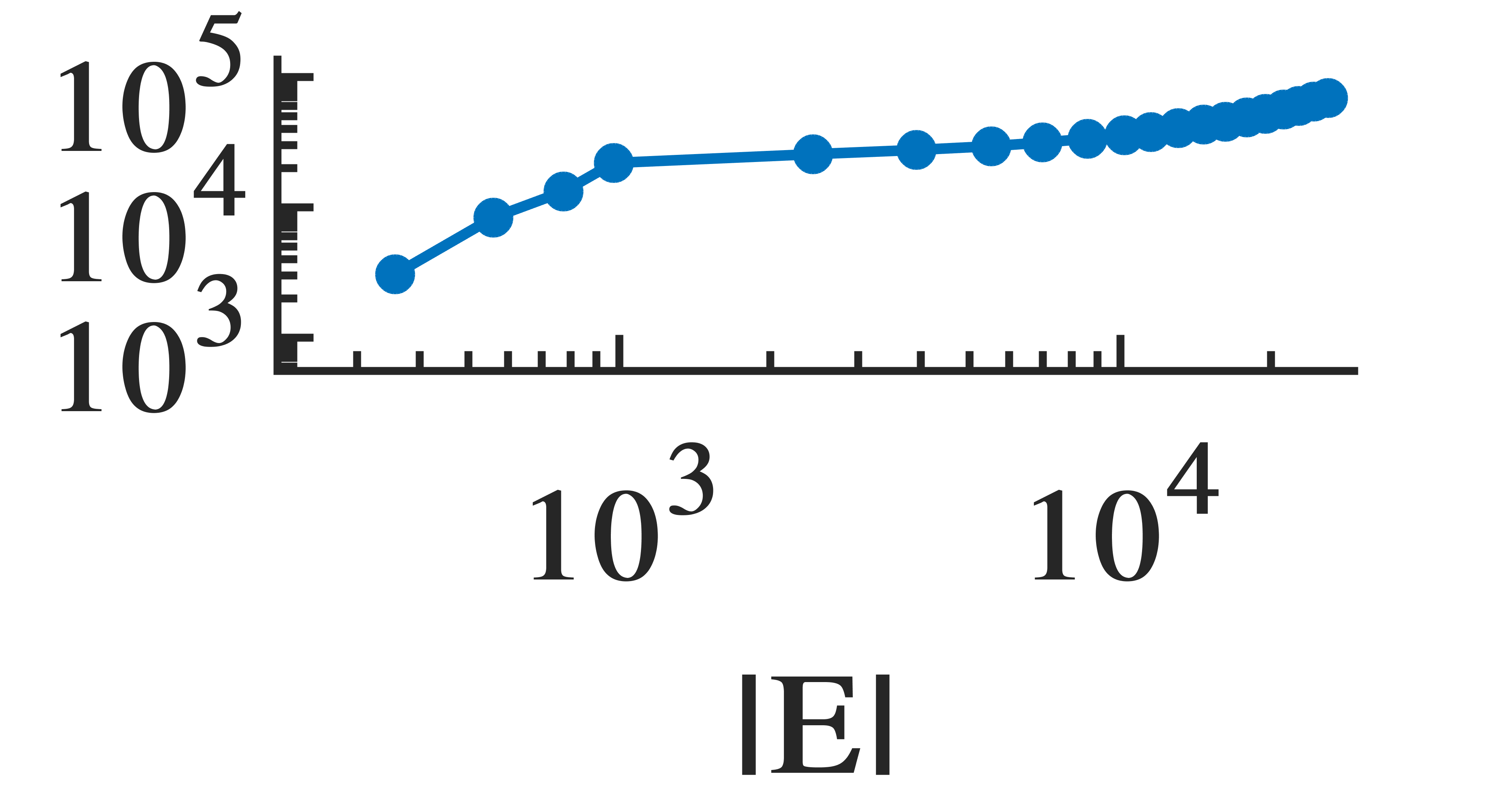}}
  \subfigure[S-ML100k]{\includegraphics[width=0.24\textwidth]{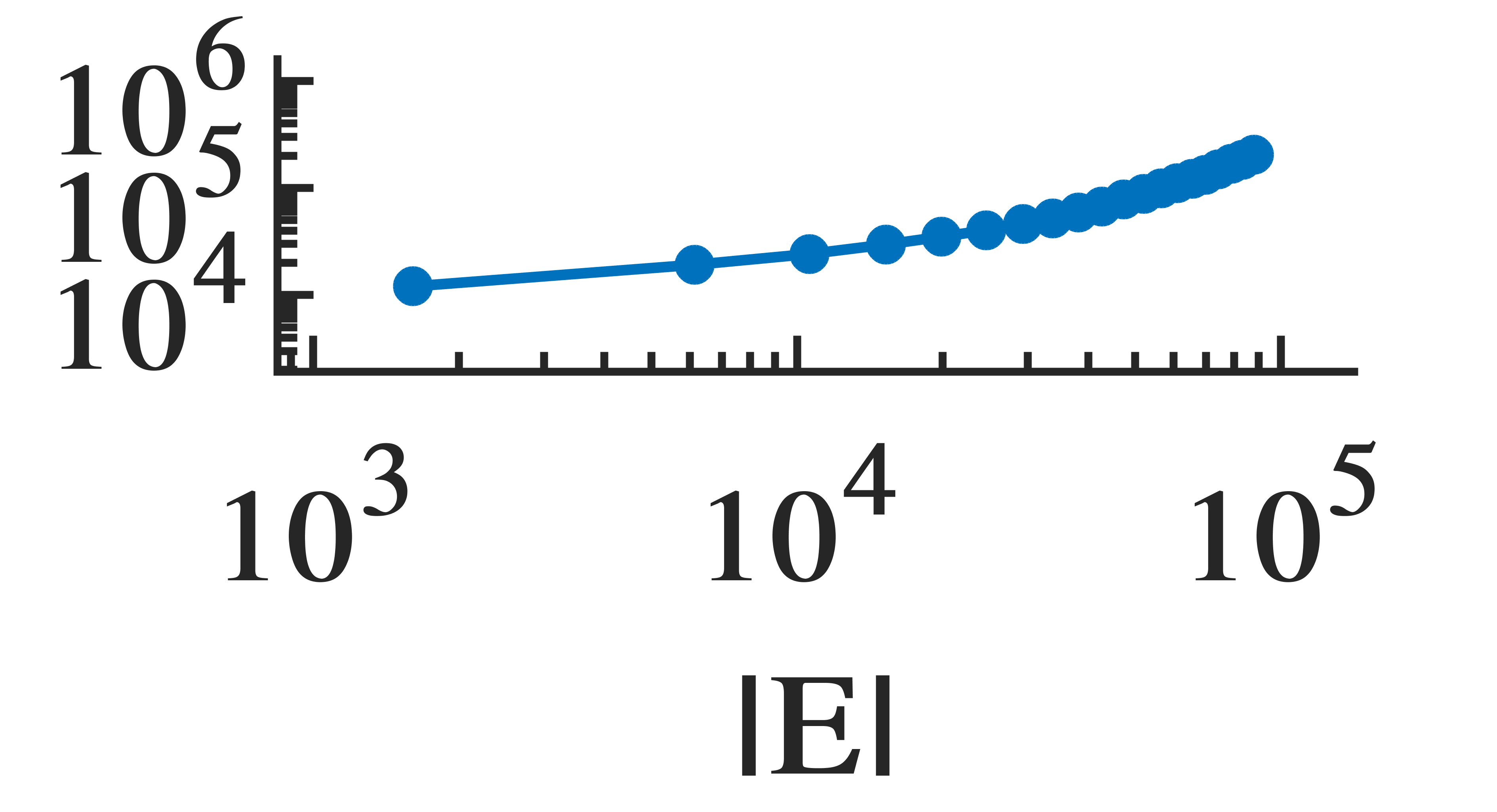}}
  \subfigure[S-ML1m]{\includegraphics[width=0.24\textwidth]{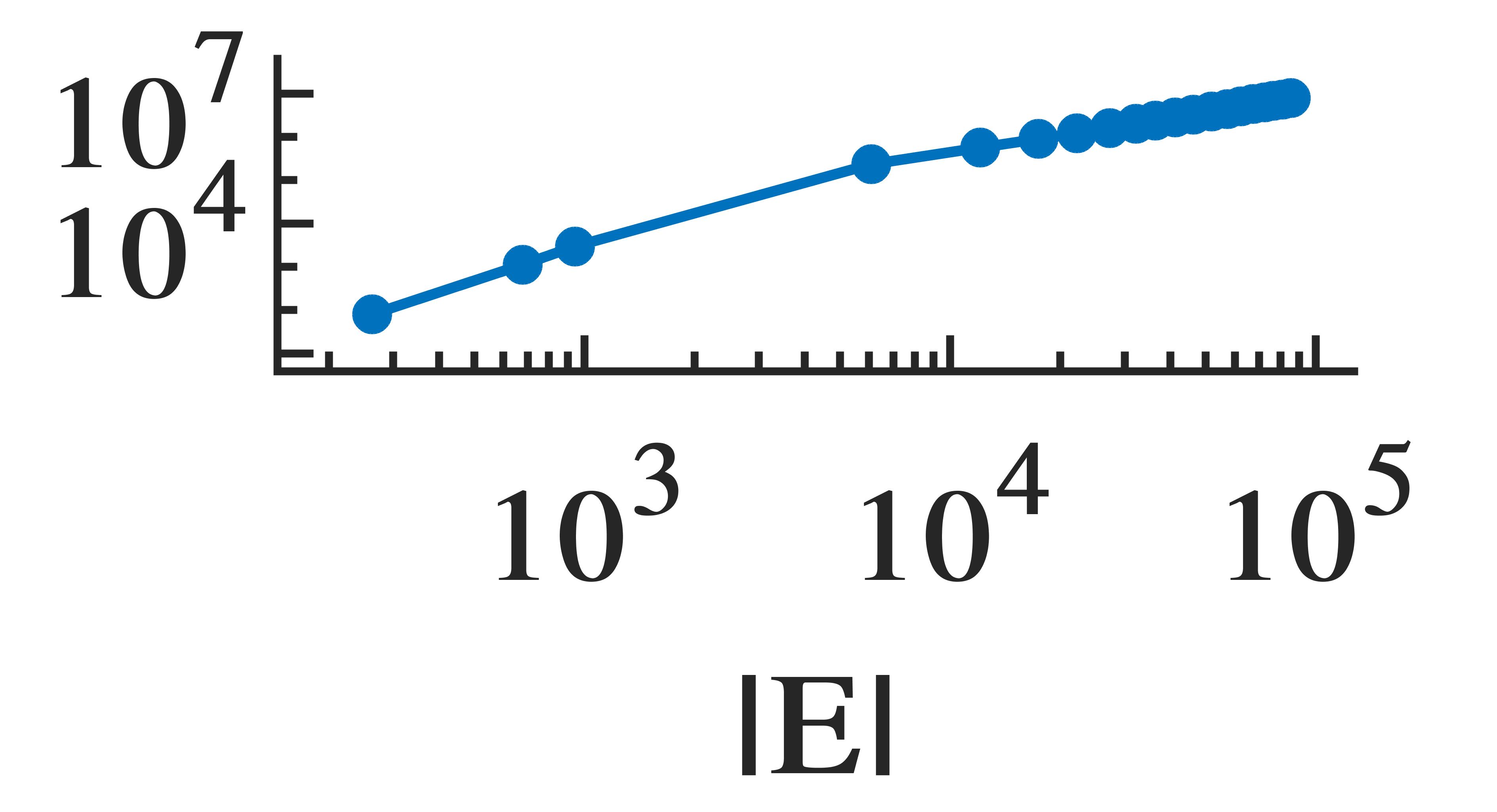}}
  \subfigure[S-Amazon]{\includegraphics[width=0.24\textwidth]{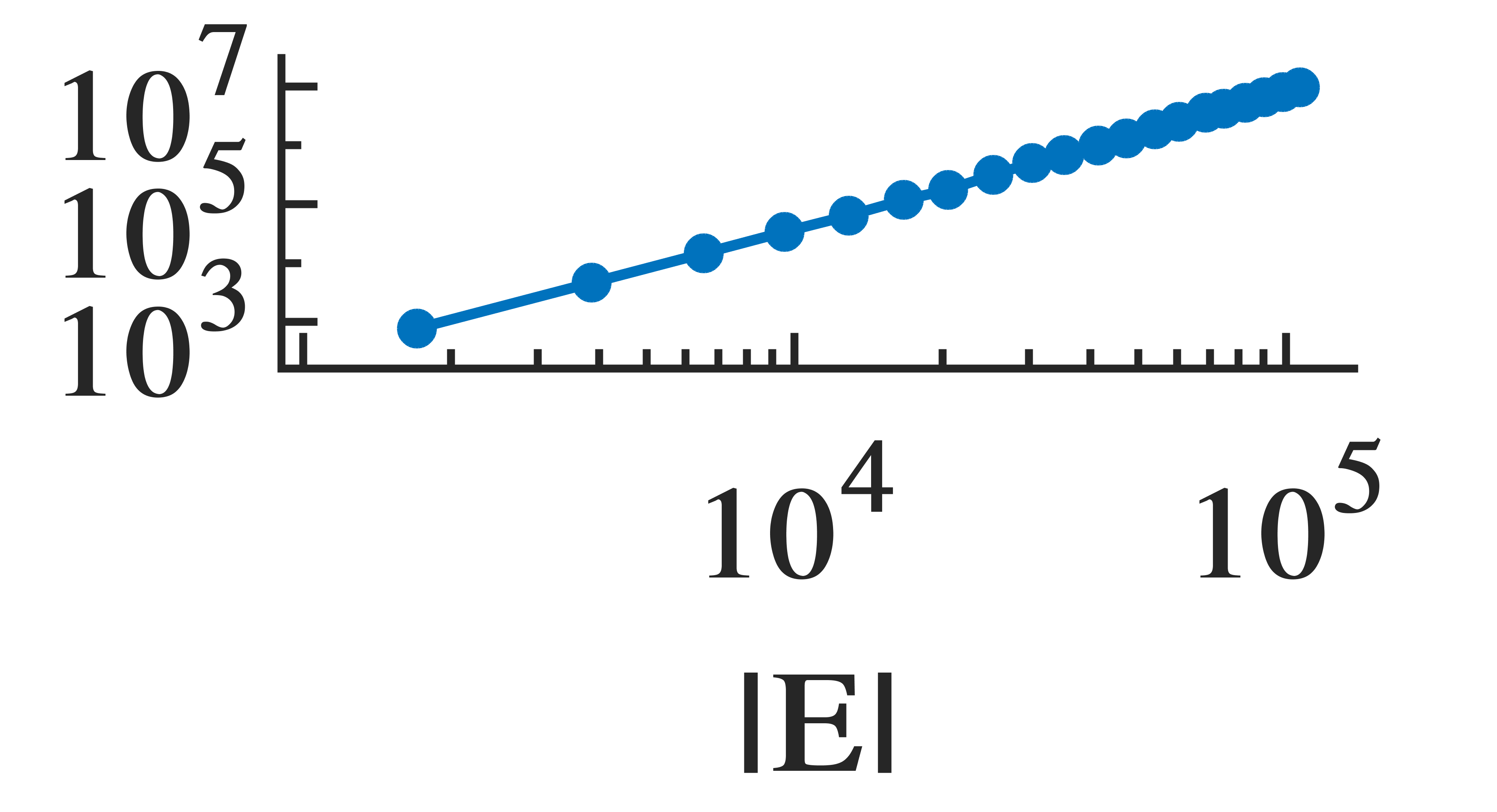}}
  \subfigure[S-Yahoo]{\includegraphics[width=0.24\textwidth]{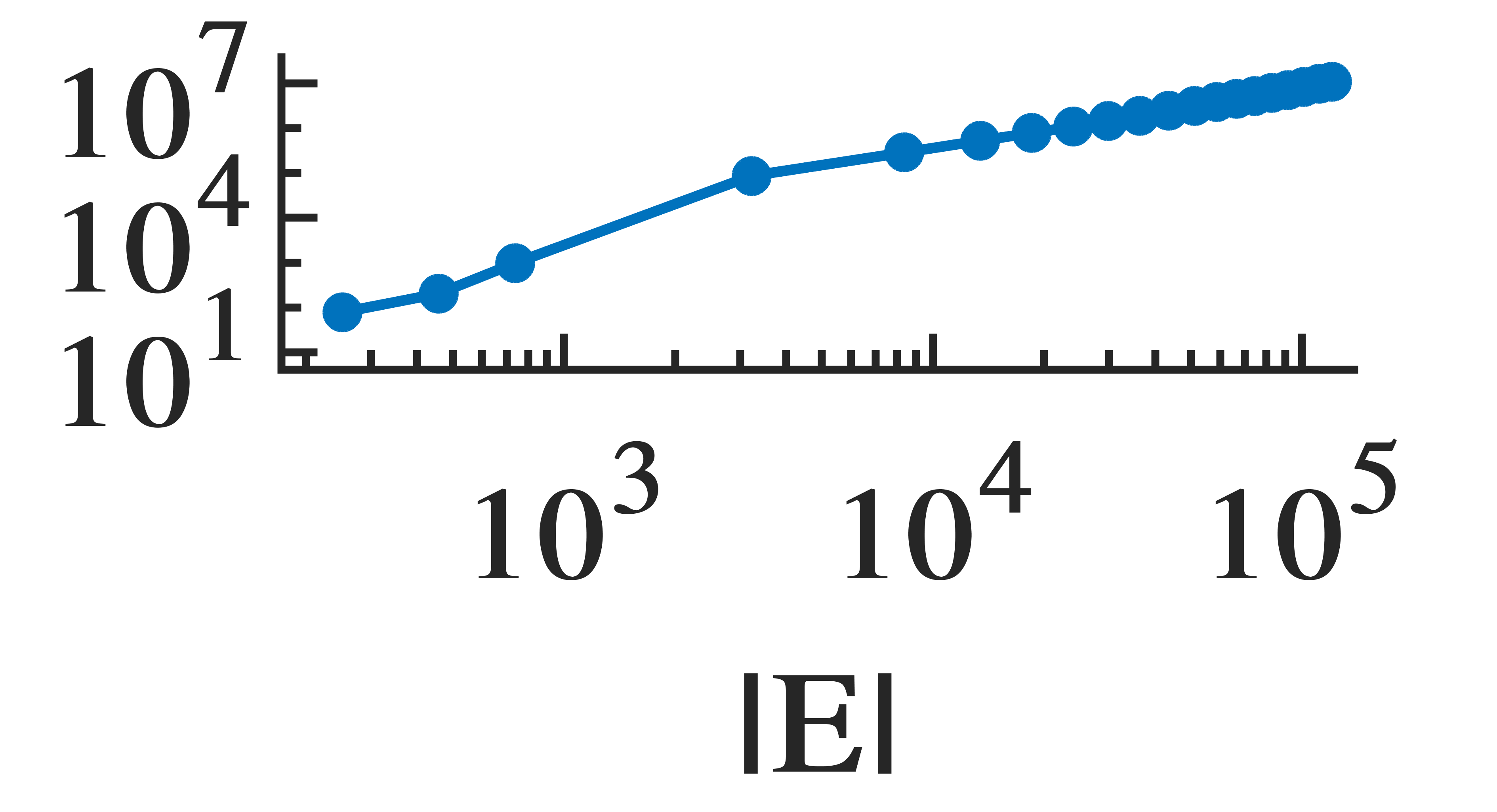}}
    \caption{The number of butterflies $\bowtie$ versus the number of edges $|E|$.} 
    \label{fig:butterflycountmodel}
\end{figure*}
\begin{figure*}[]
    \centering
    \subfigure[$S_i$, S-Ciao]{\includegraphics[width=0.24\textwidth]{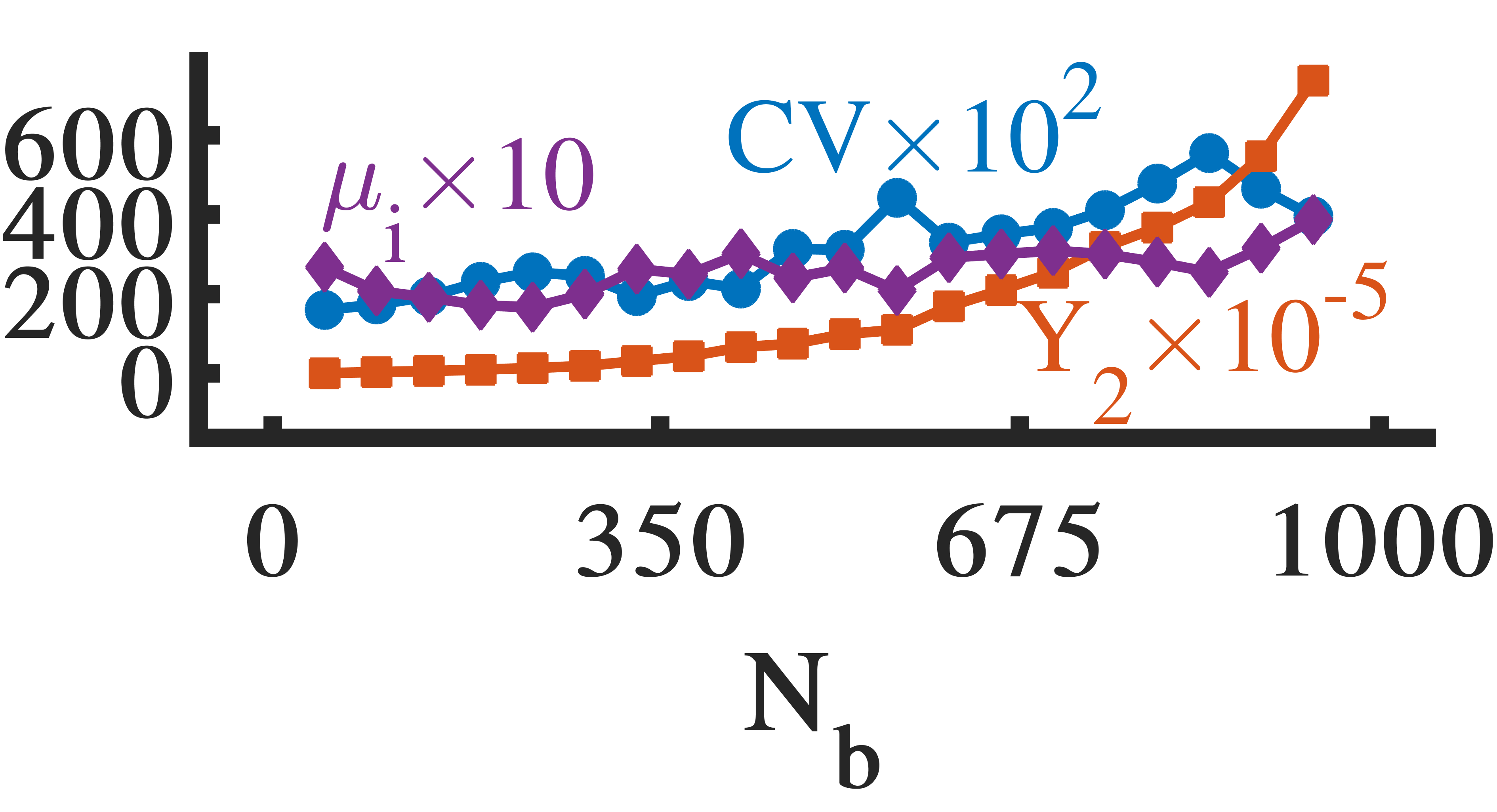}}
    \subfigure[$S_i$, S-Epinions]{\includegraphics[width=0.24\textwidth]{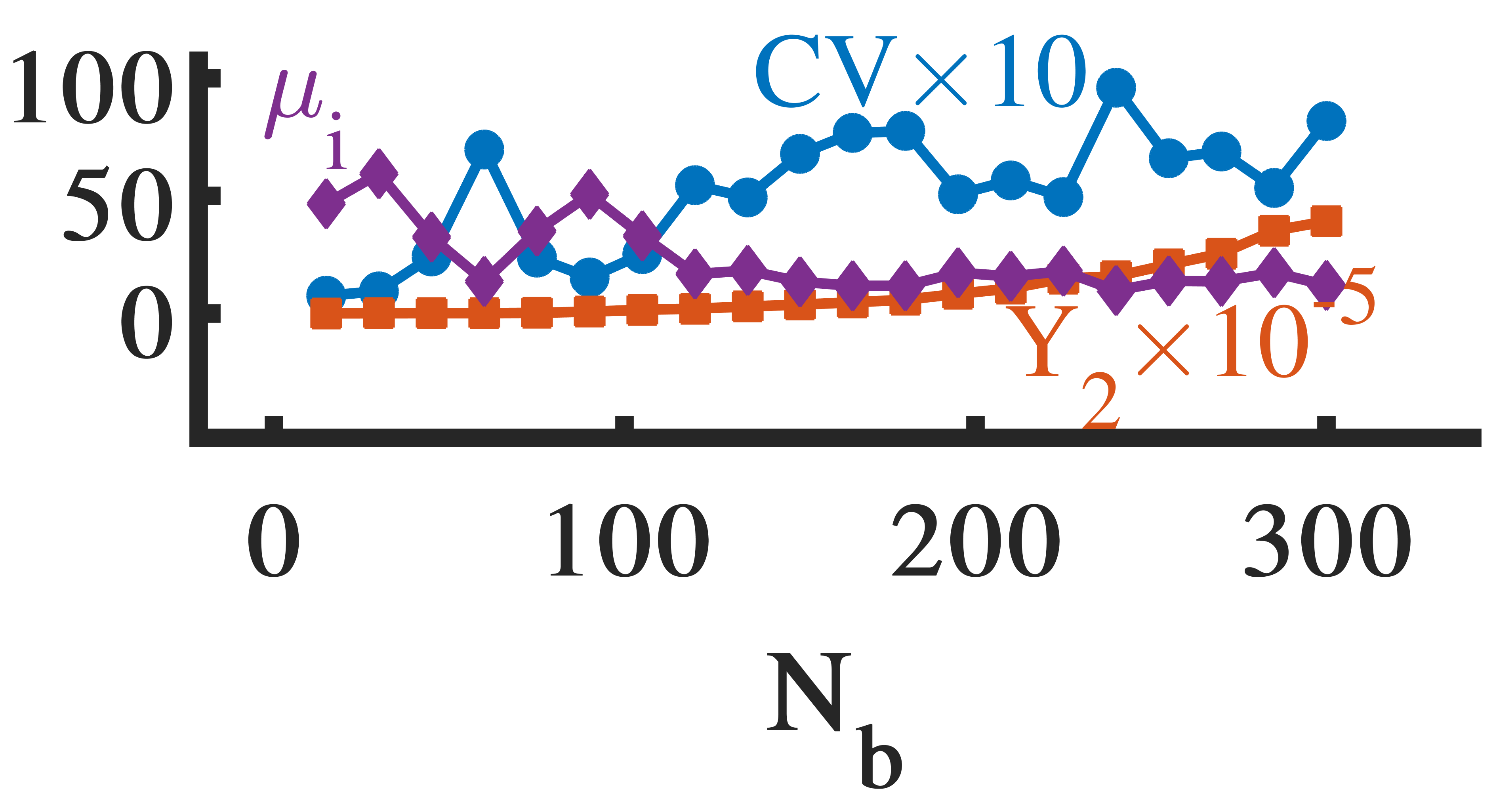}}
    \subfigure[$S_i$, S-WikiLens]{\includegraphics[width=0.24\textwidth]{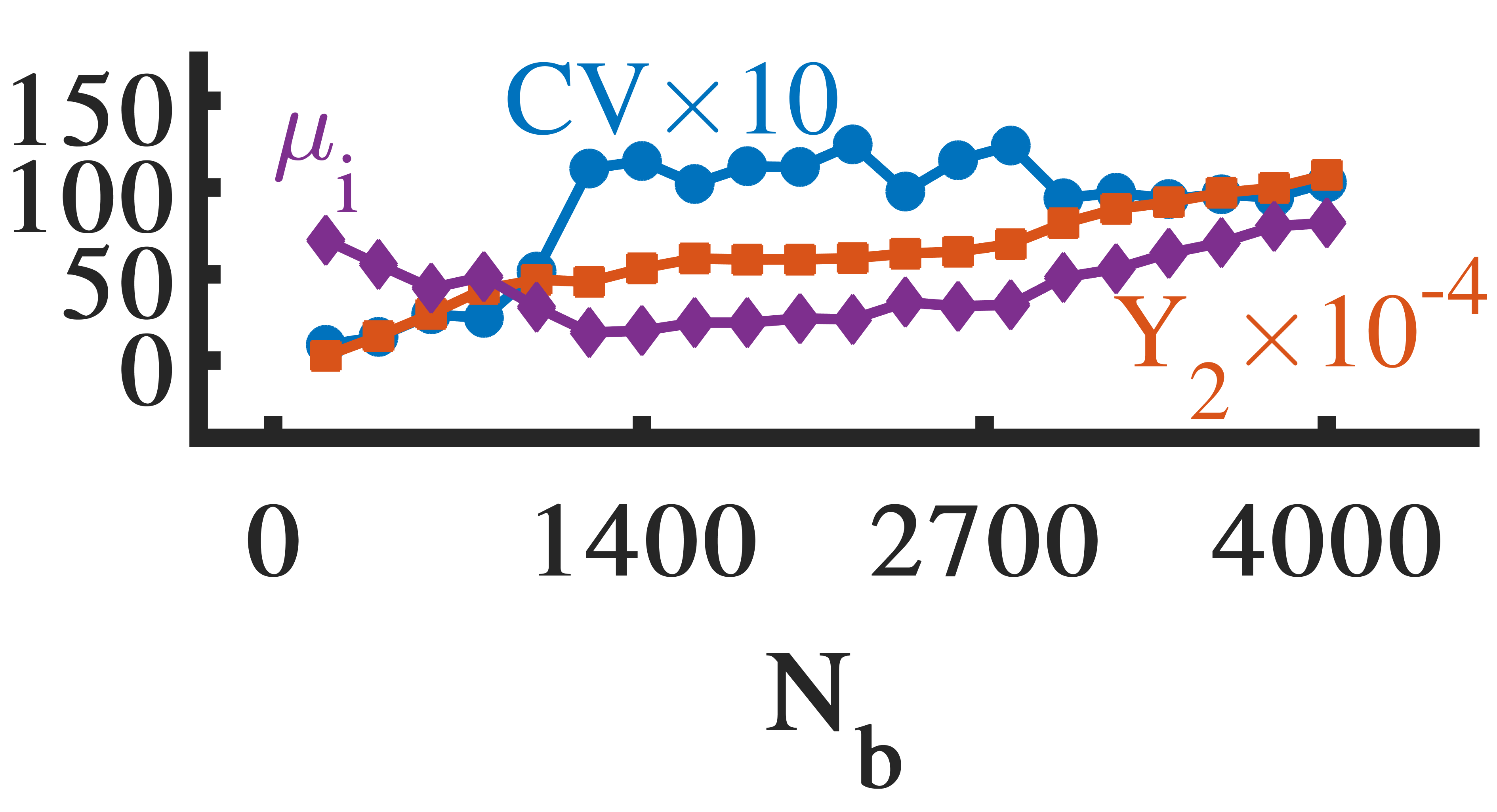}}
    \subfigure[$S_i$, S-ML100k]{\includegraphics[width=0.24\textwidth]{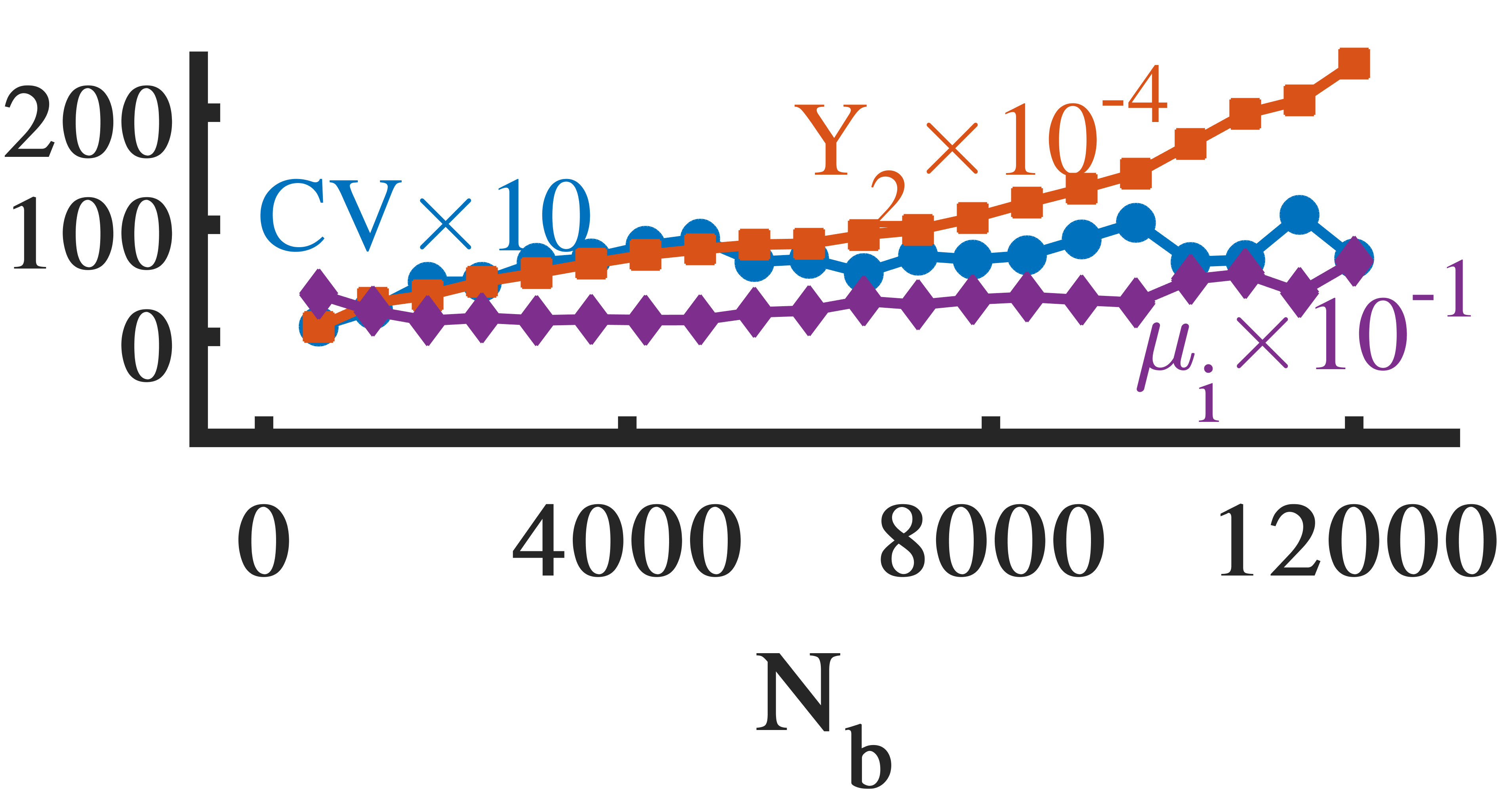}}
    \subfigure[$S_i$, S-ML1m]{\includegraphics[width=0.24\textwidth]{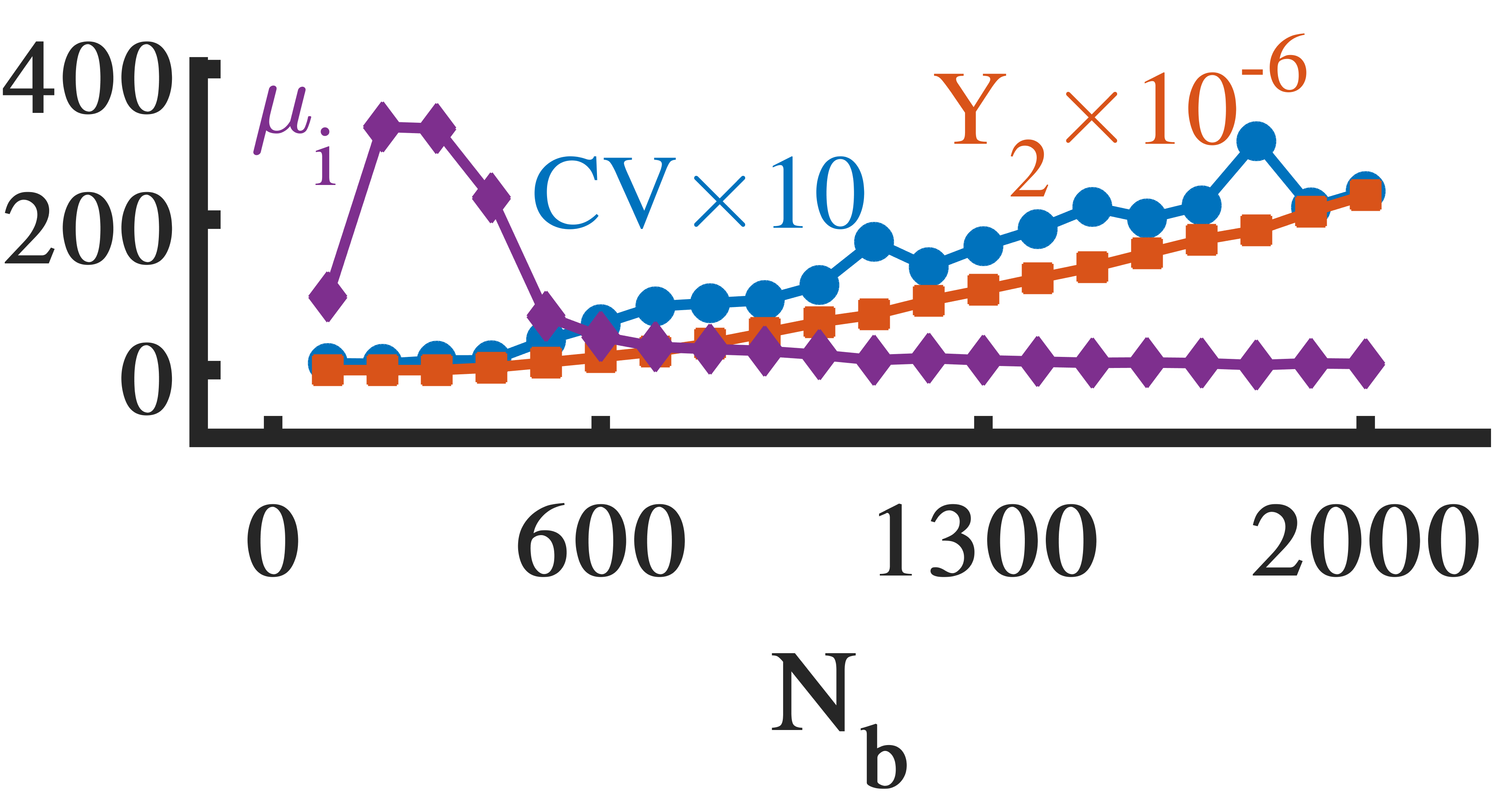}}
    \subfigure[$S_i$, S-Amazon]{\includegraphics[width=0.24\textwidth]{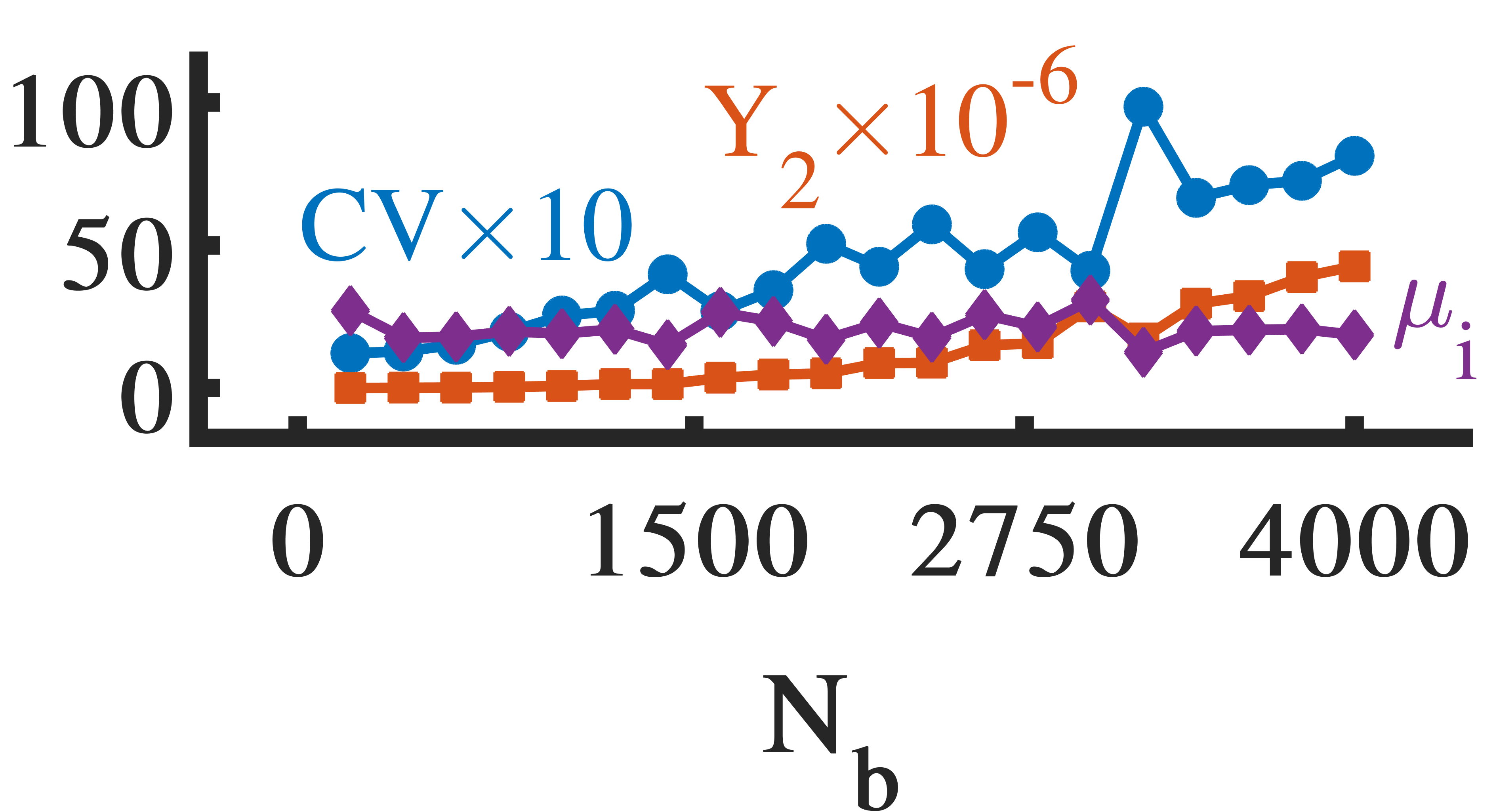}}
    \subfigure[$S_i$, S-Yahoo]{\includegraphics[width=0.24\textwidth]{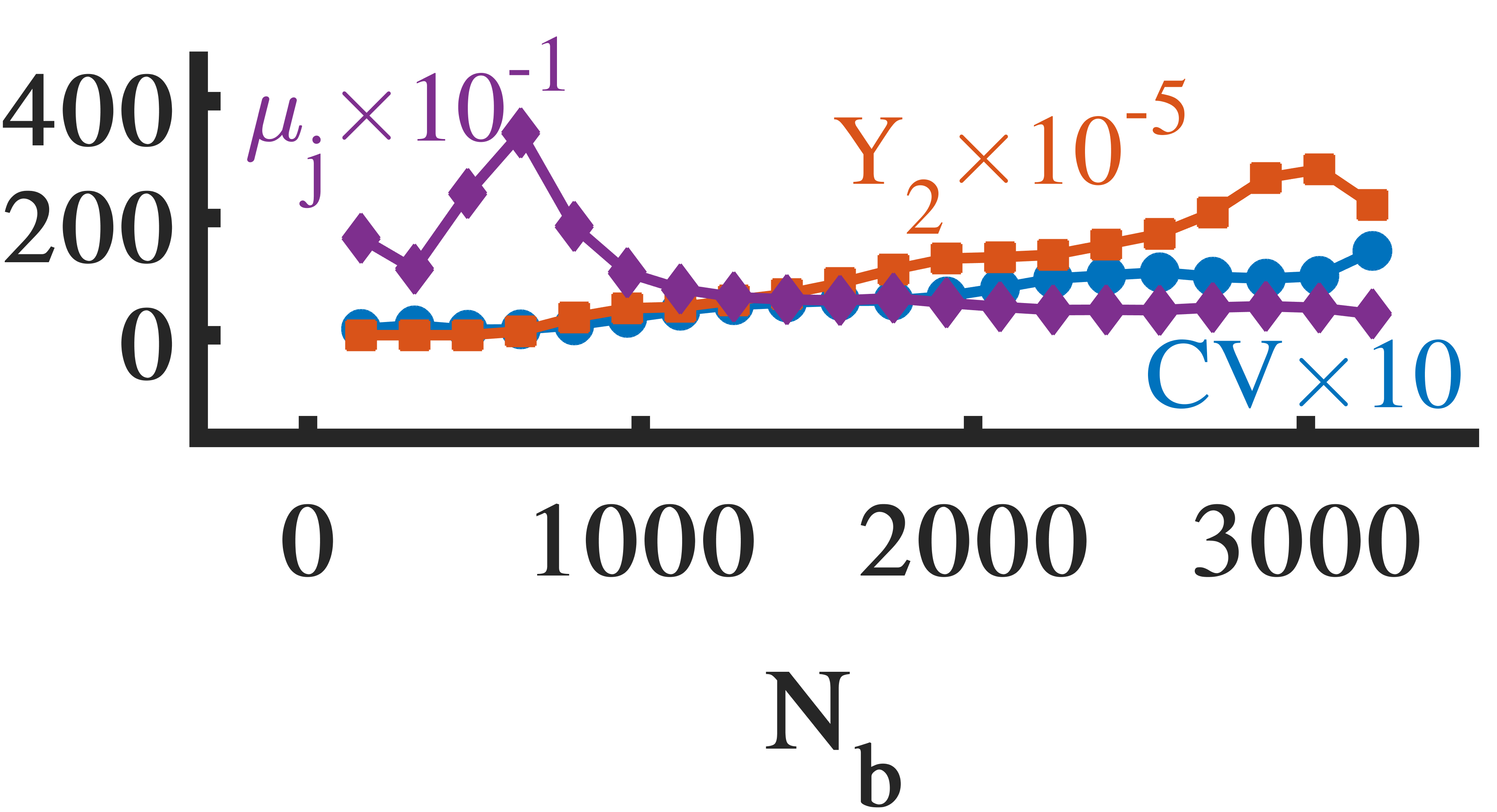}}
    
    \subfigure[$S_j$, S-Ciao]{\includegraphics[width=0.24\textwidth]{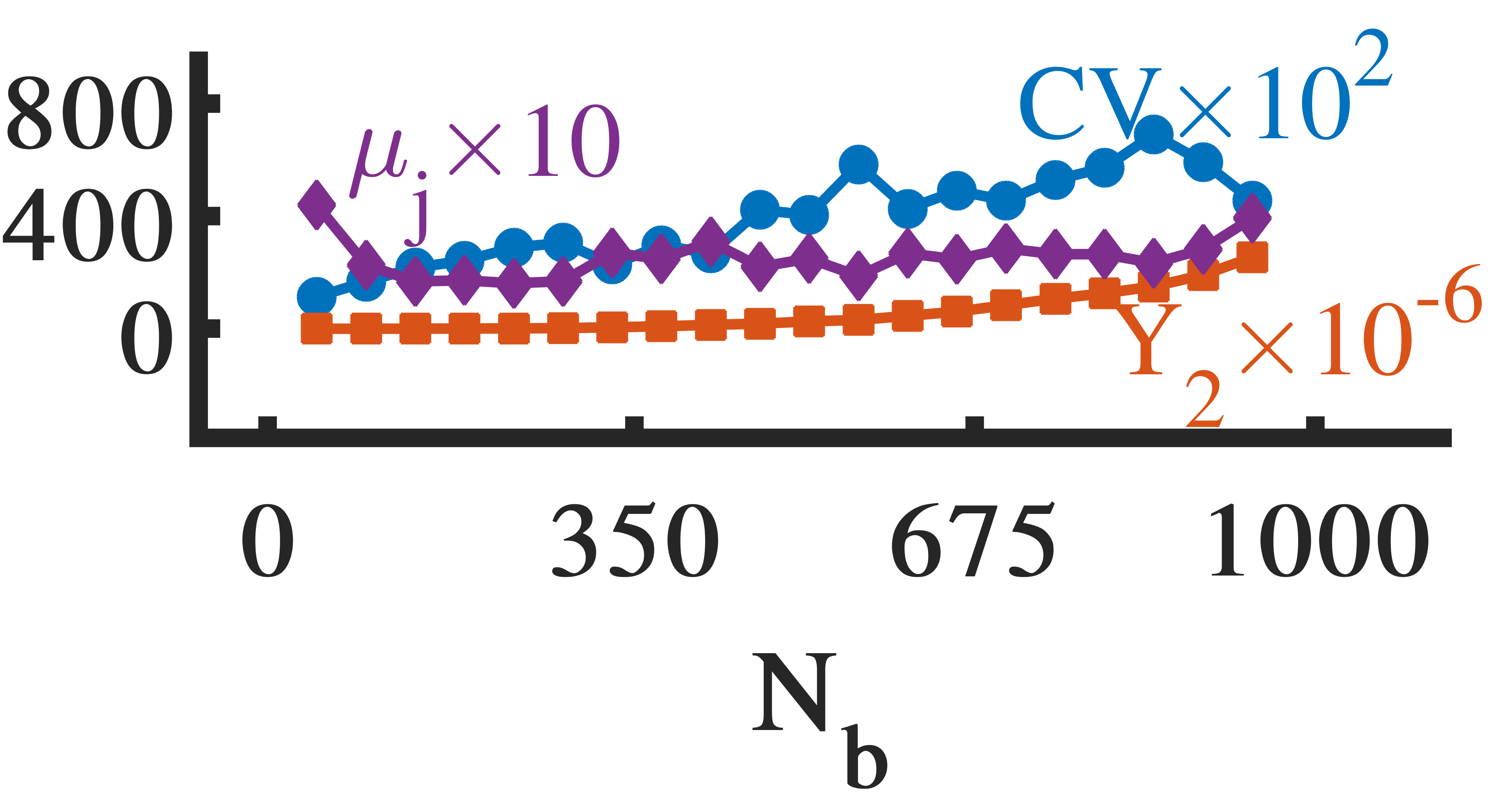}}
    \subfigure[$S_j$, S-Epinions]{\includegraphics[width=0.24\textwidth]{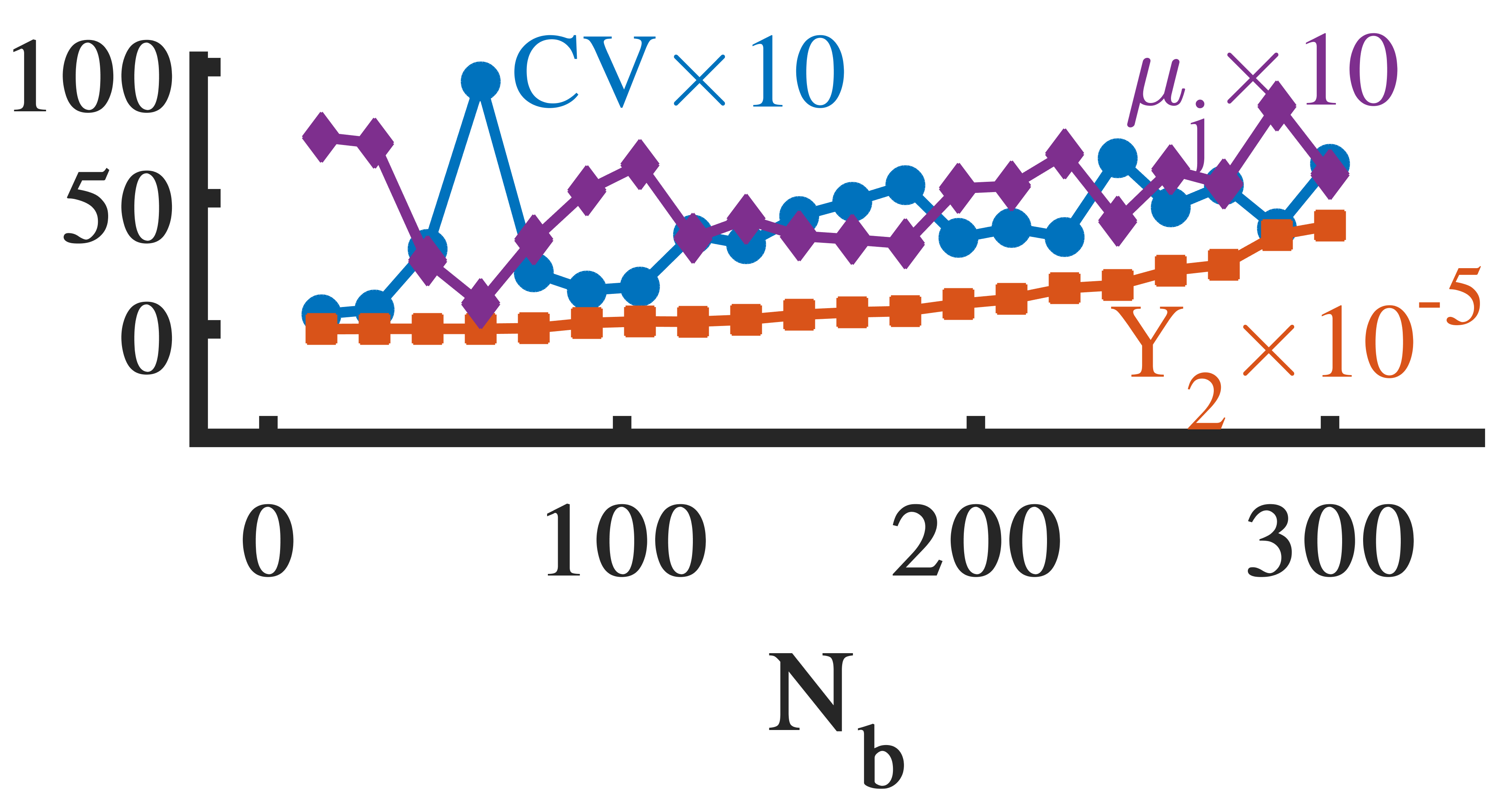}}
    \subfigure[$S_j$, S-WikiLens]{\includegraphics[width=0.24\textwidth]{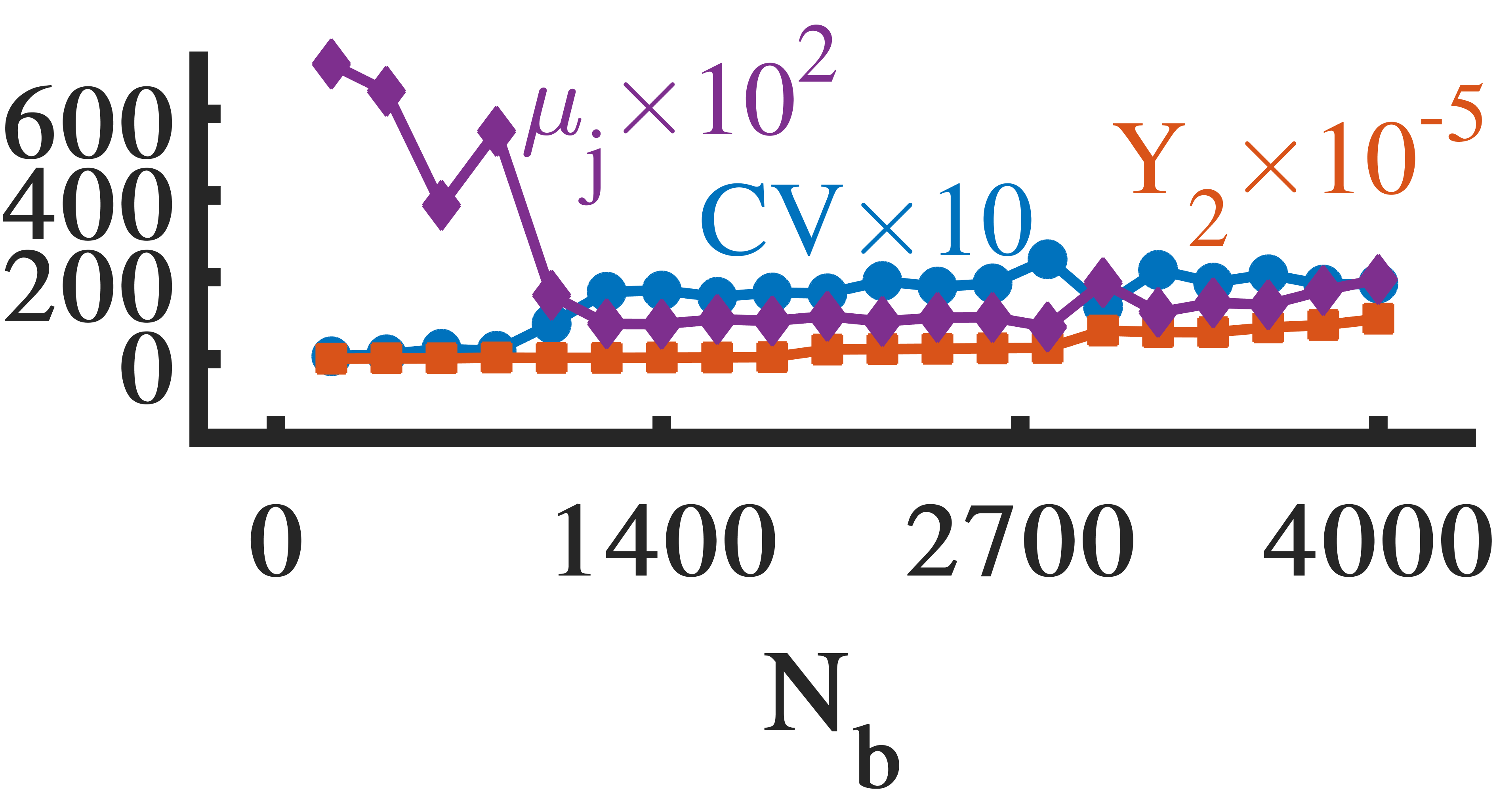}}
    \subfigure[$S_j$, S-ML100k]{\includegraphics[width=0.24\textwidth]{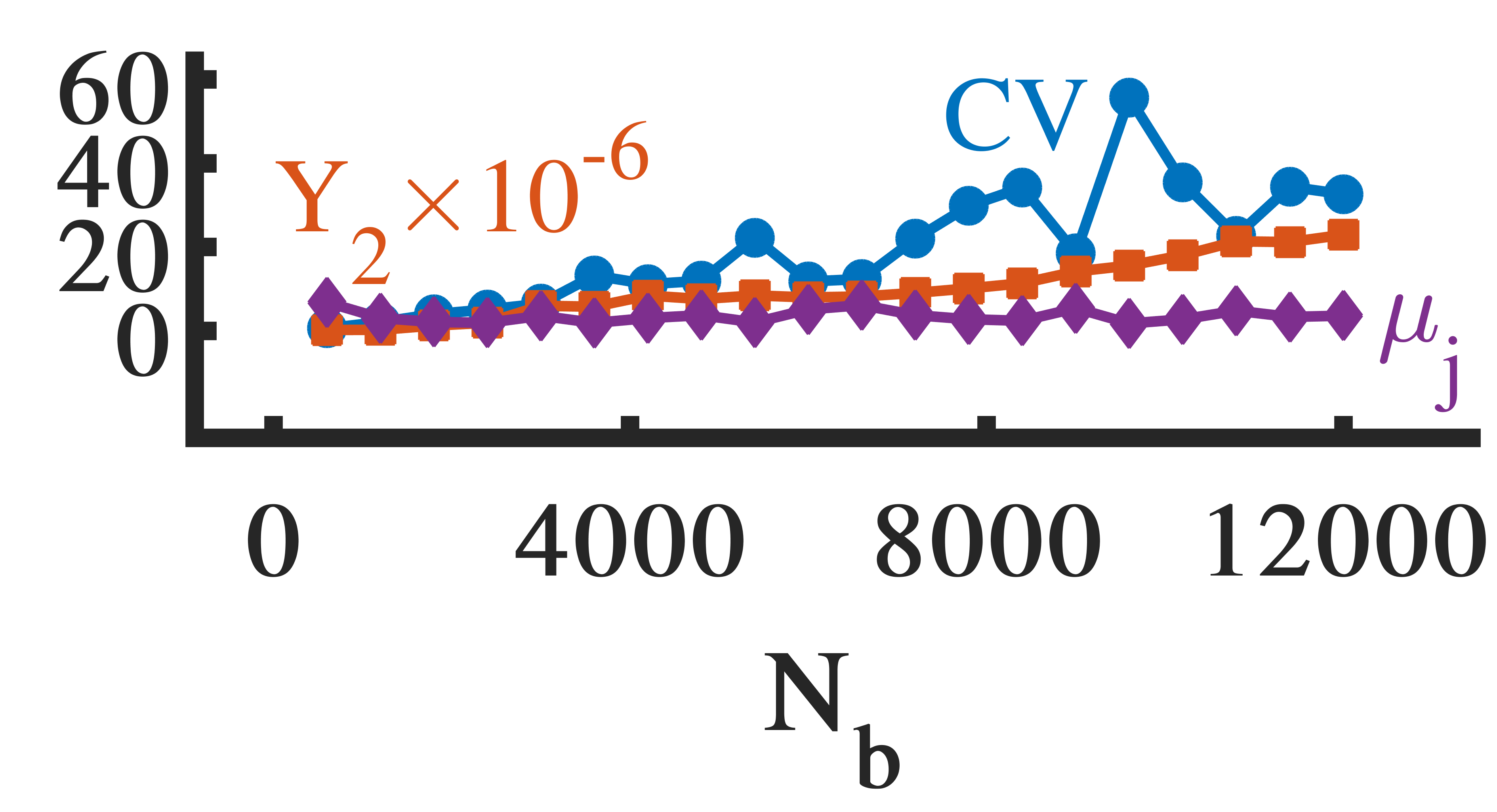}}
    \subfigure[$S_j$, S-ML1m]{\includegraphics[width=0.24\textwidth]{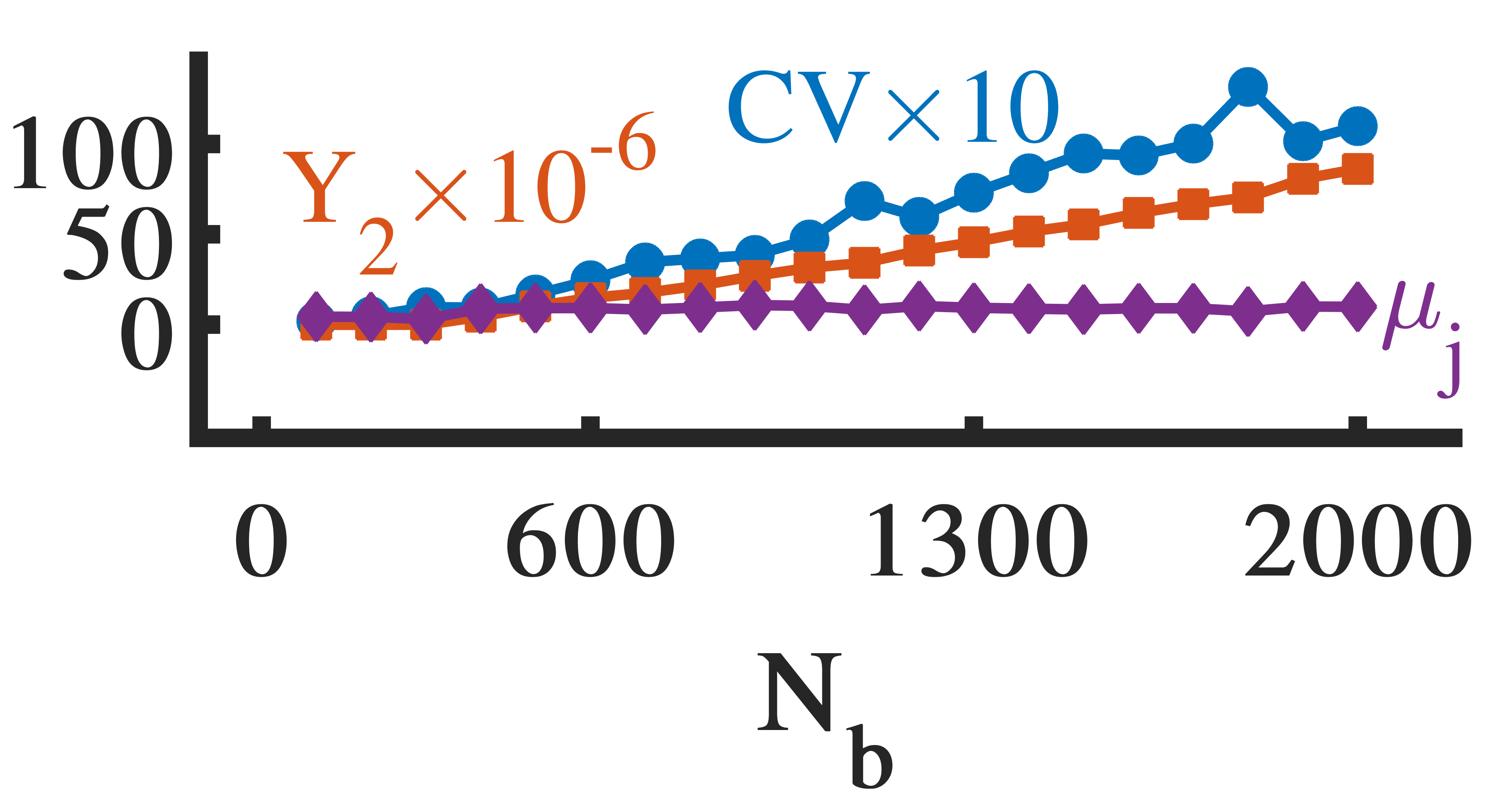}}
    \subfigure[$S_j$, S-Amazon]{\includegraphics[width=0.24\textwidth]{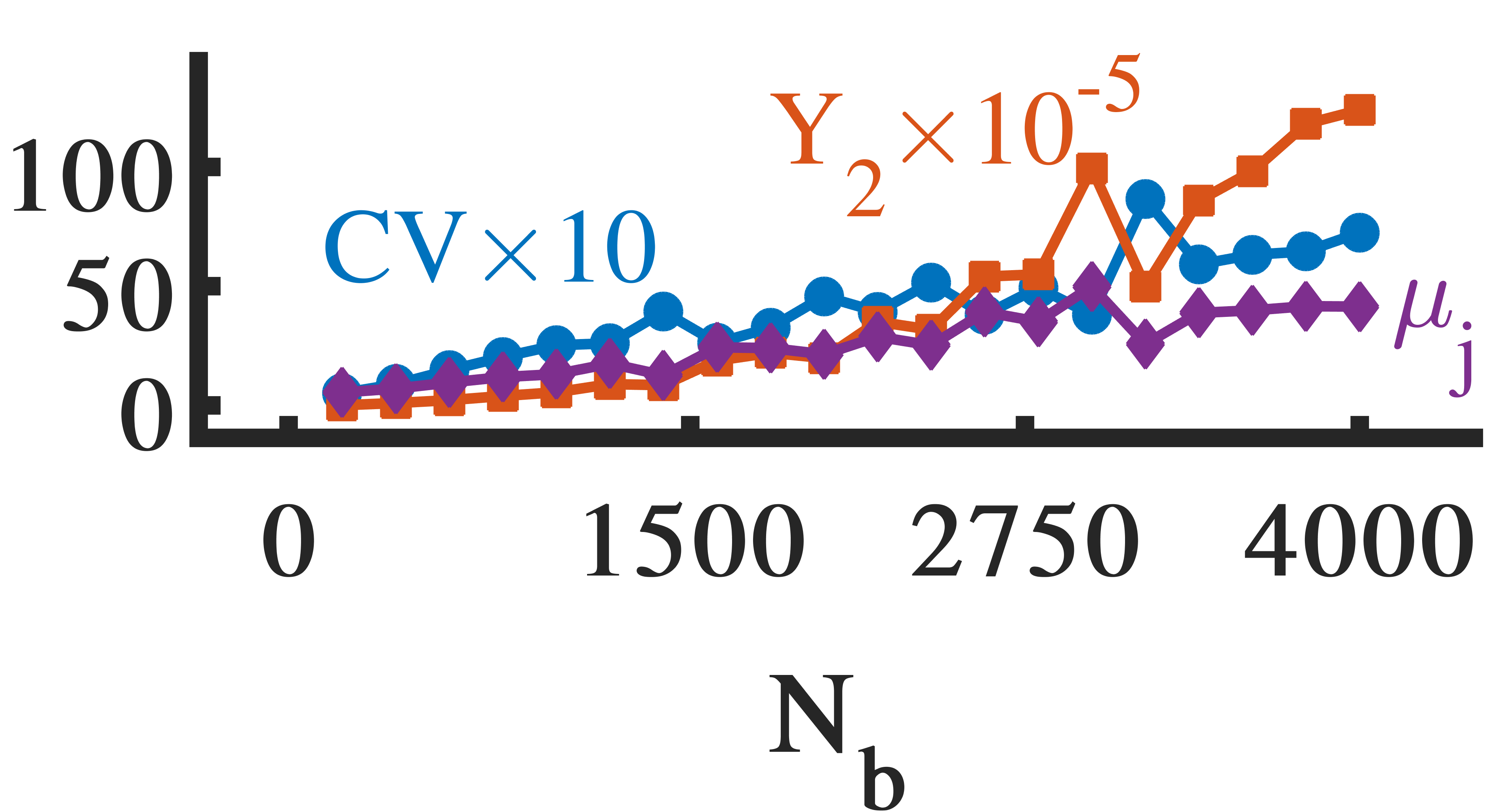}}
    \subfigure[$S_j$, S-Yahoo]{\includegraphics[width=0.24\textwidth]{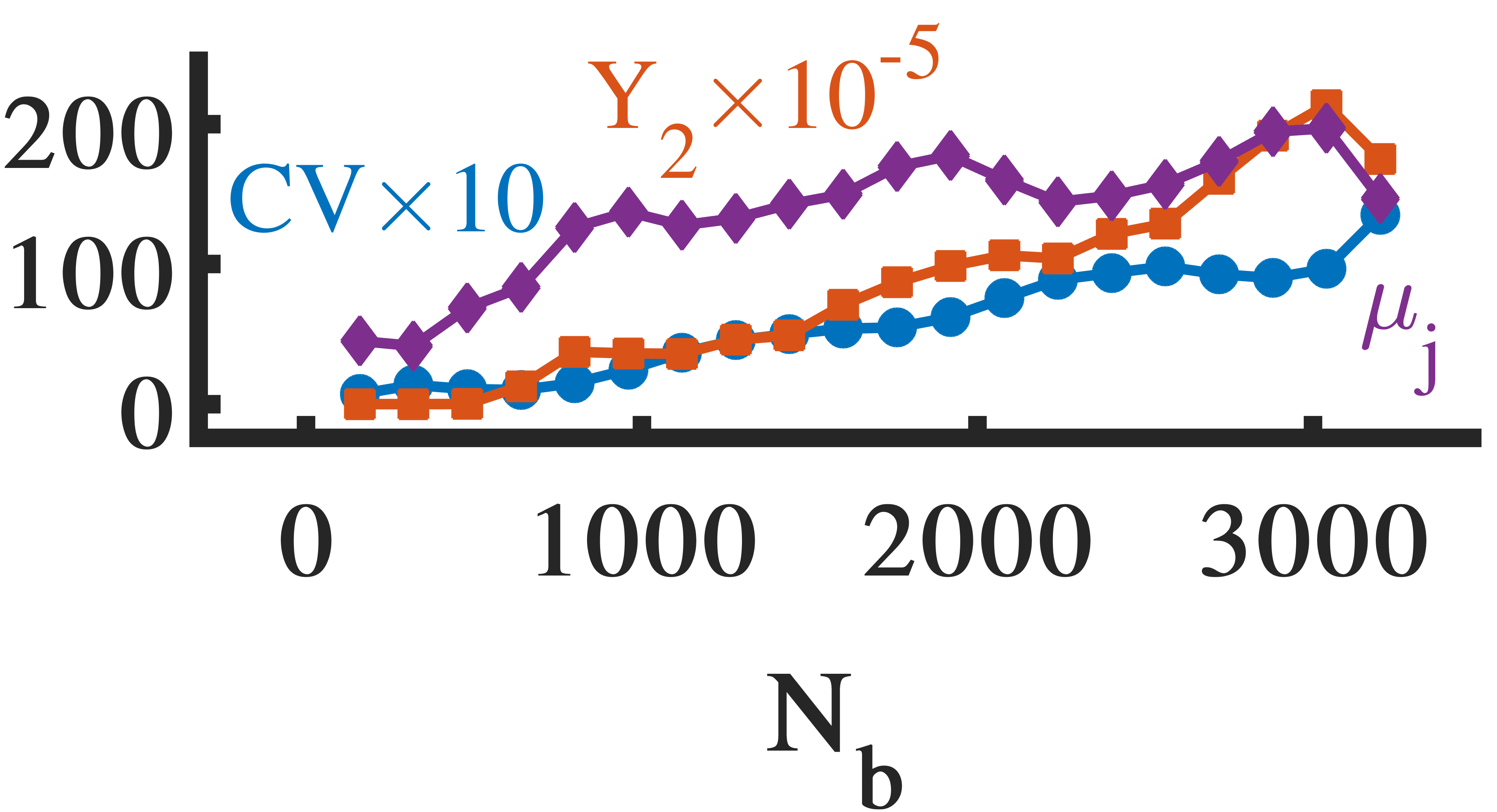}}

    \caption{Coefficient of variation (circles), excess kurtosis (squares), and mean (diamonds) of strengths of butterfly (a-g) i-vertices and (h-n) j-vertices over the timeline of burst arrivals.}
    \label{fig:sstatssynthetci}
\end{figure*}
\begin{figure*}[]
    \centering
  \subfigure[S-Ciao]{\includegraphics[width=0.24\textwidth]{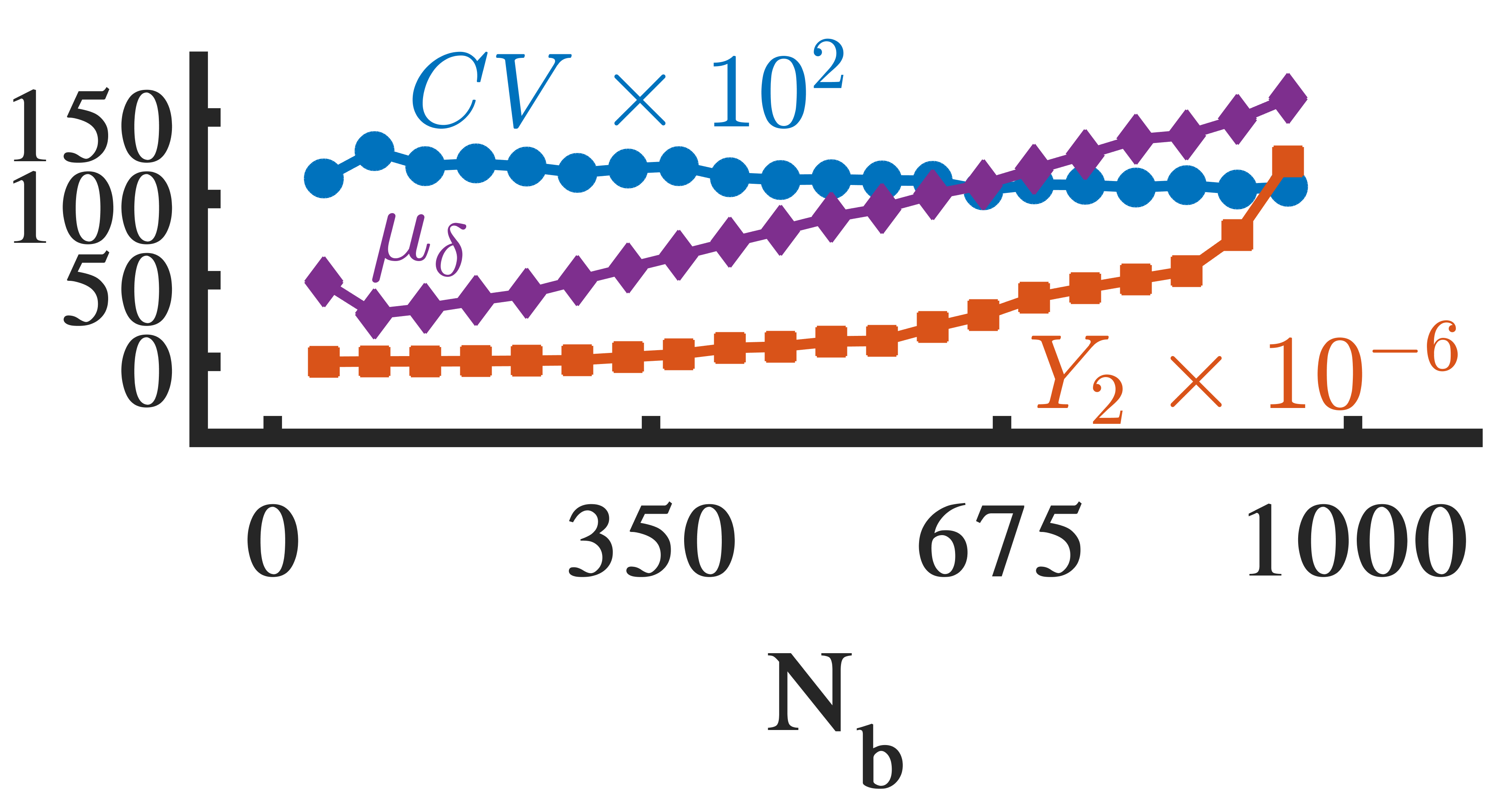}}
  \subfigure[S-Epinions]{\includegraphics[width=0.24\textwidth]{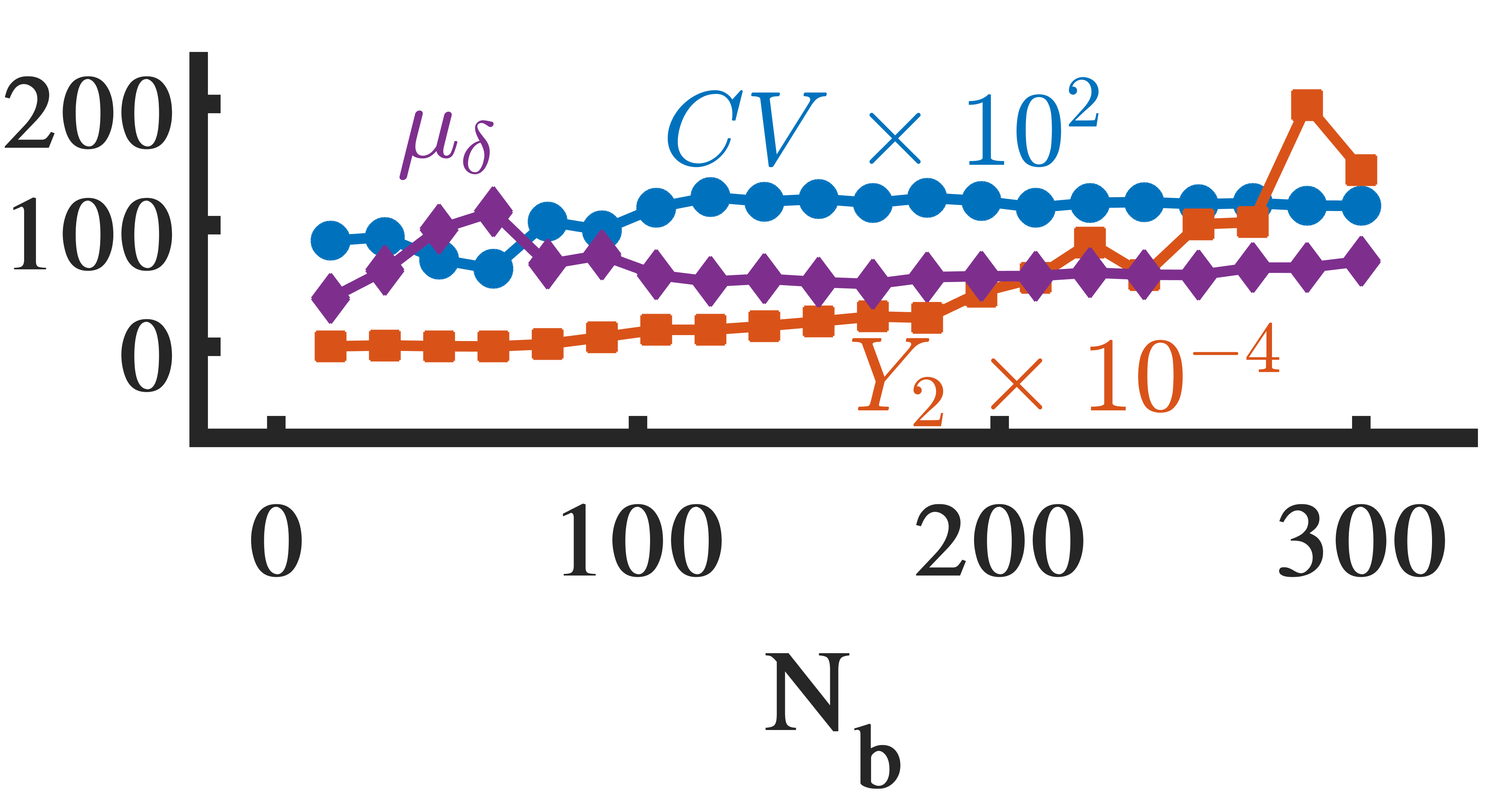}}
  \subfigure[S-WikiLens]{\includegraphics[width=0.24\textwidth]{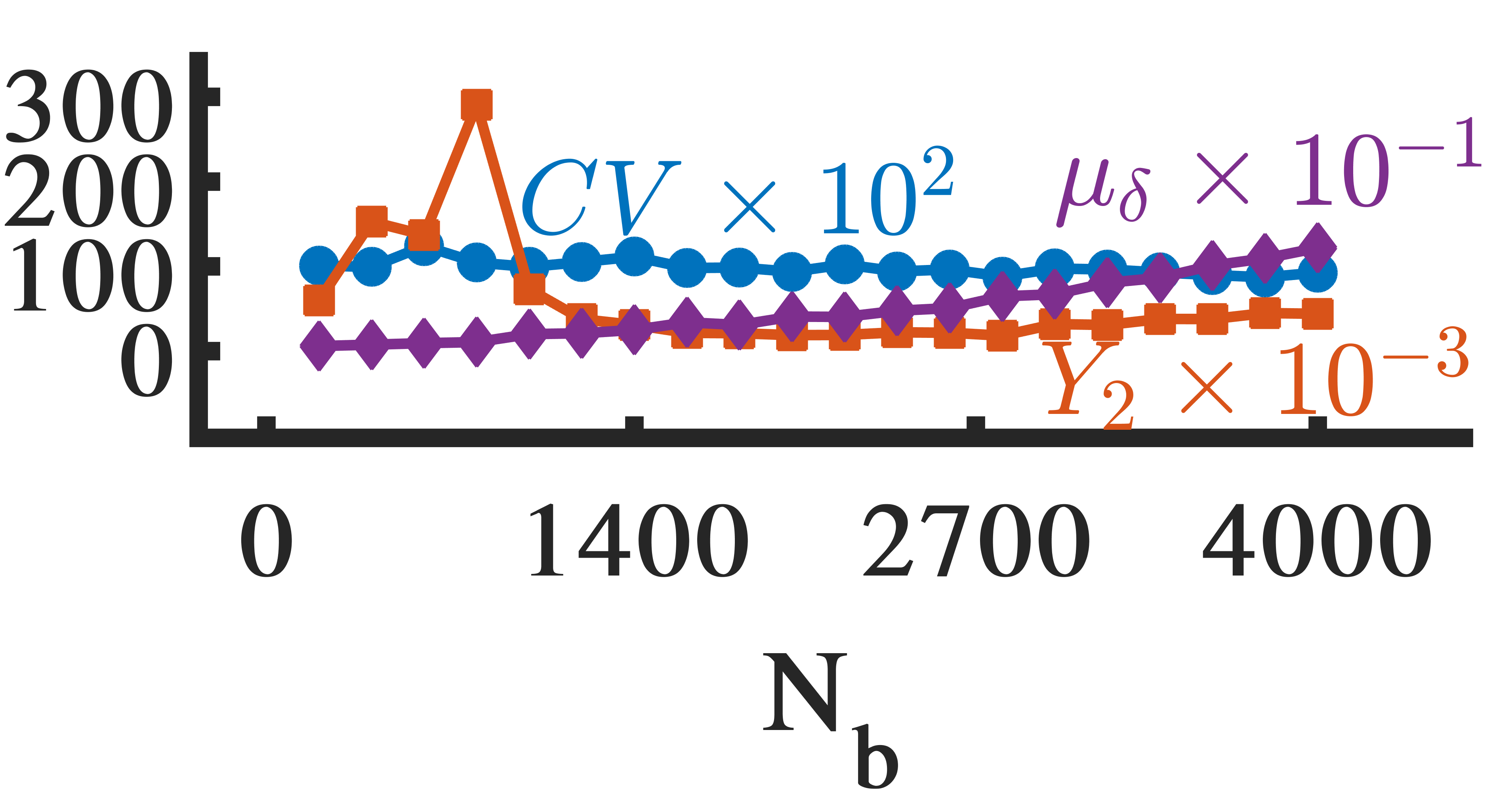}}
  \subfigure[S-ML100k]{\includegraphics[width=0.24\textwidth]{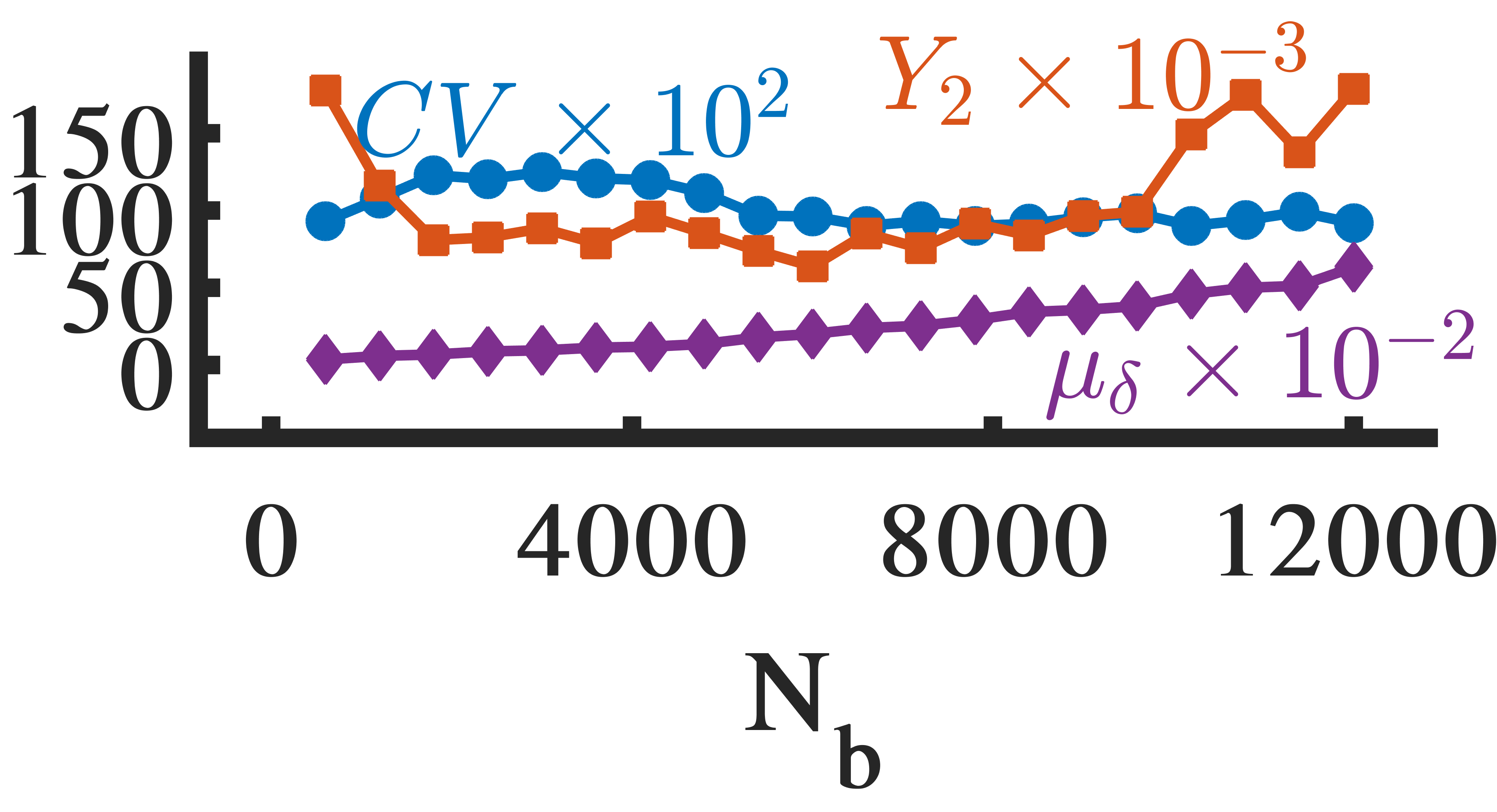}}
  \subfigure[S-ML1m]{\includegraphics[width=0.24\textwidth]{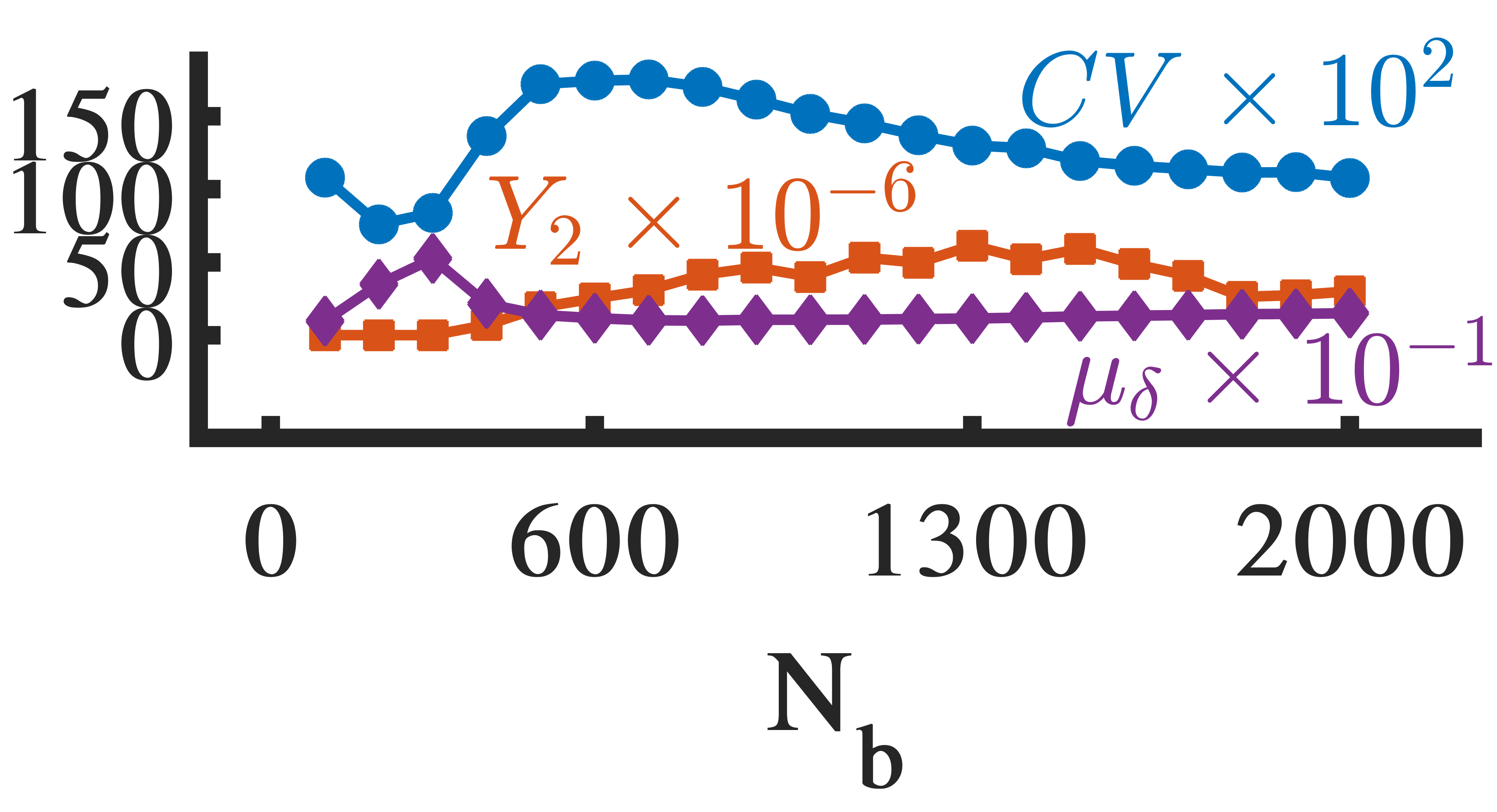}}
  \subfigure[S-Amazon]{\includegraphics[width=0.24\textwidth]{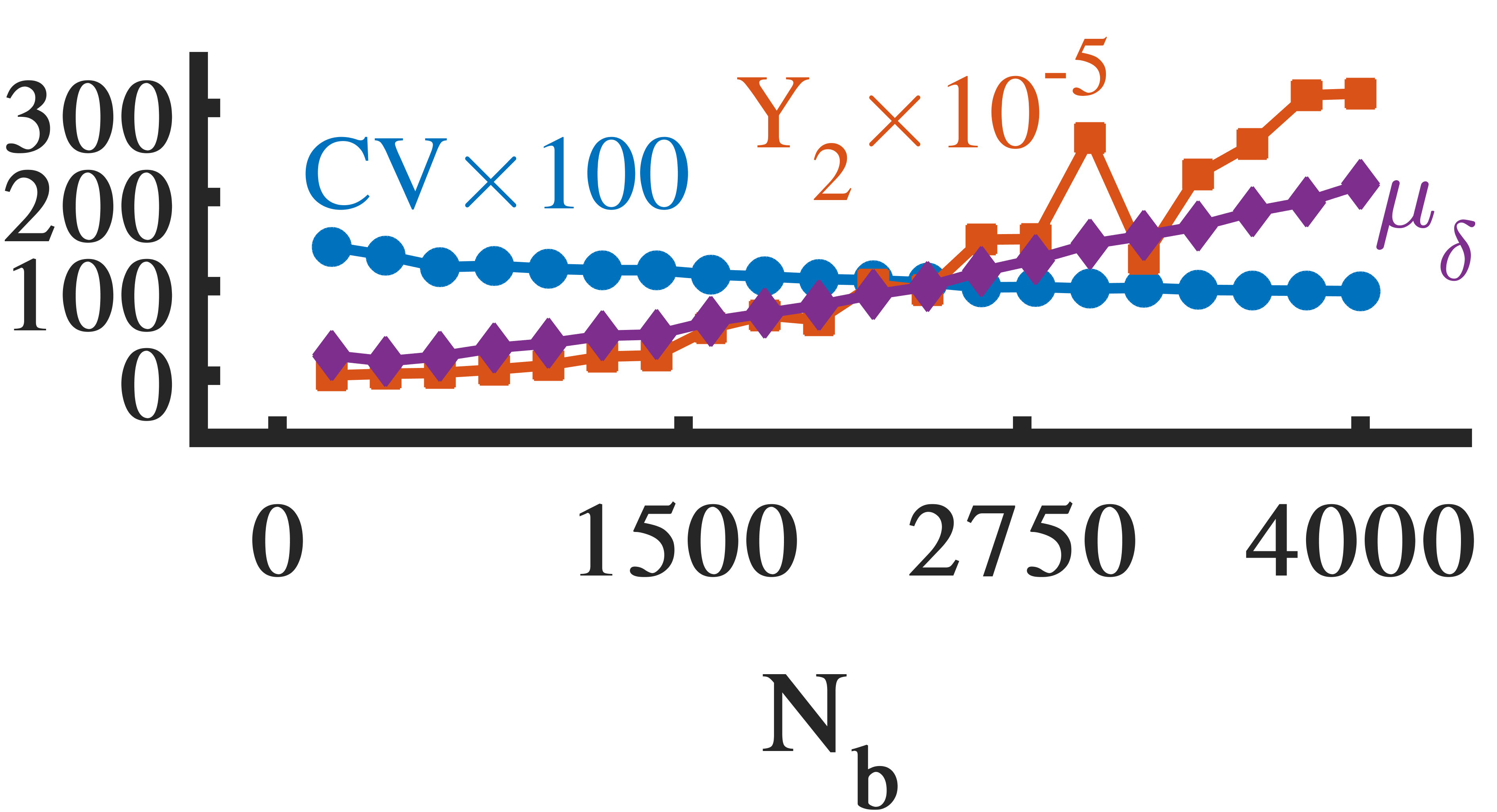}}
  \subfigure[S-Yahoo]{\includegraphics[width=0.24\textwidth]{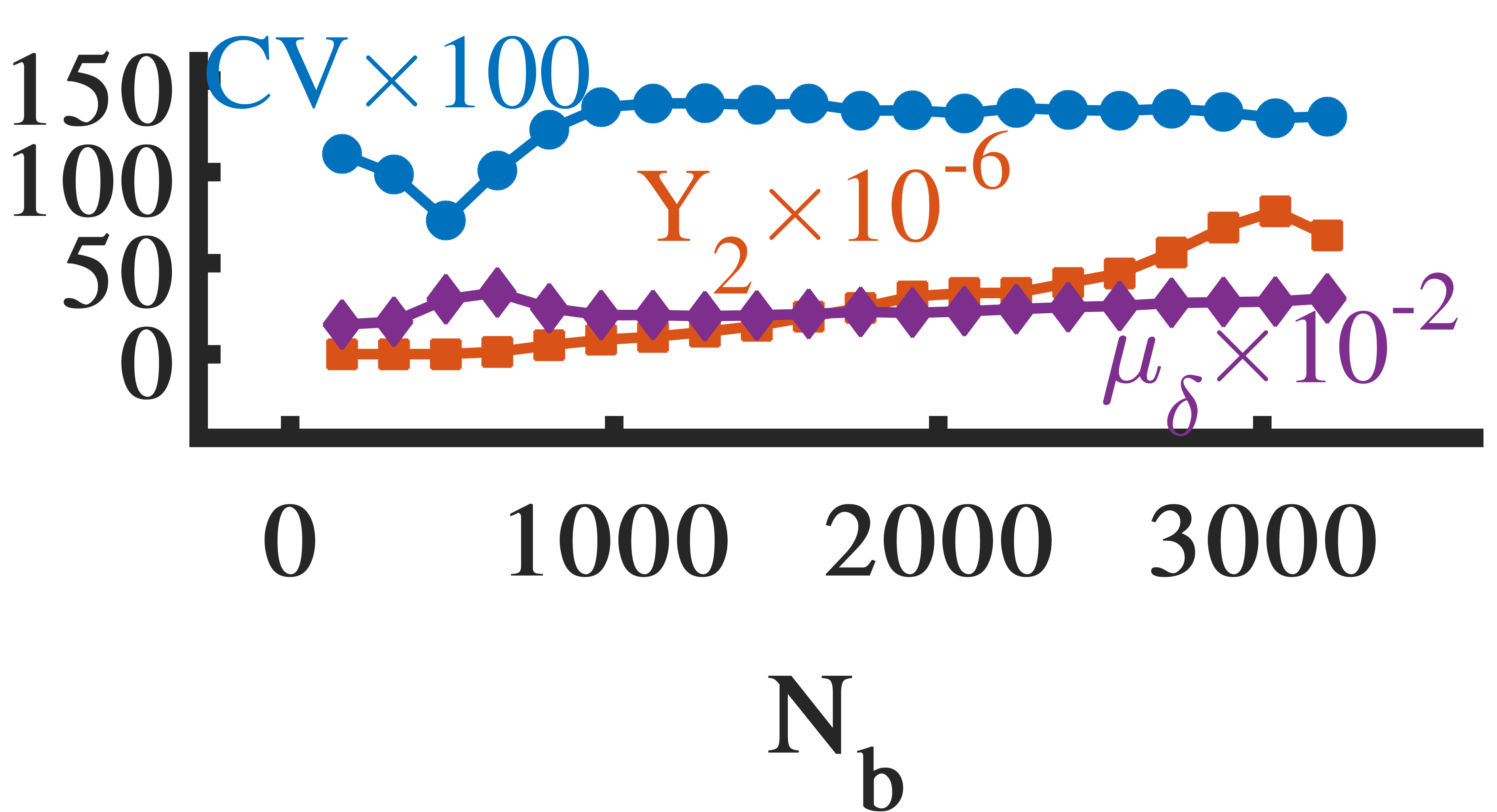}}
    \caption{Coefficient of variation $CV$ (circles), excess kurtosis $Y_2$ (squares), and mean $\mu_\delta$ (diamonds) of butterfly strength differences over the timeline of burst arrivals.} 
    \label{fig:statssynthetic}
\end{figure*}
\begin{figure*}[]
    \centering
    \subfigure[S-Ciao]{\includegraphics[width=0.24\textwidth]{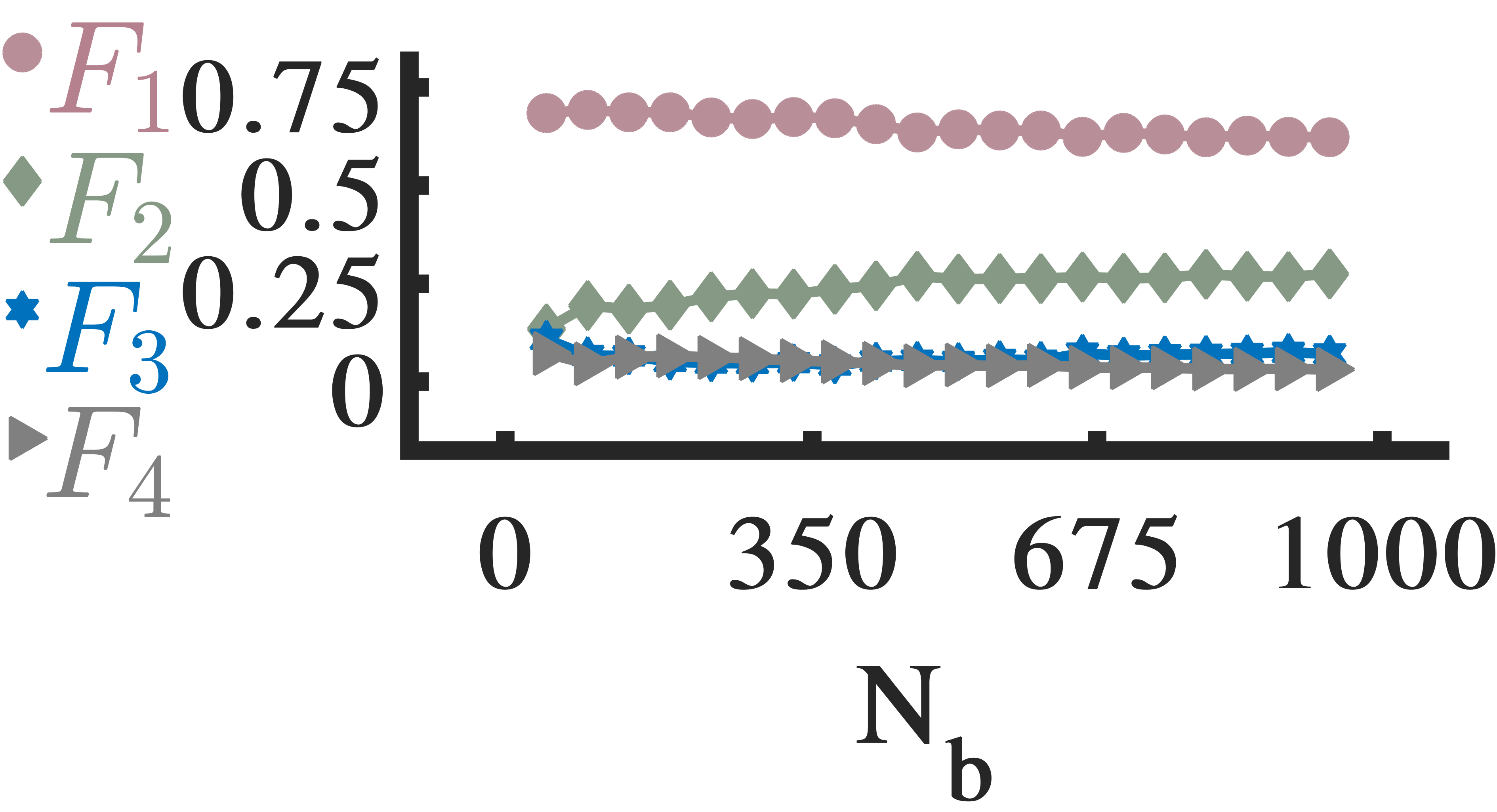}}
     \subfigure[S-Epinions]{\includegraphics[width=0.24\textwidth]{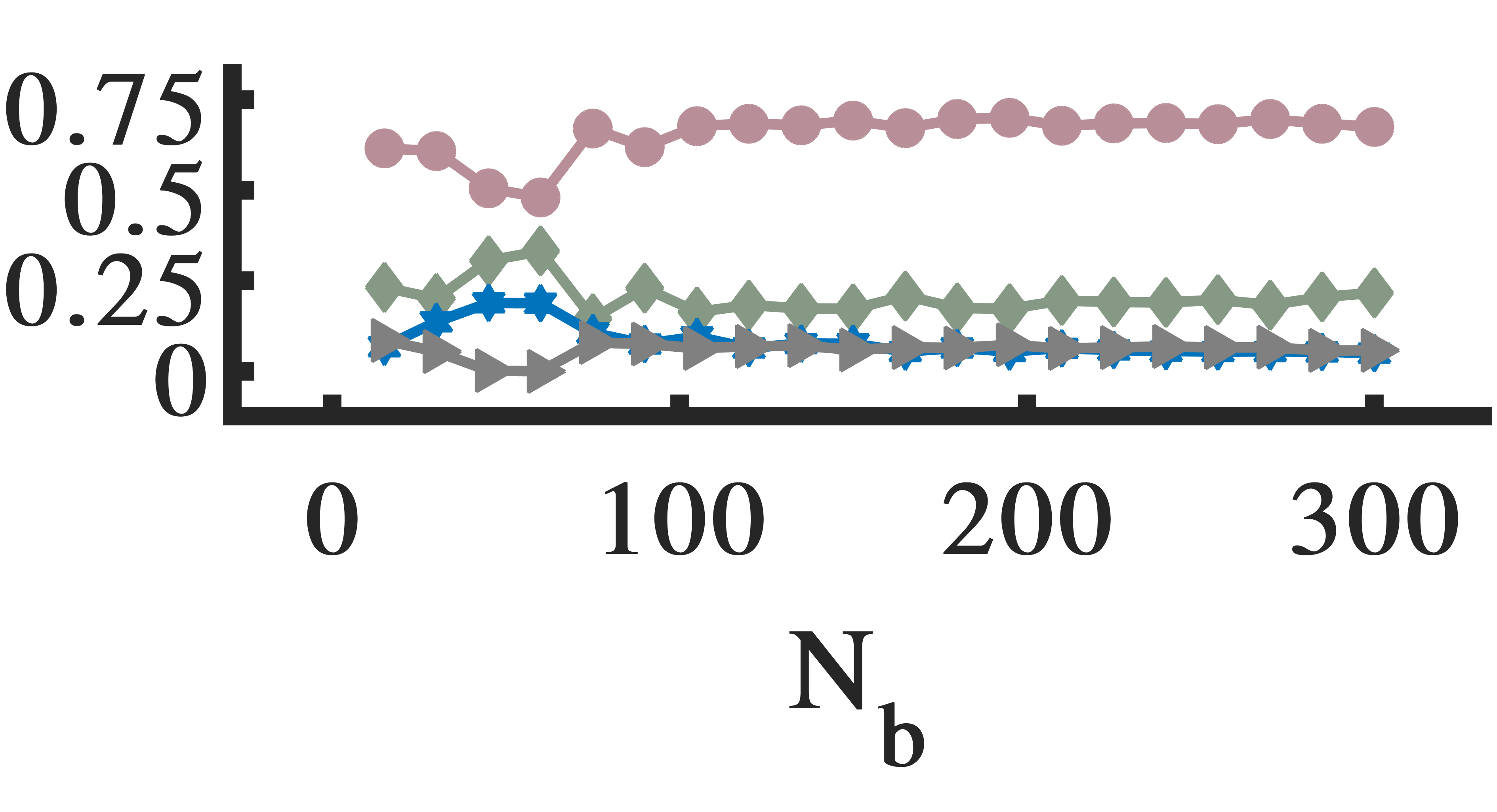}}
  \subfigure[S-WikiLens]{\includegraphics[width=0.24\textwidth]{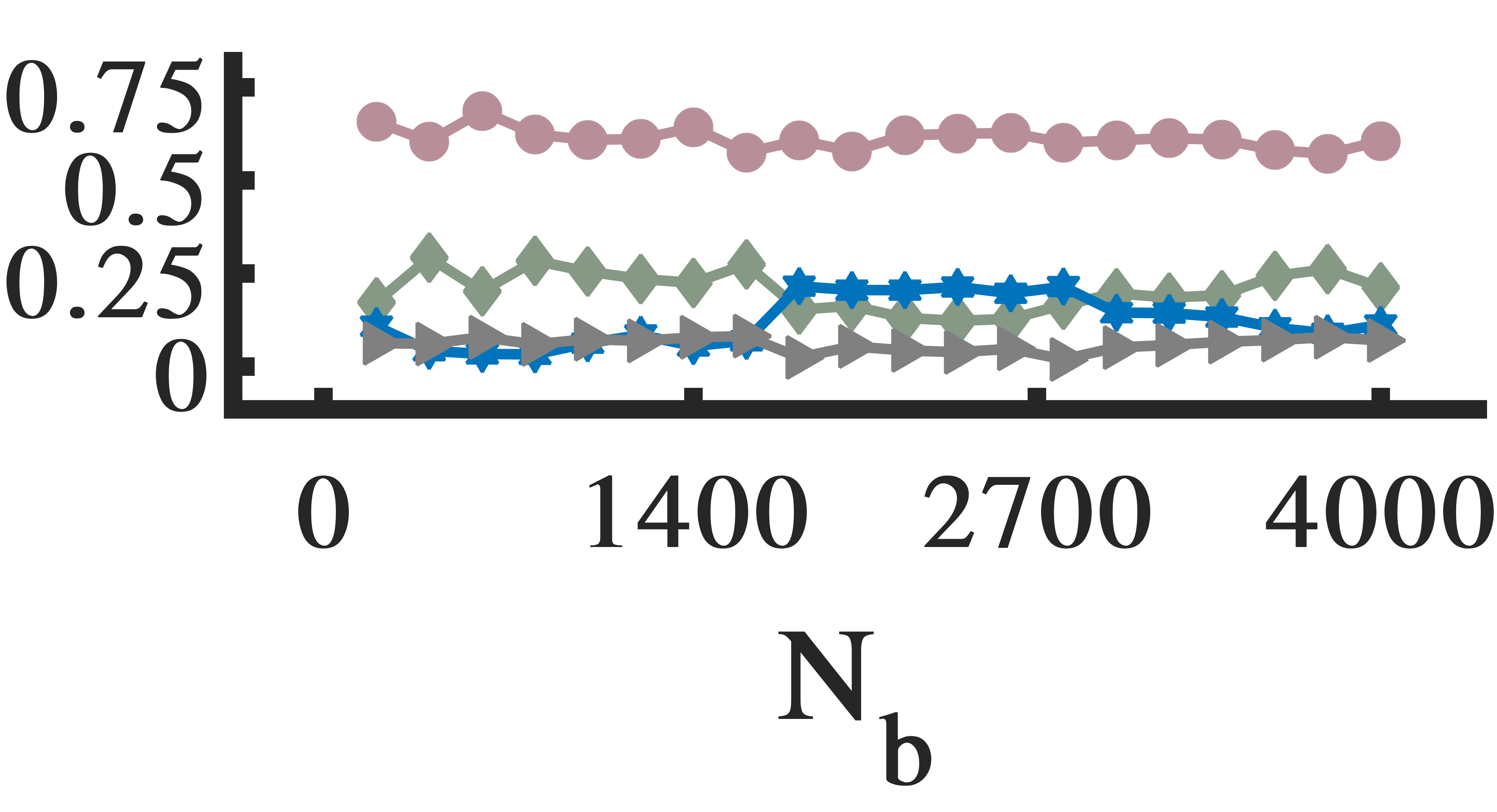}}
  \subfigure[S-ML100k]{\includegraphics[width=0.24\textwidth]{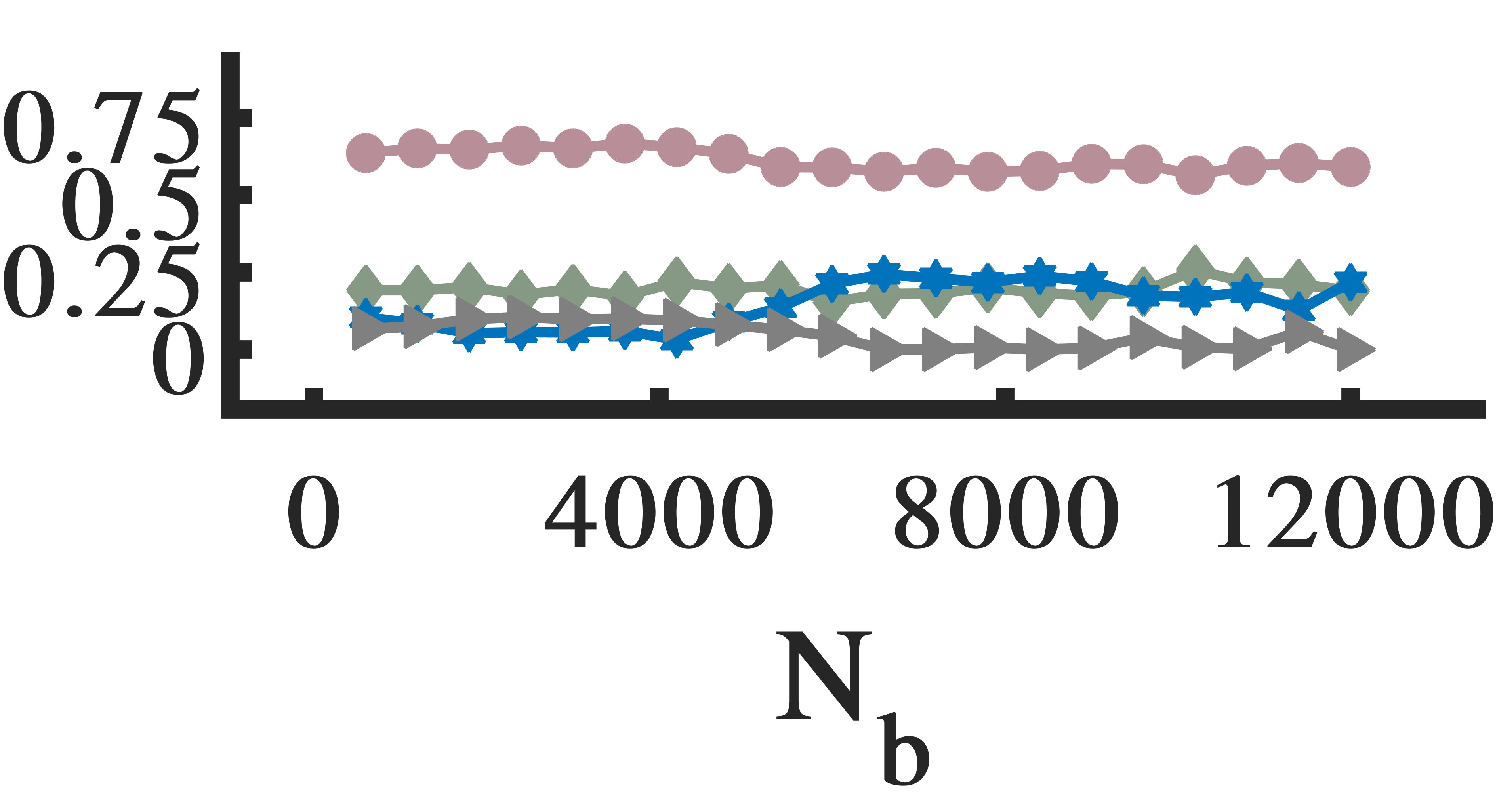}}
  \subfigure[S-ML1m]{\includegraphics[width=0.24\textwidth]{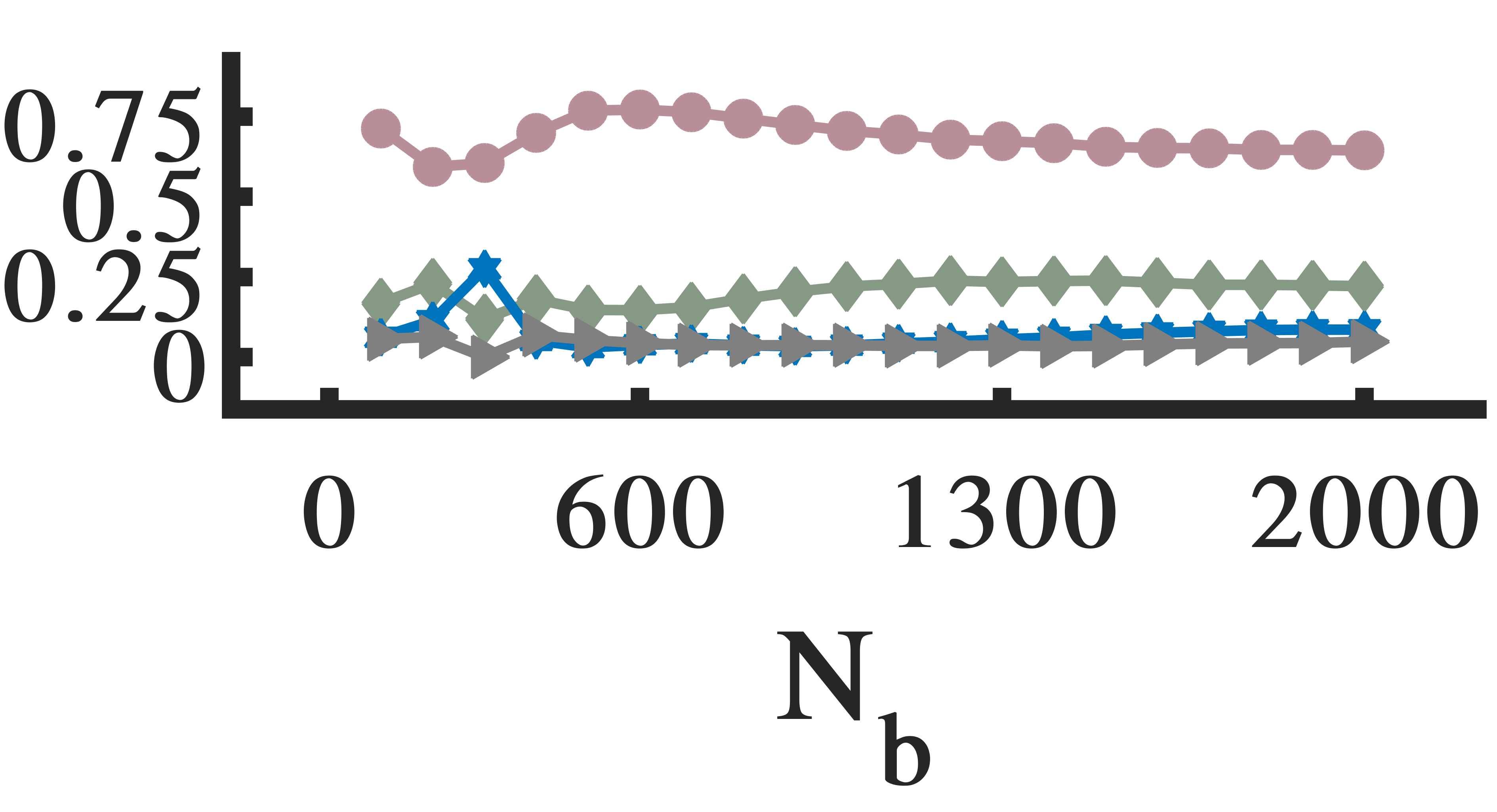}}
\subfigure[S-Amazon]{\includegraphics[width=0.24\textwidth]{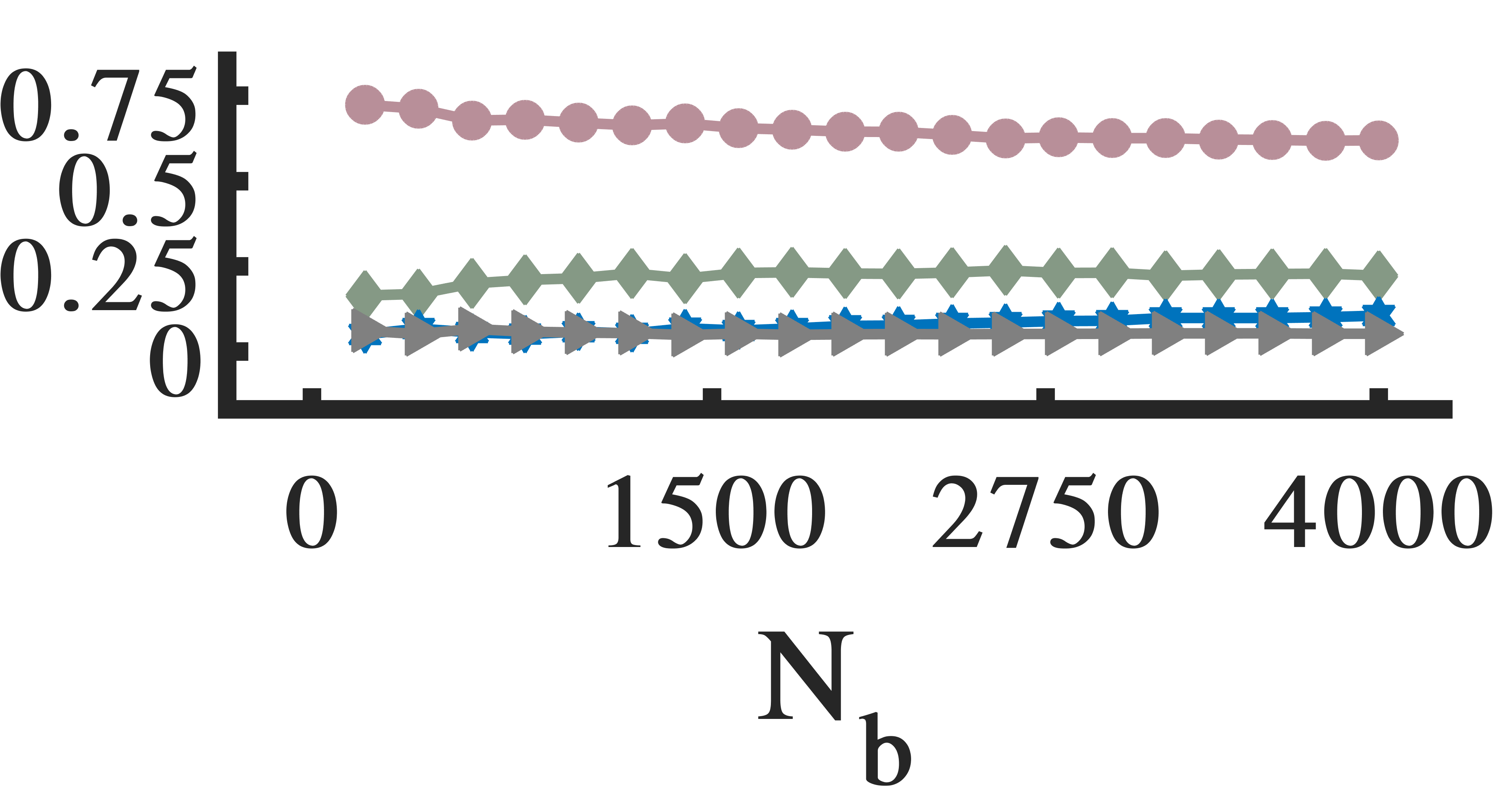}}
\subfigure[S-Yahoo]{\includegraphics[width=0.24\textwidth]{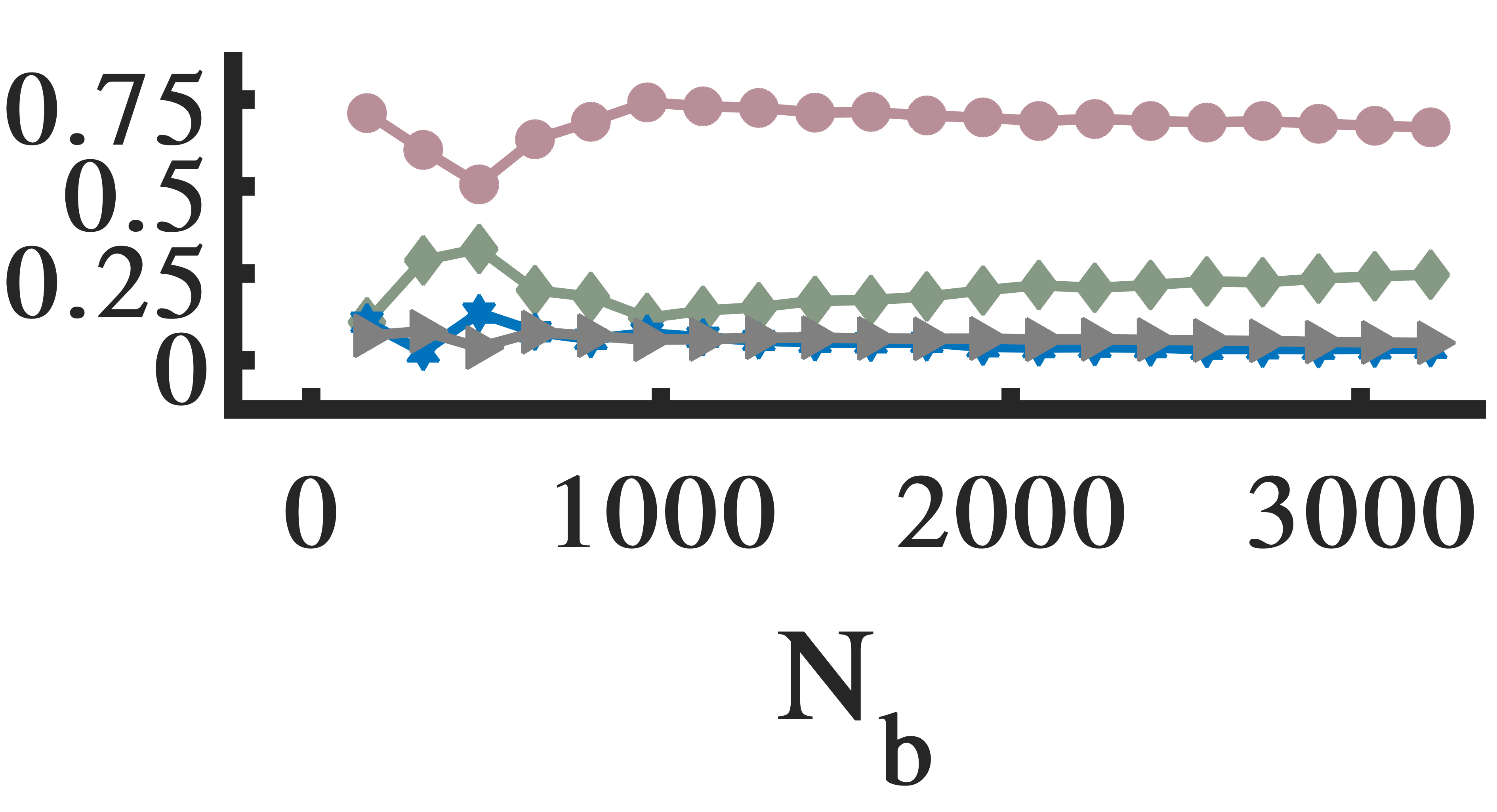}}
    \caption{F elements over the timeline of burst arrivals.}
    \label{fig:Fsyntheticgraphs}
\end{figure*}
\begin{figure*}[]
    \centering
  \subfigure[S-Ciao]{\includegraphics[width=0.24\textwidth]{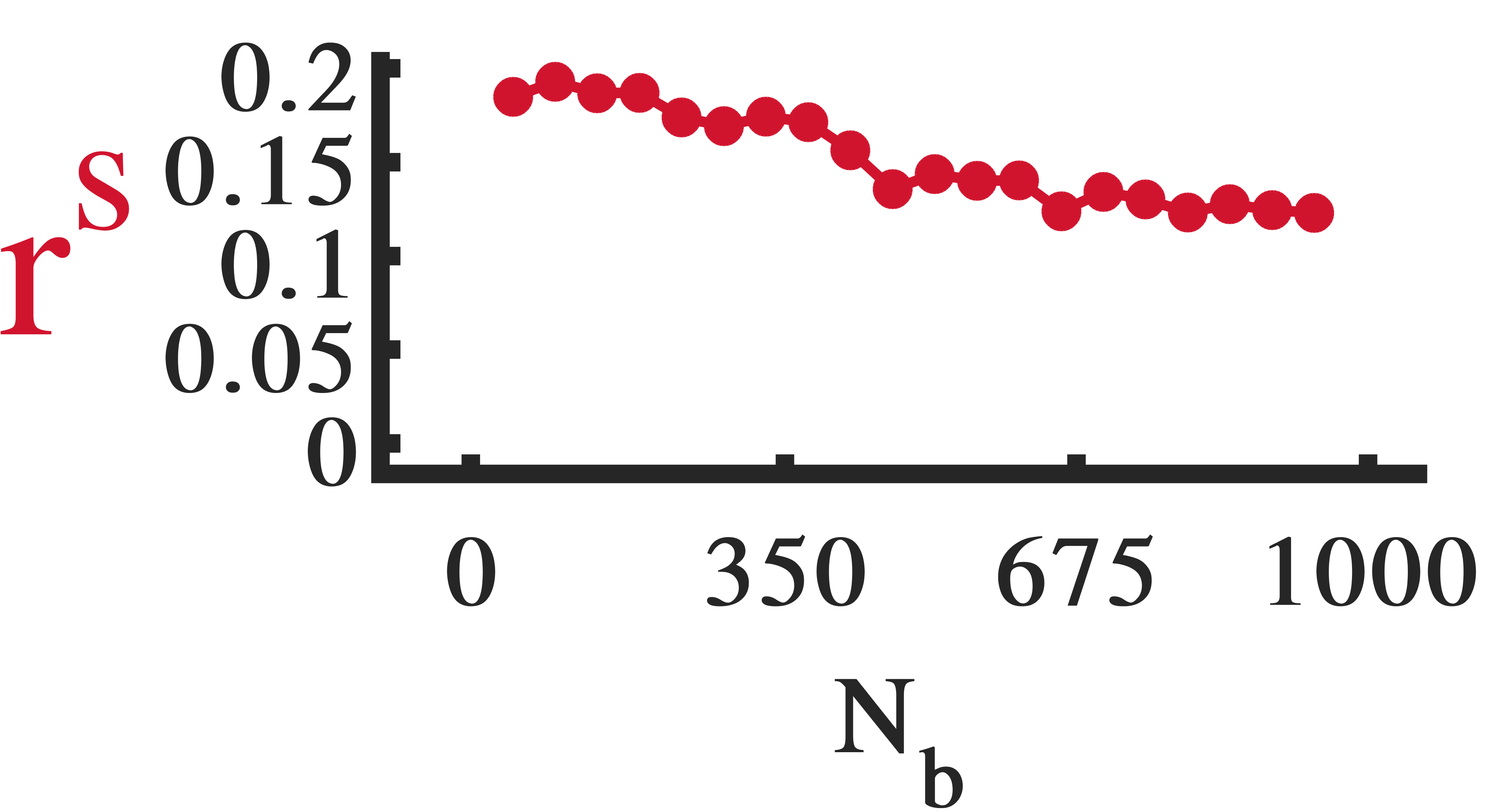}}
  \subfigure[S-Epinions]{\includegraphics[width=0.24\textwidth]{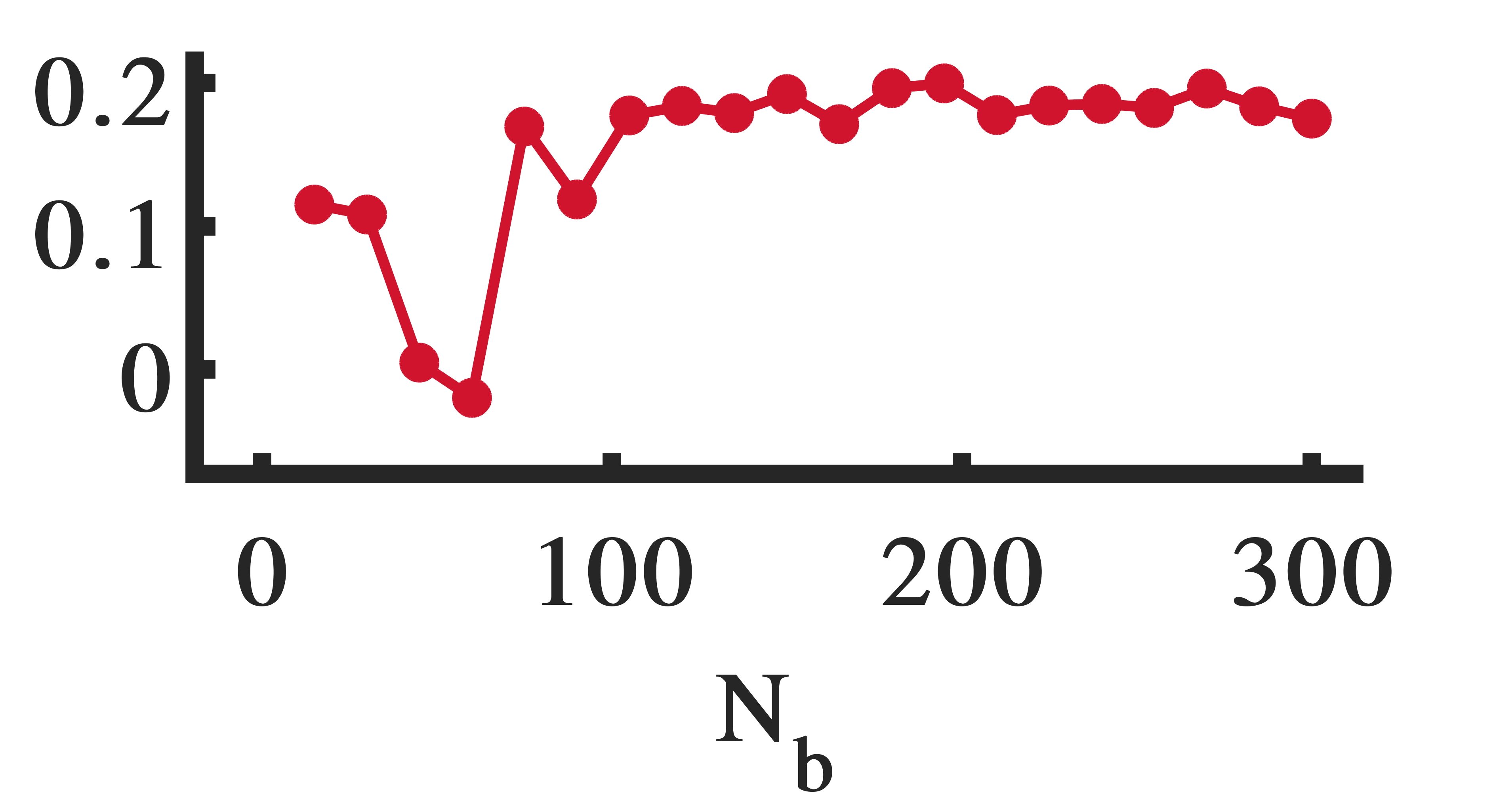}}
  \subfigure[S-WikiLens]{\includegraphics[width=0.24\textwidth]{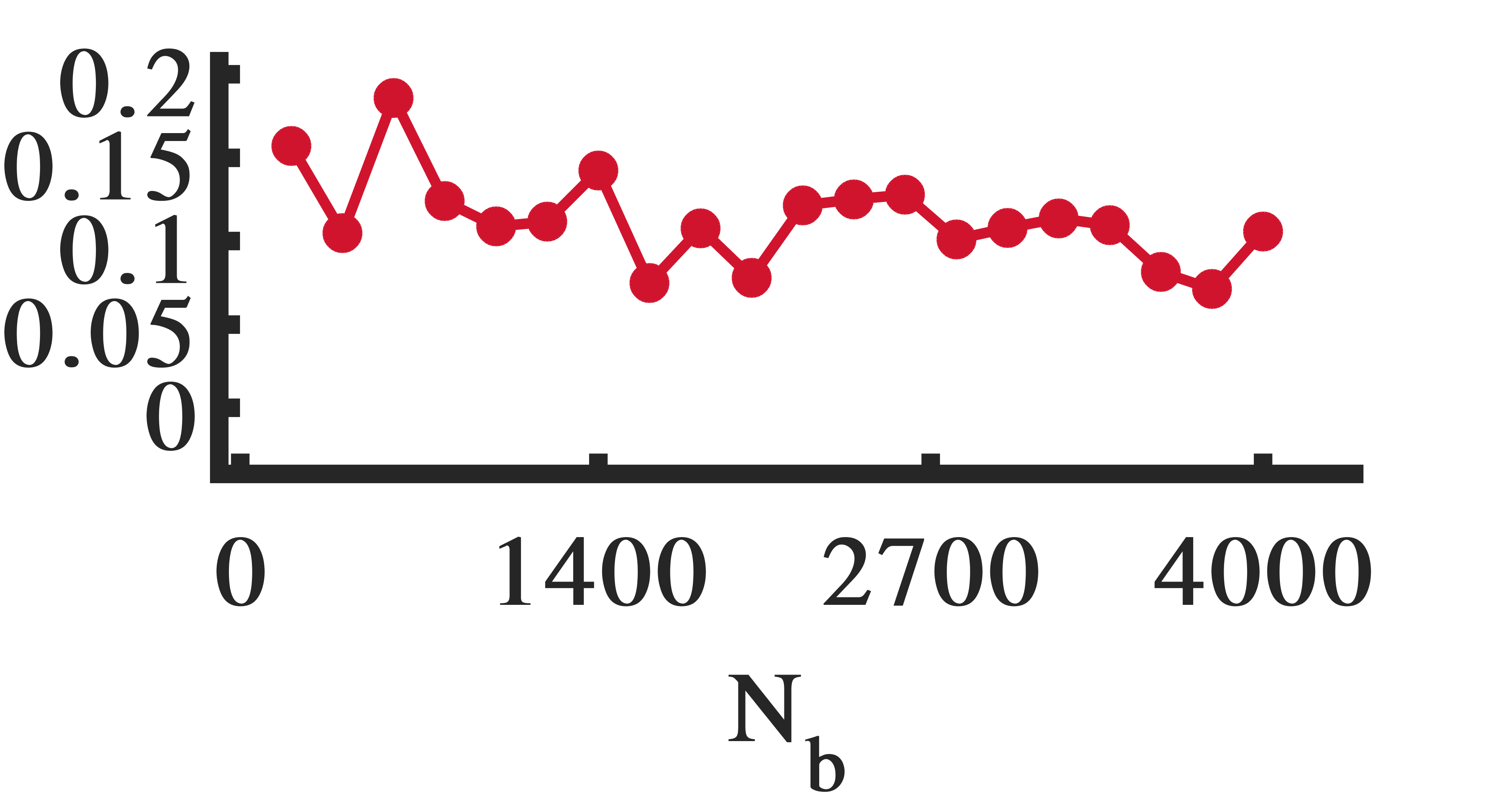}}
  \subfigure[S-ML100k]{\includegraphics[width=0.24\textwidth]{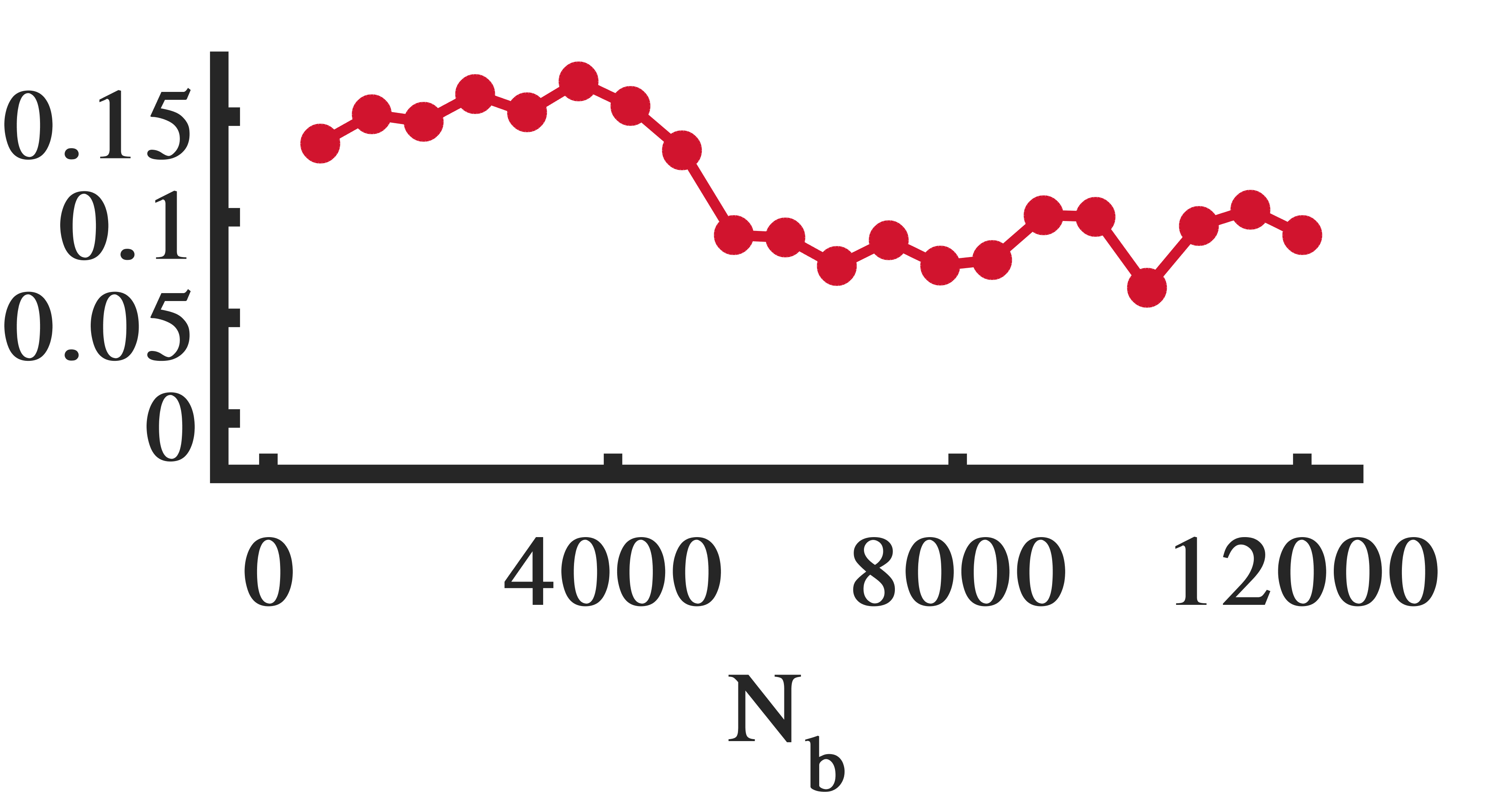}}
  \subfigure[S-ML1m]{\includegraphics[width=0.24\textwidth]{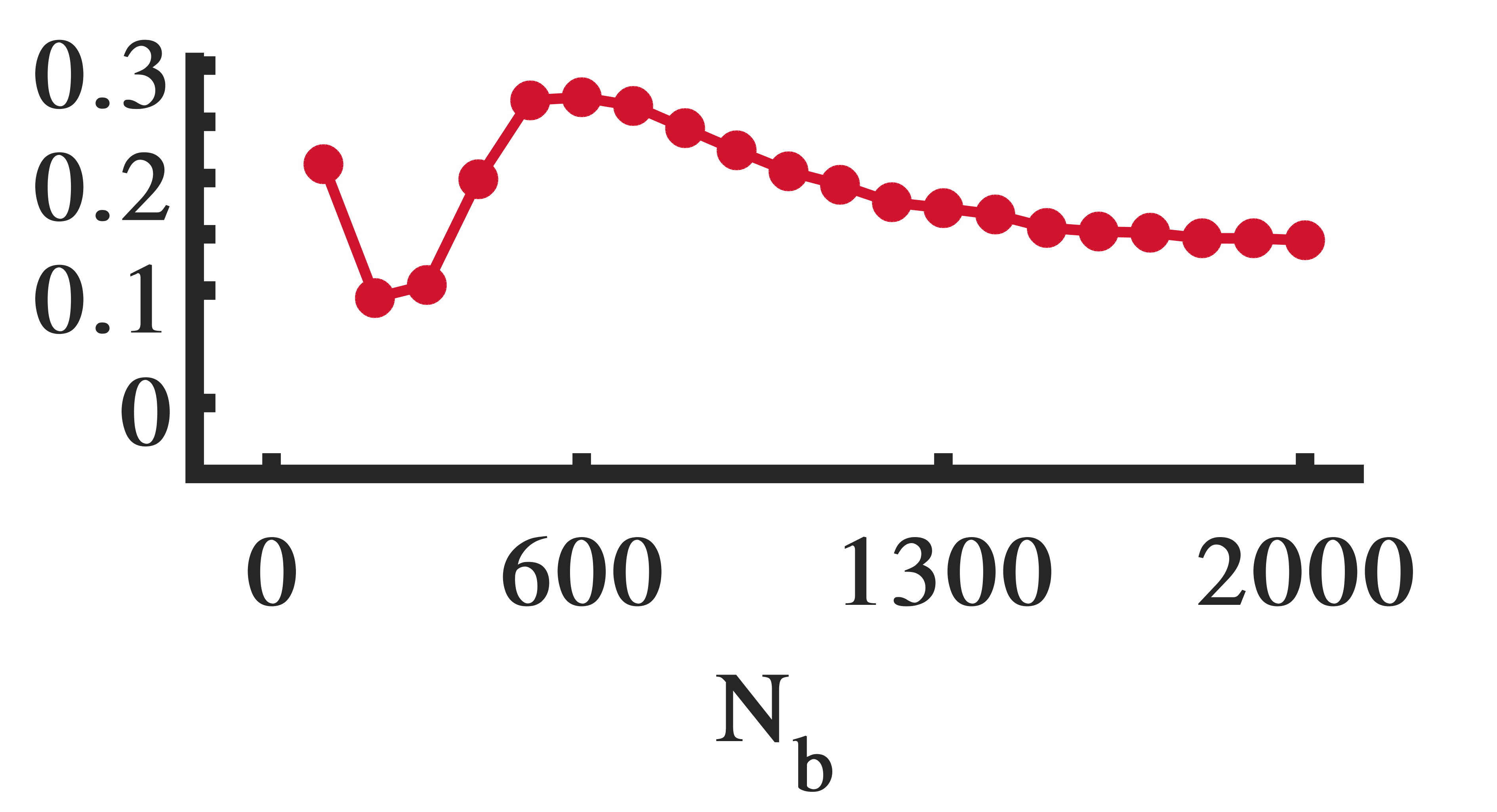}}
  \subfigure[S-Amazon]{\includegraphics[width=0.24\textwidth]{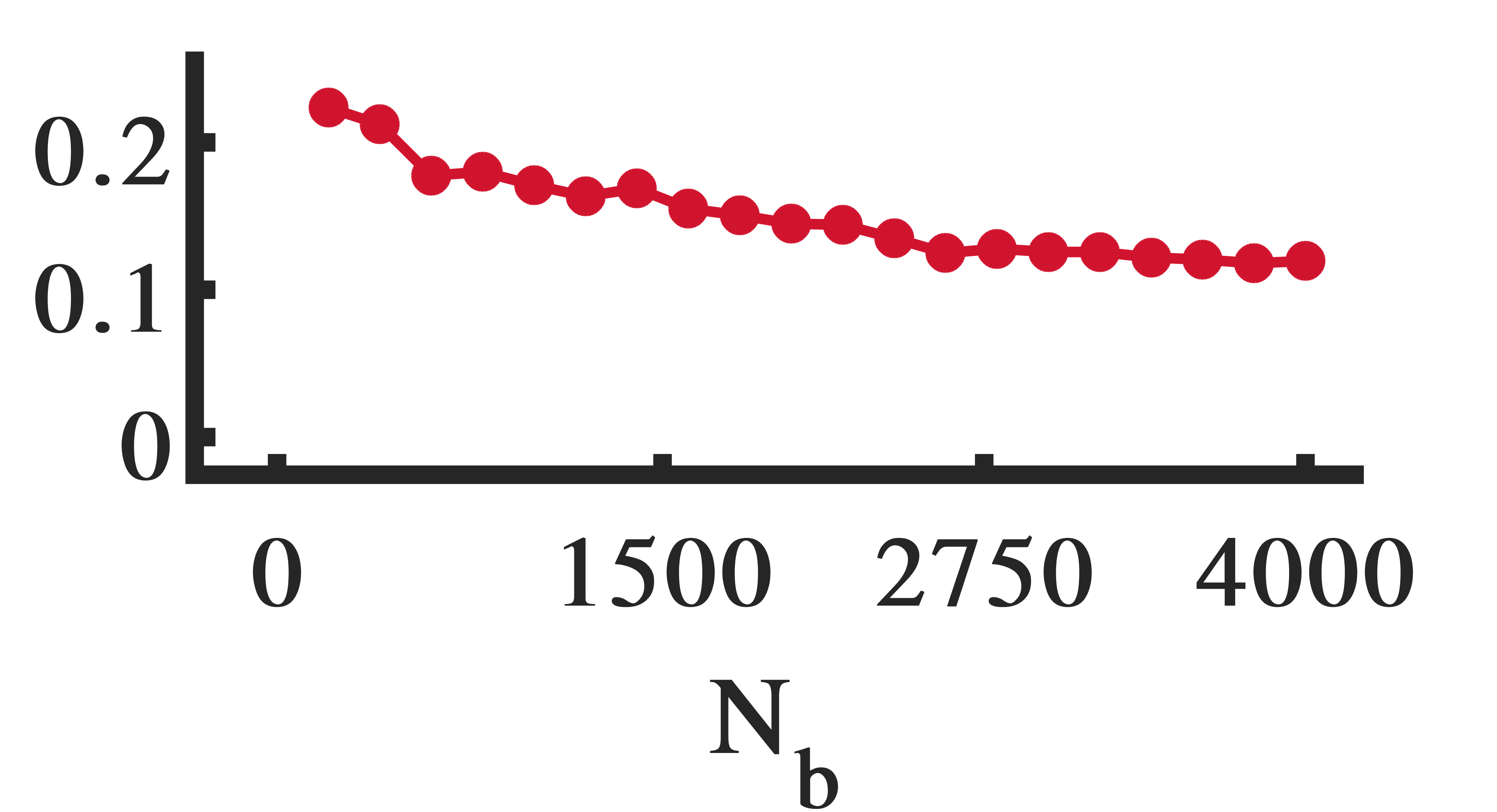}}
   \subfigure[S-Yahoo]{\includegraphics[width=0.24\textwidth]{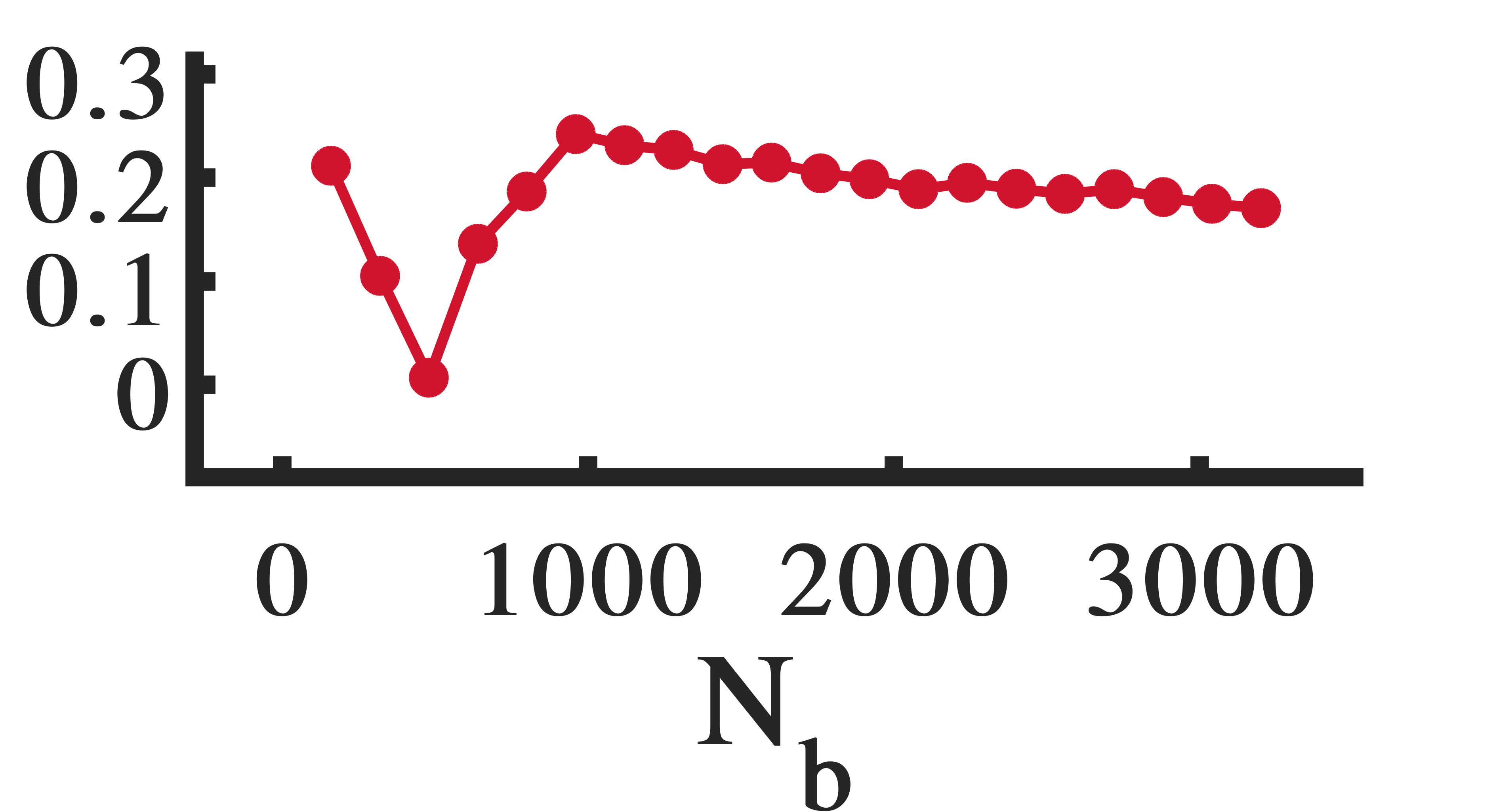}}
    \caption{Strength assortativity localization factor $r^s$ of butterfly vertices over the timeline of burst arrivals.}
    \label{fig:rssynthetic}
\end{figure*}
\begin{table*}[]\caption{Average butterfly rate in synthetic streams generated by our model.} \small \centering
    \begin{tabular}{p{2cm} p{3cm} p{1.5cm} p{1.5cm} p{1.5cm}}
    &  average butterfly rate & $|E^{20}|$ & $N_b^{20}$ & $\bowtie^{20}$  \\\hline 
   S-Ciao &  $37.2440\pm 31.4987$ & $72,966$ & $940$ & $7,430,097$
    \\S-Epinions &  $9.4358\pm 6.1519$ & $14,495$ & $300$ & $304,417$ 
    \\S-WikiLens & $6.3037\pm 5.6897$ & $25,957$ & $4,000$ & $69,302$
    \\S-ML100k &  $2.191\pm 1.3045$ & $87,951$ & $12,000$ & $204,118$
    \\S-ML1m &  $59.8963\pm29.2066$ & $85,397$ & $2,000$ & $7,860,689$ 
    \\S-Amazon & $30.8994\pm28.8517$ & $106,714$ & $4,000$ & $9,599,739$
    \\S-Yahoo & $51.1708\pm27.891$ & $120,616$ & $3,200$ & $10,847,746$
    \end{tabular}\label{tab:butterflyrateModel}
\end{table*}
\begin{table*}[h]\caption{ Mean absolute error of $r^s$ and $F$ elements.}
\small \centering
    \begin{tabular}{p{1.5cm}p{1.5cm}p{1.5cm}p{1.5cm}p{1.5cm}}
          &  $F_1$, $r^s$ &  $F_2$ &  $F_3$ &  $F_4$ \\ \hline
         S-Ciao & $0.0286$ & $0.05$ & $0.0569$ & $0.0199$\\
         S-Epinions & $0.0318$ & $0.0389$ & $0.0219$ & $0.0163$\\ 
         S-wikiLens & $0.0528$ & $0.0705$ & $0.08$ & $0.0194$\\
         S-ML100k & $0.0376$ & $0.0513$ & $0.0642$ & $0.0384$\\
         S-ML1m & $0.0617$ & $0.0445$ & $0.049$ & $0.015$\\
         S-Amazon & $0.1064$ & $0.0539$ & $0.0476$ & $0.094$\\
         S-Yahoo & $0.0578$ & $0.0782$ & $0.0519$ & $0.0316$
    \end{tabular}\label{tab:Ferrors}
\end{table*}
\subsection{Stress testing}\label{subsec:stresstesting}
We evaluate the impact of \emph{sGrow} techniques (batch of isolated edges, probabilistic connections, the random walk backbone, and the sliding window) as the following: 
\begin{itemize}
    \item We use a diverse range of parameters $M\in\{ 100,150,200,250,300 \}$, $\rho\in\{ 0.3,0.4,0.5,0.6,0.7 \}$, $L_{min}\in\{ 1,2,3,4 \}$, $L_{max}\in\{ 3,4,5 \}$, and $\beta\in\{ 5,10,15,20 \}$ to create S-Amazon stream with initial prefix of $10^3$ sgrs from the Amazon stream. In each row of Tables~\ref{tbl:tableoffiguresrs}, \ref{tbl:tableoffigures}, and \ref{tbl:burstsize}, we examine the effect of one parameter on effectiveness, efficiency, and frequency distribution of burst sizes, respectively. 
    \item We also create S-Amazon stream with $10^7$ sgrs such that after generating $5\times10^6$, one of the parameters switches from $M=100$, $\beta=5$, $L\in[1 5]$ and $\rho=0.4$ to $M=300$, $\beta=20$, $L\in[4 5]$ and $\rho=0.8$. In Figure~\ref{fig:dynamic}, we examine the effect of the parameter switch on effectiveness and efficiency.
\end{itemize}

\subsubsection{Effectiveness}\label{subsubsec:effectiveness}
In Table~\ref{tbl:tableoffiguresrs}, we check the evolution of $r^s$ as the number of butterflies grows to $10^7$ in S-Amazon stream with different parameter configurations. We observe that data points are not clustered by colors (corresponding parameters $[L_{min},L_{max}]$, $M$, and $\beta$) in rows a, c, and d, however they are clustered and ordered by $\rho$. Also, $0.1 \leq r^s \leq 0.15$ regardless of  $[L_{min},L_{max}]$, $M$, and $\beta$. The value of $r^s$ is stable at a positive value, which is higher for lower connection probabilities $\rho$. As the connection probability $\rho$ increases, the probability of establishing connections between the newly added vertices with low strength and high strength neighbors of the PRW vertices increases, therefore the assortativity level decreases. In Figure~\ref{fig:dynamic}, we check the evolution of $r^s$ as the number of butterflies grows to $10^7$ in S-Amazon stream with parameter switch in the middle of stream generation. We observe that compared to the stream with static parameters, the data points of streams with dynamic parameters follow the same pattern after the switch of $[L_{min},L_{max}]$, $M$, and $\beta$. Increasing the value of $\rho$ slightly decreases $r^s$, yet the steady state is retained after the switch.

\subsubsection{Efficiency}\label{subsubsec:efficiency} 
In Table~\ref{tbl:tableoffigures}, we check the time for generation as the stream grows to $10^7$ sgrs in S-Amazon stream with different parameter configurations.  
We observe that data points are not clustered by colors (corresponding parameters $[L_{min},L_{max}]$, $M$, and $\beta$) in rows a, c, and d, however they are clustered by $\rho$ and ordered by both $\rho$ and $[L_{min},L_{max}]$. That is, the generation time is not impacted by $M$ 
(Table~\ref{tbl:tableoffigures}.c) or $\beta$ 
(Table~\ref{tbl:tableoffigures}.d), however it is affected by $\rho$ and $[L_{min},L_{max}]$. 
As the connection probability $\rho$ increases, the generation time decreases since the number and the size of bursts created at each time step increases (Table~\ref{tbl:tableoffigures}.b).  
As the range of random walk length $L\in[L_{min},L_{max}]$  increases, the generation time decreases since the size of bursts created at each time step increases (Table~\ref{tbl:tableoffigures}.a).
In Figure~\ref{fig:dynamic}, we check the time for generation as the stream grows to $10^7$ sgrs in S-Amazon stream with parameter switch in the middle of stream generation. We observe that switching $M$ and $\beta$ gradually decreases the slope of generation time curve, while switching $L$ and $\rho$ promptly and significantly changes the slope. This confirms that $L$ and $\rho$ significantly determine the generation time.

\subsubsection{Burst Size}\label{subsubsec:burstsize} 
In Table~\ref{tbl:burstsize}, we check the frequency distribution of all burst sizes in S-Amazon with $10^7$ sgrs generated by different parameter configurations. We observe that for all values of $\rho$ and $L_{max}$ by increasing $M$, the range of burst sizes expands, therefore $M$ significantly impacts burst sizes. Also, the maximum burs size increases as $\rho$ and $L_{max}$ increase.

\subsubsection{Parameter Configuration}\label{subsec:parameterconfig}
The introduced techniques with configurable parameters enable generating realistic bipartite streaming graphs with generation time sub-linear with respect to the number of sgrs. The  burstier the stream, the lower the generation time. In the following we elaborate a reference guide for configuring the parameters.

$\boldsymbol{M}$ -- The upper bound for the random number of new batched sgrs added at each timestep does not impact the strength assortativity patterns and the generation time. This parameter can be comfortably used to adjust the level of burstiness of the streaming graph without affecting the performance of the generative algorithm.

$\boldsymbol{\rho}$ -- The probability of creating an edge between each random walk vertex and a random vertex or an edge between each new sgr and neighbors of the random walk vertices impacts the level of strength assortativity but not its steady state. It also affects the generation time. This parameter can be used to trade off scalability and the level of strength assortativity. The default value is $\rho=0.3$, yet increasing $\rho$ would decrease the generation time and strength assortativity level. Values less than or equal to $0.7$ ensure positive strength assortativity. Also, $\rho$ determines the probability of out-of-order sgrs.

$\boldsymbol{[L_{min},L_{max}]}$ -- The range for random length of PRW backbone used for establishing burst of sgrs does not impact the strength assortativity patterns. It affects the generation time. This parameter can be used to increase the scalability of the stream generation. The default range is $[1,2]$ and shifting/expanding the range by increasing the lower or upper bound would decrease the generation time.

$\boldsymbol{\beta}$ -- The sliding window parameter used as the frequency and the size of sliding does not impact the strength assortativity patterns and generation time. This parameter can be comfortably used for creating streams in which sgrs are semantically time-sensitive and need a user-specified slide parameter. The default value is $\beta=5$, while any other value can be set.
\newcolumntype{M}[1]{>{\centering\arraybackslash}m{#1}}
\begin{table}[htb!]\centering
        \begin{tabular}{cM{35mm}M{35mm}M{35mm}}
            
             & \hspace{8mm}\boldsymbol{$L_{max}$$=$$3$$-$$5$} & \hspace{8mm}\boldsymbol{$L_{min}$$=$$1$$-$$4$} &  \\       
            \toprule   
            \makecell{(a)\\$M$$=$$100$\\$\rho$$=$$0.3,0.5,0.7$\\$\beta$$=$$5$} &
            \includegraphics[width=0.288\textwidth]{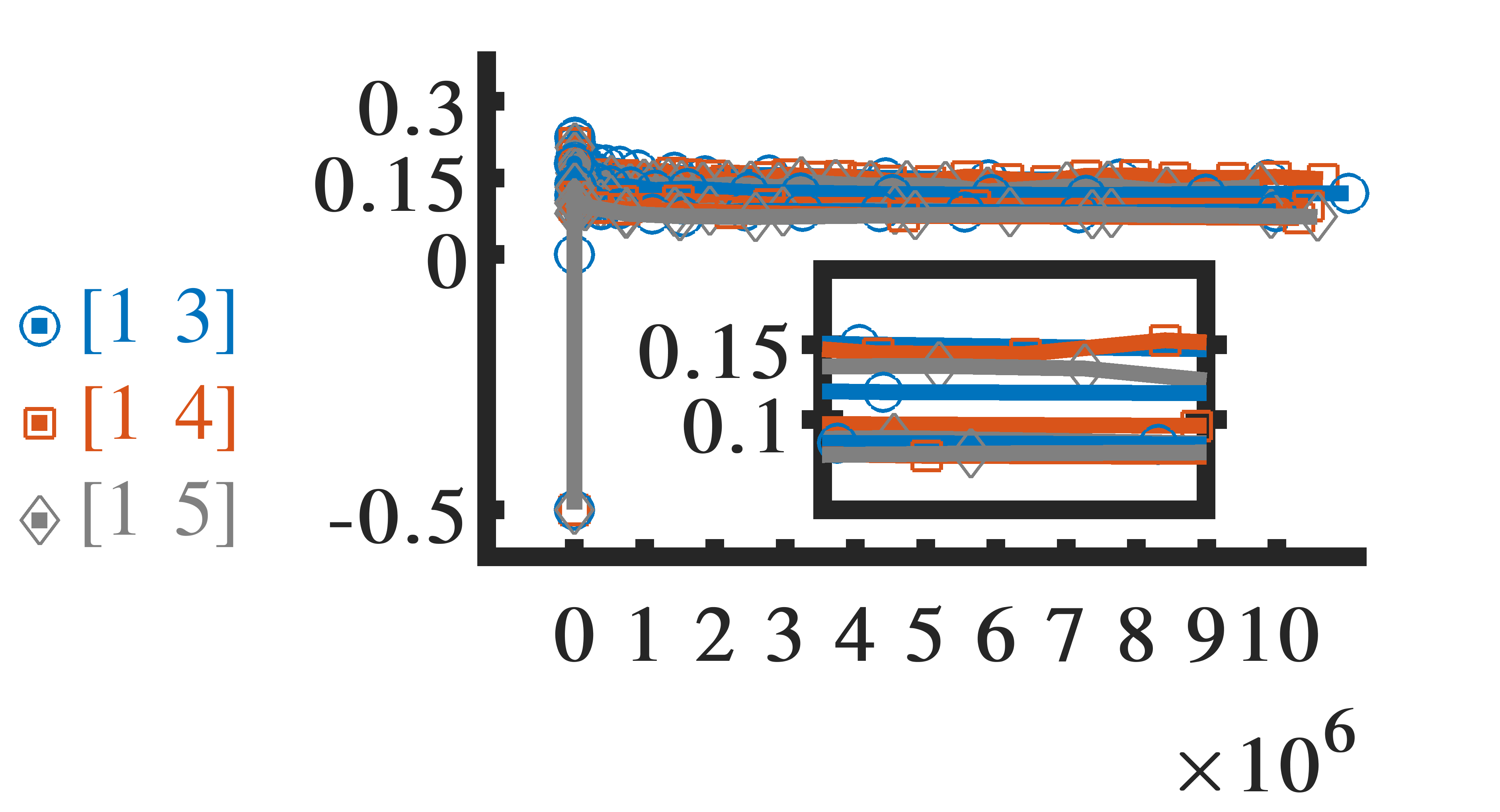} &
            \includegraphics[width=0.288\textwidth]{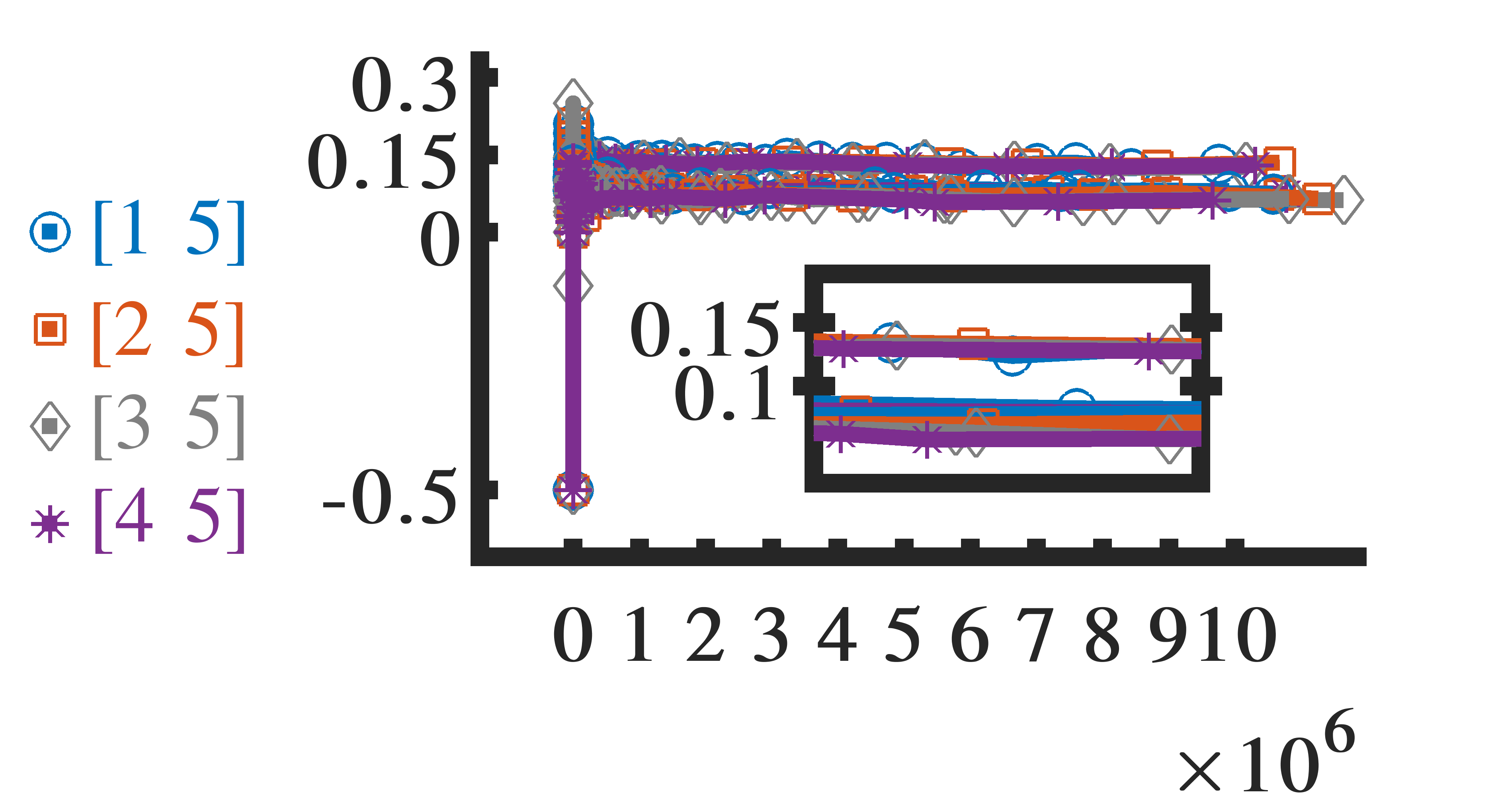} \\
            
             &\hspace{8mm} $L$$\in$$[1,3]$ & \hspace{8mm}$L$$\in$$[1,4]$ & \hspace{8mm} $L$$\in$$[1,5]$   \\
            \toprule 
            
            \makecell{(b)\\$M$$=$$100$\\\boldsymbol{$\rho$$=$$0.3-0.7$}\\$\beta$$=$$5$} & 
            \includegraphics[width=0.279\textwidth]{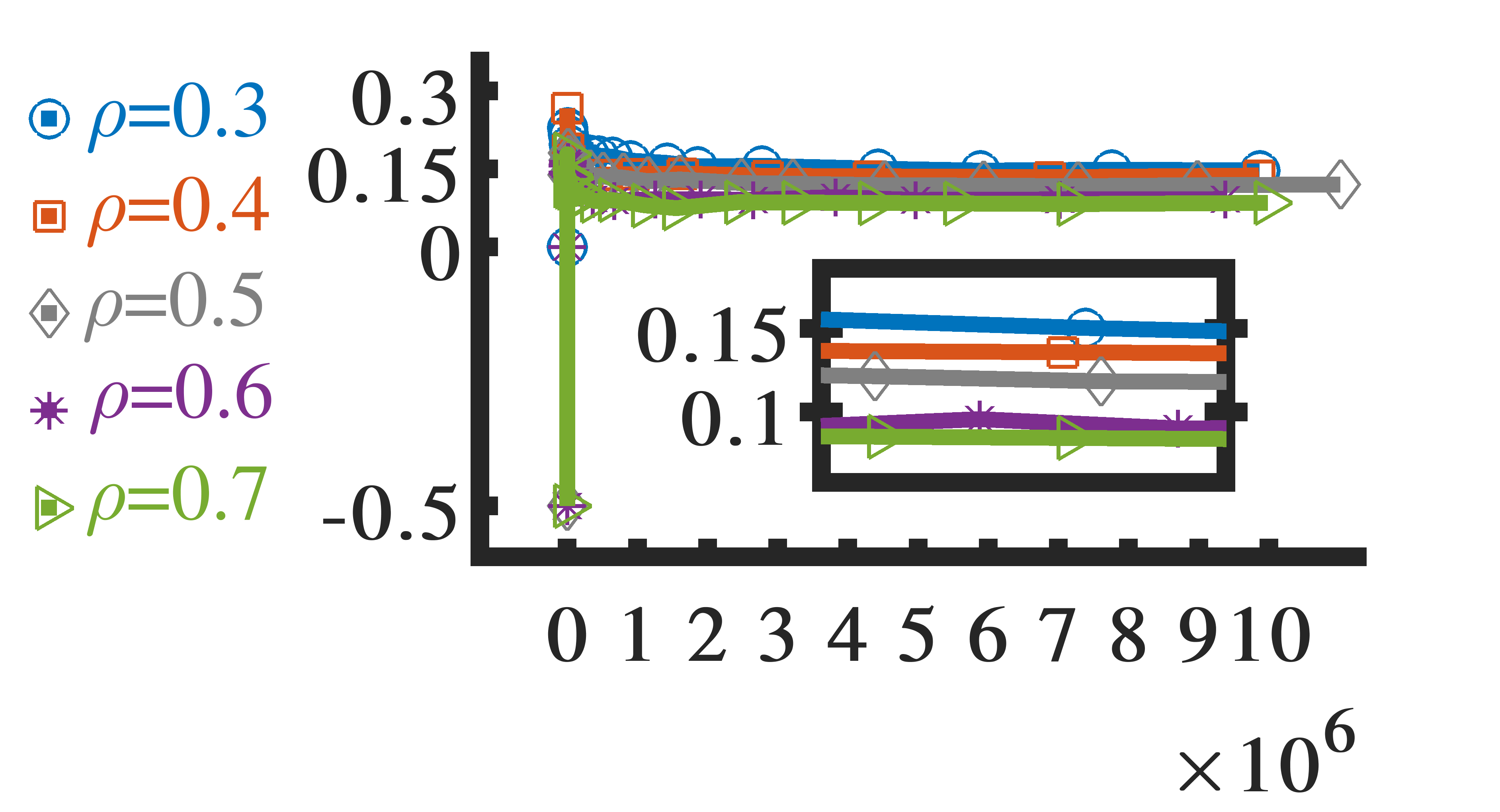} &  
            \includegraphics[width=0.288\textwidth]{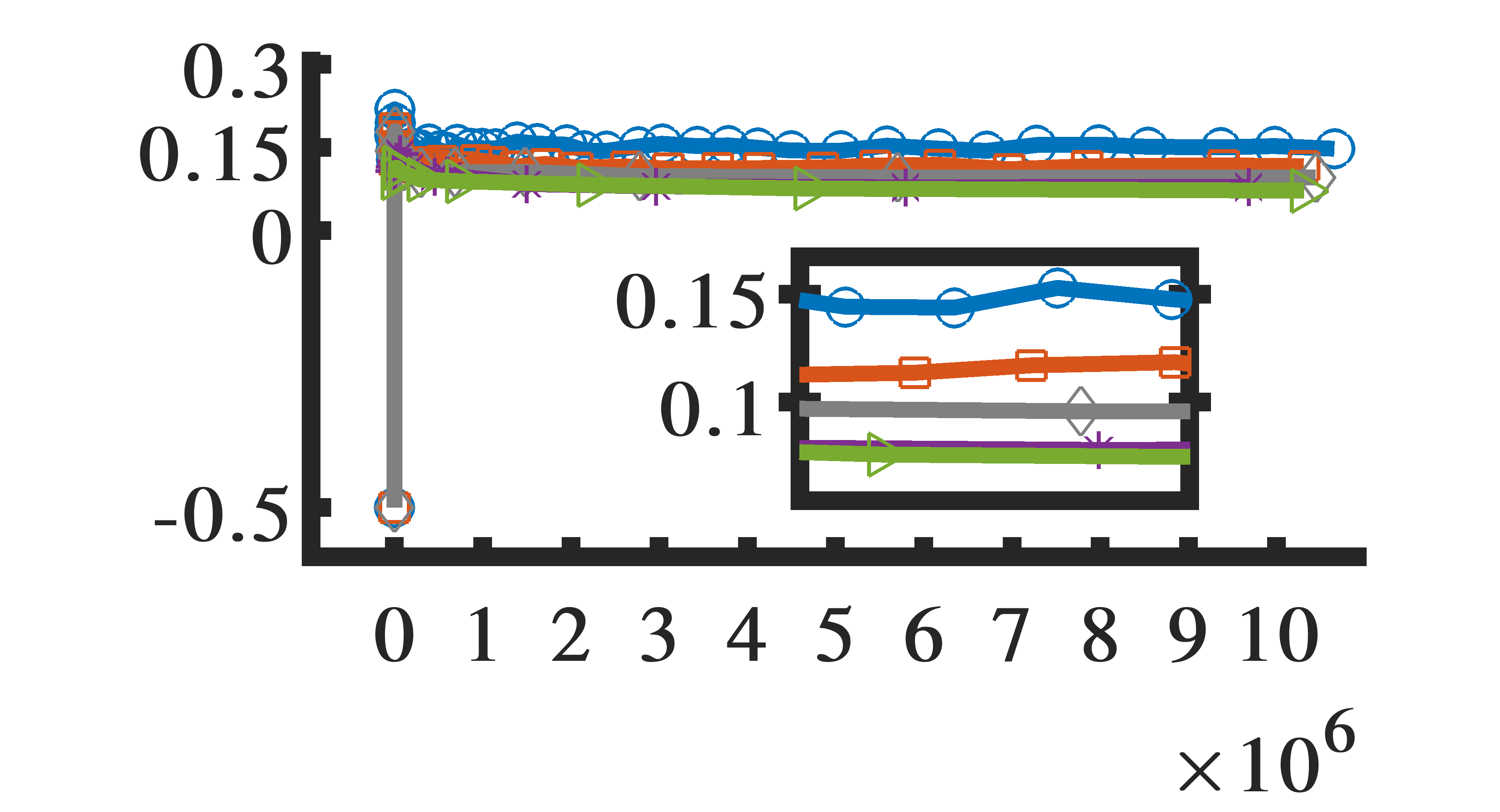}& 
            \includegraphics[width=0.288\textwidth]{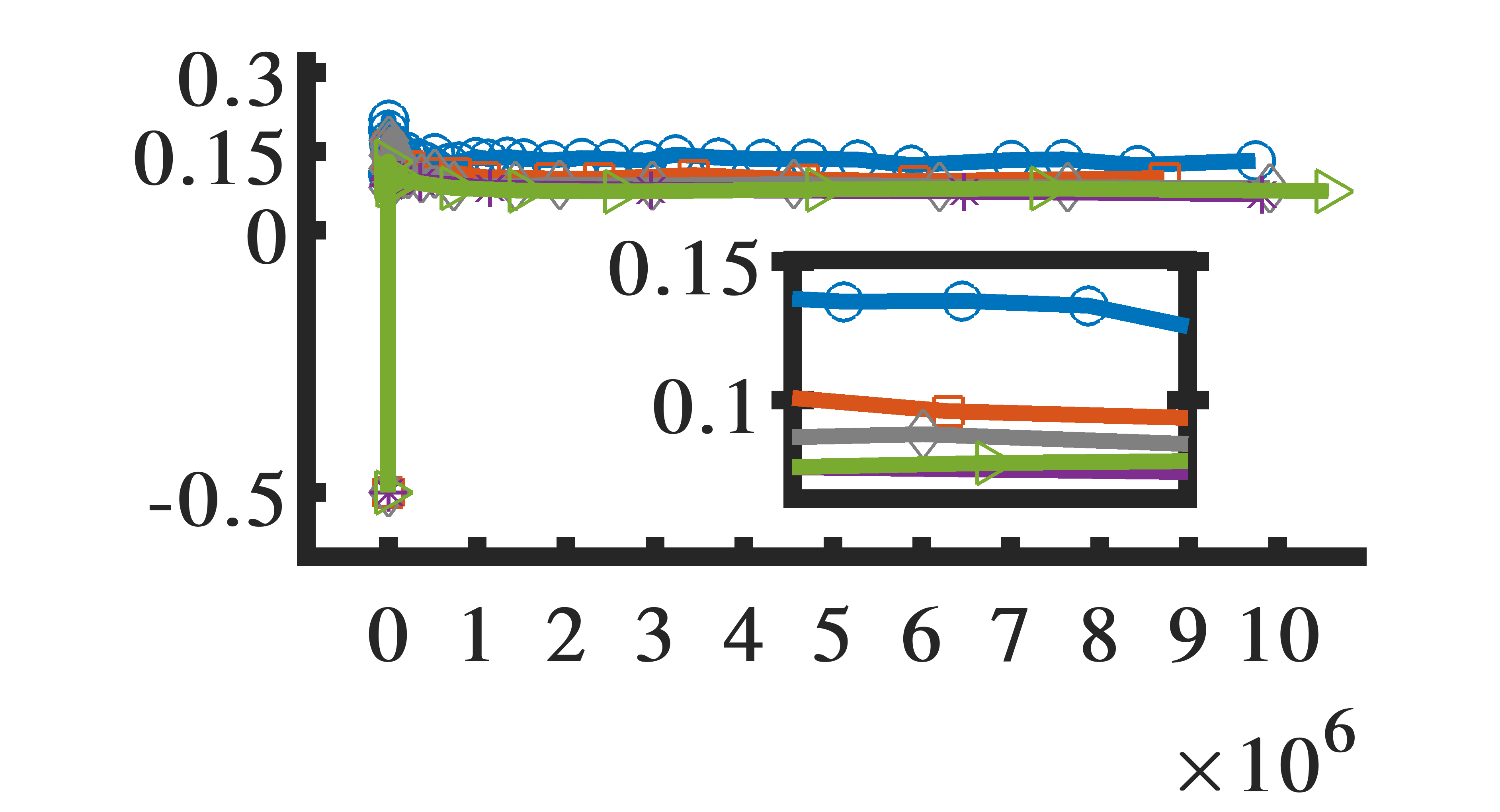}\\
            
            \midrule           
            \makecell{(c)\\\boldsymbol{$M$$=$$100-300$}\\$\rho$$=$$0.3,0.5,0.7$\\$\beta$$=$$5$} &  
           \includegraphics[width=0.28\textwidth]{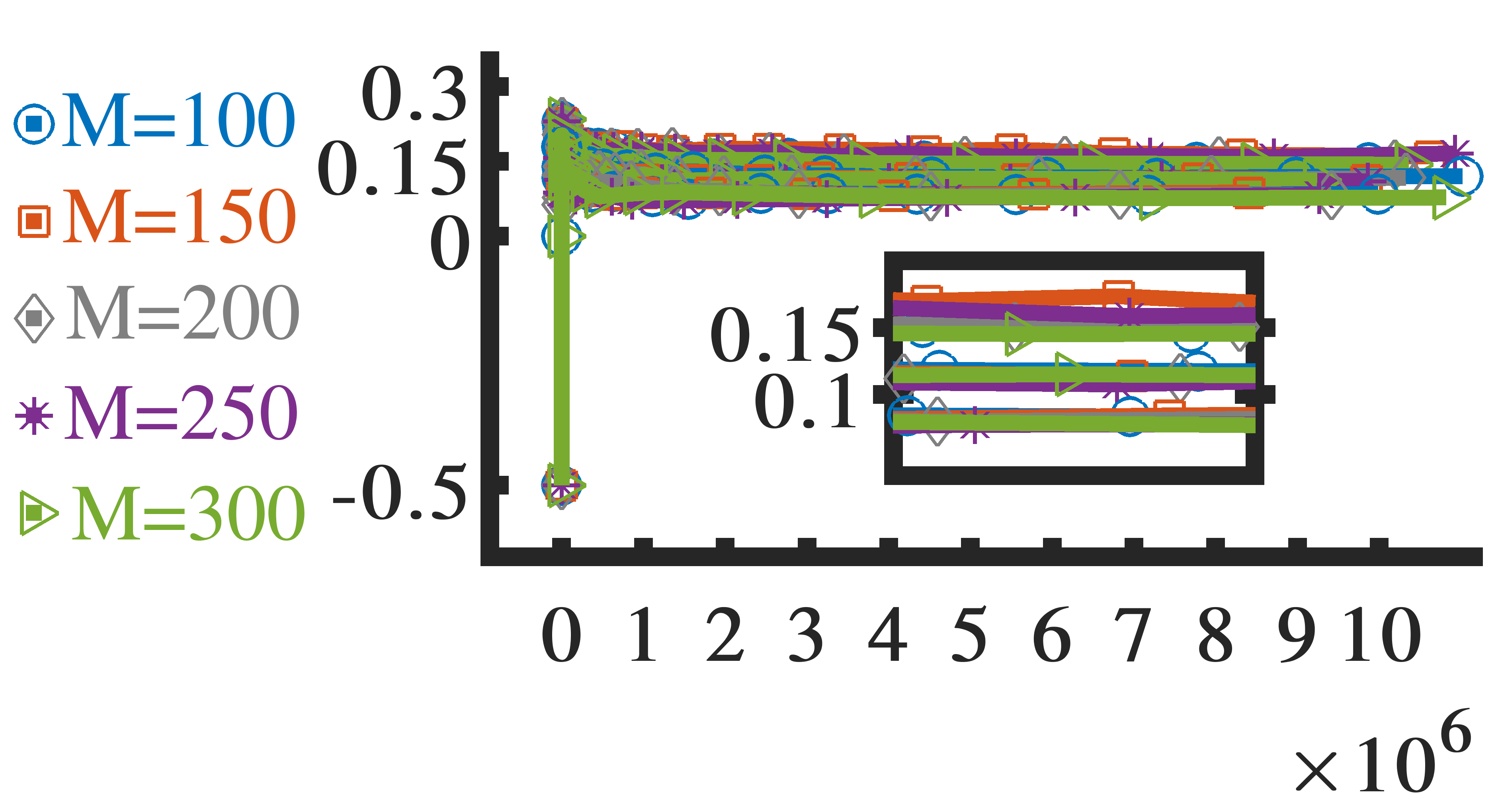} & 
            \includegraphics[width=0.288\textwidth]{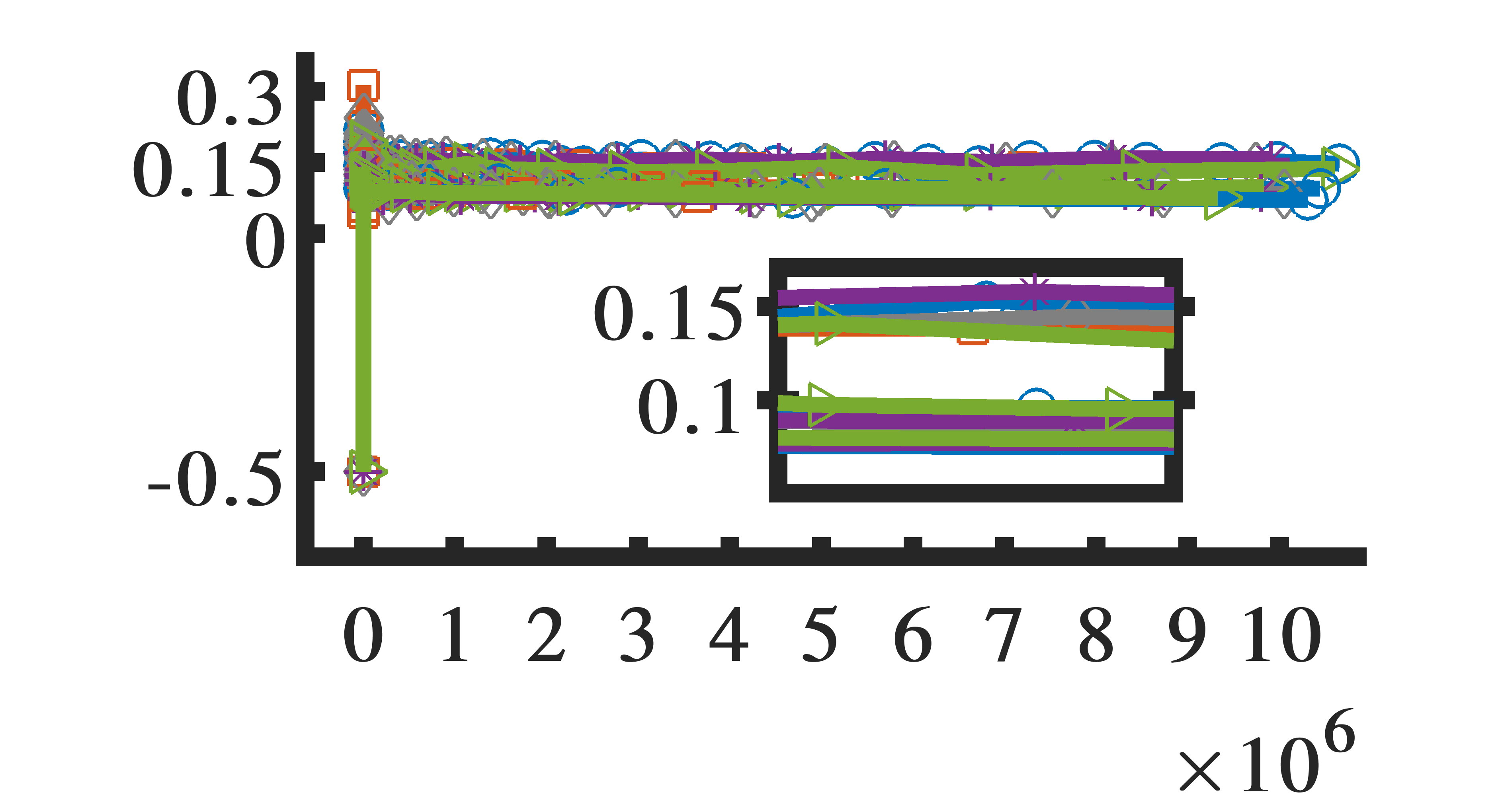} & 
            \includegraphics[width=0.288\textwidth]{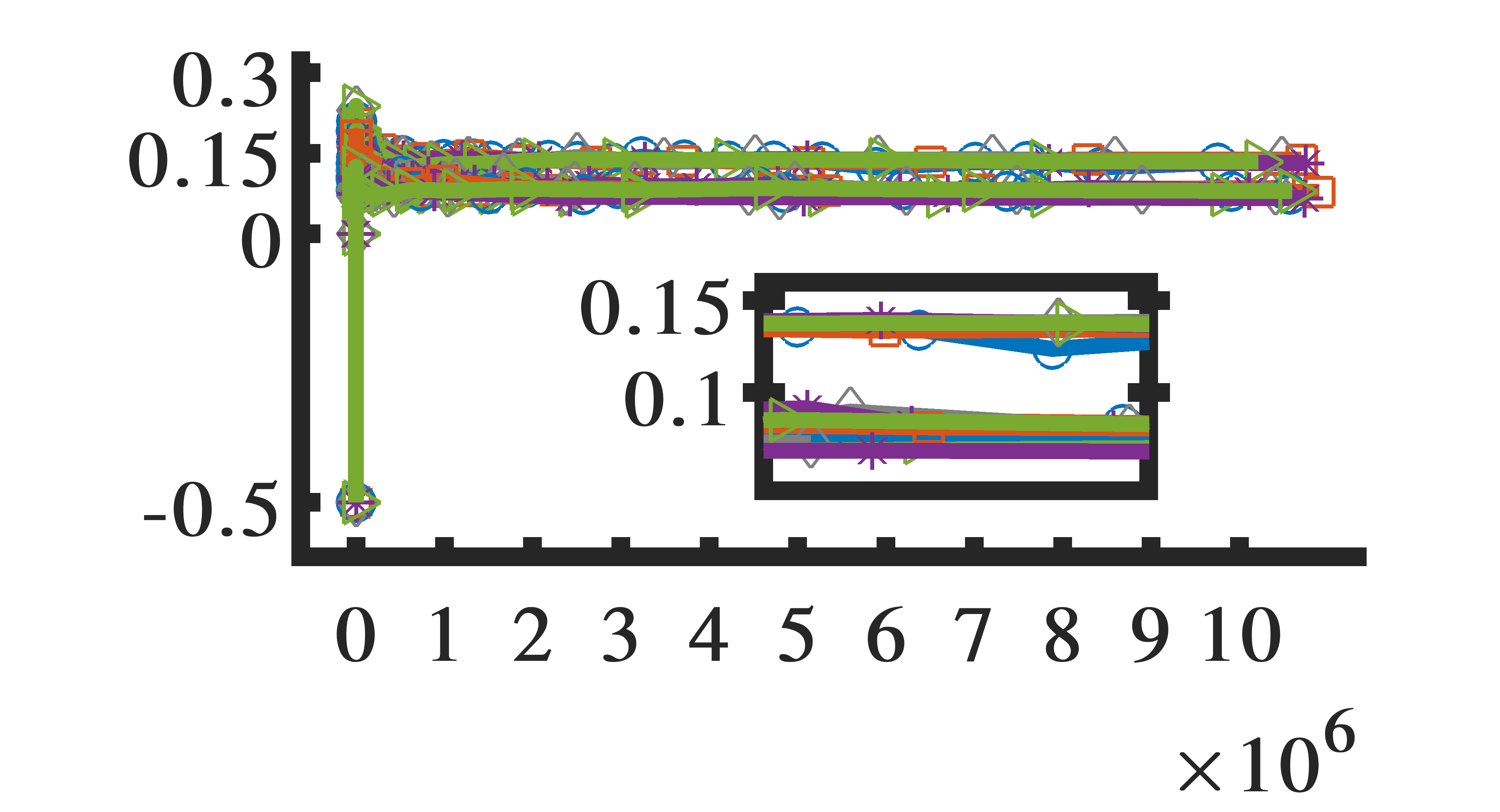}  \\            
            \midrule
            \makecell{(d)\\$M$$=$$100$\\$\rho$$=$$0.3,0.5,0.7$\\\boldsymbol{$\beta$$=$$5-20$}} &  
            \includegraphics[width=0.28\textwidth]{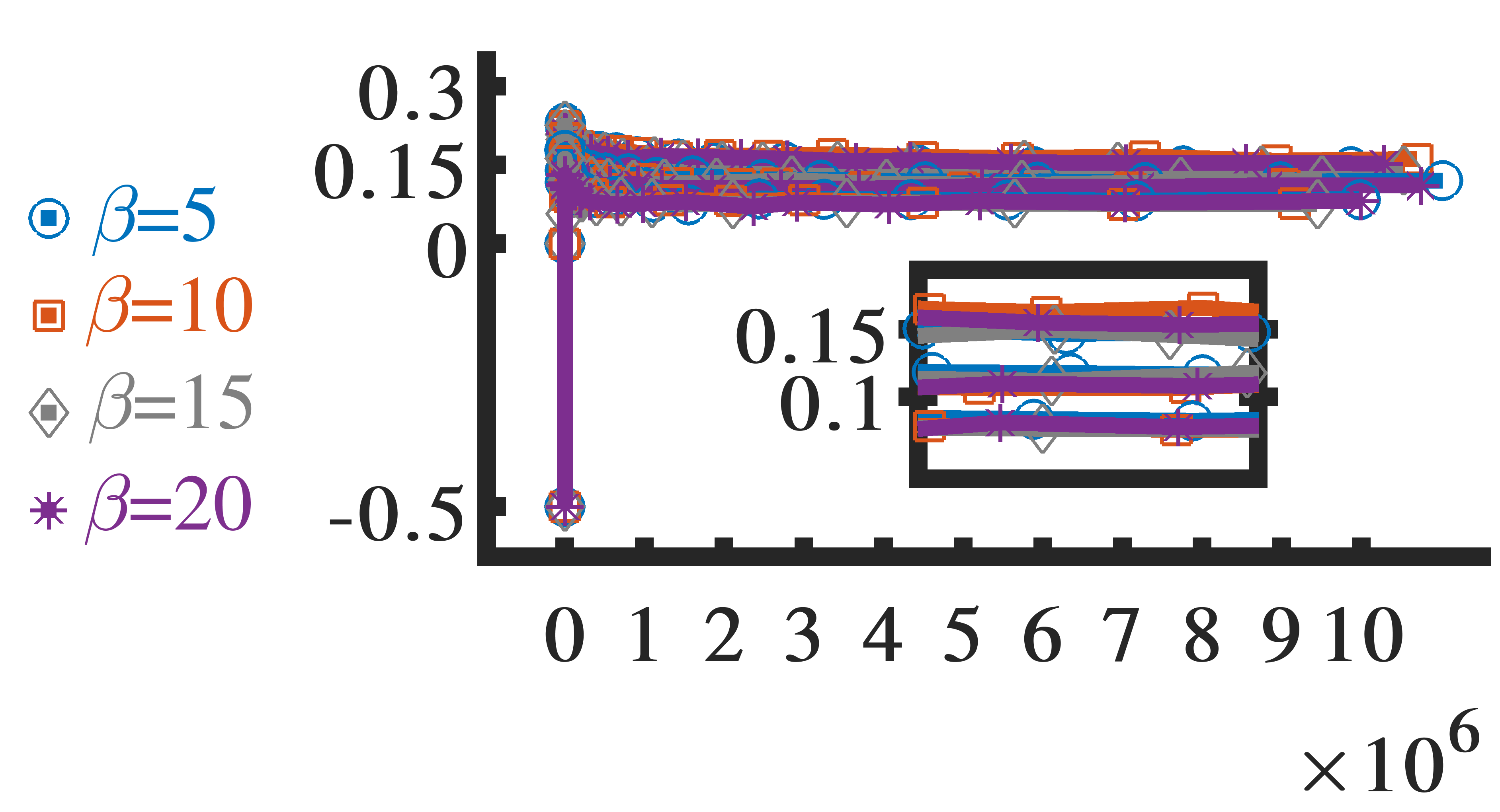} & 
            \includegraphics[width=0.288\textwidth]{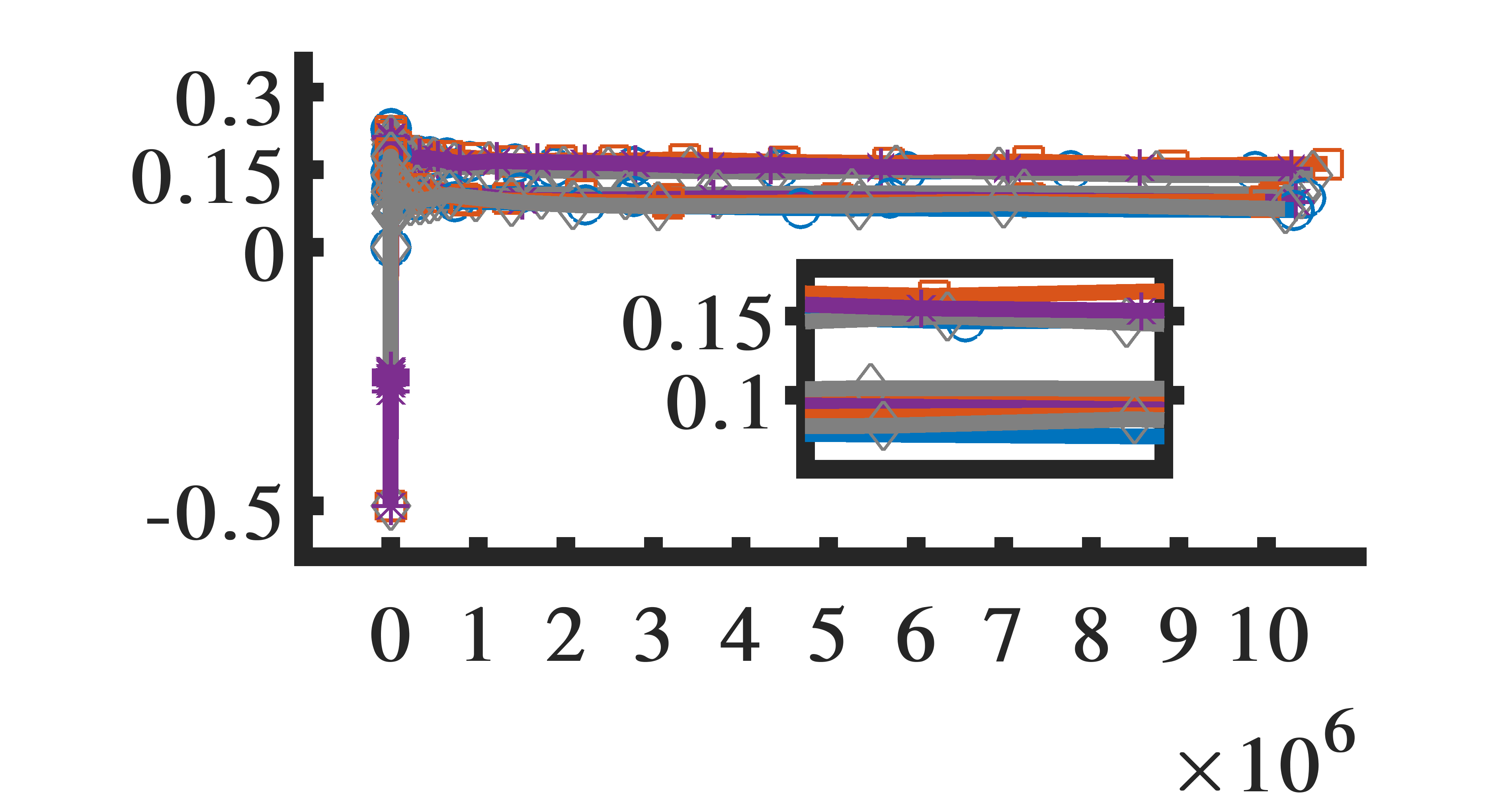} &  
            \includegraphics[width=0.288\textwidth]{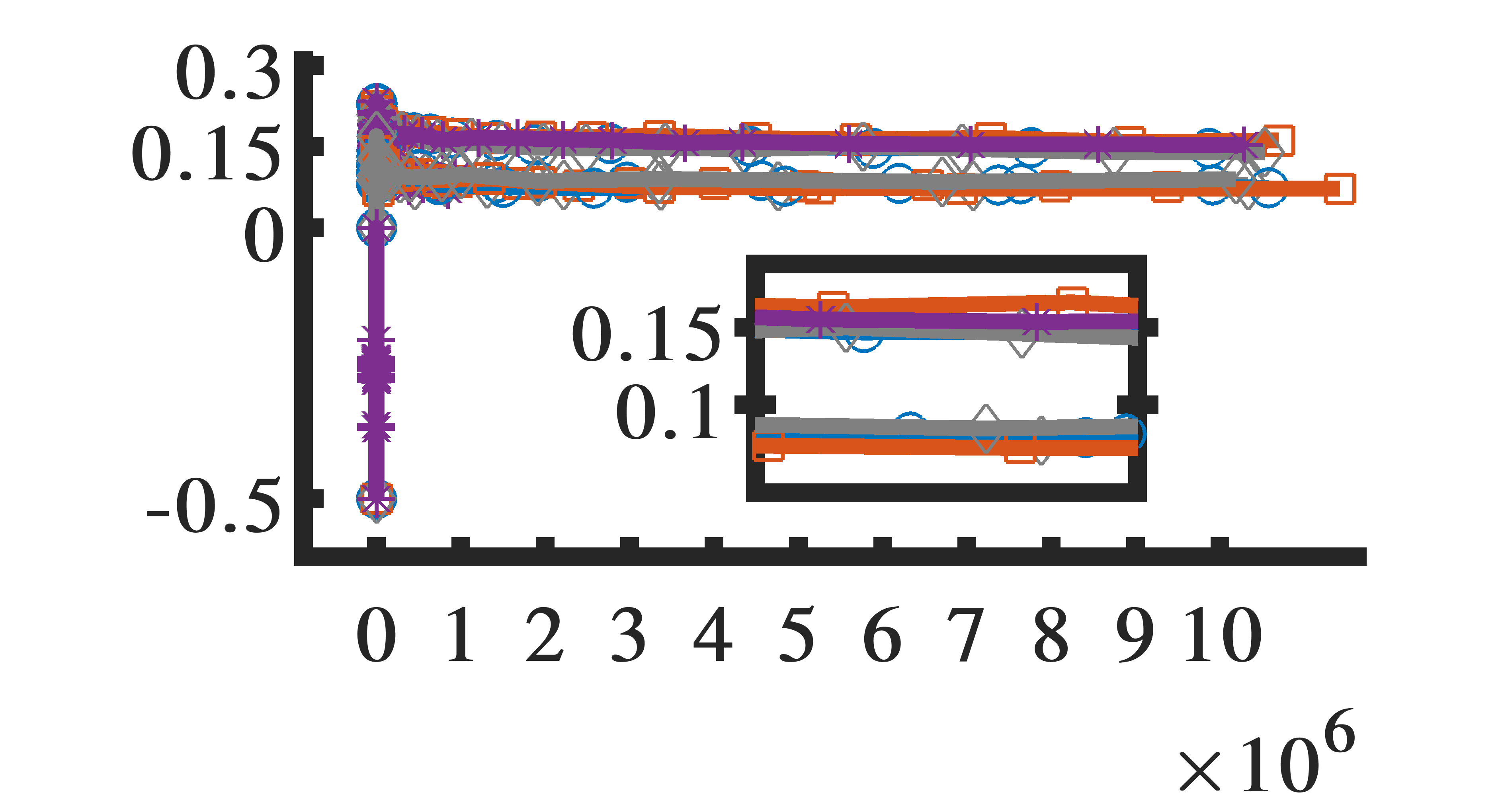}\\

        \end{tabular}
        \caption{Strength assortativity localization factor $r^s$ of sGrow (y axes) versus butterfly count $\bowtie$ (x axes) for S-Amazon with different parameter configurations.}
        \label{tbl:tableoffiguresrs}
    \end{table}
\newcolumntype{M}[1]{>{\centering\arraybackslash}m{#1}}
\begin{table}[htb!]\centering
        \begin{tabular}{cM{35mm}M{35mm}M{35mm}}
            
             & \hspace{8mm}\boldsymbol{$L_{max}$$=$$3$$-$$5$} & \hspace{8mm}\boldsymbol{$L_{min}$$=$$1$$-$$4$} &  \\       
            \toprule   
            \makecell{(a)\\$M$$=$$100$\\$\rho$$=$$0.3,0.5,0.7$\\$\beta$$=$$5$} &
            \includegraphics[width=0.288\textwidth]{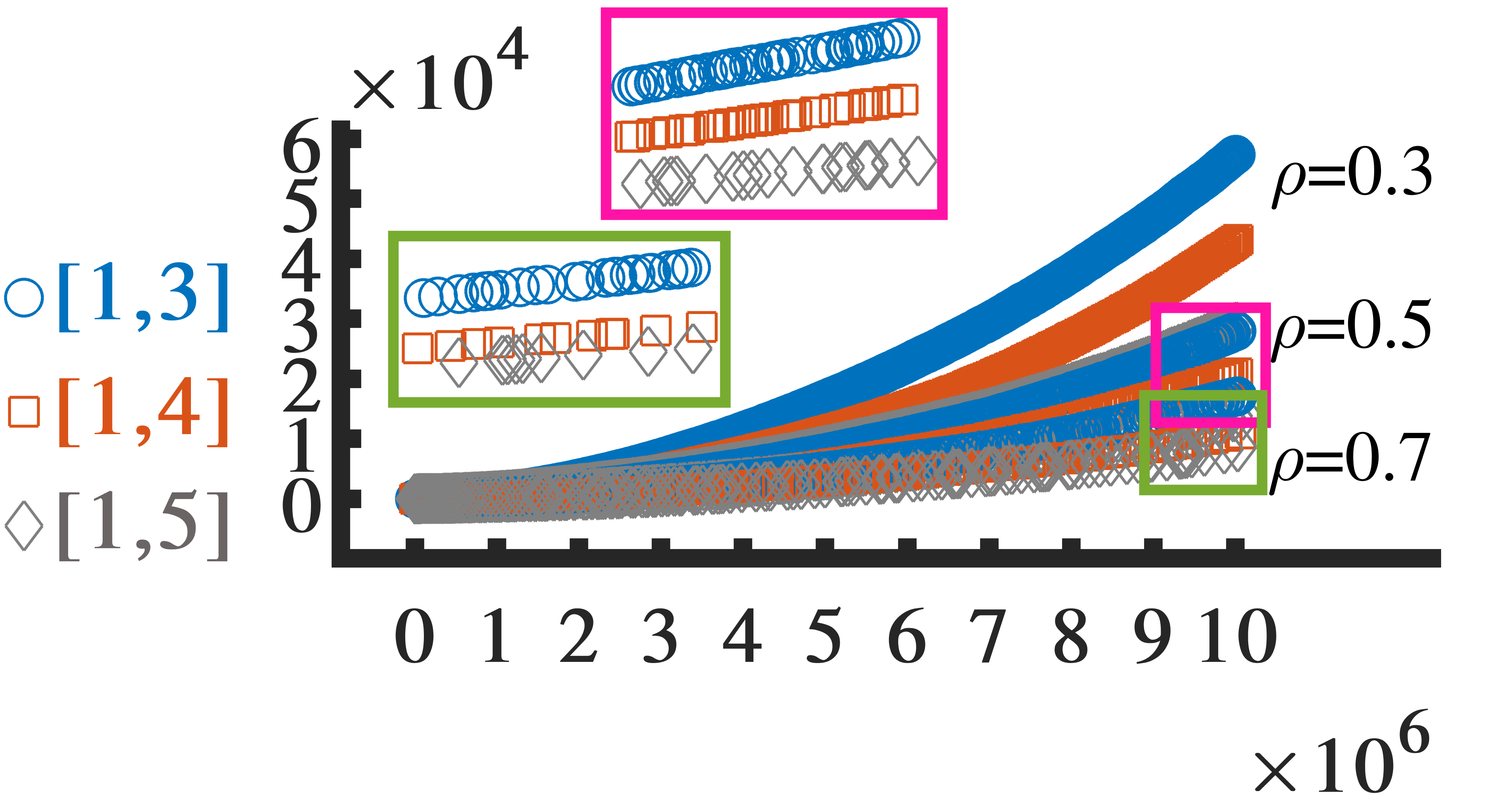} &
            \includegraphics[width=0.288\textwidth]{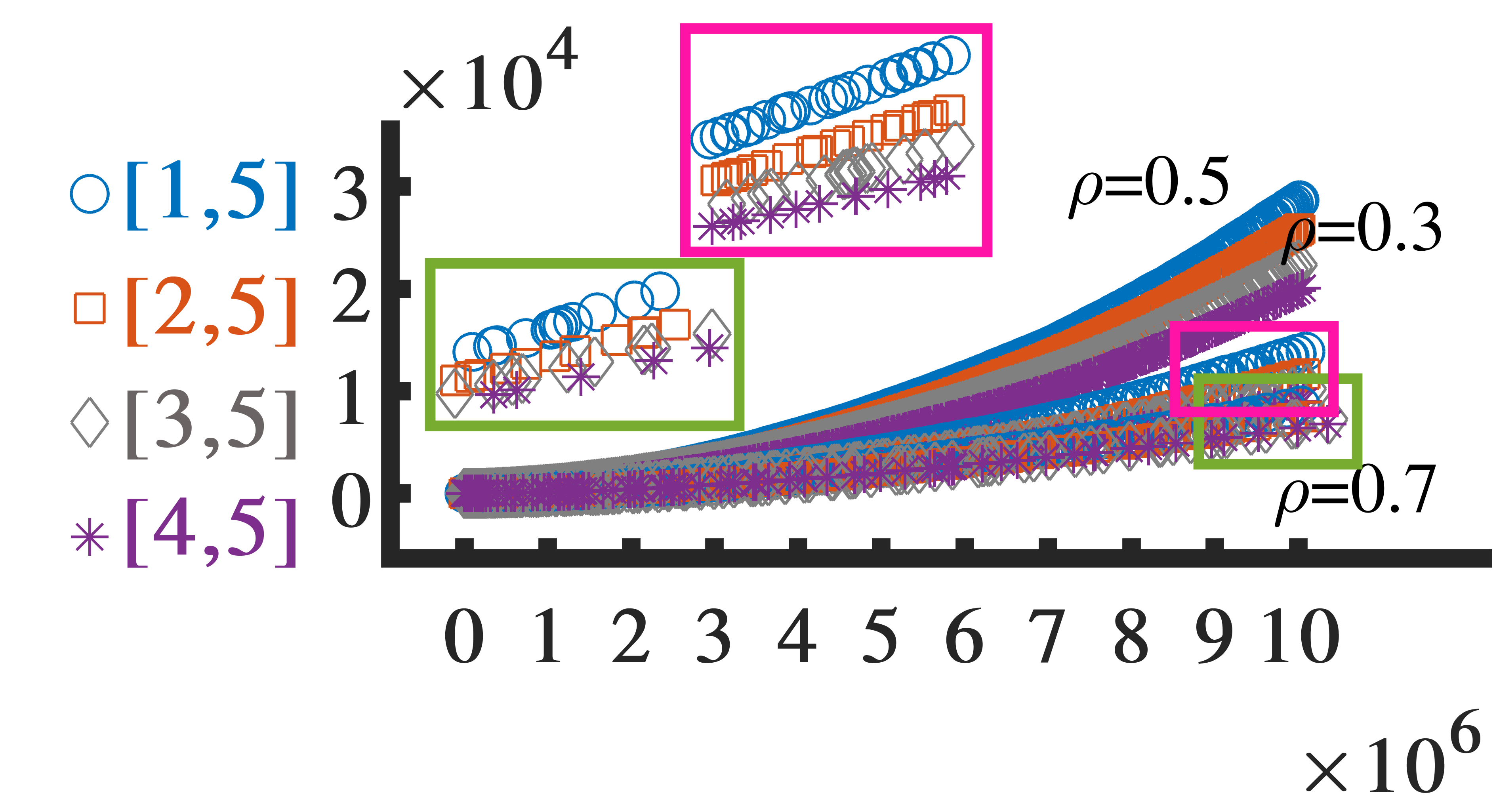} \\
            
             &\hspace{8mm} $L$$\in$$[1,3]$ & \hspace{8mm}$L$$\in$$[1,4]$ & \hspace{8mm} $L$$\in$$[1,5]$   \\
            \toprule 
            
            \makecell{(b)\\$M$$=$$100$\\\boldsymbol{$\rho$$=$$0.3-0.7$}\\$\beta$$=$$5$} & 
            \includegraphics[width=0.279\textwidth]{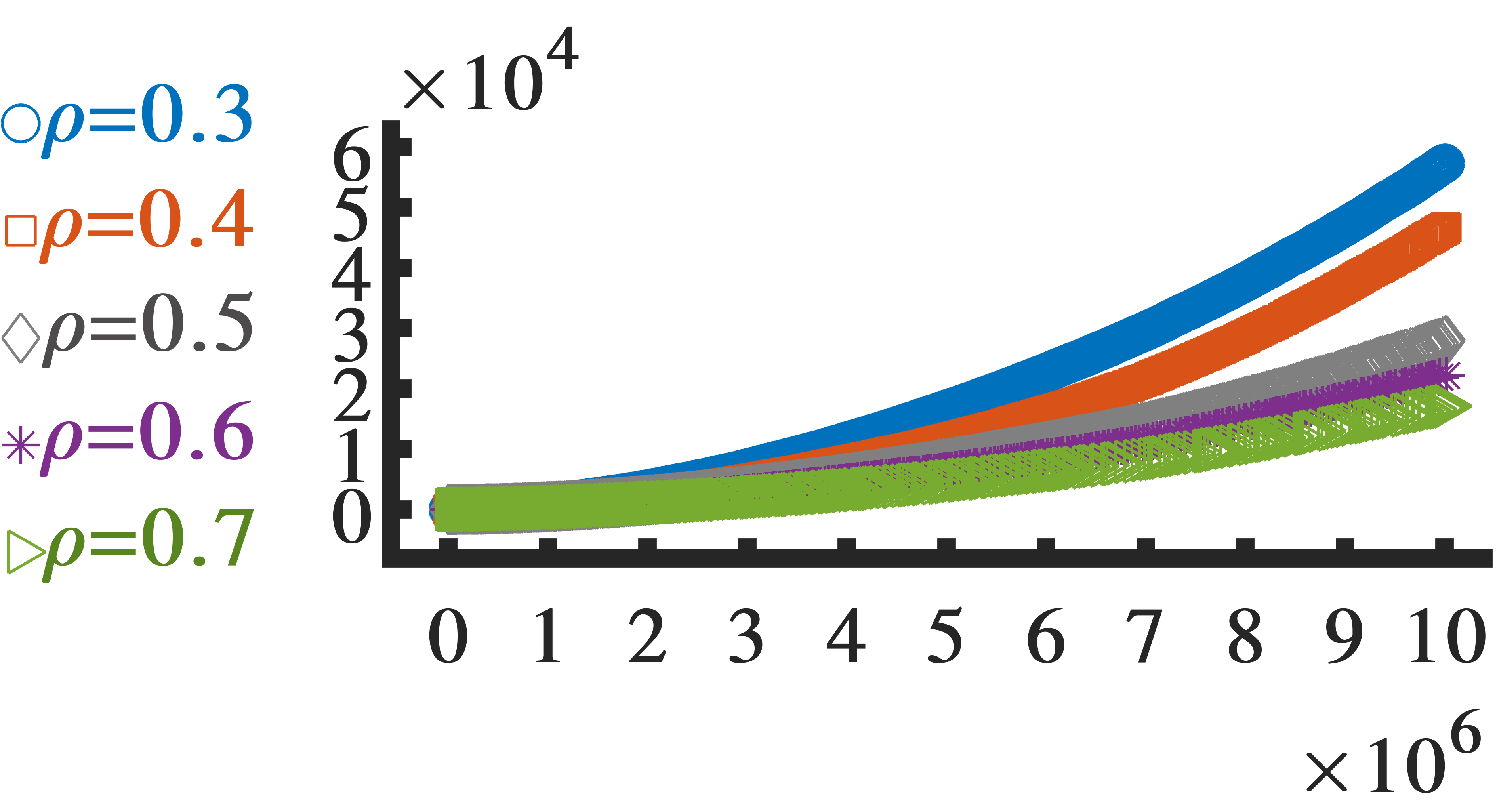} &  
            \includegraphics[width=0.288\textwidth]{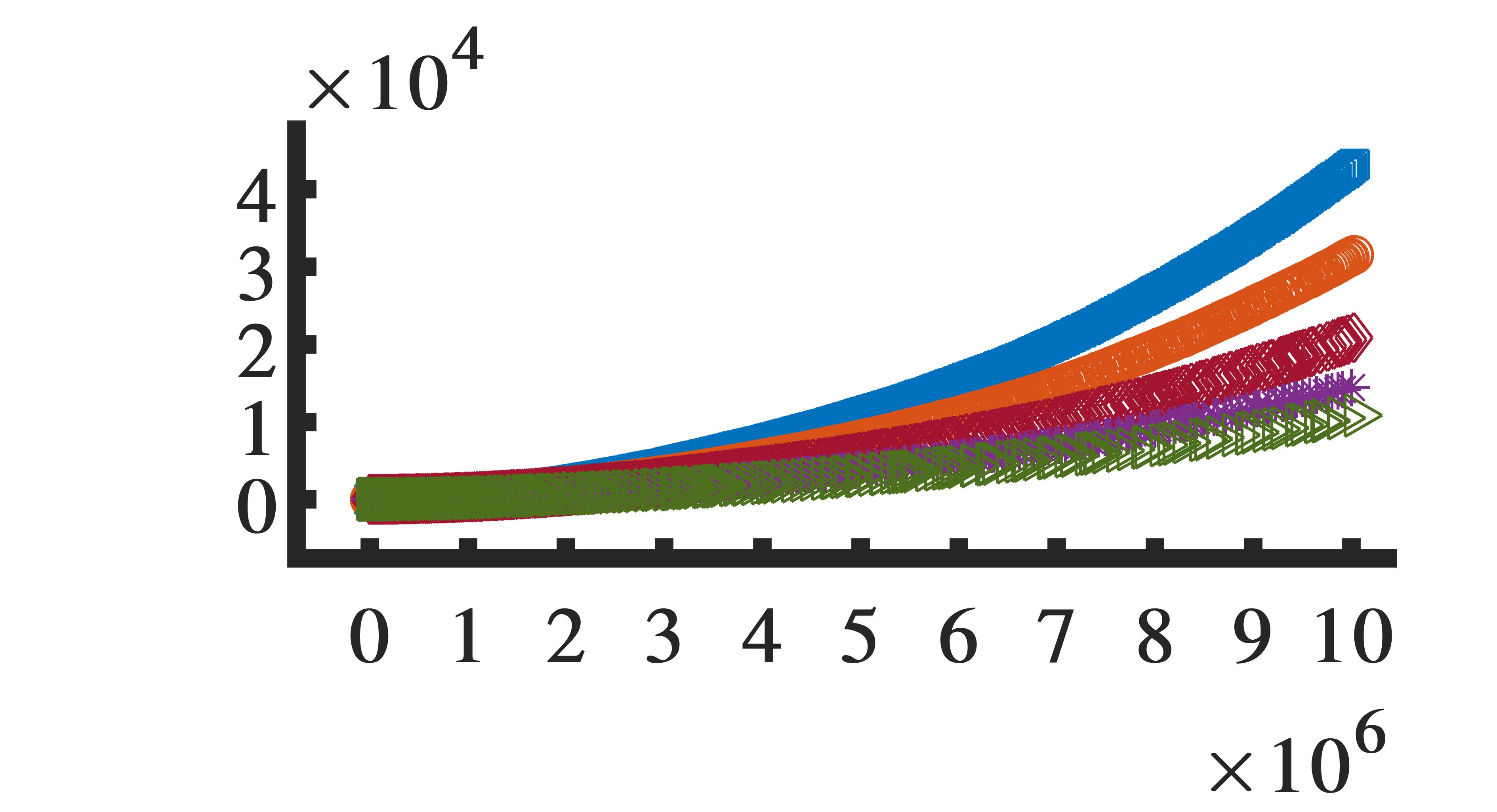}& 
            \includegraphics[width=0.288\textwidth]{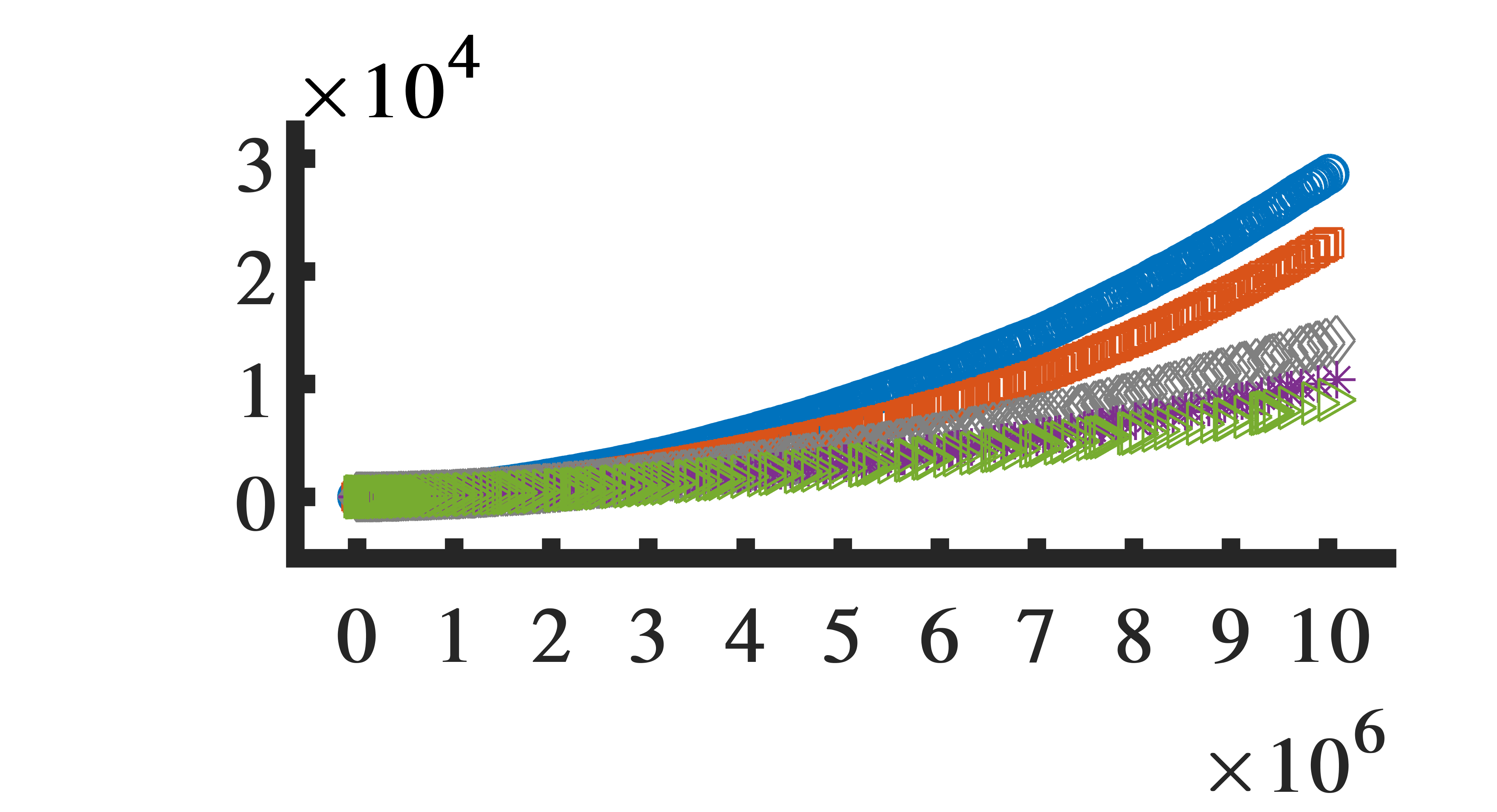}\\
            
            \midrule           
            \makecell{(c)\\\boldsymbol{$M$$=$$100-300$}\\$\rho$$=$$0.3,0.5,0.7$\\$\beta$$=$$5$} &  
           \includegraphics[width=0.28\textwidth]{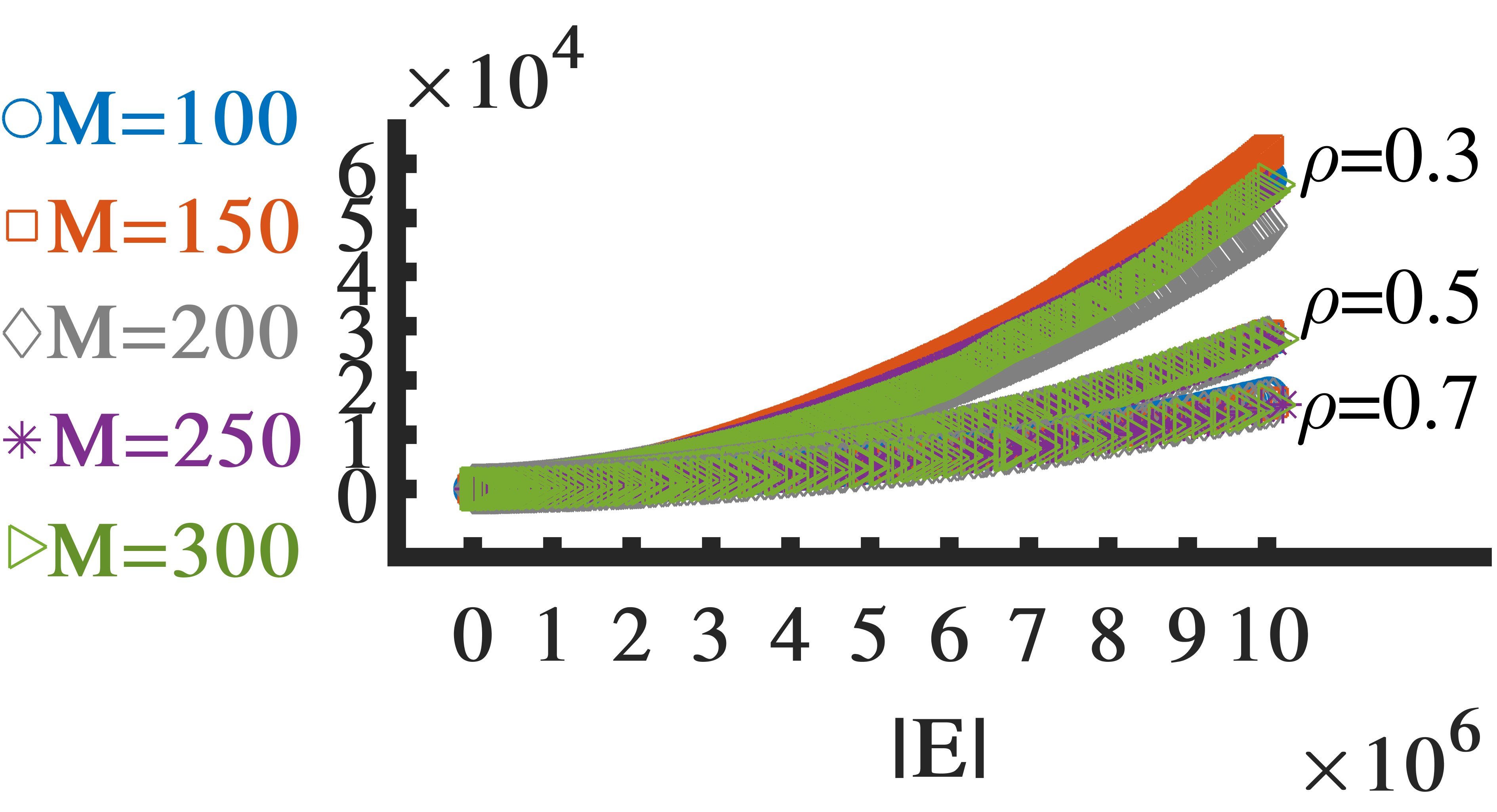} & 
            \includegraphics[width=0.288\textwidth]{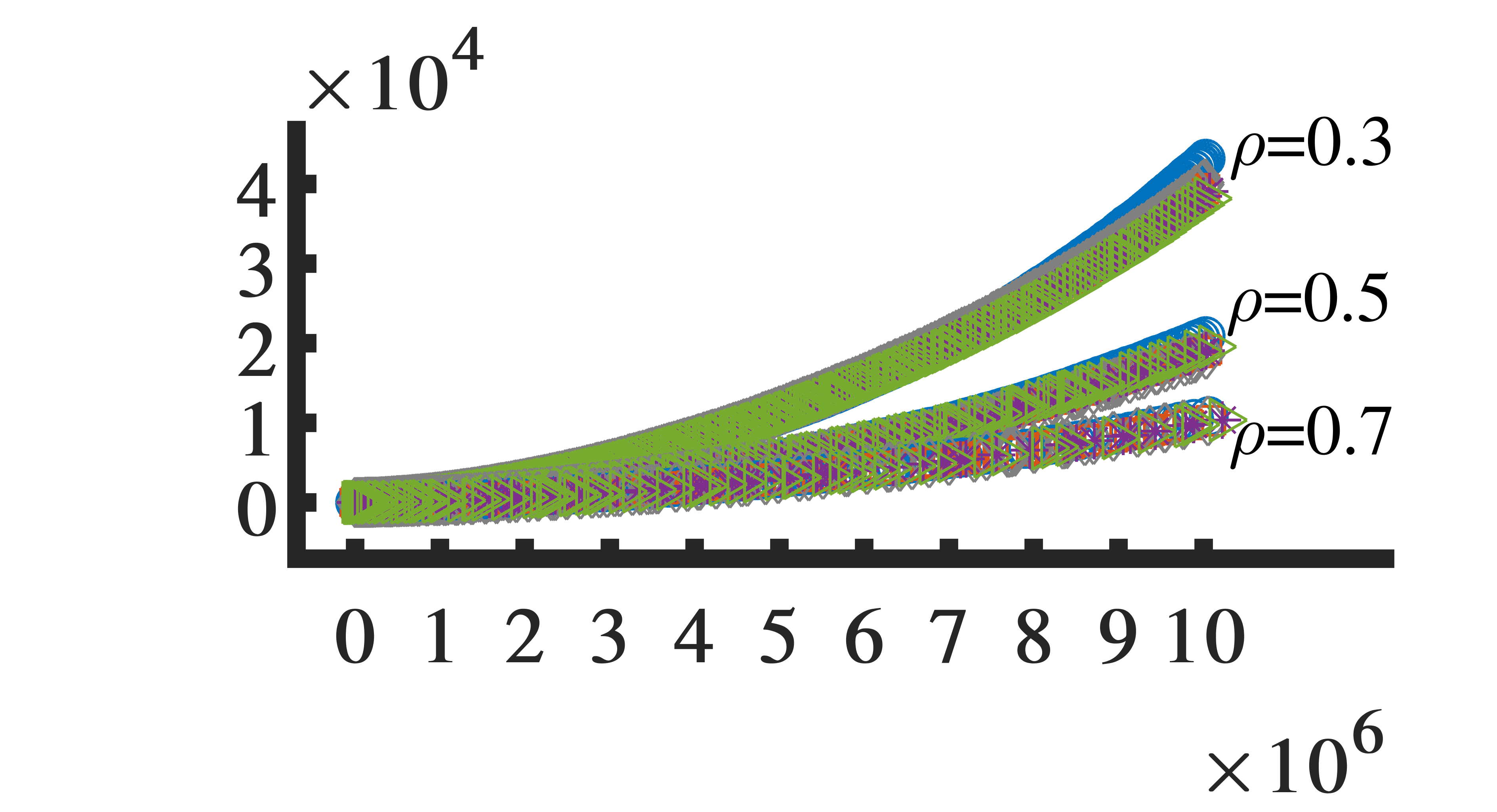} & 
            \includegraphics[width=0.288\textwidth]{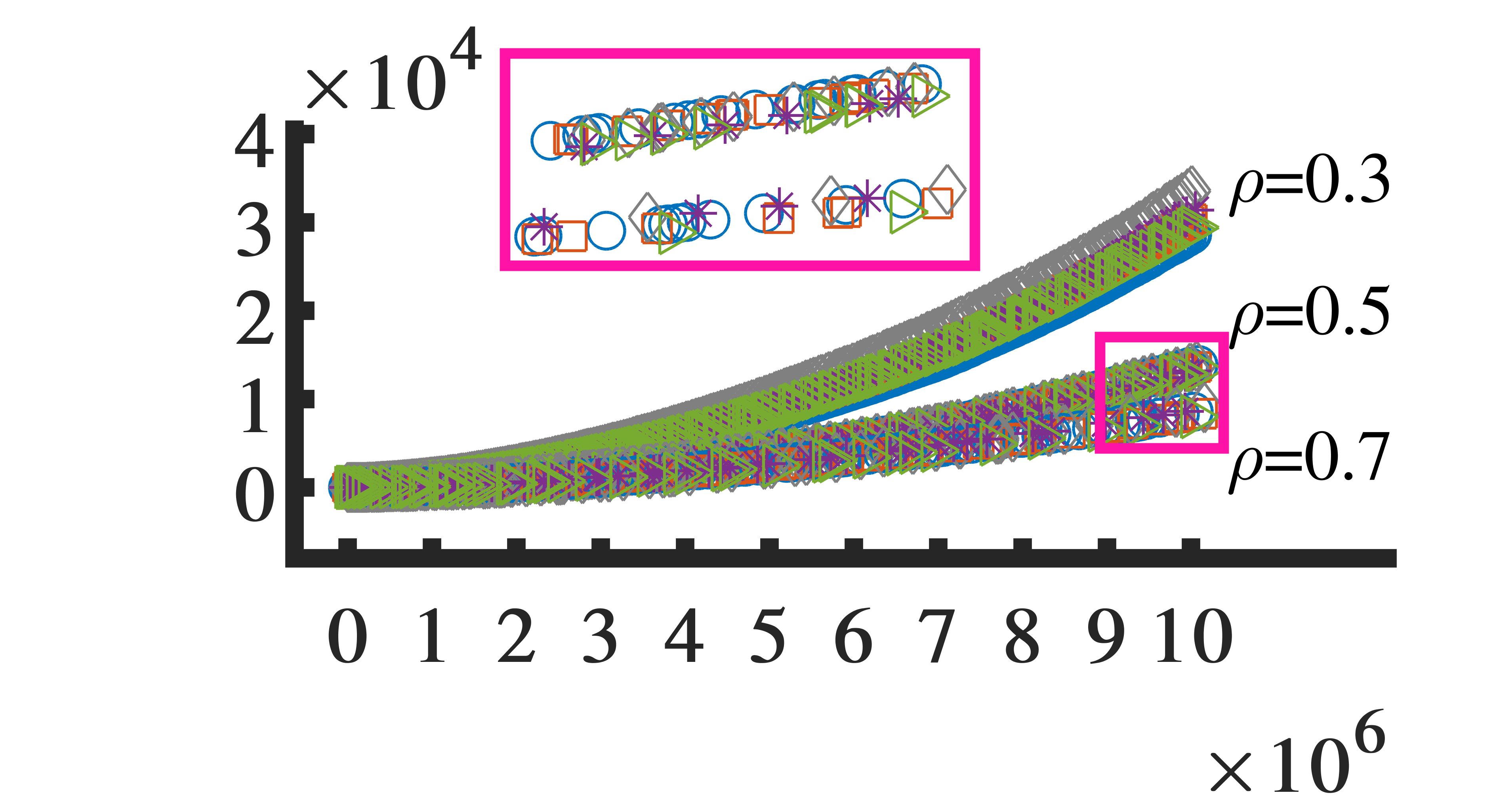}  \\            
            \midrule
            \makecell{(d)\\$M$$=$$100$\\$\rho$$=$$0.3,0.5,0.7$\\\boldsymbol{$\beta$$=$$5-20$}} &  
            \includegraphics[width=0.28\textwidth]{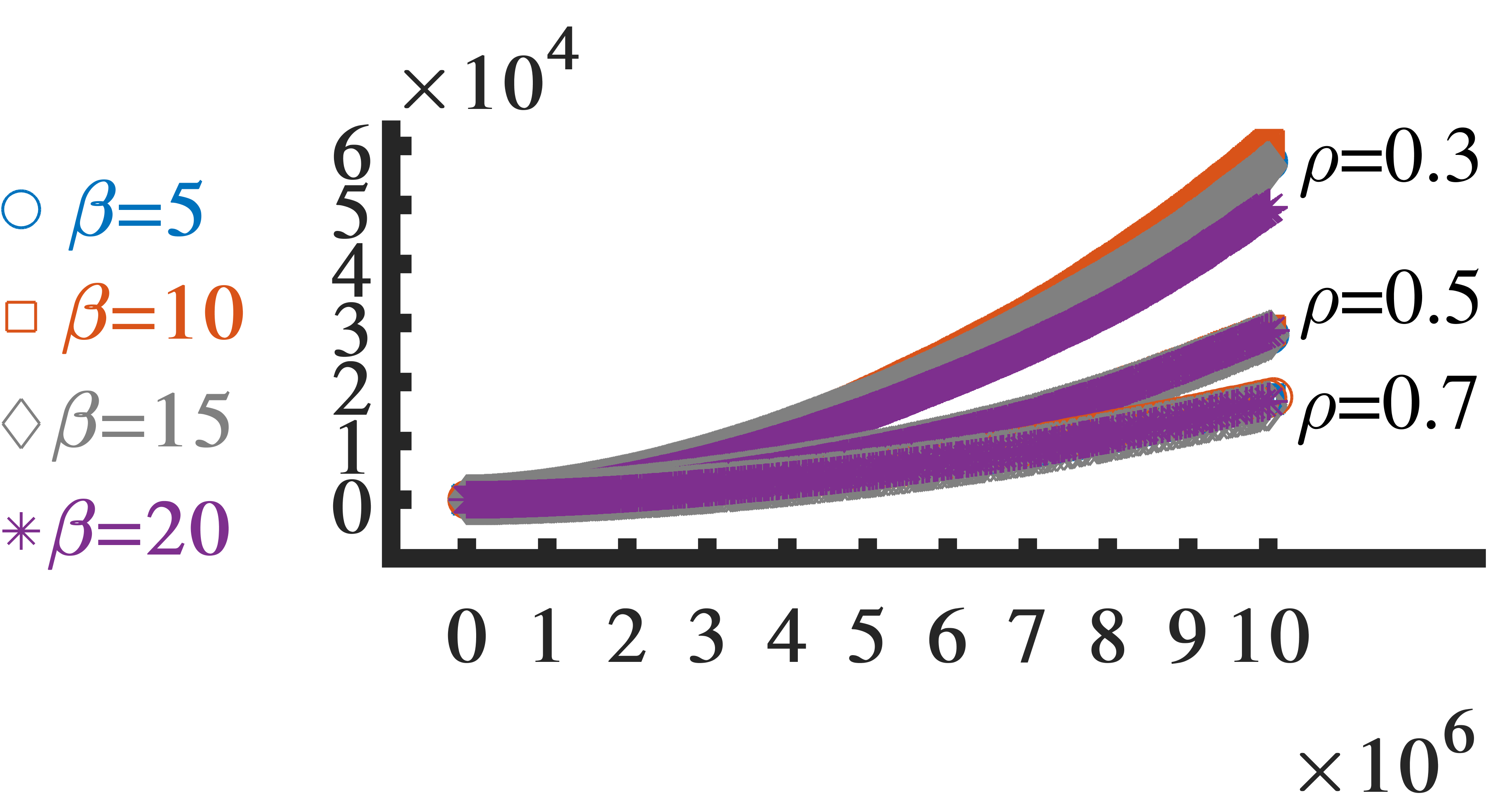} & 
            \includegraphics[width=0.288\textwidth]{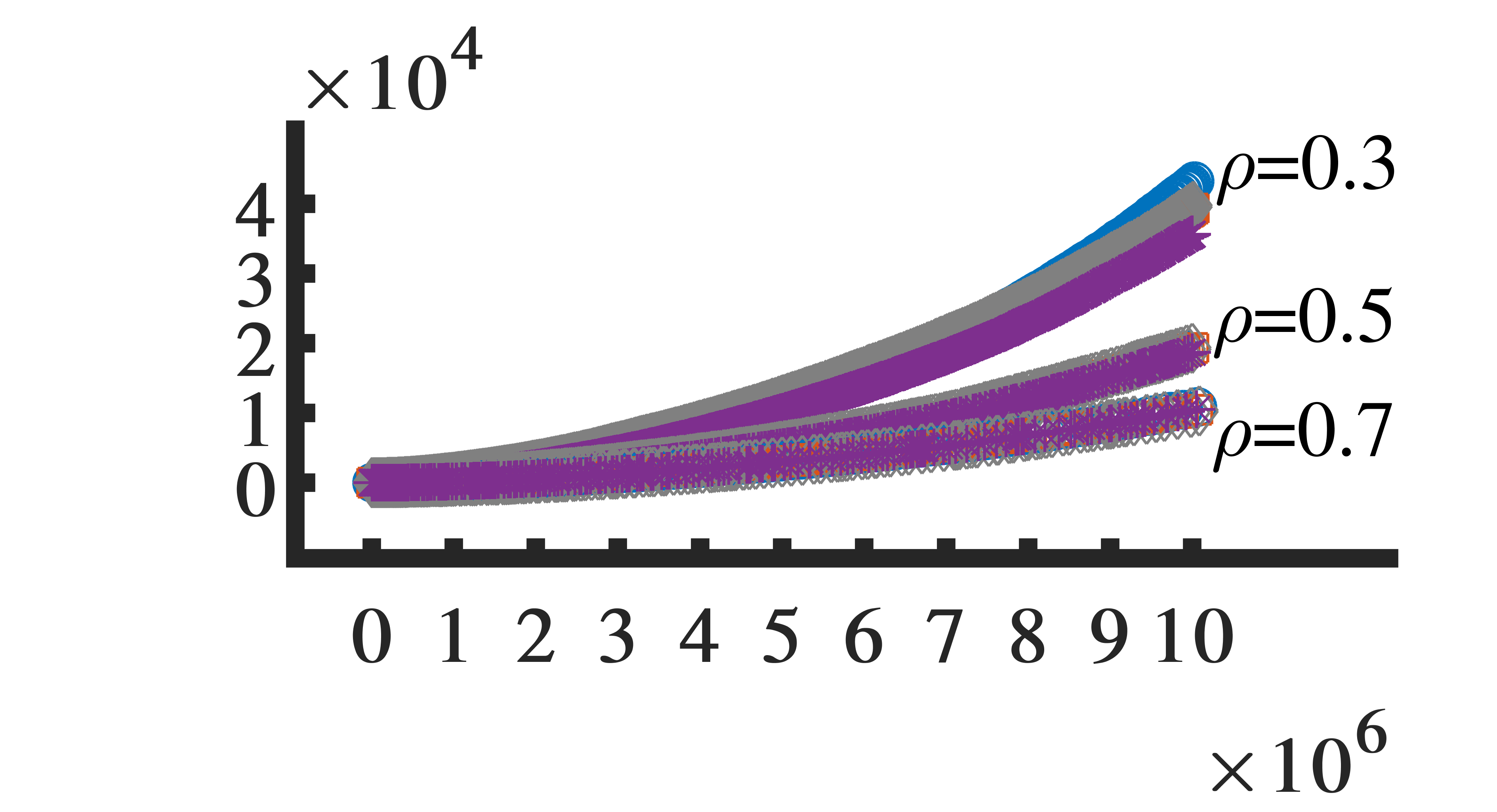} &  
            \includegraphics[width=0.288\textwidth]{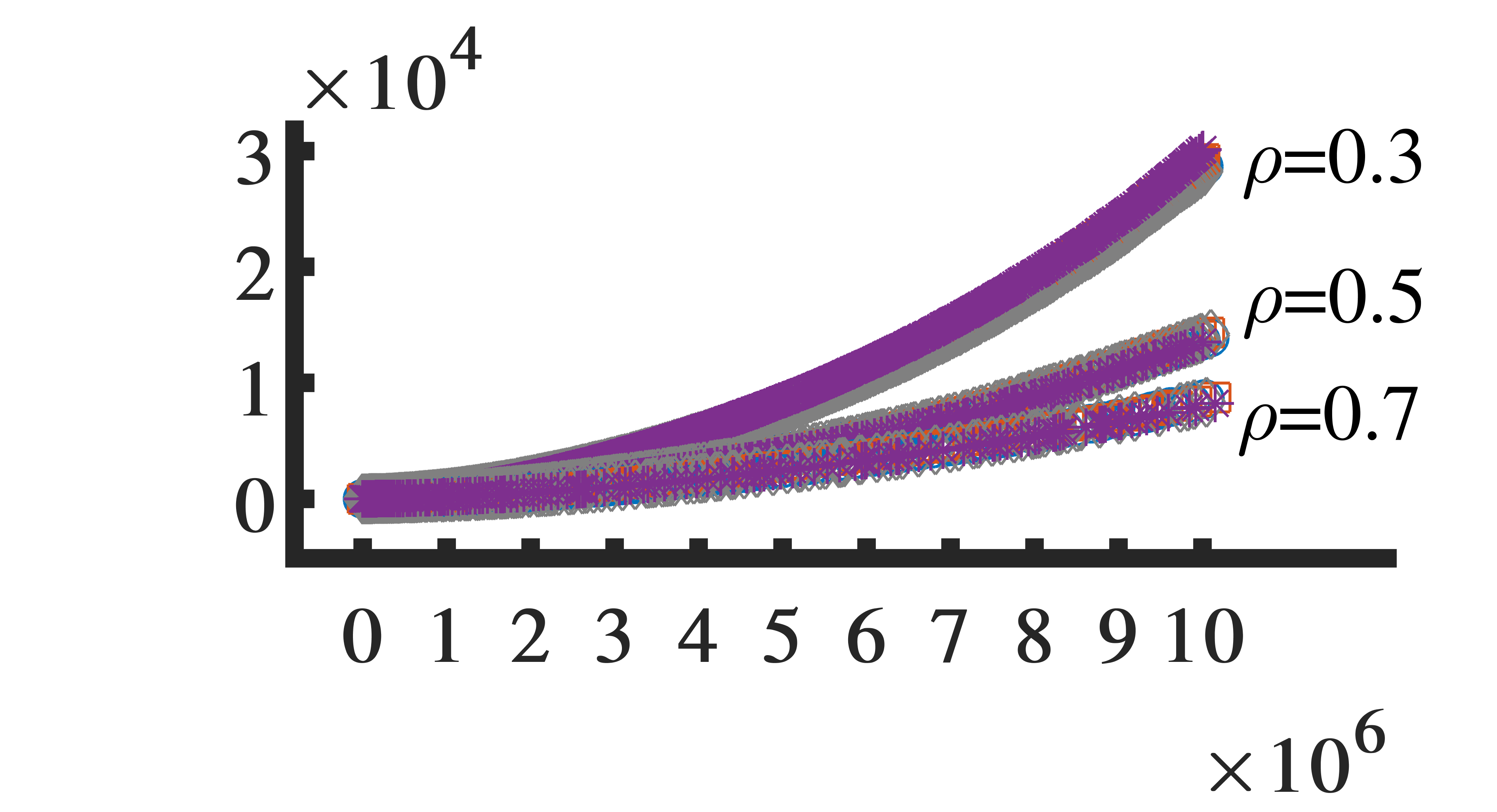}\\

        \end{tabular}
        \caption{Generation time (s) of sGrow (y axes) versus the number of sgrs $|\Re|$ (x axes) for S-Amazon with different parameter configurations.}
        \label{tbl:tableoffigures}
    \end{table}
    \begin{figure}[htb!]\centering
    \subfigure{\includegraphics[width=0.17\textwidth]{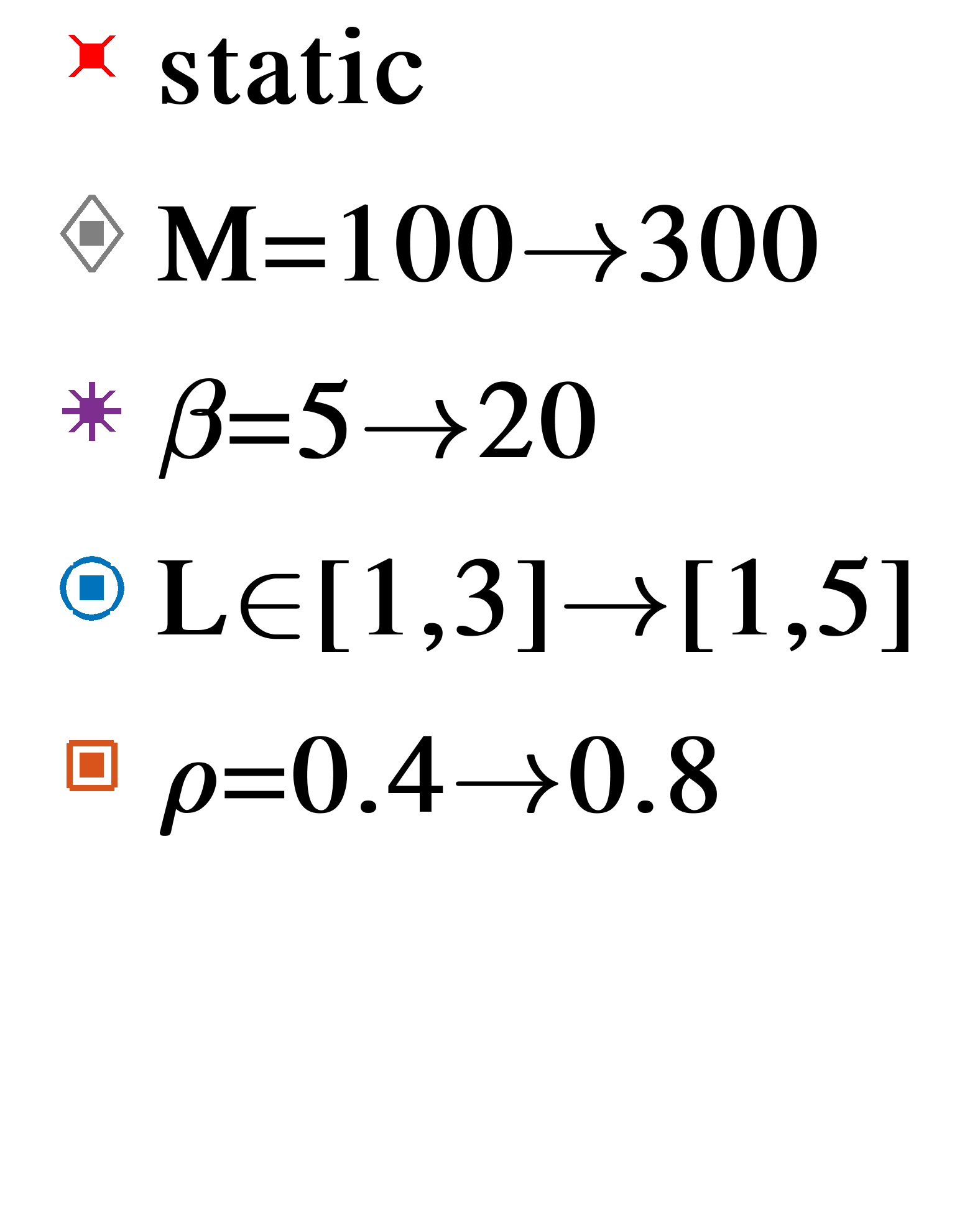}}
    \subfigure{\includegraphics[width=0.4\textwidth]{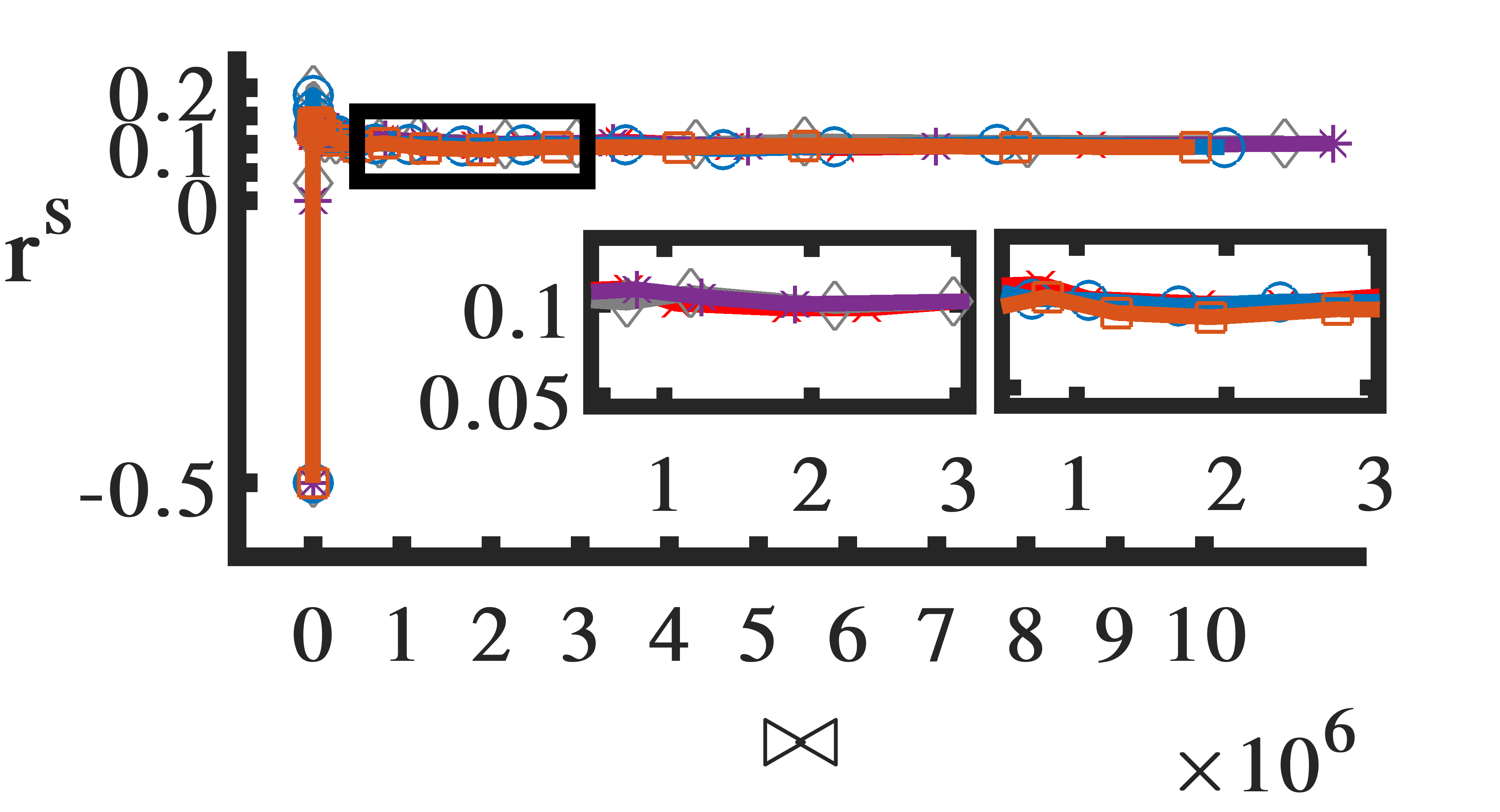}}
    \subfigure{\includegraphics[width=0.4\textwidth]{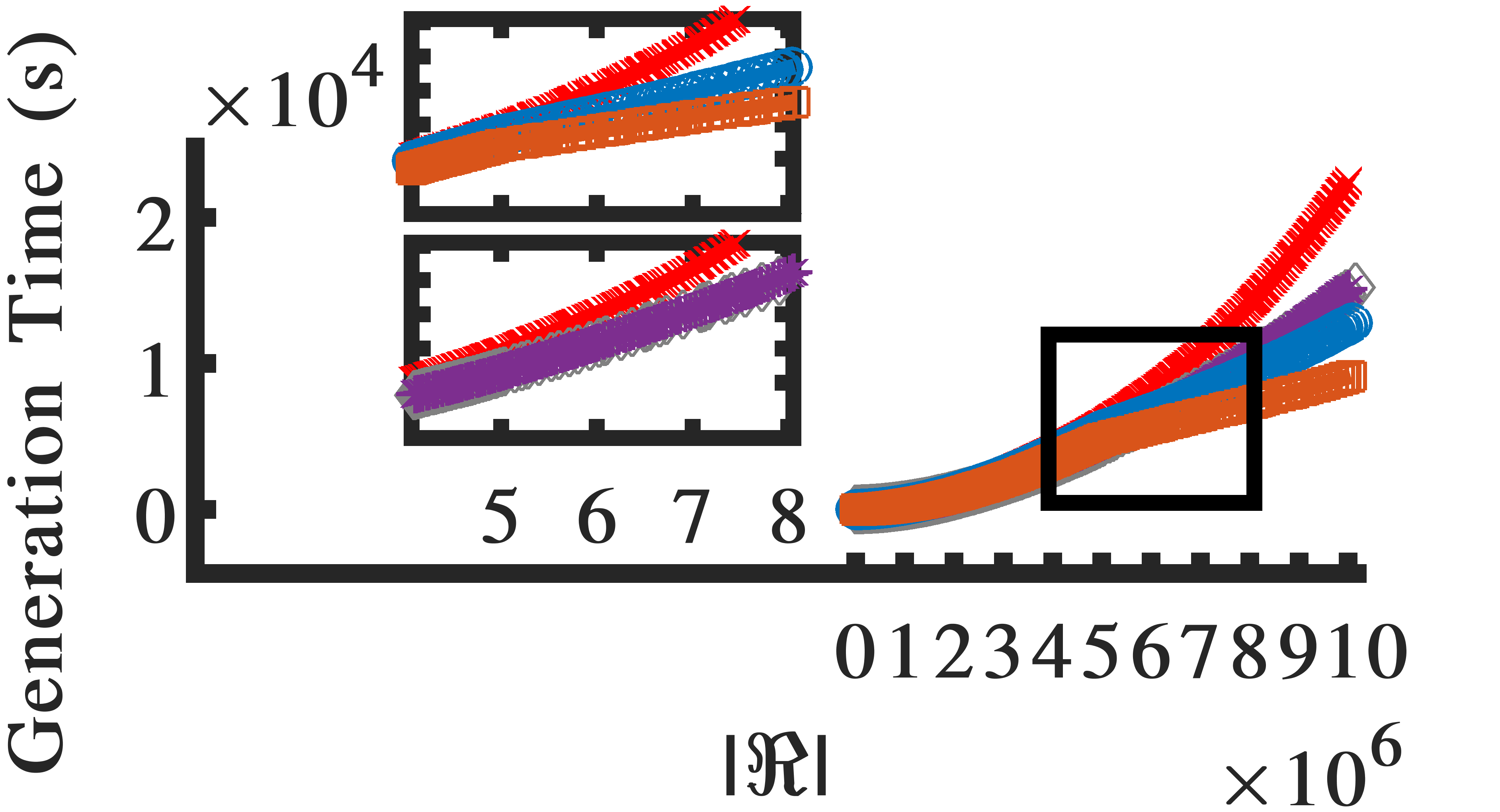}}
   \caption{ Strength assortativity localization factor $r^s$ versus butterfly count $\bowtie$ and  generation time (s) versus the number of sgrs $|\Re|$ in S-Amazon with parameter switch from $M=100$, $\beta=5$, $L\in[1,5]$ and $\rho=0.4$ to $M=300$, $\beta=20$, $L\in[4,5]$ and $\rho=0.8$ after generation of $5\times10^6$ sgrs.}\label{fig:dynamic}
\end{figure}

\newcolumntype{M}[1]{>{\centering\arraybackslash}m{#1}}
\begin{table}[htb!]\centering
        \begin{tabular}{cM{35mm}M{35mm}M{35mm}}

             &\hspace{8mm} $L$$\in$$[1,3]$ & \hspace{8mm}$L$$\in$$[1,4]$ & \hspace{8mm} $L$$\in$$[1,5]$   \\
            \toprule 
            
            \makecell{(b)\\$\rho$$=$$0.3$\\$\beta$$=$$5$} & 
            \includegraphics[width=0.28\textwidth]{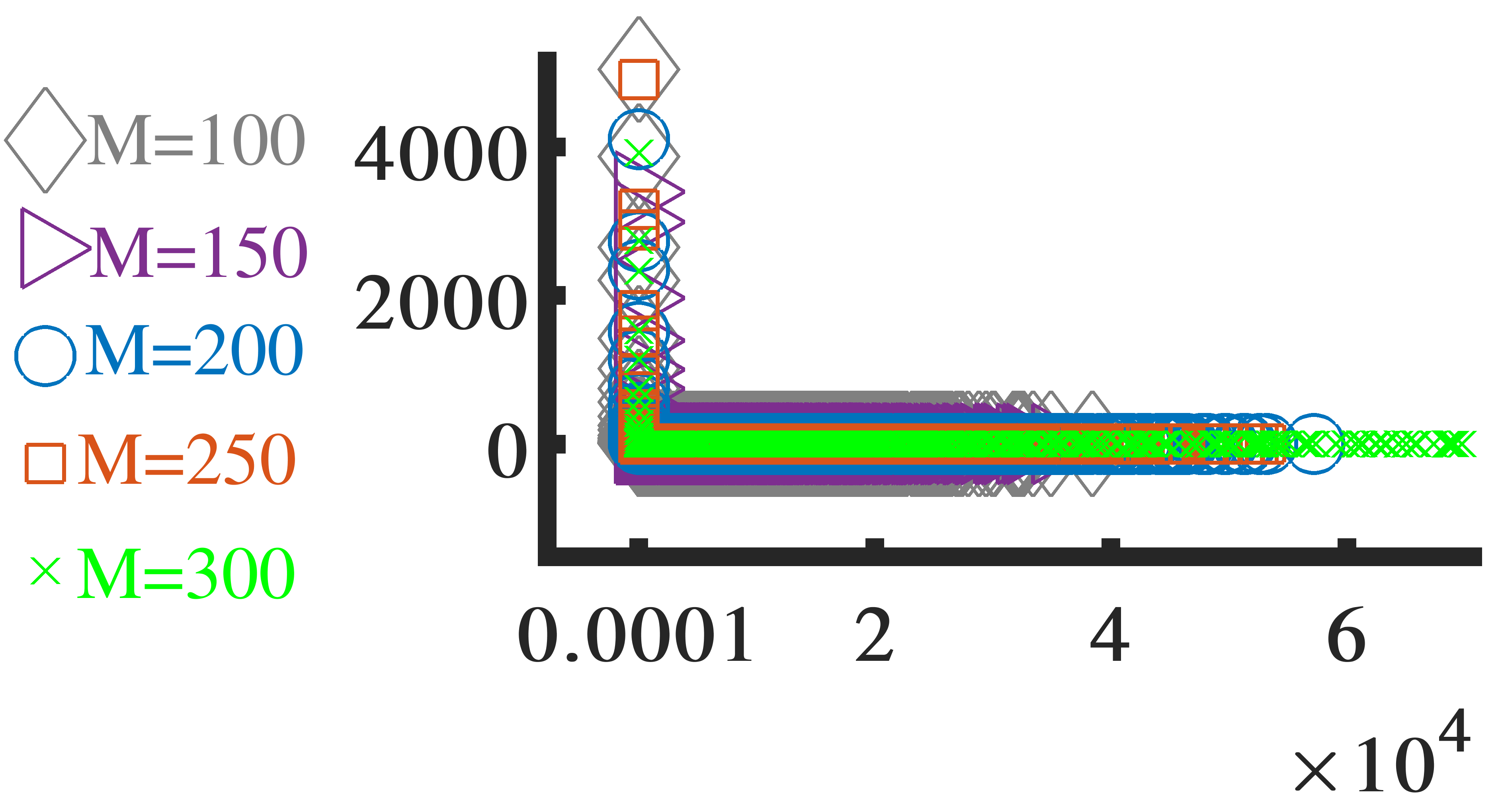} &  
            \includegraphics[width=0.288\textwidth]{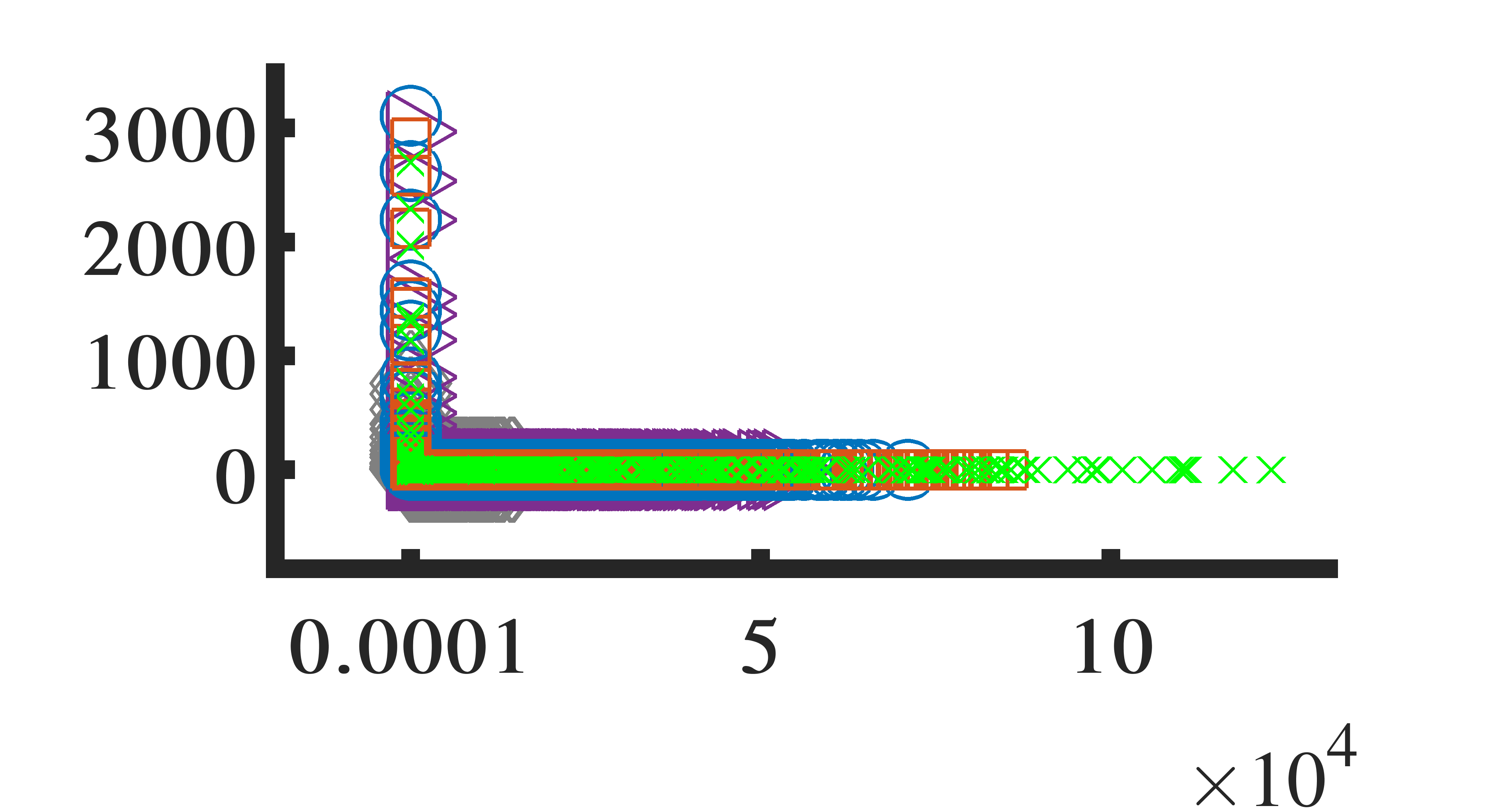}& 
            \includegraphics[width=0.288\textwidth]{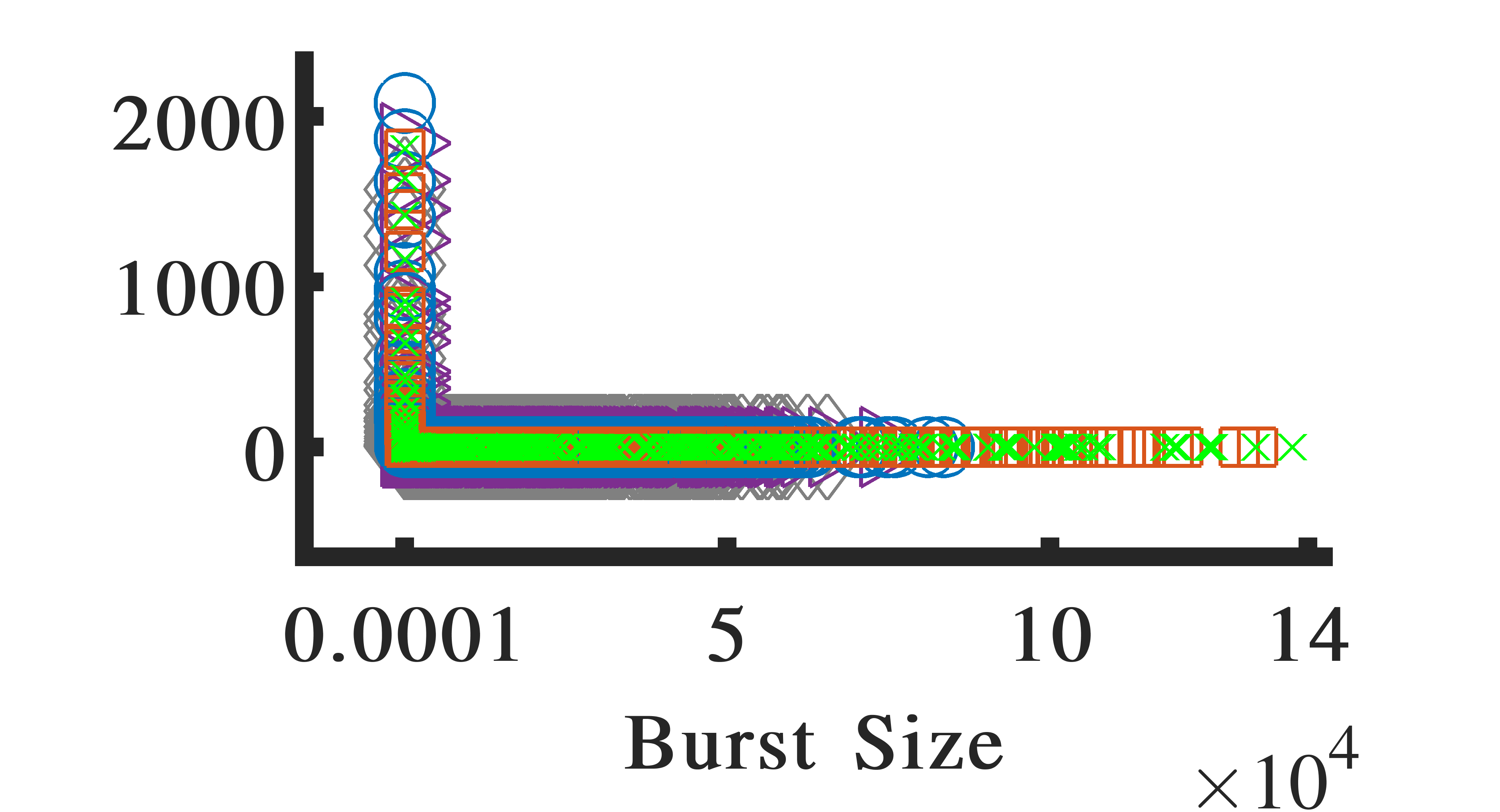}\\
            
            \midrule           
            \makecell{(c)\\$\rho$$=$$0.5$\\$\beta$$=$$5$} &  
           \includegraphics[width=0.28\textwidth]{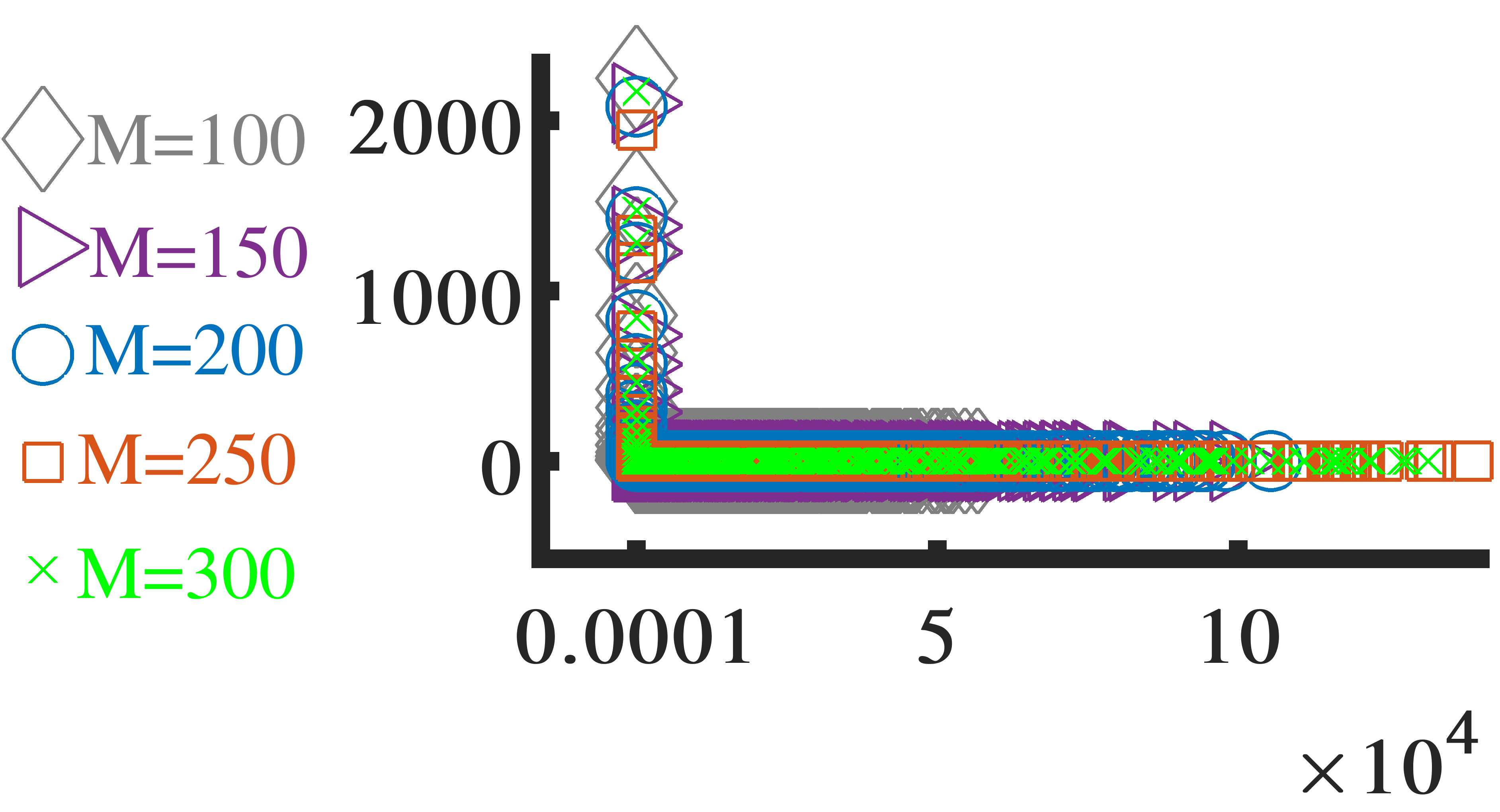} & 
            \includegraphics[width=0.288\textwidth]{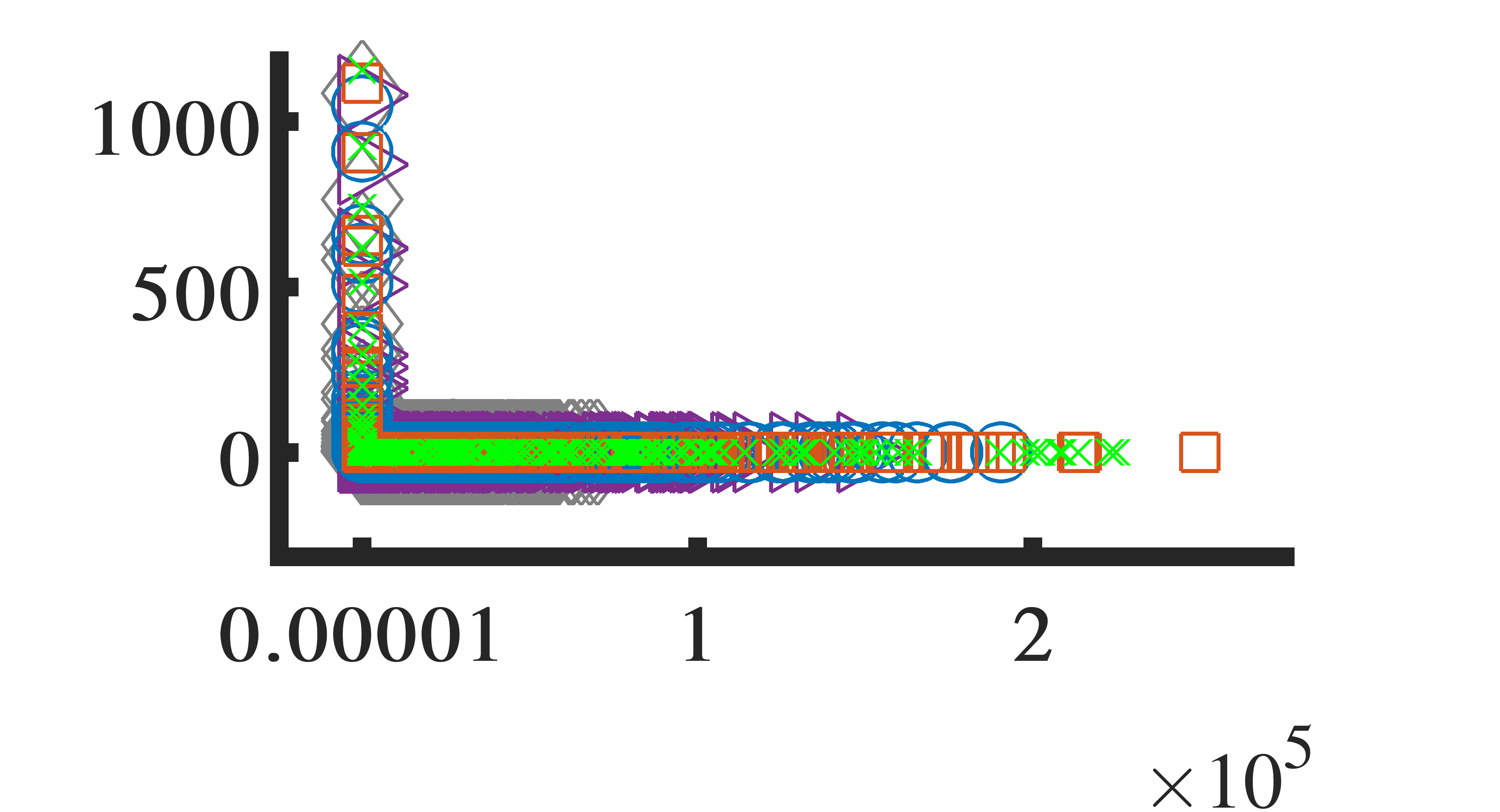} & 
            \includegraphics[width=0.288\textwidth]{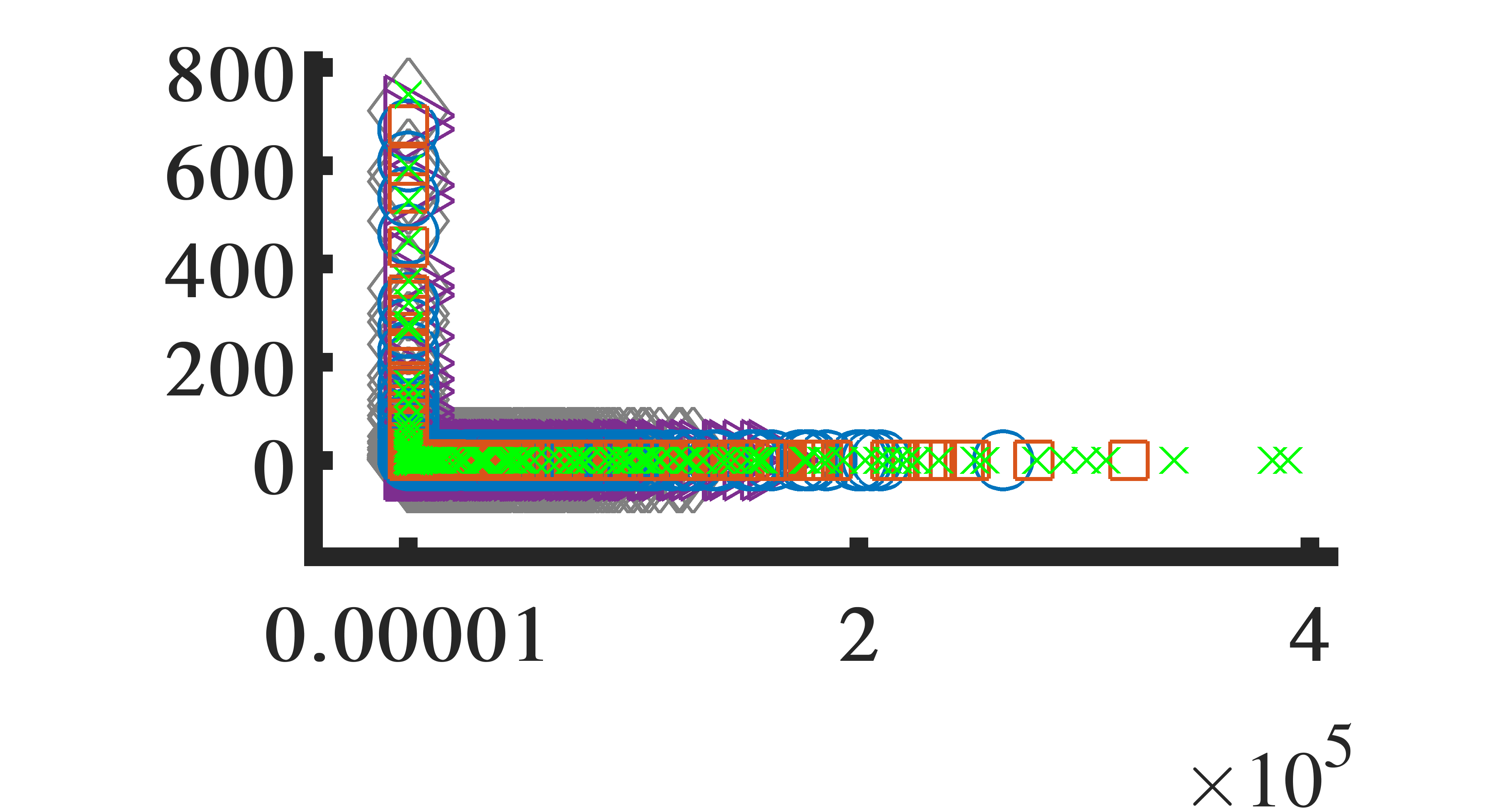}  \\            
            \midrule
            \makecell{(d)\\$\rho$$=$$0.7$\\$\beta$$=$$5$} &  
            \includegraphics[width=0.28\textwidth]{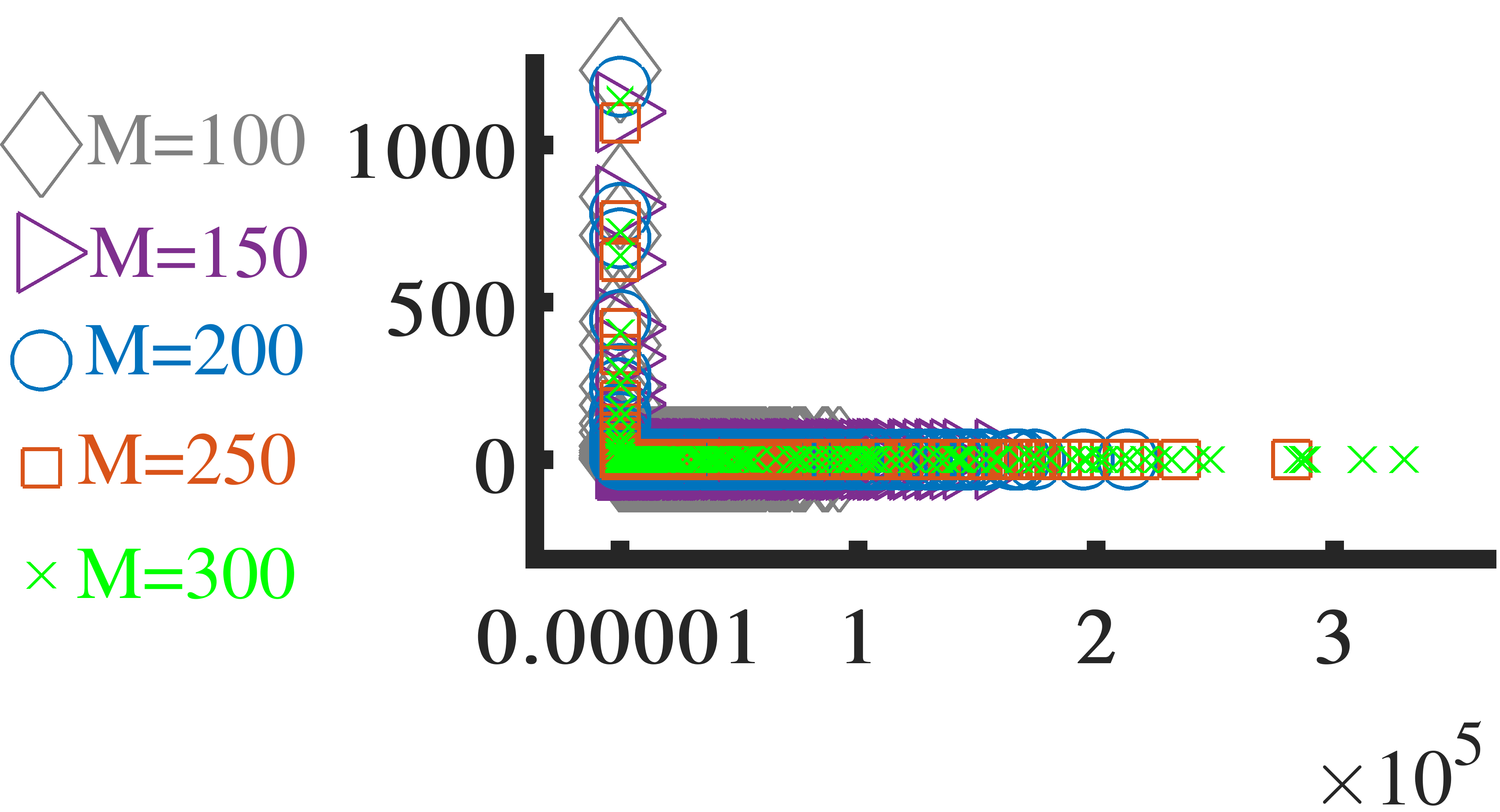} & 
            \includegraphics[width=0.288\textwidth]{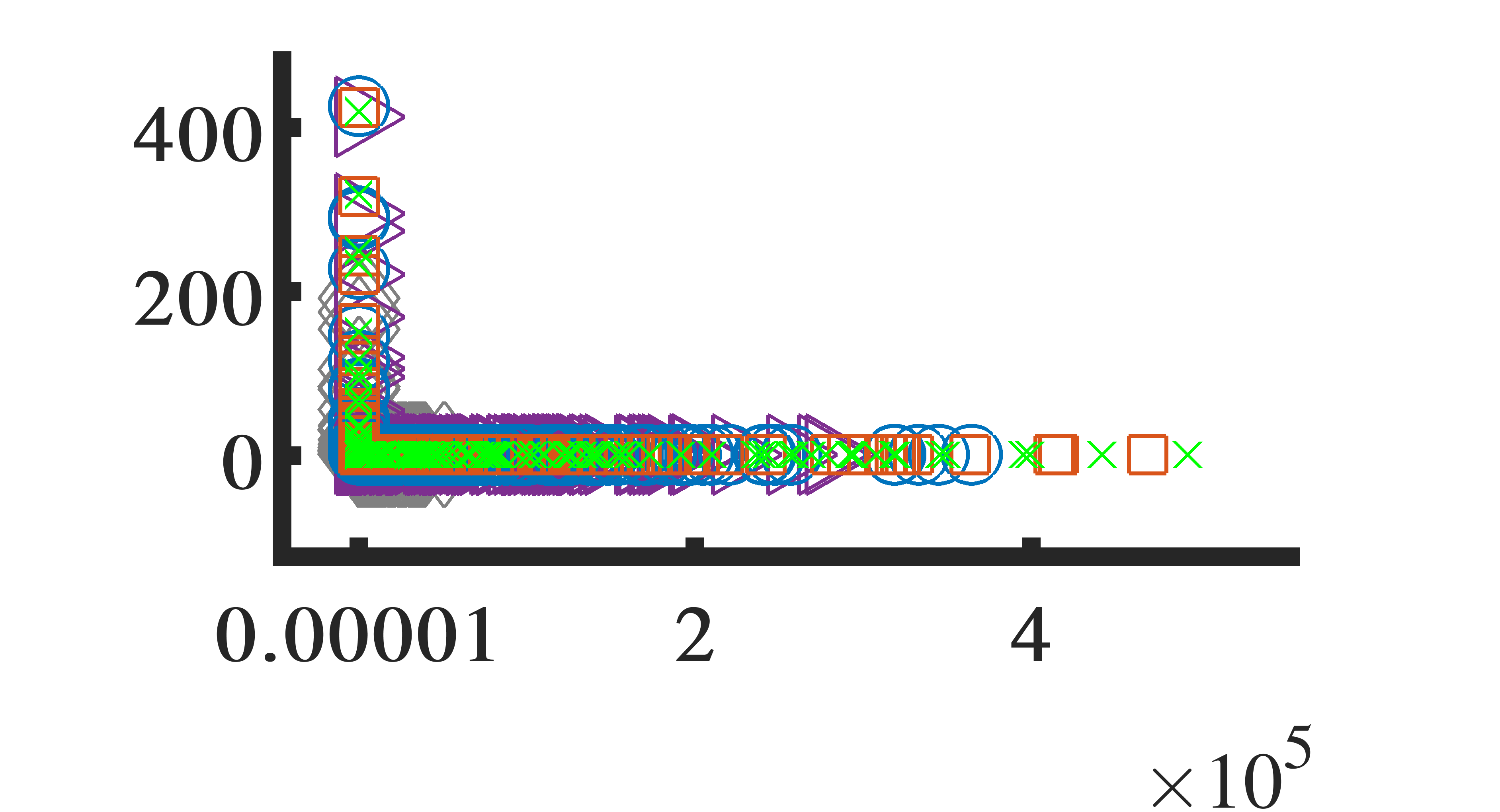} &  
            \includegraphics[width=0.286\textwidth]{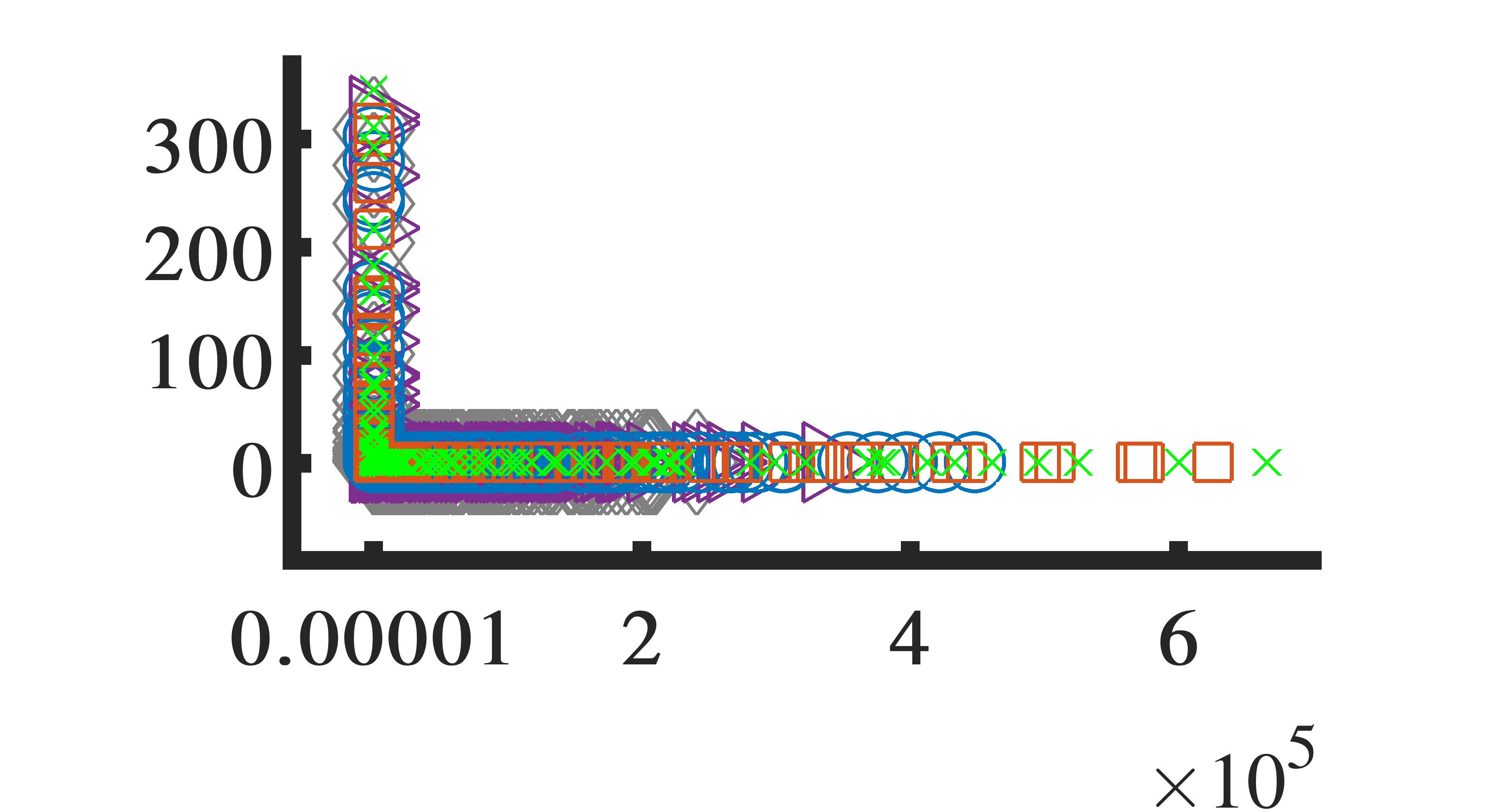}\\

        \end{tabular}
        \caption{Frequency (y axes) of burst sizes (x axes) in S-Amazon with $10^7$ sgrs.}
        \label{tbl:burstsize}
    \end{table}
\subsection{Discussion}\label{subsec:discussevaluations}
Our evaluations indicate that \emph{sGrow} efficiently and effectively reproduces the bursty emergence patterns of cohesive building blocks in the bipartite streaming graphs, regardless of the initial conditions, the scale and temporal characteristics of the generated stream, and the model configurations. This confirms that the introduced microscopic mechanisms in the body of \emph{sGrow} \emph{robustly} explain the observed streaming growth phenomenon in real-world streams. 
Our analysis also verifies the ability of \emph{sGrow} in generating realistic streaming graphs configured with user-specified properties for the scale and burstiness of the stream, level of strength assortativity, probability of-of-order streaming records, generation time, and time-sensitive connections.
\section{Conclusion}\label{sec:conclusion}
Butterflies are key building blocks of streaming bipartite graphs and their emergence patterns implies the growth patterns in streaming graphs. In this paper, we study the emergence of butterflies in streaming graphs displaying superlinear growth of butterfly count wrt the edge count. Integrating the connectivities, weights, and fine-grained temporal information, we investigate the strength assortativity of butterflies over the timeline of burst arrivals. 
We introduce a quantification approach for statistical analysis of confounding distributions that enable effective temporal analysis of strength assortativity in bipartite graphs with skewed strength distribution. Utilizing this approach, we unveil the "scale-invariant strength assortativity of streaming butterflies", a co-occurrence of three patterns: butterfly densification, strength diversification, and steady strength assortativity. We study the existing local rules for graph growth that yield skewed distributions, degree correlation, and cohesive structures to explain the observed mixing patterns. We find that implicit degree-driven preferential attachment and copying mechanisms or solely strength-driven preferential attachment with random assignment of timestamps to the edges can only partially preserve the observed patterns but are not effective enough to reproduce these patterns simultaneously. Therefore, we introduce a set of microscopic mechanisms, in the body of a proposed streaming growth model called \emph{sGrow},  based on realistic streaming graph record generation, probabilistic connections, and strength-driven preferential random walks which explain the emergence patterns of streaming butterflies. 
sGrow is based on iterative addition of bursts of edges which satisfies streaming data model,
preserves realistic patterns of butterfly emergence quantitatively and qualitatively, and 
makes the stream generation scalable. Moreover, sGrow enables generating sequence of bipartite edges attributed with timestamps and weights, isolated/out-of-order edges, and four-vertex graphlets. 
Our comprehensive evaluations validate the efficacy of sGrow in realization of streaming growth patterns effectively and independent of initial conditions, scale and temporal characteristics, and model configurations. Our analysis also verifies the robustness of sGrow in generating streaming graphs based on user-specified properties for the scale and burstiness of the stream, level of strength assortativity, probability of-of-order streaming records, generation time, and time-sensitive connections. On top of the aforementioned advantages and qualified features, sGrow suits the following applications: (1) streaming graph benchmarks by generating configurable realistic data streams supported by a reference guide for parameter configuration and stress testing analysis, (2) machine learning benchmarks by providing annotated data streams which are synthesized by realistic instance injection and suit both testing and training purposes, and (3) development of streaming algorithms and models (e.g., concept drift models) by providing microscopic mechanisms and characteristic patterns that enlighten the architecture of model/algorithm.

\bibliographystyle{plain}
\bibliography{main}

\end{document}